%% file: thesis.tex
\documentclass[11pt]{ociamthesis}  
\input{preamble}

\title{Categorical\\[1ex] Operational Physics}    
\author{Sean Tull}             
\college{St Catherine's College}  
\degree{Doctor of Philosophy}    
\degreedate{Trinity 2018}

\makeindex

\begin{document}
\setcounter{secnumdepth}{3}
\setcounter{tocdepth}{1}

\maketitle  

\input{dedication}
\input{acknowledgements}
\input{abstract}

\begin{romanpages}          
\tableofcontents            
\end{romanpages}

\input{introduction}

\input{chapter1}

\input{chapter2}
\input{chapter3}
\input{chapter4}

\input{chapter5}
\input{chapter6}

\input{outlook}
\input{categories}
\input{notation}

\addcontentsline{toc}{chapter}{Index of Subjects}
\printindex{}

\addcontentsline{toc}{chapter}{Bibliography}
\bibliography{thesis-bib}       
\bibliographystyle{alpha}

\end{document}

%% file: preamble.tex
\usepackage{amsmath,amssymb,stmaryrd,xspace,enumerate}
\usepackage{amsthm}
\usepackage{tikz-cd}
\usepackage{multirow}
\usepackage{mathtools}
\usepackage{url}
\usepackage{makeidx}
\usepackage{relsize}
\usepackage{bm} 
\usepackage[inline]{enumitem}
\setenumerate{label=\arabic*.}
\newcounter{counter}
\input{./style-options}

\newcommand{\storecounter}{\setcounter{counter}{\value{enumi}}}
\newcommand{\resumecounter}{\setcounter{enumi}{\value{counter}}}

\newcommand{\indef}[1]{\deff{#1}} 
\newcommand{\emps}[1]{\emph{#1}} 
\newcommand{\deff}[1]{{\boldmath\textbf{#1}}} 

\newenvironment{exampleslist}
    {\begin{enumerate}[leftmargin=*,label=\arabic*., ref=\arabic*]}
    {\end{enumerate}}

\newenvironment{inlinelist}
   {\begin{enumerate*}[label=\arabic*), ref=\arabic*)]}
    {\end{enumerate*}}

\newcommand{\kleene}{\simeq}

\newcommand{\x}{i}
\newcommand{\y}{j}
\newcommand{\z}{k}
\newcommand{\w}{l}
\newcommand{\oi}{j} 
\newcommand{\X}{X}
\newcommand{\Y}{Y}

\newcommand{\effpullbone}{\text{(a)}}
\newcommand{\effpullbtwo}{\text{(b)}}

\theoremstyle{plain}
\newtheorem{assumption}{Axiom}

\newtheorem{theorem}{Theorem}[chapter]
\newtheorem{proposition}[theorem]{Proposition}
\newtheorem{corollary}[theorem]{Corollary}
\newtheorem{lemma}[theorem]{Lemma}

\theoremstyle{definition}
\newtheorem{definition}[theorem]{Definition}
\newtheorem{example}[theorem]{Example}
\newtheorem{examples}[theorem]{Examples}
\newtheorem{remark}[theorem]{Remark}

\newcommand{\notetoself}[1]{}
\newcommand{\omitfornow}[1]{}
\newcommand{\omitthis}[1]{}

\newcommand{\Set}{\cat{Set}}
\newcommand{\VecSp}{\cat{Vec}}
\newcommand{\VecC}{\VecSp_\mathbb{C}}
\newcommand{\Veck}{\VecSp_\fieldk}
\newcommand{\FVeck}{\cat{FVec}_\fieldk}
\newcommand{\FHilb}{\cat{FHilb}}
\newcommand{\Hilb}{\cat{Hilb}}
\newcommand{\Rel}{\cat{Rel}}
\newcommand{\Grp}{\cat{Grp}}
\newcommand{\PFun}{\cat{PFun}}
\newcommand{\KlDT}{\cat{Class}}
\newcommand{\KlSD}{\cat{Class_{\mathsf{p}}}}
\newcommand{\KlD}{\Kleisli{D}}
\newcommand{\KlMn}{\Kleisli{\multiset_\mathbb{N}}}
\newcommand{\Class}{\cat{FClass}}
\newcommand{\Kleisli}[1]{\mathsf{Kl}(#1)}

\newcommand{\CStarop}{\cat{CStar}^\op}
\newcommand{\CStarSUop}{\cat{CStar}^\op_{\mathsf{su}}}
\newcommand{\CStarUop}{\cat{CStar}^\op_{\mathsf{u}}}
\newcommand{\FCStar}{\cat{FCStar}}
\newcommand{\vNop}{\cat{vNA}^\op}
\newcommand{\vNSUop}{\cat{vNA}^\op_{\mathsf{su}}}
\newcommand{\vNUop}{\cat{vNA}^\op_{\mathsf{u}}}

\newcommand{\Quant}[1]{\cat{Quant}_{#1}}
\newcommand{\QuantSU}{\cat{Quant}_{\mathsf{sub}}}
\newcommand{\HilbP}{\Hilb_\sim}
\newcommand{\FHilbP}{\FHilb_\sim}
\newcommand{\Mat}{\cat{Mat}} 
\newcommand{\MatS}{\Mat_S}
\newcommand{\Rep}{\cat{Rep}} 
\newcommand{\Spek}{\cat{Spek}}
\newcommand{\MSpek}{\cat{MSpek}}
\newcommand{\SpekOrMSpek}{\cat{(M)Spek}}
\newcommand{\IV}{IV}

\newcommand{\GSet}{\cat{G}\text{-}\Set}
\newcommand{\VecP}{{\VecSp{}}_{\sim}}
\newcommand{\FVecP}{\cat{FVec}_{\sim}}
\newcommand{\VecProj}{\cat{Proj}_k} 

\newcommand{\catG}{\cat{G}}

\newcommand{\CMon}{\cat{CMon}} 
\newcommand{\DCM}{\cat{DCM}} 
\newcommand{\PCM}{\cat{PCM}}
\newcommand{\PCMD}{\cat{Par}} 
\newcommand{\CMonD}{\cat{Tot}}

\newcommand{\CPM}{\ensuremath{\mathsf{CPM}}\xspace}
\newcommand{\Split}[1]{\ensuremath{\mathrm{Split}}\xspace}

\newcommand{\cat}[1]{\ensuremath{\mathbf{#1}}\xspace}
\newcommand{\catJ}{\cat{J}}
\newcommand{\catD}{\cat{D}}
\newcommand{\catA}{\cat{A}}
\newcommand{\catC}{\cat{C}}
\newcommand{\catB}{\cat{B}}

\newcommand{\ClassDet}{\mathsf{ClassDet}}
\newcommand{\ClassProb}{\mathsf{ClassProb}}
\newcommand{\MatOT}{\mathsf{Mat}} 
\newcommand{\QuantOT}{\mathsf{Quant}}
 
\newcommand{\FinCStarOTC}{\mathsf{FCStar}} 
\newcommand{\CStarOTC}{\mathsf{CStar}} 
\newcommand{\CStarOT}{\mathsf{CStar}} 
\newcommand{\RelOT}{\mathsf{Rel}}
\newcommand{\Kl}{\cat{Kl}} 

\newcommand{\OTCat}{\cat{OT}}
\newcommand{\OTRep}{\OTCat^+}
\newcommand{\OTCatOrRep}{\OTCat^{(+)}}

\newcommand{\OTRepStrict}{\OTRep_{\mathsf{strict}}}
\newcommand{\OpCats}{\cat{OCat}} 
\newcommand{\OpCatsComp}{\OpCats_{\mathsf{comp}}} 
\newcommand{\TestCats}{\cat{TestCat}} 

\newcommand{\PTest}{\cat{PTest}} 
\newcommand{\Test}{\cat{Test}} 

\newcommand{\pos}{\mathrm{pos}} 
\newcommand{\causal}{\mathsf{caus}}
\newcommand{\subcausal}{\mathsf{sc}}
\newcommand{\pure}{\mathsf{pure}}
\newcommand{\prepure}{\mathsf{p}} 
\newcommand{\sa}{\mathsf{sa}} 
\newcommand{\Rplus}{\mathbb{R}^+} 
\newcommand{\Rpos}{\Rplus} 
\newcommand{\Qpos}{\mathbb{Q}^+} 

\newcommand{\dkernel}{kernel} 

\newcommand{\id}[1]{\ensuremath{\mathrm{id}_{#1}}}

\newcommand{\Aut}{\mathrm{Aut}} 

\newcommand{\op}{\ensuremath{\mathrm{\rm op}}}

\newcommand{\Multicat}{\mathsf{M}} 
\newcommand{\mcat}[1]{\cat{#1}} 
\newcommand{\partialcomp}{\hspace{3.5pt} \partialsymbolcomp \hspace{2.5pt}} 
\newcommand{\tikzcdpartialright}[1]{\arrow{r}[description, inner sep = 0]{\partialsymbollargearrow}[above]{{#1}}}

\newcommand{\partialsmallarrow}[1][1.2pt]{ 
\begin{tikzpicture}[baseline=-0.58ex]
\draw[black, fill=white,radius={#1}] (0,0) circle ; 
\draw[black, fill=white,radius=0.1pt] (0,0) circle ; 
\end{tikzpicture}
}
\newcommand{\circlearrow}[1]{\ensuremath{\mathrlap{\hspace{2.5pt}\to}{\hspace{6pt}{#1}\hspace{6pt}}}}

\newcommand{\partialto}{\circlearrow{\partialsmallarrow[2.5pt]}}

\newcommand{\partialsymbolcomp}{ 
\begin{tikzpicture}[baseline=-0.5ex]
\draw[black, fill=white,radius=3.6pt] (0,0) circle ; 
\draw[black, fill=white,radius=0.6pt] (0,0) circle ; 
\end{tikzpicture}
}
\newcommand{\partialsymbollargearrow}{ 
\begin{tikzpicture}[baseline=-0.6pt]
\draw[black, fill=white,radius=3pt] (0,0) circle ; 
\draw[black, fill=white,radius=0.2pt] (0,0) circle ; 
\end{tikzpicture}
}

\newcommand{\discard}[1]{\ensuremath{\tinygroundnew_{#1}}}

\newcommand{\discardflip}[1]{\ensuremath{\tinygroundflipnew_{#1}}}
\newcommand{\ts}[1]{(#1)} 

\DeclareFontFamily{U}{mathux}{\hyphenchar\font45}
\DeclareFontShape{U}{mathux}{m}{n}{
      <5> <6> <7> <8> <9> <10>
      <10.95> <12> <14.4> <17.28> <20.74> <24.88>
      mathux10
      }{}
\DeclareSymbolFont{mathux}{U}{mathux}{m}{n}

\DeclareMathSymbol{\bigovee}{1}{mathux}{"8F}

\DeclareMathSymbol{\bigperp}{1}{mathux}{"4E}

\DeclareFontFamily{U}{mathb}{\hyphenchar\font45}
\DeclareFontShape{U}{mathb}{m}{n}{
      <5> <6> <7> <8> <9> <10> gen * mathb
      <10.95> mathb10 <12> <14.4> <17.28> <20.74> <24.88> mathb12
      }{}
\DeclareSymbolFont{mathb}{U}{mathb}{m}{n}
\DeclareFontSubstitution{U}{mathb}{m}{n}
\DeclareMathSymbol{\mylgroup}{\mathbin}{mathb}{'160}
\DeclareMathSymbol{\myrgroup}{\mathbin}{mathb}{'161}

\newcommand{\totaspar}[1]{ \mylgroup {#1} \myrgroup } 

\newcommand{\supp}{\mathsf{supp}}
\newcommand{\dring}[1]{D(#1)}

\newcommand{\sem}[1]{\ensuremath{\llbracket #1 \rrbracket}}

\newcommand{\pfun}{\relto} 
\newcommand{\fieldk}{k} 

\DeclareMathOperator{\Tr}{Tr}
\newcommand{\Dbl}[1]{\widehat{#1}}
\newcommand{\bra}[1]{\ensuremath{\langle #1 |}}
\newcommand{\ket}[1]{\ensuremath{| #1 \rangle}}

\newcommand{\Eff}{\text{Eff}}

\newcommand{\isomto}{\xrightarrow{\sim}}

\newcommand{\To}{\mathbb{T}}
\newcommand{\Ro}{R} 
\newcommand{\multiset}{\mathcal{M}}

\newcommand{\Ker}{\mathrm{Ker}}
\newcommand{\Coker}{\mathrm{Coker}}

\newcommand{\img}{\mathrm{im}}
\newcommand{\Img}{\mathrm{Im}}
\newcommand{\CoIm}{\mathrm{Coim}}
\newcommand{\coim}{\mathrm{coim}}
\newcommand{\coker}{\mathrm{coker}}
\newcommand{\DKer}{\mathrm{DKer}} 

\newcommand{\enc}{E}
\newcommand{\dec}{D}
\newcommand{\idob}{F}

\newcommand{\ext}[1]{\mdil{#1}}
\newcommand{\Ext}[2]{\MDil{#2}}
\newcommand{\MDil}[1]{M({#1})} 
\newcommand{\mdil}[1]{\mathrm{min}(#1)}

\newcommand{\med}{\nu} 

\newcommand{\coproj}{\kappa}
\newcommand{\biprod}{\oplus}
\newcommand{\pinit}{0}
\newcommand{\pcoprod}{\mathbin{\dot{+}}}
\newcommand{\pprod}{\mathbin{\dot{\times}}}
\newcommand{\pbiprod}{\mathbin{\dot{\oplus}}}

\newcommand{\pcoproj}{\coproj}
\newcommand{\pproj}{\pi}
\newcommand{\simP}{\sim} 
\newcommand{\quotP}{\mathbb{P}} 
\newcommand{\quot}[2]{{#1} \!/\! {#2}}
\newcommand{\Gen}{I} 
\newcommand{\TrI}{\mathbb{T}} 

\newcommand{\plusI}[1]{\mathsf{GP}({#1})}
\newcommand{\plusIdag}[1]{\mathsf{GP}^{\dagger}({#1})}

\newcommand{\tens}{\hat{\otimes}}
\newcommand{\dtens}{\tens} 
\newcommand{\obb}[1]{{\bm{#1}}}

\newcommand{\aalpha}{\obb{\alpha}}
\newcommand{\llambda}{\obb{\lambda}}
\newcommand{\ssigma}{\obb{\sigma}}
\newcommand{\rrho}{\obb{\rho}}

\newcommand{\Partial}{\mathsf{Par}}

\newcommand{\Events}{\cat{Event}} 

\newcommand{\otcplus}{+} 
\newcommand{\OT}{\mathsf{OT_B}} 
\newcommand{\OTC}{\mathsf{OT}} 

\newcommand{\hilbH}{\mathcal{H}} 
\newcommand{\hilbK}{\mathcal{K}} 
\newcommand{\hilbL}{\mathcal{L}} 

\newcommand{\tc}{\sim} 
\newcommand{\tp}{\mathbb{P}} 
\newcommand{\cta}{\catA} 
\newcommand{\ctb}{\catB} 

\newcommand{\evec}{f} 

\usepackage{tikz,xypic}

\usetikzlibrary{decorations.pathreplacing,decorations.markings,arrows.meta,backgrounds}
\pgfdeclarelayer{edgelayer}
\pgfdeclarelayer{nodelayer}
\pgfsetlayers{background,edgelayer,nodelayer,main}

\tikzstyle{cdot}=[circle, draw=black, fill=black!25, inner sep=.4ex] 
\tikzstyle{bigdot}=[dot, inner sep=0pt]
\tikzstyle{whitedot}=[circle, draw=black, fill=white, inner sep=.4ex]
\tikzstyle{greydot}=[circle, draw=black, fill=black!25, inner sep=.4ex] 
\tikzstyle{blackdot}=[circle, draw=black, fill=black, inner sep=.4ex]
\tikzset{arrow/.style={decoration={
    markings,
    mark=at position #1 with \arrow{>[length=2pt, width=3pt]}},
    postaction=decorate},
    reverse arrow/.style={decoration={
    markings,
    mark=at position #1 with {{\arrow{<[length=2pt, width=3pt]}}}},
    postaction=decorate}
}

\newcommand\relto[1][]{\mathbin{\smash{
\begin{tikzpicture}[baseline={([yshift=-1pt]
current bounding box.south)}]
    \node (A) at (0,0) [inner xsep=0pt, inner ysep=1pt, minimum width=0.15cm] {\ensuremath{\scriptstyle #1}};
    \draw [->, line width=0.4pt, line cap=round]
        ([xshift=-2.5pt] A.south west)
        to ([xshift=3pt] A.south east);
    \draw [line width=.4pt] ([xshift=-1pt,yshift=-2pt]A.south) to ([xshift=1pt,yshift=2pt]A.south);
\end{tikzpicture}}}}

\newcommand{\tinymultflip}[1][cdot]{
\smash{\raisebox{-2pt}{\hspace{-5pt}\ensuremath{\begin{pic}[scale=0.4,yscale=1]
    \node (0) at (0,0) {};
    \node[#1, inner sep=1.5pt] (1) at (0,0.55) {};
    \node (2) at (-0.5,1) {};
    \node (3) at (0.5,1) {};
    \draw (0.center) to (1.center);
    \draw (1.center) to [out=left, in=down, out looseness=1.5] (2.center);
    \draw (1.center) to [out=right, in=down, out looseness=1.5] (3.center);
    \node[#1, inner sep=1.5pt] (1) at (0,0.55) {};
\end{pic}
}\hspace{-3pt}}}}

\newcommand{\tinycomult}[1][cdot]{
\smash{\raisebox{-2pt}{\hspace{-5pt}\ensuremath{\begin{pic}[scale=0.4,yscale=1]
    \node (0) at (0,0) {};
    \node[#1, inner sep=1.5pt] (1) at (0,0.55) {};
    \node (2) at (-0.5,1) {};
    \node (3) at (0.5,1) {};
    \draw (0.center) to (1.center);
    \draw (1.center) to [out=left, in=down, out looseness=1.5] (2.center);
    \draw (1.center) to [out=right, in=down, out looseness=1.5] (3.center);
    \node[#1, inner sep=1.5pt] (1) at (0,0.55) {};
\end{pic}
}\hspace{-3pt}}}}

\newcommand{\tinyunit}[1][cdot]{
\smash{\raisebox{3pt}{\hspace{-3pt}\ensuremath{\begin{pic}[scale=0.4,yscale=-1]
    \node (0) at (0,0) {};
    \node[#1, inner sep=1.5pt] (1) at (0,0.55) {};
    \draw (0.center) to (1.north);
\end{pic}
}\hspace{-1pt}}}}

\newcommand{\tinycounit}[1][cdot]{
\smash{\raisebox{-5pt}{\hspace{-3pt}\ensuremath{\begin{pic}[scale=0.4,yscale=1]
    \node (0) at (0,0) {};
    \node[#1, inner sep=1.5pt] (1) at (0,0.55) {};
    \draw (0.center) to (1.north);
        \node[#1, inner sep=1.5pt] (1) at (0,0.55) {};
\end{pic}
}\hspace{-1pt}}}}

\newcommand{\tinycup}{\smash{\raisebox{2pt}{\hspace{-2pt}\ensuremath{\begin{pic}[scale=0.2]
   \pgftransformscale{1.5} \draw[arrow=.6, scale = 1] (0,0) to[out=-90,in=-90,looseness=1.5] (1.5,0);
\end{pic}}}}}

\newcommand{\tinycap}{\smash{\raisebox{-3pt}{\hspace{-2pt}\ensuremath{\begin{pic}[scale=0.2, yscale=-1]
   \pgftransformscale{1.5} \draw[arrow=.6, scale = 1] (0,0) to[out=-90,in=-90,looseness=1.5] (1.5,0);
\end{pic}}}}}

\newenvironment{pic}[1][] {\begin{aligned}\begin{tikzpicture}[scale=2.0, font=\tiny,#1]}{\end{tikzpicture}\end{aligned}} 

\newif\ifvflip\pgfkeys{/tikz/vflip/.is if=vflip}
\newif\ifhflip\pgfkeys{/tikz/hflip/.is if=hflip}
\newif\ifhvflip\pgfkeys{/tikz/hvflip/.is if=hvflip}

\newenvironment{picc}[1][]
{\begin{aligned}\begin{tikzpicture}[font=\tiny,#1]}
{\end{tikzpicture}\end{aligned}}

\newlength\minimummorphismwidth
\newlength\stateheight
\newlength\minimumstatewidth
\newlength\connectheight
\tikzset{colour/.initial=white}

\tikzstyle{pure}=[line width=.7pt]

\makeatletter

\pgfdeclareshape{groundd}
{
    \savedanchor\centerpoint
    {
        \pgf@x=0pt
        \pgf@y=0pt
    }
    \anchor{center}{\centerpoint}
    \anchorborder{\centerpoint}

    \anchor{north}
    {
        \pgf@x=0pt
        \pgf@y=0.16\stateheight
    }
    \anchor{south}
    {
        \pgf@x=0pt
        \pgf@y=0pt
    }
    \saveddimen\overallwidth
    {
        \pgfkeysgetvalue{/pgf/minimum width}{\minwidth}
        \pgf@x=\minimumstatewidth
        \ifdim\pgf@x<\minwidth
            \pgf@x=\minwidth
        \fi
    }
    \backgroundpath
    {
        \begin{pgfonlayer}{main} 
        \pgfsetstrokecolor{black}
        \pgfsetlinewidth{1.25pt}
        \ifhflip
            \pgftransformyscale{-1}
        \fi
        \pgftransformscale{0.5}
        \pgfpathmoveto{\pgfpoint{-0.5*\overallwidth}{0pt}}
        \pgfpathlineto{\pgfpoint{0.5*\overallwidth}{0pt}}
        \pgfpathmoveto{\pgfpoint{-0.33*\overallwidth}{0.33*\stateheight}}
        \pgfpathlineto{\pgfpoint{0.33*\overallwidth}{0.33*\stateheight}}
        \pgfpathmoveto{\pgfpoint{-0.16*\overallwidth}{0.66*\stateheight}}
        \pgfpathlineto{\pgfpoint{0.16*\overallwidth}{0.66*\stateheight}}
        \pgfpathmoveto{\pgfpoint{-0.02*\overallwidth}{\stateheight}}
        \pgfpathlineto{\pgfpoint{0.02*\overallwidth}{\stateheight}}
        \pgfusepath{stroke}
        \end{pgfonlayer}
    }
}

\usepackage{tikz,xypic}
\usetikzlibrary{decorations.pathreplacing,decorations.markings,arrows.meta,backgrounds,shapes}
\usetikzlibrary{circuits.ee.IEC}
\pgfdeclarelayer{edgelayer}
\pgfdeclarelayer{nodelayer}
\pgfsetlayers{background,edgelayer,nodelayer,main}
\tikzstyle{none}=[inner sep=0mm]
\tikzstyle{every loop}=[]
\tikzstyle{mark coordinate}=[inner sep=0pt,outer sep=0pt,minimum size=3pt,fill=black,circle]

\tikzset{arrow/.style={decoration={
    markings,
    mark=at position #1 with \arrow{>[length=2pt, width=3pt]}},
    postaction=decorate},
    reverse arrow/.style={decoration={
    markings,
    mark=at position #1 with {{\arrow{<[length=2pt, width=3pt]}}}},
    postaction=decorate}
}

\tikzstyle{upground}=[circuit ee IEC,thick,ground,rotate=90,scale=1.5]
\tikzstyle{downground}=[circuit ee IEC,thick,ground,rotate=-90,scale=1.5]

\newcommand{\mapminh}{5mm} 
\newcommand{\stateminh}{5mm}
\newcommand{\maplw}{0.7pt} 
\newcommand{\stateshift}{-0.2pt}
\newcommand{\effectshift}{-0.2pt}

\tikzstyle{box}=[map]
\tikzstyle{medium box}=[medium map]
\tikzstyle{dot}=[inner sep=0mm,minimum width=2mm,minimum height=2mm,draw,shape=circle]  
\tikzstyle{black dot}=[dot,fill=black]
\tikzstyle{white dot}=[dot,fill=white,,text depth=-0.2mm]
\tikzstyle{grey dot}=[dot,fill=black!25] 

\tikzstyle{corner1}=[box,fill=white, font=\footnotesize] %
\tikzstyle{corner2}=[dot,fill=white, font=\footnotesize] %
\tikzstyle{corner3}=[dot,fill=black!25, font=\footnotesize] %
\tikzstyle{corner4}=[dot,fill=black, font=\footnotesize] %

\tikzstyle{scalar}=[circle,draw,inner sep=2pt, line width=\maplw] 

\usetikzlibrary{shapes.misc, positioning}

\tikzset{stateshape/.style={append after command={
   \pgfextra
        \draw[sharp corners, fill=white, line width = \maplw]%
    (\tikzlastnode.west)%
    [rounded corners=0pt] |- (\tikzlastnode.north)%
    [rounded corners=0pt] -| (\tikzlastnode.east)%
    [rounded corners=5pt] |- (\tikzlastnode.south)%
    [rounded corners=5pt] -| (\tikzlastnode.west);
   \endpgfextra}}}

\tikzset{effectshape/.style={append after command={
   \pgfextra
        \draw[sharp corners, fill=white, line width = \maplw]%
    (\tikzlastnode.west)%
    [rounded corners=0pt] |- (\tikzlastnode.south)%
    [rounded corners=0pt] -| (\tikzlastnode.east)%
    [rounded corners=5pt] |- (\tikzlastnode.north)%
    [rounded corners=5pt] -| (\tikzlastnode.west);
   \endpgfextra}}}

 \tikzstyle{map}=[draw,shape=rectangle, inner sep=2pt,minimum height=\mapminh, minimum width=7mm,fill=white]

\tikzstyle{point}=[stateshape,inner sep=2pt, minimum width=6mm, minimum height=\stateminh, yshift=\stateshift]
\tikzstyle{copoint}=[effectshape,inner sep=.2pt, minimum width=6mm, minimum height=\stateminh, yshift=-\effectshift]

\tikzstyle{wide point}=[point, minimum width=12mm]
\tikzstyle{wide copoint}=[copoint, minimum width=12mm]

 \tikzstyle{map}=[draw,shape=rectangle, inner sep=2pt,minimum height=\mapminh, minimum width=7mm,fill=white, line width = \maplw]

\tikzstyle{medium map} = [map, minimum width = 12mm] 
\tikzstyle{semilarge map} = [map, minimum width = 15mm] 
\tikzstyle{large map} = [map, minimum width = 18mm] 

\tikzstyle{kpoint} =[point]
\tikzstyle{kpointadj} =[copoint]
\tikzstyle{kpointconj}=[dagpointconj] 

\makeatletter
\newcommand{\boxshape}[3]{%
\pgfdeclareshape{#1}{
\inheritsavedanchors[from=rectangle] 
\inheritanchorborder[from=rectangle]
\inheritanchor[from=rectangle]{center}
\inheritanchor[from=rectangle]{north}
\inheritanchor[from=rectangle]{south}
\inheritanchor[from=rectangle]{west}
\inheritanchor[from=rectangle]{east}
\backgroundpath{
\southwest \pgf@xa=\pgf@x \pgf@ya=\pgf@y
\northeast \pgf@xb=\pgf@x \pgf@yb=\pgf@y

\@tempdima=#2
\@tempdimb=#3

\pgfpathmoveto{\pgfpoint{\pgf@xa - 5pt + \@tempdima}{\pgf@ya}}
\pgfpathlineto{\pgfpoint{\pgf@xa - 5pt - \@tempdima}{\pgf@yb}}
\pgfpathlineto{\pgfpoint{\pgf@xb + 5pt + \@tempdimb}{\pgf@yb}}
\pgfpathlineto{\pgfpoint{\pgf@xb + 5pt - \@tempdimb}{\pgf@ya}}
\pgfpathlineto{\pgfpoint{\pgf@xa - 5pt + \@tempdima}{\pgf@ya}}
\pgfpathclose
}
}}

\boxshape{NEbox}{0pt}{3pt} 
\boxshape{SEbox}{0pt}{-3pt}
\boxshape{NWbox}{3pt}{0pt}
\boxshape{SWbox}{-3pt}{0pt}
\boxshape{rec-box}{0pt}{0pt}
\makeatother

\tikzstyle{cloud}=[shape=cloud,draw,minimum width=1.5cm,minimum height=1.5cm]

\tikzstyle{dagmap}=[draw,shape=NEbox,inner sep=2pt,minimum height=\mapminh,fill=white, line width = \maplw] %
\tikzstyle{dashedmap}=[draw,dashed,shape=NEbox,inner sep=2pt,minimum height=\mapminh,fill=white, line width = \maplw]
\tikzstyle{mapdag}=[draw,shape=SEbox,inner sep=2pt,minimum height=\mapminh,fill=white, line width = \maplw]
\tikzstyle{mapadj}=[draw,shape=SEbox,inner sep=2pt,minimum height=\mapminh,fill=white, line width = \maplw]
\tikzstyle{maptrans}=[draw,shape=SWbox,inner sep=2pt,minimum height=\mapminh,fill=white, line width = \maplw]
\tikzstyle{mapconj}=[draw,shape=NWbox,inner sep=2pt,minimum height=\mapminh,fill=white, line width = \maplw]

\tikzstyle{medium dagmap}=[draw,shape=NEbox,inner sep=2pt,minimum height=\mapminh,fill=white,minimum width=7mm, line width = \maplw]
\tikzstyle{semilarge dagmap}=[draw,shape=NEbox,inner sep=2pt,minimum height=\mapminh,fill=white,minimum width=9.5mm, line width = \maplw]
\tikzstyle{large dagmap}=[draw,shape=NEbox,inner sep=2pt,minimum height=\mapminh,fill=white,minimum width=12mm, line width = \maplw]

\makeatletter

\pgfdeclareshape{cornerpoint}{
\inheritsavedanchors[from=rectangle] 
\inheritanchorborder[from=rectangle]
\inheritanchor[from=rectangle]{center}
\inheritanchor[from=rectangle]{north}
\inheritanchor[from=rectangle]{south}
\inheritanchor[from=rectangle]{west}
\inheritanchor[from=rectangle]{east}
\backgroundpath{
\southwest \pgf@xa=\pgf@x \pgf@ya=\pgf@y
\northeast \pgf@xb=\pgf@x \pgf@yb=\pgf@y

\pgfmathsetmacro{\pgf@shorten@left}{\pgfkeysvalueof{/tikz/shorten left}}
\pgfmathsetmacro{\pgf@shorten@right}{\pgfkeysvalueof{/tikz/shorten right}}

\pgfpathmoveto{\pgfpoint{0.5 * (\pgf@xa + \pgf@xb)}{\pgf@ya - 5pt}}
\pgfpathlineto{\pgfpoint{\pgf@xa - 8pt + \pgf@shorten@left}{\pgf@yb - 1.5 * \pgf@shorten@left}}
\pgfpathlineto{\pgfpoint{\pgf@xa - 8pt + \pgf@shorten@left}{\pgf@yb}}
\pgfpathlineto{\pgfpoint{\pgf@xb + 8pt - \pgf@shorten@right}{\pgf@yb}}
\pgfpathlineto{\pgfpoint{\pgf@xb + 8pt - \pgf@shorten@right}{\pgf@yb - 1.5 * \pgf@shorten@right}}
\pgfpathclose
}
}

\pgfdeclareshape{cornercopoint}{
\inheritsavedanchors[from=rectangle] 
\inheritanchorborder[from=rectangle]
\inheritanchor[from=rectangle]{center}
\inheritanchor[from=rectangle]{north}
\inheritanchor[from=rectangle]{south}
\inheritanchor[from=rectangle]{west}
\inheritanchor[from=rectangle]{east}
\backgroundpath{
\southwest \pgf@xa=\pgf@x \pgf@ya=\pgf@y
\northeast \pgf@xb=\pgf@x \pgf@yb=\pgf@y

\pgfmathsetmacro{\pgf@shorten@left}{\pgfkeysvalueof{/tikz/shorten left}}
\pgfmathsetmacro{\pgf@shorten@right}{\pgfkeysvalueof{/tikz/shorten right}}

\pgfpathmoveto{\pgfpoint{0.5 * (\pgf@xa + \pgf@xb)}{\pgf@yb + 5pt}}
\pgfpathlineto{\pgfpoint{\pgf@xa - 8pt + \pgf@shorten@left}{\pgf@ya + 1.5 * \pgf@shorten@left}}
\pgfpathlineto{\pgfpoint{\pgf@xa - 8pt + \pgf@shorten@left}{\pgf@ya}}
\pgfpathlineto{\pgfpoint{\pgf@xb + 8pt - \pgf@shorten@right}{\pgf@ya}}
\pgfpathlineto{\pgfpoint{\pgf@xb + 8pt - \pgf@shorten@right}{\pgf@ya + 1.5 * \pgf@shorten@right}}
\pgfpathclose
}
}

\makeatother

\pgfkeyssetvalue{/tikz/shorten left}{0pt}
\pgfkeyssetvalue{/tikz/shorten right}{0pt}

\tikzstyle{dagpoint common}=[draw,fill=white,inner sep=1pt, line width = \maplw, minimum height = 4mm, yshift=1.2pt] 
\tikzstyle{dagpoint sc}=[shape=cornerpoint,dagpoint common]
\tikzstyle{dagpoint adjoint sc}=[shape=cornercopoint,dagpoint common]
\tikzstyle{dagpoint}=[shape=cornerpoint,shorten left=4pt,dagpoint common]
\tikzstyle{dagpointadj}=[shape=cornercopoint,shorten left=5pt,dagpoint common]
\tikzstyle{dagpointconj}=[shape=cornerpoint,shorten right=5pt,dagpoint common]
\tikzstyle{dagpointtrans}=[shape=cornercopoint,shorten right=5pt,dagpoint common]
\tikzstyle{dagpointsymm}=[shape=cornerpoint,shorten left=5pt,shorten right=5pt,dagpoint common]

\tikzstyle{widedagpoint}=[dagpoint, minimum width=1 cm, inner sep=2pt]
\tikzstyle{widedagpointadj}=[dagpointadj, minimum width=1 cm, inner sep=2pt]

\tikzstyle{every picture}=[baseline=-0.25em,scale=0.5]
\tikzstyle{label}=[font=\footnotesize,text height=1ex, text depth=0.15ex]

\usetikzlibrary[shapes]

\tikzset{
sidetriangle/.style = {regular polygon, regular polygon sides = 3, aspect = 1, shape border rotate = 90, draw, inner sep = 0, minimum width = 1.2cm}
}

\tikzset{
isoc/.style = {shape=isosceles triangle, shape border rotate = 180, isosceles triangle stretches = true, minimum width = 1.2cm, minimum height= 1.5cm, inner sep = 0.3}}

\tikzset{
coarse/.style = {shape = circle, fill = white, draw, inner sep = 0, minimum width =0.125cm}
}
\tikzset{
coarsesymbol/.style = {shape = circle, fill = white, inner sep = -0.7, minimum width = 0.125cm}
}

\tikzstyle{sidetriangle2}=[sidetriangle, minimum width = 2cm, fill=white]
\tikzstyle{sideisocsmall}]=[style=isoc, minimum width = 1cm, minimum height = 0.8cm, draw, fill=white, font=\Large]
\tikzstyle{sideisoc}]=[style=isoc, minimum width = 2cm, draw, fill=white, font=\Large]
\tikzstyle{sideisocmid}]=[style=isoc, minimum width = 2.5cm, draw, fill=white, font=\Large]
\tikzstyle{sideisocmedium}]=[style=isoc, minimum width = 3cm, draw, fill=white, font=\Large]

\newcommand{\tinygroundnew}{
\smash{
{\hspace{-3pt}
\ensuremath{
\begin{picc}[scale=1.0] 
    \node[upground, xscale=0.8, yscale=0.7] (1) at (0,0.16) {};
    \draw (0,0.03) to (0,-0.25);
\end{picc}
}\hspace{-1pt}}}}

\newcommand{\tinygroundflipnew}{
\smash{
{\hspace{-3pt}
\ensuremath{
\begin{picc}[yscale=-1.0] 
    \node[upground, xscale=0.8, yscale=-0.7] (1) at (0,0.10) {};
    \draw (0,-0.03) to (0,-0.31);
\end{picc}
}\hspace{-1pt}}}}

%% file: style-options.tex
\usepackage{fancyhdr}


%% file: dedication.tex
\begin{dedication}
\Large
\emph{For Lion.}
\end{dedication}

%% file: acknowledgements.tex
\begin{acknowledgements}
Firstly, I wish to thank my supervisors Bob Coecke and Chris Heunen. 

To Bob I am grateful for first suggesting the goal of a categorical quantum reconstruction, which motivated this entire project, for his guidance in the academic world, and for fostering the kind of lively atmosphere which leads to invitations to Barbados with two days notice. 

Chris deserves equal thanks as my main teacher of the practice of research,  and for all the time and excellent advice he has given me, remaining a constantly available source of feedback and collaboration even after moving to the distant town of Edinburgh. 

Next I wish to thank my examiners Paolo Perinotti and Dan Marsden for an enjoyable viva and for helpful feedback which has improved the presentation of this thesis.

I thank Bart Jacobs and Aleks Kissinger for immediately making me feel at home during two visits to their group in Nijmegen which were very influential on this work, and everyone there for enlightening discussions, notably Bas and Bram Westerbaan and John van de Wetering. Extra thanks goes to Bart for suggesting the copowers in Chapter 2, and to Kenta Cho for the collaborative work underpinning Chapter 3.  

Next I extend my gratitude to Paolo, once again, and Mauro D'Ariano, who warmly welcomed me on an interesting visit to Pavia University and whose work with Giulio Chiribella motivates much of this thesis.

I thank Marino Gran for hosting several enjoyable and productive visits to UCLouvain discussing purer categorical topics. I also thank my other collaborators on work not in this thesis but which has enriched my studies; Bob, Chris, John Selby, Aleks, Bas, and Pau Enrique Moliner. 

In Oxford I thank all the colleagues who've made this time so enjoyable and influenced my thinking; including Stefano Gogioso, Sam Staton for suggesting connections with multicategories, Robin Lorenz and visitor Johannes Kleiner for distracting me with talks about consciousness, and Christoph Dorn for encouraging me to see a bit of college life. I am grateful also to the EPSRC for all of their financial support.

Lastly, I wish to thank my friends and family for all of their help through the years; particularly my parents, whose support made this possible, and Jon for his help in my final year. Special thanks goes to Jon and Carol Field for so often putting up with me in their home, where some of the main results of this thesis were reached. Most importantly, I thank Lottie for being with me every step of the way, and bringing so much joy to these years.
\end{acknowledgements}

%% file: abstract.tex
\begin{abstract} 
Many insights into the quantum world can be found by studying it from amongst more general \emph{operational theories} of physics. In this thesis, we develop an approach to the study of such theories purely in terms of the behaviour of their \emph{processes}, as described mathematically through the language of \emph{category theory}. This extends a framework for quantum processes known as \emph{categorical quantum mechanics} (CQM) due to Abramsky and Coecke.

We first consider categorical frameworks for operational theories. We introduce a notion of such theory, based on those of Chiribella, D'Ariano and Perinotti (CDP), but more general than the probabilistic ones typically considered. We establish a correspondence between these and what we call \emph{operational categories}, using features introduced by Jacobs et al.~in \emph{effectus theory}, an area of categorical logic to which we provide an operational interpretation. We then see how to pass to a broader category of \emph{super-causal} processes, allowing for the powerful diagrammatic features of CQM.

Next we study operational theories themselves. We survey numerous principles that a theory may satisfy, treating them in a basic diagrammatic setting, and relating notions from probabilistic theories, CQM and effectus theory. Particular focus is paid to the quantum-like features of \emph{purifications} and \emph{superpositions}. We provide a new description of superpositions in the category of pure quantum processes, using this to give an abstract construction of the more well-behaved category of Hilbert spaces and linear maps.

Finally, we reconstruct finite-dimensional quantum theory itself. More broadly, we give a recipe for recovering a class of generalised quantum theories, before instantiating it with operational principles inspired by an earlier reconstruction due to CDP. This reconstruction is fully categorical, not requiring the usual technical assumptions of probabilistic theories. Specialising to such theories recovers both standard quantum theory and that over real Hilbert spaces. 
\end{abstract}

%% file: introduction.tex
\chapter*{Introduction} \addcontentsline{toc}{chapter}{Introduction} \markboth{Introduction}{Introduction}

The state of contemporary physics is one of contradiction. Our deepest insights into nature come from quantum mechanics, yet even a century after its conception the underlying reality this theory describes remains deeply mysterious, with debates over its proper interpretation continuing to this day. 

At the same time, quantum theory provides us with experimental predictions of unprecedented accuracy, and in more recent years it has emerged that quantum systems can be incredibly \emps{useful}, allowing one to quickly perform computations that may take vastly longer using classical computers.

Together, these facts have encouraged many to take an \emps{operationalist} perspective on physical theories. In this approach, one studies a theory in terms of the \emps{operations} it allows one to perform through physical experiments, rather than any underlying reality that it may describe. Though this could be seen as a denial that any such reality exists, the operational approach may simply be taken as a practical one, allowing physics to progress in the absence of any such clear underlying picture of the world. 

Central to the operational perspective is the notion of a \emps{process} between two physical systems. Examples include the preparation of a system into a particular state, the evolution of a system over time, and the performing of measurements. The mathematical language of such composable processes is \emps{category theory}, a powerful and very general one which can also be used to study connections between different fields and ideas, and even as a foundation for mathematics~\cite{mac1978categories}. 

 Over the past decade and a half, the categorical perspective has led to a new approach to the study of physical theories purely in terms of their process-theoretic properties. Categories provide an intuitive calculus for reasoning about these processes using diagrams~\cite{selinger2011survey}, and lie at the heart of new connections emerging between the foundations of physics, quantum information, mathematics and computer science~\cite{baez1995higher,abramskycoecke:categoricalsemantics,BaezRosettaStone,abramsky2011introduction,coecke2011categories}. 

The greatest successes of the categorical method in physics so far have been in the study of quantum theory itself, and particularly its `pure' processes as captured by the now well-understood category of \emps{Hilbert spaces}~\cite{heunen2009embedding}, including the development of a high-level diagrammatic formalisation of quantum computation~\cite{coecke2011interacting,CKbook}. However, in more recent years, categorical methods relevant to the study of more general theories, including classical physics, have begun to emerge~\cite{NewDirections2014aJacobs,EffectusIntro}. The goal of this thesis is to develop such a categorical approach to the study of operational theories of physics. 

\subsection*{Categories of processes}

Let us now be a bit more precise about the kinds of theories we will be considering. The basic ingredients are physical systems and processes between them. We depict a process $f$ which takes us from a system of type $A$ to one of type $B$ as a box
\[
\scalebox{0.8}{\input{./figures/process-intro.tikz}}
\]
Operationally, we might wish to think of a process as a piece of experimental apparatus in our laboratory. Like these, processes can be plugged together and placed alongside each other, allowing us to form \emps{circuit diagrams}\index{circuit diagram} like: 
\[
\scalebox{0.8}{\input{./figures/circuit-SMC.tikz}} 
\]
It is well-known that such a specification of processes corresponds simply to a \emps{symmetric monoidal category}, whose \emps{objects} are systems and \emps{morphisms} are the processes. The use of these diagrammatic methods in physics was pioneered by Abramsky and Coecke~\cite{abramskycoecke:categoricalsemantics} in a field of research now known as \emps{categorical quantum mechanics} (CQM)\index{CQM|see {categorical quantum mechanics}}\index{categorical quantum mechanics}. 

Since categories are very general, more we will be required in order for us to view a given category as being of an `operational' nature. A particular characteristic of the operational perspective is that, given any system, we should always have some process which simply \emps{discards} it, which we may depict as 
\[
\scalebox{0.8}{\input{./figures/discard.tikz}}
\]
Such symmetric monoidal \emps{categories with discarding} provide a very general framework for reasoning about operational procedures, and will be the basic setting throughout this work. 

Examples include quantum theory, in which morphisms are given by so-called \emps{completely positive maps} between Hilbert spaces, as well as classical probabilistic or possibilistic physics, and even more exotic theories such as \emps{Spekkens toy model}~\cite{spekkens2007evidence,coecke2012spekkens}. In Chapter~\ref{chap:CatsWDiscarding} we introduce this categorical framework more formally and provide numerous such examples.

\subsection*{Tests and operational theories}

Along with the structure of processes, there are further features which are typically included in notions of operational theories. At a basic level, the only way in which we may actually interact with systems in such a theory is through experimental \emps{tests} or measurements. Such a procedure takes a given system and returns one of a range of possible \emps{outcomes}, which the experimenter then records, perhaps by reading the value of a pointer on some device:
\[
\scalebox{0.6}{\input{./figures/box2.tikz}}
\]
Each possible outcome corresponds to the occurrence of a particular physical process or \emps{event}, so that a test such as the above is given simply by an indexed collection of events from $A$ to $B$. Imagining that an experimenter should be free to choose which test to perform next based on outcomes of earlier experiments, however, quickly leads one to realise that tests should more generally take the form 
 \begin{equation*} 
\left(
\scalebox{0.8}{\input{./figures/process_i-var.tikz}}\!\!\right)^n_{i = 1}
\end{equation*}
allowing for varying output systems (though this is not always standard, see e.g.~\cite[p.12-13]{chiribella2010purification}). 

Tests should satisfy some basic rules reflecting our interpretation; for example that like processes we should be allowed to place them side-by-side to form new ones. Moreover, given any test, we may also imagine an experimenter choosing to not care which out of two (or more) of its events, say $f$ and $g$, occur, thus merging them into a new \emps{coarse-grained} event denoted
\[
\scalebox{0.8}{\input{./figures/coarse-grain-merge.tikz}}
\]  
One may then define an \emps{operational theory} to be a collection of events, given by a symmetric monoidal category with discarding, along with a specification of tests and such a partially defined addition $\ovee$, satisfying suitable axioms. Examples include quantum theory, in which tests are given by so-called \emps{quantum instruments}~\cite{nielsen2010quantum}, as well as classical and possibilistic theories.

Now, the typical approach in physics is to only consider \emps{probabilistic} such theories, which come with extra structure explicitly relating tests to probabilistic experiments, along with technical assumptions ensuring that the processes of any given type generate a finite-dimensional real vector space~\cite{chiribella2010purification,Barrett2007InfoGPTs}. In this thesis we will not use these assumptions, showing that operational theories may in fact be studied in a fully categorical manner, much in the spirit of CQM. 

As a first step, it is useful to know that the full structure of an operational theory may in fact be studied in terms of the properties of a single category. This may be done by considering its \emps{partial tests}, i.e.~subsets of tests, which form a category with discarding in a straightforward manner. 

In doing so we gain the ability to represent the features of tests, their outcomes and coarse-graining all using categorical features called \emps{coproducts} $A + B$. In particular, any (partial) test may now be represented as a single morphism of the form 
\[
\scalebox{0.8}{\input{./figures/process_test_extra.tikz}}
\]
Conversely, any suitable category with coproducts in fact defines a whole operational theory in this way. 

The use of these features comes from a categorical formalism for classical, probabilistic and quantum computation known as \emps{effectus theory}~\cite{EffectusIntro}, which gains a new operational interpretation from this perspective. The two categorical formalisms we have mentioned can be compared in terms of their main features as follows. 
\[
\begin{tabular}{c | c c c}
 & \multicolumn{3}{c}{Main Feature Description} \\ 
Formalism & Categorical & Logical & Operational  \\ 
\hline
CQM  & $\otimes$ & And & Parallel Processes \\
Effectus Theory & $+$ & Or & Tests 
\end{tabular}
\]
In Chapter~\ref{chap:OpCategories} we properly define operational theories and study their correspondence with certain categories with coproducts which we call \emps{operational categories}, along with connections to effectus theory.

\subsection*{Beyond sub-causal processes}

From first principles we have seen how a physical theory may be described by a category coming with a partial addition $\ovee$ on its morphisms. 
The fact that $\ovee$ is typically only partially defined relates to the assumption that every morphism $f$ belongs to a test, and so is \emps{sub-causal} meaning that 
\[
\scalebox{0.8}{\input{./figures/sub-causal.tikz}}
\]
for some process $e$. For example in quantum theory the only maps with a direct interpretation, satisfying the above, are those which are trace non-increasing.

However, it is often much easier to instead work with a \emps{total} addition operation $f + g$ on morphisms. To do so, we must consider more general \emps{super-causal} processes. In Chapter~\ref{chap:totalisation} we present a general construction, which given any category $\catC$ with a suitable partial addition operation, constructs a new one $\To(\catC)$ with a total addition, its \emps{totalisation}, within which $\catC$ sits as the sub-category of sub-causal morphisms. This construction can be seen to connect the effectus and CQM formalisms, which typically study sub-causal and super-causal processes respectively.

Working with super-causal processes also allows us to consider powerful extra diagrammatic features which are central to the CQM approach; most notably that our category is \emps{dagger-compact}~\cite{abramskycoecke:categoricalsemantics,Selinger2007139}. In diagrams, this means that we may `flip pictures upside-down', made visible through the use of pointed boxes, and also `bend wires' to exchange inputs and outputs of our morphisms, and so produce diagrams like
\[
\scalebox{0.8}{\input{./figures/compact-circuit.tikz}}
\]
In Chapter~\ref{chap:totalisation} we introduce and study the $\To(\catC)$ construction, before recalling these extra diagrammatic features.

\subsection*{Principles for operational theories}

A major benefit of the study of generalised physical theories is the ability they provide to isolate particular physical principles, and examine their consequences. Several surprising aspects of the quantum world, such as the famous \emps{no-cloning theorem}, have been found to in fact hold in all non-classical probabilistic theories~\cite{Barnum2007NoBroadcast}, while others such as \emps{quantum teleportation} have been found to be more special~\cite{barnum2012teleportation}.

For example, a principle which has been shown to lead to many quantum-like features in the setting of probabilistic theories is the ability to write every process in terms of those which are `maximally informative' in the following sense~\cite{chiribella2010purification}. We call a morphism $f$ \emps{pure} when any \emps{dilation} of it is trivial:
\[
\scalebox{0.8}{\input{./figures/wholei-intro.tikz}}
\text{ for some $\rho$ with }
\quad 
\scalebox{0.8}{\input{./figures/causal-rho.tikz}}
\]
and we say that \emps{purification} holds when every morphism has a dilation which is pure. 
Quantum theory has particularly well-behaved purifications given by the \emps{Stinespring dilation} of any completely positive map. 

In contrast, the following principle is much more general, holding in both the quantum and classical settings. Firstly, many categories come with \emps{zero morphisms}, special morphisms $0 \colon A \to B$ with which every morphism composes to give $0$. Such a category then has \emps{kernels} when every morphism $f$ comes with another $\ker(f)$, satisfying
 \[
\scalebox{0.8}{\input{./figures/kernel-intro1.tikz}}
\ 0 \ \ 
\implies
\
(\exists ! h) 
\
\scalebox{0.8}{\input{./figures/kernel-intro2.tikz}}
 \]
The existence of certain such kernels in fact captures the essential structure of subspaces found in classical and quantum theory, as historically treated in the field of \emps{quantum logic}~\cite{heunen2010quantum}. 

Many principles, such as purification, have typically only been studied in the context of probabilistic theories, while others such as kernels only appear in specific categorical settings. In Chapter~\ref{chap:principles} we study a range of principles for operational theories, seeing that they may in fact be treated in the very general setting of symmetric monoidal categories with discarding. In doing so we find close relations between features that have arisen in the frameworks of probabilistic theories, categorical quantum mechanics and effectus theory. 

\subsection*{Superpositions and phases}

In order to move our attention away from general theories and towards quantum theory itself, we will require an account of arguably its most characteristic feature; the ability to form \emps{superpositions} of pure processes. The most famous example is of course Schr\"odinger's cat, which exists in a superposition of the pure states 
\[
\scalebox{0.8}{\input{./figures/schro.tikz}}
\]

In fact there is already a well-known categorical description of superpositions; abstractly, they are given by an addition operation on morphisms in the category of Hilbert spaces and linear maps. In turn this arises from the existence of \emps{biproducts} $\hilbH \biprod \hilbK$ in this category, which are given concretely by the direct sum of Hilbert spaces~\cite{Selinger2007139}. Indeed states of such a direct sum are precisely superpositions of states of $\hilbH$ with those of $\hilbK$. 

However, there is a problem. Pure quantum processes are not simply given by linear maps between Hilbert spaces, since physically we must identify any two such maps whenever they are equal up to some \emps{global phase}  $e^{i \theta}$, for real-valued $\theta$.

In fact, in the category of pure quantum processes $\hilbH \biprod \hilbK$ is no longer a biproduct. Nonetheless, it has similar properties which we are able to capture using the new notion of a \emps{phased biproduct}, or more general \emps{phased coproduct} $A \pcoprod B$ in a category. These resemble coproducts, but come with extra isomorphisms called \emps{phases}. In quantum theory their presence reflects the fact that we may equally have replaced the state of Schr\"odinger's cat with any one of the form 
\[
\scalebox{0.8}{\input{./figures/schro-phase.tikz}}
\]

In Chapter~\ref{chap:superpositions} we introduce and study phased coproducts, showing that from any suitable category $\catC$ with them we may construct a new one $\plusI{\catC}$ with coproducts from which it arises by quotienting out some `global phases' as above. In particular this lets us recover the category of Hilbert spaces and linear maps from that of pure quantum processes. 

\subsection*{Reconstructing quantum theory}

The primary motivation for the study of operational theories has always been to find new understandings of the quantum world. Just a short time after giving the first precise formulation of quantum theory in the language of Hilbert spaces~\cite{von1955mathematical}, von Neumann himself expressed his dissatisfaction with this formalism~\cite{redei1996john}, and since then there have been many attempts to \emps{reconstruct} the full apparatus of the theory from instead more basic operational statements about experimental procedures. 

Early results were given in terms of quantum logic~\cite{birkhoff1975logic,piron1976foundations,soler1995characterization}, and various versions of the `convex probabilities' framework pursued by Mackey, Ludwig and many others~\cite{mackey1963mathematical,ludwig1985foundations,gudder1999convex,foulis1981empirical,davies1970operational}. Unfortunately, each of these results relied on some technicalities which could not be said to be fully operational. 

The birth of quantum information led to a renewed interest in these questions and, after a proposal by Fuchs~\cite{fuchs2002quantum}, a goal to understand quantum theory in terms of information-theoretic principles. The first form of such a reconstruction of finite-dimensional quantum theory was provided by Hardy~\cite{Hardy2001QTFrom5}, and the first entirely operational reconstruction by Chiribella, D'Ariano and Perinotti~\cite{PhysRevA.84.012311InfoDerivQT}, using purification as its primary principle. Along with these other such reconstructions have been presented in various frameworks~\cite{clifton2003characterizing,wilce2009four,d2010probabilistic,Hardy2011a,fuchs2011quantum,masanes2011derivation,wilceRoyal,hohn2017quantum,selby2018reconstructing,van2018reconstruction}.

However, these reconstructions all typically rely on the standard technical assumptions of probabilistic theories. We may wonder whether these features are integral to the process of recovering quantum theory, or whether instead a purely process-theoretic reconstruction is possible.

In Chapter~\ref{chap:recons} we provide such a categorical reconstruction of quantum theory. We show that any suitable category with discarding which is non-trivial and:
\begin{itemize}  \setlength\itemsep{0em}
\item is dagger-compact;
\item has essentially unique purifications;
\item has kernels;
\end{itemize}
and whose \emps{scalars} satisfy a basic \emps{boundedness} property is in fact equivalent to that of a generalised quantum theory $\Quant{S}$ over a certain ring $S$. When our scalars have an extra feature - the presence of square roots - we find that $S$ resembles either the real or complex numbers. Specialising to probabilistic theories we then immediately obtain either standard quantum theory or more unusually that over \emps{real Hilbert spaces}. 

Recovering quantum theory in this manner provides us with a new elementary axiomatization of the theory which will hopefully be of use in the formalisation of quantum computation, thanks to the many established uses of categories from across computer science~\cite{abramsky2011introduction}. More speculatively, it suggests that future theories of physics may be formulated in a manner which takes processes as their most fundamental ingredients.

\section*{Prerequisites}

Throughout we will assume a very basic knowledge of category theory, though we aim to introduce all key definitions for our purposes, including simple notions such as coproducts. For later reference, some standard ones we will use are as follows. 

In any category a morphism $f \colon A \to B$ is \indef{monic} \index{monomorphism|see {monic}} when $f \circ g =  f \circ h \implies g = h$, \indef{epic} \index{epic} \index{epimorphism|see {epic}} when $g \circ f = h \circ f \implies g = h$, and an \indef{isomorphism}\index{isomorphism} when there exists a morphism $f^{-1}$ with $f \circ f^{-1} = \id{B}$ and $f^{-1} \circ f = \id{A}$. The appropriate notion of mapping $F \colon \catC \to \catD$ between categories is that of a \indef{functor}\index{functor}, and between these is that of a \indef{natural transformation}\index{natural transformation}. 

A pair of functors $F \colon \catC \to \catD$ and $G \colon \catD \to \catC$ form an \indef{equivalence}\index{equivalence} of categories $\catC \simeq \catD$\label{not:equivalence} when there are natural isomorphisms  $G \circ F \simeq \id{\catC}$ and $F \circ G \simeq \id{\catD}$, and an \indef{isomorphism}\index{isomorphism!of categories} when these are strict equalities. Assuming choice, an equivalence may also be given simply by a functor $F \colon \catC \to \catD$ which is is  \indef{full}\index{full} (every $g \colon F(A) \to F(B)$ has $g = F(f)$ for some $f \colon A \to B$), \indef{faithful}\index{faithful} ($F(f) = F(g) \implies f = g$), and has that every object of $\catD$ is  isomorphic to one of the form $F(A)$. By an \indef{embedding}\index{embedding} we will simply mean a faithful functor. Occasionally we will also mention the concept of an \indef{adjunction}\index{adjunction} between categories. 

The standard text on category theory is~\cite{mac1978categories}, while friendlier introductions are given by~\cite{abramsky2011introduction,leinster2014basic} and the physicist-targeted~\cite{coecke2011categories}. 

\paragraph{Statement of originality}
All work here is my own, unless otherwise stated. The results of Section~\ref{sec:totalisation} are in collaboration with Kenta Cho. This thesis is based on the papers~\cite{mainpaper}, \cite{superpos}, \cite{catreconstruction} and new material. During my DPhil I also co-authored the articles~\cite{CatsOfRelations,ppqkd,11709,EPTCS266.25,coecke2017two,connectors}.


%% file: figures/process-intro.tikz
\begin{tikzpicture}
	\begin{pgfonlayer}{nodelayer}
		\node [style=none] (0) at (0, -1) {};
		\node [style=none] (1) at (0, 1) {};
		\node [style=label] (2) at (0, -1.5) {$A$};
		\node [style=label] (3) at (0, 1.5) {$B$};
		\node [style=map] (4) at (0, -0) {$f$};
	\end{pgfonlayer}
	\begin{pgfonlayer}{edgelayer}
		\draw [style=none] (0.center) to (1.center);
	\end{pgfonlayer}
\end{tikzpicture}

%% file: figures/discard.tikz
\begin{tikzpicture}
	\begin{pgfonlayer}{nodelayer}
		\node [style=none] (0) at (0, -1) {};
		\node [style=label] (1) at (0, -1.5) {$A$};
		\node [style=upground] (2) at (0, 0.25) {};
	\end{pgfonlayer}
	\begin{pgfonlayer}{edgelayer}
		\draw [style=none] (0.center) to (2);
	\end{pgfonlayer}
\end{tikzpicture}

%% file: figures/process_i-var.tikz
\begin{tikzpicture}
	\begin{pgfonlayer}{nodelayer}
		\node [style=none] (0) at (0, -1) {};
		\node [style=none] (1) at (0, 1) {};
		\node [style=label] (2) at (0, -1.5) {$A$};
		\node [style=label] (3) at (0, 1.5) {$B_i$};
		\node [style=map] (4) at (0, -0) {$f_i$};
	\end{pgfonlayer}
	\begin{pgfonlayer}{edgelayer}
		\draw [style=none] (0.center) to (1.center);
	\end{pgfonlayer}
\end{tikzpicture}

%% file: chapter1.tex
\chapter{Categories of Processes} \label{chap:CatsWDiscarding}

In the process-theoretic approach to physics, we imagine a physical theory simply as a specification of certain \emps{systems}\index{system} and \emps{processes}\index{process} that may occur between them. A general process may be depicted
\[
\scalebox{0.8}{\input{./figures/process.tikz}}
\] 
and thought of as a physical occurrence which transforms a system of type $A$ into one of type $B$. Given another process taking as input the system $B$ we should be able to compose them to form a new process
\[
\scalebox{0.8}{\input{./figures/process-composition.tikz}}
\]
which we typically interpret as `$f$ occurs, and then $g$ occurs'.The formal structure capturing this notion of composable processes is the following. Recall that a \indef{category} \index{category} $\catC$ \label{not:category} consists of:
\begin{itemize} 
\item \label{not:ob}
 a collection of \indef{objects} \index{object} $A, B, C \dots$;
\item \label{not:morphism}
for each pair of objects $A, B$
a collection $\catC(A,B)$ of  \index{morphism}\indef{morphisms} $f \colon A \to B$; 
\end{itemize}
along with a rule $\circ$ for composing any pair of morphisms $f \colon A \to B$, $g \colon B \to C$ to give a morphism $g \circ f \colon A \to C$. Some basic axioms are also satisfied; composition is associative, with $(h \circ g) \circ f = h \circ (g \circ f)$, and every object comes with an \indef{identity morphism} \index{identity morphism} $\id{A} \colon A \to A$ satisfying $\id{B} \circ f = f = f \circ \id{A}$ for all $f \colon A \to B$. 

Along with the notation $f \colon A \to B$, morphisms may be drawn just like our processes above, with identities and composition depicted
\[
\scalebox{0.8}{\input{./figures/identity-morphism.tikz}}
\qquad \qquad 
\scalebox{0.8}{\input{./figures/composition-morphism.tikz}}
\]
so that the identity and associativity rules become trivial diagrammatically, e.g.~for associativity we have
\[
\scalebox{0.8}{\input{./figures/associativity-morphism.tikz}}
\]
When interpreting a category physically, it is natural to assume we also have a `spatial' composition $A, B \mapsto A \otimes B$, $f, g \mapsto f \otimes g$ allowing us to place objects (systems) and morphisms (processes) `side-by-side' in diagrams:
\[
\scalebox{0.8}{\input{./figures/tensor-morphism.tikz}}
\]
We also often wish to consider processes with `no input'. This is expressed by having some object $I$ interpreted as `nothing', and depicted by the empty diagram:
\[
\scalebox{0.8}{\input{./figures/trivial-object-stuff.tikz}}
\]
As is well-known, these features are captured by the following extra structure on a category. Recall that a \indef{monoidal category} \index{monoidal category} $(\catC, \otimes)$\label{not:moncat} is a category $\catC$ together with
\begin{itemize}
\item 
 a functor $\otimes \colon \catC \times \catC \to \catC$;
 \item  \label{not:unitobject}
 a distinguished object $I$ called the \indef{unit object}\index{unit object};
 \item \label{not:coherenceiso}
 natural \indef{coherence isomorphisms}\index{coherence isomorphisms}
\[
\begin{tikzcd}
(A \otimes B) \otimes C 
\rar{\alpha_{A,B,C}}[swap]{\sim}
& A \otimes (B \otimes C)
\end{tikzcd}
\qquad
\begin{tikzcd}
I \otimes A 
\rar{\lambda_A}[swap]{\sim}
& 
A 
&
A \otimes I \lar[swap]{\rho_A}[swap]{\sim}
\end{tikzcd}
\]
\end{itemize}
satisfying some equations~\cite{coecke2011categories}. 

The diagrammatic notation above in fact forms a precise \emps{graphical calculus} \index{graphical calculus} for reasoning about monoidal categories~\cite{selinger2011survey}, allowing one in practice to avoid the technicalities of the coherence isomorphisms, and making many facts about monoidal categories immediately apparent. 

In any monoidal category, we call morphisms $\rho \colon I \to A$, $e \colon A \to I$ and $s \colon I \to I$ \indef{states, effects} \index{state}\index{effect}\index{scalar} and \indef{scalars} respectively. Since (the identity on) $I$ is given by an empty picture, these are respectively depicted as:
\[
\scalebox{0.8}{\input{./figures/state-effect-scalar.tikz}}
\]

The scalars $s \colon I \to I$ in any monoidal category form a commutative monoid under composition. This is surprising from the formal definition of a monoidal category, but immediate from the graphical calculus since we have:
\[
\scalebox{0.8}{\input{./figures/scalar-commutative.tikz}}
\]
They also allow us to define a scalar multiplication $f \mapsto s \cdot f$ on morphisms by 
\[
\scalebox{0.8}{\input{./figures/scalar-mult.tikz}}
\]
We may have alternatively chosen to multiply by scalars on the other side. However, in categories arising from physical theories the order in which we compose via $\otimes$ is typically unimportant, thanks to the following extra structure. 

Recall that a \indef{symmetric} monoidal category \index{symmetric monoidal category} is one coming with a natural `swap' isomorphism $\sigma_{A, B} \colon A \otimes B \simeq B \otimes A$ \label{not:swap}
satisfying $\sigma_{B,A} \circ \sigma_{A, B} = \id{A \otimes B}$, along with some coherence equations. We depict $\sigma$ by crossing wires, so that naturality and this equation become:
\[
\scalebox{0.8}{\input{./figures/swap-nat.tikz}}
\qquad
\qquad
\scalebox{0.8}{\input{./figures/swap-symmetry2.tikz}}
\]

\subsection*{Categories with discarding}

In this work our focus will be on categories with an interpretation as operational processes one may perform within some domain of physics; such categories have also been called \emps{process theories}~\cite{CKbook,selbythesis}. A distinguishing feature of this operational setting is the ability that any agent should have to simply \emps{discard} or `ignore' a sub-system which is no longer of interest. This leads to the following central notion of this thesis.

\begin{definition} \label{def:cat-with-discarding}
A \deff{category with discarding} \index{category!with discarding}is a category $\catC$ with a distinguished object $I$ and a chosen morphism $\discard{A} \colon A \to I$ for each object $A$, with $\discard{I} = \id{I}$. \label{not:discarding}
A \deff{monoidal category with discarding} is one for which $\catC$ is monoidal, with $I$ being the monoidal unit, and such that 
\begin{equation*} 
\scalebox{0.8}{\input{./figures/discard_axiom1.tikz}}
\end{equation*}
for all objects $A, B$.
\end{definition}

The presence of discarding reflects the perspective of an experimenter who may choose to only examine a smaller part of a larger process or system, as opposed to that of the underlying physics of the world which is typically taken to be reversible and so lack any such notion of discarding a system. We capture this idea of restricting to smaller parts of processes by saying that a morphism $f$ is a \indef{marginal}\index{marginal} of another morphism $g$ when 
\[
\scalebox{0.8}{\input{./figures/dilation.tikz}}
\]
and in this case we refer to $g$ as a \indef{dilation}\index{dilation} of $f$.

The existence of a unique way to discard a system has also been found to be closely related to notions of \emps{causality} in a physical theory~\cite[p.~10]{chiribella2010purification}~\cite{coecke2013causal,coecke2014terminality}, leading to the following definition. 

\begin{definition}~\cite{coecke2015categorical}
 \label{def:causal}
In any category with discarding, a morphism $f \colon A \to B$ is called \deff{causal}\index{causality} when it satisfies 
\[
\scalebox{0.8}{\input{./figures/causal.tikz}}
\]
\end{definition}

Intuitively, if $f$ is a causal process it should have no influence on earlier processes and so make no difference whether we first discard our system or first perform $f$ and then discard its output. 

\begin{lemma} \label{lem:coherence-isoms-causal}
Let $\catC$ be a (symmetric) monoidal category with discarding. Then all coherence isomorphisms $\alpha, \lambda, \rho$, $\sigma$ are causal, and the collection of causal morphisms forms a monoidal subcategory $\catC_{\causal}$. \label{not:causalsubcat}
\end{lemma}

\begin{proof}
Clearly all identities are causal, and if $f, g$ are then so is $g \circ f$. The coherence isomorphisms $\rho_A$ are all causal by naturality since 
\[
\scalebox{0.8}{\input{./figures/coher-causal.tikz}}
\]
Simple naturality argument show that the $\alpha_{A,B,C}$ and $\lambda_A$ are all causal also. 
Finally, whenever $f \colon A \to C$ and $g \colon B \to D$ are causal then so is $f \otimes g$, since: 
\[
\scalebox{0.8}{\input{./figures/causal-arg-elem.tikz}}
\]
\end{proof}

\section{Examples} \label{sec:examples}

Let's now meet our main examples of symmetric monoidal categories both with and without discarding.

\subsection*{Deterministic classical physics}
\begin{exampleslist}
\item \label{ex:SetAndPFun}
There is a category $\Set$\label{cat:set} whose objects are sets $A, B, C \dots$ and morphisms are functions $f \colon A \to B$. This forms the causal subcategory of the symmetric monoidal category with discarding $\PFun$\label{cat:pfun} whose morphisms are now partial functions $f \colon A \pfun B$ between sets. The monoidal structure is given by the Cartesian product $A \times B$ of sets and (partial) functions, with the unit object being the singleton set $I = 1 = \{ \star\}$. 

In this category the scalars may be seen as simply $0$ and $1$. Effects on an object $A$ are found to correspond to subsets $B \subseteq A$, while a state of $A$ is either empty or corresponds to a unique element $a \in A$. Discarding is given by the unique  function $\discard{A} \colon A \to \{\star\}$, so that a morphism is causal precisely when it is total, i.e.~belongs to $\Set$.
\storecounter
\end{exampleslist}

\subsection*{Algebraic examples}

\begin{exampleslist}
\resumecounter

\item \label{ex:com_monoid}
Any commutative monoid $(M, \cdot)$ forms a symmetric monoidal category with one object $\star$ in which morphisms are elements $m \in M$, with $\circ$ and $\otimes$ being multiplication in $M$. Here every morphism is a scalar.

\item \label{ex:MatS}
Let $S$ be a semi-ring (a `ring without subtraction') which is commutative. There is a symmetric monoidal category $\Mat_S$\label{cat:Mats} whose objects are natural numbers $n \in \mathbb{N}$ and morphisms $M \colon n \to m$ are $m \times n$ matrices $M_{i,j}$ with elements in $S$. Such a matrix composes with another $N \colon m \to k$ by standard matrix multiplication
\[
(M \circ N)_{i,k} = \sum^m_{j=1} N_{j,k} \cdot M_{i,j}
\]
using multiplication and addition in the semi-ring $S$. The identity morphism on $n$ is the $n \times n$ matrix with $1$ as each diagonal entry and $0$ elsewhere. The monoidal product $\otimes$ is given on objects by $n \otimes m = n \times m$ and on morphisms by the usual Kronecker product of matrices
\[
M \otimes N = 
\begin{pmatrix}
a_{11} \cdot N & \dots & a_{1m} \cdot N \\
\vdots & \ddots & \vdots \\ 
a_{n1} \cdot N & \dots & a_{nm} \cdot N
\end{pmatrix}
\]
with $I = 1$. The scalars in $\MatS$ correspond to elements $s \in S$, while states and effects on $n$ are $n$-tuples of elements of $S$, seen as column and row vectors respectively. $\MatS$ has a choice of discarding given by $\discard{n} = (1, \dots 1) \colon n \to 1$, so that a matrix $M$ is causal whenever each of its columns sum to $1$. 
\storecounter
\end{exampleslist}

\subsection*{Classical probability theory}

\begin{exampleslist} 
\resumecounter 
\item \label{ex:KLD}
In the category $\KlDT$\label{cat:KlDT} the objects are sets and morphisms $f \colon A \to B$ are functions sending each element $a \in A$ to a finite `distribution' over elements of $B$ with values in the positive real numbers $\Rplus := \{r \in \mathbb{R} \mid r \geq 0 \}$. That is, they are functions $f \colon A \times B \to \Rplus$ for which $f(a,b)$ is non-zero only for only finitely many values of $B$, for each $a \in A$. 

Alternatively, we may view such morphisms $A \to B$ as `$A \times B$ matrices', in which each `column' has finitely many non-zero entries. The composition of $f \colon A \to B$ and $g \colon B \to C$ is then that of matrices
\[
(g \circ f)(a)(c) = \sum_{b \in B} f(a)(b) \cdot g(b)(c)
\] 
This category is symmetric monoidal with $I = \{\star\}$, $A \otimes B = A \times B$ and $f \otimes g$ defined as for the Kronecker product of matrices. The scalars here are given by the `unnormalised probabilities' $\Rplus$. $\KlDT$ has discarding given by the unique map $\discard{A} \colon A \to \{\star\}$ with $\discard{A}(a)(\star) = 1$ for all $a \in A$. Then a morphism $f$ is causal precisely when it sends each element $a \in A$ to a probability distribution, i.e.~for all $a \in A$ we have
\[
\sum_{b \in B} f(a)(b) = 1
\]
In particular, causal states of an object $A$ are simply finite probability distributions over $A$. 
More broadly, at an operational level we are often interested in the sub-category $\KlSD$\label{cat:KlsD} of morphisms $f \colon A \to B$ which send each element to a finite \emps{sub-distribution}, i.e.~ for all $a \in A$
\[
\sum_{b \in B} f(a)(b) \leq 1
\]
In $\KlSD$ the scalars are then probabilities $p \in [0,1]$, and an effect on an object $A$ simply assigns a probability $e(a)$ to each element $a \in A$.
Abstractly we may describe $\KlDT$ and $\KlSD$ as \emps{Kleisli categories}, of the \emps{$\Rpos$-multiset} and \emps{sub-distribution} \emps{monad} respectively~\cite{EffectusIntro}. 
More generally, for continuous probability we can consider the Kleisli category $\Kl(G)$ of the Giry monad $G$ on measure spaces~\cite{MeasSpaces2013Jac,NewDirections2014aJacobs}.

\item \label{ex:Class}
Restricting the above example to finite sets is equivalent to considering the category $\Class := \Mat_{\mathbb{R}^+}$\label{cat:class}, a special case of Example~\ref{ex:MatS}. The scalars here are given by $\mathbb{R}^+$, and causal morphisms are precisely (transposed) Stochastic matrices. 
\storecounter
\end{exampleslist}

\subsection*{Quantum theory}

\begin{exampleslist}
\resumecounter 
\item \label{ex:Hilb}
In the symmetric monoidal category $\Hilb$\label{cat:hilb} objects are complex Hilbert spaces $\hilbH, \hilbK \dots$ and morphisms are bounded linear maps $f \colon \hilbH \to \hilbK$. The monoidal structure is given by the usual tensor product $\hilbH \otimes \hilbK$ of Hilbert spaces, with unit object $I = \mathbb{C}$. Then states $\omega$ of an object $\hilbH$ correspond to elements $\psi \in \hilbH$ by taking $\psi = \omega(1)$, and so by taking adjoints so do effects. In particular the scalars are given by $\mathbb{C}$. 

We write $\FHilb$\label{cat:fhilb} for the full subcategory given by restricting to finite-dimensional Hilbert spaces. Both categories can be seen to describe `pure' quantum theory, which thanks to the \emps{no-deleting} theorem~\cite{pati2000impossibility} comes with no canonical choice of discarding. 
\storecounter
\end{exampleslist}
We may extend this example to include discarding and so describe more general quantum operations as follows.
\begin{exampleslist}
\resumecounter
\item \label{ex:Quant}
In the symmetric monoidal category $\Quant{}$\label{cat:quant}, objects are finite-dimensional complex Hilbert spaces and morphisms $\hilbH \to \hilbK$ are \emps{completely positive} \index{completely positive map} linear maps $f \colon B(\hilbH) \to B(\hilbK)$ between their spaces of operators.
 The monoidal structure $\otimes$ is the usual one for such maps, inherited from that of Hilbert spaces, again with $I=\mathbb{C}$. Scalars $\mathbb{C} \to \mathbb{C}$ now correspond to elements $r \in \Rpos$. By Gleason's Theorem, states and effects on an object $\hilbH$ now correspond to unnormalised density matrices $\rho \in B(\hilbH)$. 

 This category has a canonical choice of discarding with $\discard{\hilbH}$ being the map sending each $a \in B(\hilbH)$ to its trace $\Tr(a) \in \mathbb{C}$. Then a morphism $f$ is causal whenever it is trace-preserving as a completely positive map, and causal states are simply density matrices in the usual sense.

From an operational perspective we are often interested in the subcategory $\QuantSU$\label{cat:quantsu} of trace \emps{non-increasing} completely positive maps, in which the scalars are probabilities $p \in [0,1]$.

There is a functor $\FHilb \to \Quant{}$ which sends each linear map $f \colon \hilbH \to \hilbK$ to the induced \emps{Kraus map} \index{Kraus map}
\[
\Dbl{f} := f \circ (-) \circ f^\dagger \colon B(\hilbH) \to B(\hilbK)
\]
Any two linear maps $f, g$ induce the same such map whenever they are equal up to \emps{global phase}, i.e.~when $f = e^{i \theta} \cdot g$ for some $\theta \in [0,2\pi)$. Hence the subcategory of all such Kraus maps is equivalent to the category $\FHilbP$\label{cat:fhilbP} of equivalence classes $[f]$ of morphisms in $\FHilb$ under equality up to global phase. More broadly we define $\HilbP$\label{cat:hilbP} to be the category of equivalence classes $[f]$ of maps in $\Hilb$ up to global phase, in the same way.

\item \label{ex:CStarVNalg}
Extending our previous example to infinite dimensions, and unifying it with our classical examples, we may consider the category $\CStarop$\label{cat:cstarop} of unital complex C*-algebras, where morphisms $A \to B$ are completely positive linear maps $f \colon B \to A$. Note that we work in the opposite category, with maps going the other way to morphisms. 

There are several different tensors available for (infinite-dimensional) operator algebras; we will take as $\otimes$ the so-called minimal tensor product of C*-algebras. Here $I=\mathbb{C}$, so that scalars are given by elements of $\Rpos$. States on an object $A$ correspond to those $\omega \colon A \to \mathbb{C}$ on the algebra in the usual sense, while effects are positive elements $e \in A$. Discarding $\discard{A}$ is given by the unique completely positive map $\mathbb{C} \to A$ sending $1$ to $1_A$. Then a morphism $A \to B$ is causal whenever its corresponding completely positive map $f \colon B \to A$ is \emps{unital}, with $f(1_B) = 1_A$. More generally the maps with a direct operational interpretation are those which are \emps{sub-unital}, with $f(1_B) \leq 1_A$, forming the subcategory $\CStarSUop$\label{cat:cstarsu}.

When working in finite dimensions one often simply takes morphisms to go in the same direction as maps; we write $\FCStar$\label{cat:FCStar} for the category of finite-dimensional C*-algebras with morphisms $A \to B$ being completely positive maps $f \colon A \to B$. This is symmetric monoidal just as for $\CStarop$. Every finite-dimensional C*-algebra comes with a trace, so that $\discard{A} \colon A \to \mathbb{C}$. here is given by $a \mapsto \Tr(a)$. There is an embedding $\FCStar \hookrightarrow \CStarop$ sending trace non-increasing maps to sub-unital ones.

$\CStarop$ contains a version of classical probabilistic theory given by restricting to the full subcategory of all commutative C*-algebras, with $\Class$ equivalent to the respective subcategory of $\FCStar$.

To model quantum theory we can alternatively restrict to those algebras given by the bounded operators $B(\hilbH)$ of some Hilbert space $\hilbH$. In particular this gives an embedding $\Quant{} \hookrightarrow \FCStar$.


\item 

A particularly well-behaved class of C*-algebras are those which are \emps{von Neumann algebras}\index{von Neumann algebra}. We write $\vNop$\label{cat:vnop} for the (opposite of) the subcategory of $\CStarop$ given by all von Neumann algebras and \emps{normal} completely positive maps between them, as studied in depth in~\cite{EffectusIntro}. We are also often interested in its subcategory $\vNSUop$\label{cat:vnsu} of sub-unital morphisms. 
\storecounter
\end{exampleslist}

Our main examples of categories with discarding so far are either \emps{deterministic}, with scalars $\{0,1\}$, or more generally \emps{probabilistic}, with scalars belonging to $\Rpos$. It is common in the foundations of physics to work only with such \emps{general probabilistic theories}\index{general probabilistic theory}, and to make some extra assumptions. The first, \emps{tomography}\index{tomography}, ensures that morphisms are determined entirely by the probabilities they produce:
\[
\left(
\scalebox{0.8}{\input{./figures/tomography-1.tikz}}
\forall \omega, e 
\right)
\implies
\scalebox{0.8}{\input{./figures/tomography-2.tikz}}
\]
This in turn ensures that maps of any given type generate a real vector space (up to some size issues)~\cite{EPTCS172.1}. Secondly, tomography is assumed to be \emps{finite}, meaning that  this space is finite-dimensional.

In this thesis we will not make any of these assumptions, aiming to work in a purely process-theoretic manner. In particular this allows us to consider more general theories whose scalars are not given by probabilities, such as the following.

\subsection*{Possibilistic examples}

\begin{exampleslist} 
\resumecounter 
\item \label{ex:Rel}
There is a category $\Rel$\label{cat:rel} whose objects are sets and whose morphisms $R \colon A \to B$ are relations $R \subseteq A \times B$. Composition of $R \colon A \to B$ and $S \colon B \to C$ is given by 
\[
S \circ R = \sem{(a,c) \mid \exists b \text{ such that } (a,b) \in R \text{ and } (b,c) \in S} \subseteq A \times C
\]
Here $\otimes$ is given by the Cartesian product, with $I$ being the singleton set $\{\star\}$. The scalars are the Booleans $\mathbb{B} := \{\bot,\top\}$, with states and effects on an object $A$ each corresponding to subsets of $A$. There is a canonical choice of discarding given by the relation $\discard{A} \colon A \to \{ \star\}$ relating every $a \in A$ with $\star$. Then a relation $R \colon A \to B$ is causal when it relates every element of $A$ to some element of $B$.  

\item \label{ex:RelC}
The previous example can be greatly generalised. For any category $\catC$ which is \emps{regular}~\cite{bourn2004regular} we may similarly define a symmetric monoidal category with discarding $\Rel(\catC)$\label{cat:relC} of \emps{internal relations} \index{relation} in $\catC$ in the same way. 

For some examples, $\Rel$ is the special case where $\catC = \Set{}$. Taking $\catC$ to be the category $\VecSp_k$ of vector spaces over a field $k$ gives the category $\Rel(\VecSp_k)$ of \emps{linear relations} over $k$, i.e.~subspaces $R \leq V \times W$. Setting $\catC$ instead to be the category $\Grp$ of groups leads to relations which are subgroups $R \leq G \times H$.

The author explored $\Rel(\catC)$ with Chris Heunen in~\cite{CatsOfRelations}, and with Marino Gran also in~\cite{connectors}, applying its diagrammatic features to topics in categorical algebra. 

More generally still, any such category $\Rel(\catC)$ is a special case of a \emps{bicategory of relations} in the sense of Carboni and Walters~\cite{carboni1987cartesian}.

\item  \label{ex:Spek}
A physically interesting possibilistic example somewhere in-between $\Rel$ and quantum theory is provided by \emps{Spekkens toy model}~\cite{spekkens2007evidence} \index{Spekkens toy model}. Spekkens originally presented the theory in terms of its states, which are subsets of sets of the form $\IV^n$, where $\IV= \{1, 2, 3, 4\}$, obeying the so-called `knowledge balance principle'. The theory was then given an inductive categorical definition in~\cite{coecke2012spekkens,edwards2009non}. 

We write $\Spek$\label{cat:spek} for the smallest symmetric monoidal subcategory of $\Rel$ closed under $\circ, \otimes$, identities, swap maps and relational converse, and containing the objects $I = \{ \star \}$ and $\IV= \{1, 2, 3, 4\}$, all permutations $\IV \to \IV$, and the relations
\[
\scalebox{0.8}{\input{./figures/spekkens-state.tikz}} :: \star \mapsto 1, 3
\qquad 
\qquad
\scalebox{0.8}{\input{./figures/spekkens-copy.tikz}} :: 
\begin{array}{ccc}
1 & \mapsto & (1,1), (2,2)\\ 
2 & \mapsto & (1,2), (2,1) \\ 
3 & \mapsto & (3,3), (4,4) \\ 
4 & \mapsto & (3,4), (4,3)
\end{array}
\]

$\Spek$ contains many similar features to $\FHilb$, closely resembling \emps{stabilizer quantum mechanics}~\cite{pusey2012stabilizer,backens2016complete}. In the original paper~\cite{spekkens2007evidence} (which uses only functional relations as morphisms) quantum features such as steering and teleportation are studied in the theory. It may be extended to a category with discarding $\MSpek$\label{cat:mspek}~\cite{coecke2012spekkens}, defined to be the smallest monoidal subcategory of $\Rel$ closed under relational converse and containing $\Spek$ as well as the discarding morphisms from $\Rel$.

\storecounter
\end{exampleslist}

\subsection*{Morphisms of categories with discarding} 

At times we will also consider mappings between categories. By a \indef{morphism} \index{morphism!of categories with discarding} $F \colon (\catC, \discard{}) \to (\catD, \discard{})$ of categories with discarding we mean a functor $F \colon \catC \to \catD$ which \indef{preserves discarding} in that $\discard{F(I)}$ is an isomorphism 
and $F(\discard{A})$ is causal for all objects $A$. When $\catC$ and $\catD$ are (symmetric) monoidal with discarding we moreover require $F$ to be a strong (symmetric) monoidal functor and that its structure isomorphism $I \simeq F(I)$ is causal; from this it follows that those isomorphisms $F(A) \otimes F(B) \simeq F(A \otimes B)$ will be causal also, similarly to Lemma~\ref{lem:coherence-isoms-causal}. 

In either case a morphism $F$ is an \indef{equivalence} \index{equivalence!of categories with discarding} $\catC \simeq \catD$ when it is full and faithful, and every object of $\catD$ is causally isomorphic to one of the form $F(A)$.

%% file: figures/process.tikz
\begin{tikzpicture}
	\begin{pgfonlayer}{nodelayer}
		\node [style=none] (0) at (4, 1) {};
		\node [style=map] (1) at (4, -0) {$f$};
		\node [style=label] (2) at (4, -1.5) {$A$};
		\node [style=label] (3) at (4, 1.5) {$B$};
		\node [style=none] (4) at (4, -1) {};
	\end{pgfonlayer}
	\begin{pgfonlayer}{edgelayer}
		\draw [style=none] (4.center) to (0.center);
	\end{pgfonlayer}
\end{tikzpicture}

%% file: figures/spekkens-state.tikz
\begin{tikzpicture}
	\begin{pgfonlayer}{nodelayer}
		\node [style=white dot] (0) at (0, -0.5) {};
		\node [style=none] (1) at (0, 0.5) {};
		\node [style=none] (2) at (0, 0.5) {};
		\node [style=none] (3) at (0, 1) {$IV$};
	\end{pgfonlayer}
	\begin{pgfonlayer}{edgelayer}
		\draw [style=none] (0) to (1.center);
	\end{pgfonlayer}
\end{tikzpicture}

%% file: chapter2.tex
\chapter{Operational Theories and Categories}  \label{chap:OpCategories}

Aside from the categorical structure of processes, there are other features which are typically included as basic components of an operational theory of physics. Most notably, such a theory should also describe multiple-outcome experimental procedures or \emps{tests} which we may perform on our systems, along with the outcome data obtained from these experiments.

A framework combining these features with the categorical approach is found in the notion of an `operational-probabilistic theory' due to Chiribella, D'Ariano and Perinotti~\cite{chiribella2010purification}. Such a theory is given by a (strict) symmetric monoidal category of processes,  along with additional structure specifying which processes form admissible tests, modelling the use of experimental outcome data, and allowing one to assign probabilities to these outcomes.

In this chapter, we introduce a similar general notion of such an \emps{operational theory} of physics. We then see how such theories may in fact be presented entirely categorically, simply through the properties a single category which we call an \emps{operational category}. This provides categorical descriptions of all of the main features of operational-probabilistic theories, such as the ability to form convex combinations of physical events, and allows us to extend these notions beyond the probabilistic setting.

In fact the categorical features we will use are not themselves new, being based on \emps{effectus theory}, an area of categorical logic developed by Jacobs and collaborators for the study of classical, probabilistic and quantum computation~\cite{NewDirections2014aJacobs,EffectusIntro}. We will see a correspondence between basic properties of a theory and its associated category, in particular providing effectus theory with an operational interpretation. 

\section{Operational Theories} \label{sec:optheories}

\subsection{Basic operational theories}

Let us begin by introducing a basic framework for what may be described as an operational theory of physics. As outlined in Chapter~\ref{chap:CatsWDiscarding}, we will start with a symmetric monoidal category, whose objects here we call \indef{systems}\index{system} and morphisms $f \colon A \to B$ we call \indef{events}\index{event}. As we have seen this means that events may be composed to form circuit diagrams like 
\[
\scalebox{0.8}{\input{./figures/circuit-SMC3.tikz}}
\]

\paragraph{Tests}

On top of this category, an operational theory concerns experimental procedures which we call \indef{tests}\index{test}. Formally, a test is given by a finite non-empty collection\label{not:test}
\begin{equation} \label{eq:test-display} 
\big(
\begin{tikzcd}
A \rar{{f_\x}} & B
\end{tikzcd}
\big)_{\x \in \X}
\end{equation}
of events of the same type. Such a test is to be thought of as an operation we may perform on a system of type $A$, leaving us with a system of type $B$, with finitely many possible \indef{outcomes}\index{test!outcome} indexed by the non-empty set $\X$. On any run of the test precisely one event $f_\x$ will occur, with the outcome $\x$ then recorded. 

Our theory will specify which finite collections $\ts{f_\x \colon A \to B}_{\x \in \X}$ form admissible tests. More generally we call a finite non-empty collection $\ts{f_\x}_{\x \in \X}$ a \indef{partial test}\index{partial test} when it forms a sub-collection of a test $\ts{f_\y}_{\y \in \Y}$, with $\X \subseteq \Y$. We require some basic properties of tests.

\begin{assumption} \label{assump:tests}
 Tests satisfy the following:
\begin{itemize}
\item 
every event belongs to some test;
\item 
tests are closed under relabellings of outcomes; 
\item 
whenever $\ts{f_\x}_{\x \in \X}$ and $\ts{g_\y}_{\y \in \Y}$ are tests, so is
\[
\left(
\scalebox{0.8}{\input{./figures/test-tensor.tikz}}
\right)_{\x \in \X, \y \in \Y}
\]
\end{itemize}
\end{assumption}
The latter assumption states that, like events, we may place tests `side-by-side' to form new ones. Another way we may expect to form new tests is by using outcome data from earlier ones as input, which we capture as follows. 

\begin{assumption}[\textbf{Basic Control}] \label{assump:control} \index{control!basic}
Let $\ts{f_\x \colon A \to B}_{\x \in \X}$ be a test and, for each of its outcomes $\x$, let $\ts{g({\x},{\y}) \colon B \to C}_{\y \in \Y_\x}$ be a test. Then the following is a test:
\[
\big(
\begin{tikzcd}
A \rar{f_\x} & B \rar{g(\x,\y)} & C
\end{tikzcd}
\big)_{\x \in \X, \y \in \Y_\x}
\]
\end{assumption}

We refer to the above as a \indef{controlled test}, interpreting it as performing the test $\ts{f_\x}_{\x \in \X}$ and then depending on the outcome $\x \in \X$ choosing which test $g(\x,-)$ to perform next. This axiom appears as an optional assumption in the framework of~\cite{PhysRevA.84.012311InfoDerivQT}, which allows for theories without any simple causal structure and hence any such straightforward notion of conditioning. 

\paragraph{Coarse-graining} \index{coarse-graining}
A second way in which an agent should be able to make use of the outcome data from a test is simply to discard it, thus `merging' several of its events. Call a collection of events of the same type $\ts{f_\x \colon A \to B}_{\x \in \X}$ \indef{compatible} \index{event!compatible events}when they form a partial test. An operational theory should come with a rule for merging any compatible pair of events $f, g \colon A \to B$ into a \indef{coarse-grained} event $f \ovee g \colon A \to B$, which we interpret as `either $f$ or $g$ occurs' \label{not:coarse-grain}. The partial operation $\ovee$ should fulfill some basic rules to match this interpretation.

\begin{assumption} \label{assump:cg}
The operation $\ovee$ satisfies the following. 
\begin{itemize}
\item if $\ts{f,g, h_1, \dots h_n}$ is a test, $f \ovee g$ is defined and $\ts{f \ovee g, h_1, \dots h_n}$ is a test;

\item $f \ovee g = g \ovee f$ for all compatible $\ts{f, g}$;

\item $(f \ovee g) \ovee h = f \ovee (g \ovee h)
$ for all compatible $\ts{f, g, h}$;
\item 
for all compatible $\ts{g,h}$ and events $f, k$ we have
\begin{align*}
f \circ (g \ovee h) &= (f \circ g) \ovee (f \circ h) \\
(g \ovee h) \circ k &= (g \circ k) \ovee (h \circ k)   \\ 
f \otimes (g \ovee h) &= (f \otimes g) \ovee (f \otimes h)
\end{align*}
\end{itemize}
\end{assumption}

Each of the above requirements has a straightforward operational interpretation. For example, the first of the final three equations above states that the events `either $g$ or $h$, then $f$' and `either $g$ then $f$, or $h$ then $f$' coincide. Note that both sides of the equations above are indeed well-defined thanks to our assumptions about tests. These properties allows us to define the coarse-graining of any non-empty compatible collection of events by 
\[
\bigovee^n_{i = 1} f_i := f_1 \ovee (f_2 \ovee (\ldots \ovee f_n))
\] 
It will also be helpful to assume the existence of units $0 \colon A \to B$ for coarse-graining, which we think of as the unique \indef{impossible event} \index{event!impossible} between any two systems. Recall that a category has \indef{zero morphisms} \index{zero morphism} \index{zero arrow|see {zero morphism}} \label{not:zeromorphism} when it has a (necessarily unique) family of morphisms $0 = 0_{A,B} \colon A \to B$ satisfying $0 \circ f = 0 = g \circ 0$ for all morphisms $f, g$, and in the monoidal setting we also similarly require  $f \otimes 0 = 0 = 0 \otimes g$.  

\begin{assumption} \label{assump:zeroes}
 The category of events has zero morphisms. Moreover a tuple
$\ts{f_1, \dots f_n}$ forms a test iff $\ts{f_1, \dots, f_n, 0}$ does also, and we have $f \ovee 0 = f$ for all events $f$.
\end{assumption}

Finally we will require the operational ability to discard systems as well as outcome data. The presence of such discarding maps will also allow us to specify tests in terms of partial tests. 

\begin{assumption}[\textbf{Causality}] \label{assump:causality} \index{causality}
The category of events has discarding, and a partial test $\ts{f_\x}_{\x \in \X}$ is a test precisely when it satisfies
\begin{equation} \label{eq:test-causal}
\bigovee_{\x \in \X} \discard{} \circ f_\x = \discard{}
\end{equation}
\end{assumption}

Intuitively, a test should be a partial test which always returns some outcome, as a whole being causal in our earlier sense. Note that in particular the above tells us that $\discard{A}$ is the unique effect on any system which forms a test on its own. As remarked in Chapter~\ref{chap:CatsWDiscarding}, this is indeed closely related to notions of causality in probabilistic theories~\cite{PhysRevA.84.012311InfoDerivQT}.

\begin{definition}
A \deff{basic operational theory} \index{operational theory!basic}$\Theta$ consists of a symmetric monoidal category $\Events_{\Theta}$ \label{not:catevents} with discarding, a choice of tests, and coarse-graining operations $\ovee$ satisfying Axioms 1-\ref{assump:causality}.  
\label{not:catEvents}
\end{definition}

\begin{remark}
Alternatively, one may instead define such a theory in terms of partial tests and coarse-graining, then defining tests as those satisfying~\eqref{eq:test-causal}. However we view tests as a more primitive notion so have used them as our starting point. 
\end{remark}

Many of our motivating examples of operational theories will be \emps{probabilistic}, \index{operational theory!probabilistic} here meaning that their scalars are given by probabilities $p \in [0,1]$, with $p \ovee q := p + q$ being defined whenever this value is $\leq 1$. This is assumed in frameworks such as~\cite{chiribella2010purification}. 

More generally scalars in a theory behave much like probabilities, forming a commutative monoid with a similar partial addition $\ovee$. For example, we may call a test consisting of scalars $\ts{p_i \colon I \to I}^n_{i=1}$ a \emps{distribution} \index{distribution}, by analogy with finite probability distributions. Given any collection of $n$ events $f_i \colon A \to B$ we may then consider their \emps{convex combination} \index{convex combination}
\[
\bigovee^n_{i=1}
\quad
\left(
\scalebox{0.8}{\input{./figures/convex-comb.tikz}}
\right)
\]
which is well-defined thanks to the control axiom. One may go on to define many typical notions from the study of probabilistic theories such as `completely mixed' states, reasoning much like in~\cite{chiribella2010purification}.

\subsection{Extending the notion of test}

So far we have taken the common approach of defining tests as collections of events of the same type $\ts{f_\x \colon A \to B}_{\x \in \X}$, as in e.g.~\cite{chiribella2010purification,gogioso2017categorical}. However, there are standard operational procedures which cannot immediately be described in this manner (typically requiring extra structure to do so \cite[Remark, p.12-13]{chiribella2010purification}). 

For example consider an agent who first performs such a test and then, depending on the outcome $\x$, chooses between performing one of several tests having different output systems $C_\x$. A simple case would be, conditioned on the outcome of a coin flip, preparing some state $\omega$ of a system $A$ or $\rho$ of another system $B$:
\[
\Big(
	\begin{tikzcd}
	I \rar{\text{`heads'}} & I \rar{\omega} & A
	\end{tikzcd}
	,
	\begin{tikzcd}
	I \rar{\text{`tails'}} & I \rar{\rho} & B
	\end{tikzcd}
\Big)
\]
To account for such procedures, we must allow tests to have the general form\label{not:test2}
\begin{equation} \label{eq:new-test-form}
\big(
\begin{tikzcd}
A \rar{{f_\x}} & B_\x
\end{tikzcd}
\big)_{\x \in \X}
\end{equation}
for finite sets $X$, now with varying output systems.

Operational theories of this new sort may be defined just as previously. As before, such a theory specifies a category of events, certain collections of which form tests or partial tests. We now include the empty collection as a partial test of any given type. 

\label{not:coarse-grain2}
Coarse-graining $f \ovee g$ should still only be defined on events of the same type $f, g \colon A \to B$ which belong to some test, whose other events may have different types. More generally a collection of events of the same type $\ts{f_i \colon A \to B}^n_{i=1}$ are again called \indef{compatible} \index{compatible events} when they form a partial test, and their coarse-graining will be definable as before, with that of the empty partial test now set to $0$. To include the procedures discussed above we now require a stronger control axiom. 

\begin{assumption}[Control] \label{assump:strong-control} \index{control}
Let $\ts{f_\x \colon A \to B_\x}_{\x \in \X}$ be a test and, for each of its outcomes $\x$, let $\ts{g(\x,\y) \colon B \to C_{\x,\y}}_{\y \in  \Y_\x}$ be a test. Then the following is a test:
\[
\big(
\begin{tikzcd}
A \rar{f_\x} & B_\x \rar{g({\x,\y})} & C_{\x,\y}
\end{tikzcd}
\big)_{\x \in \X, \y \in \Y_\x}
\]
\end{assumption}
The rest of our earlier axioms were carefully worded to apply immediately to theories of this new form, which we refer to simply as follows.

\begin{definition} \label{def:OTfull}
An \deff{operational theory} \index{operational theory}\index{operational theory!proper} $\Theta$ is given by a symmetric monoidal category with discarding $\Events_{\Theta}$ along with a specification of tests of the form of~\eqref{eq:new-test-form}, and operations $\ovee$ satisfying Axioms~\ref{assump:tests}, \ref{assump:cg}, \ref{assump:zeroes}, \ref{assump:causality} and \ref{assump:strong-control}.
\end{definition}

To distinguish these from basic theories, we sometimes call such theories \indef{proper} operational theories. Because of the common practice of taking tests the form~\ref{eq:test-display}, in this chapter we will consider both kinds of theories. Despite their name, the axioms of proper operational theories are in some sense weaker than those of basic ones, by the following. 

\begin{lemma} \label{lem:get-causal-states}
Let $f \colon A \to B$ be an event in a theory of either kind.
\begin{enumerate}
\item \label{enum:OpTh}
In an operational theory $f$ belongs to a test $\ts{f, e}$ for some $e\colon A \to I$.
\item \label{enum:bOpTh}
In a basic operational theory $f$ belongs to a test of the form $\ts{f, g \colon A \to B}$, and every object has a causal state.
\end{enumerate}
\end{lemma}
\begin{proof}
\ref{enum:OpTh}
Any $f$ belongs to some test $\ts{f \colon A \to B, g_1 \colon A \to C_1, \dots, g_n \colon A \to C_n}$. Then using control $\ts{f, e}$ is a test where $e = \bigovee^n_{i=1} \discard{} \circ g_i$.  

\ref{enum:bOpTh}
Here by assumption $f$ belongs to some test $\ts{f, g_1, \dots, g_n}$ with each $g_i \colon A \to B$. Then $g = \bigovee^n_{i=1} g_i$ is well-defined and $\ts{f,g}$ is a test. For the second statement take $f$ to be the zero state.
\end{proof}

\subsection{Examples}  \label{examples:OpTheories}

Many of our examples of categories from Chapter~\ref{chap:CatsWDiscarding} extend to form operational theories. In each case these also form basic operational theories by restricting to tests of the form~\eqref{eq:test-display} and excluding objects such as $\emptyset$ or $0$ which lack causal states.

\begin{exampleslist}
\item The theory $\ClassDet$ of deterministic classical physics has category of events $\PFun$. Here a collection of partial functions $\ts{f_\x \colon A \to B_\x}_{\x \in \X}$ form a test when their domains are disjoint and partition $A$, with $\ovee$ being disjoint union.
\item 
The classical probabilistic theory $\ClassProb$ instead has category of events $\KlSD$. Tests are collections $\ts{f_\x \colon A \to B_\x}_{\x \in \X}$ satisfying
\[
\sum_{\x \in \X} \sum_{b \in B_\x} f_\x(a)(b) = 1
\]
for all $a \in A$, with $\ovee$ being element-wise addition.

\item 
Finite-dimensional quantum theory $\QuantOT$ has category of events $\QuantSU$ with events given by trace non-increasing completely positive maps. Tests are collections $\ts{f_\x \colon \hilbH \to \hilbK_\x}_{\x \in \X}$ whose sum is trace-preserving. When the $\hilbK_\x$ do not vary these are also known as \emps{quantum instruments}~\cite{nielsen2010quantum}. Here $\ovee$ is the usual addition of such maps. 
More broadly this extends to a theory $\CStarOTC$ with category of events $\CStarSUop$, with tests being collections of maps whose sum is unital.

\item 
The possibilistic classical theory $\RelOT$ has category of events $\Rel$. Here any collection of relations $\ts{R_\x \colon A \to B_\x}_{\x \in \X}$ form a partial test, making the coarse-graining operational total, and we set $R \ovee S = R \vee S$.
More generally, one may take unions of relations in any regular category $\catC$ which is \emps{coherent}~\cite{johnstone2002sketches}, and then $\Rel(\catC)$ extends to an operational theory $\RelOT(\catC)$ in the same way.

\item For any unital commutative semi-ring $S$, we define a theory $\MatOT_S$ whose category of events $\Mat_{S^{\leq 1}}$\label{cat:MatSless} is the subcategory of $\Mat_S$ consisting of those matrices with values in the set $S^{\leq 1} := \{a \in S \mid (\exists b \in S) \ \ \ a + b = 1 \}$. A collection of such matrices forms a test when their sum is causal in $\Mat_S$, with $\ovee$ given by such addition of matrices.
The scalars in this theory are $S^{\leq 1}$; for example in $\MatOT_\mathbb{Z}$ they are simply the integers $\mathbb{Z}$.
\end{exampleslist}

\section{Operational Categories} \label{sec:opcats-first}

The full definition of a (basic) operational theory can be quite unwieldy, requiring the extra specification of both tests and coarse-graining rules. In fact the essential structure of these kinds of theory can be captured internally to a single category. 

\begin{definition}\label{not:PtestTheta}
Let $\Theta$ be an operational theory. We define a symmetric monoidal category with discarding $\PTest(\Theta)$ as follows:
\begin{itemize}
\item 
objects are finite indexed collections $\ts{A_\x}_{\x \in \X}$ of systems of $\Theta$;
\item 
morphisms $M \colon \ts{A_\x}_{\x \in \X} \to \ts{B_\y}_{\y \in  \Y}$ are collections, indexed by $\x \in \X$, of partial tests $\ts{M(\x,\y) \colon A_\x \to B_\y}_{\y \in  \Y}$. 
\end{itemize}
Such morphisms may be thought of as matrices of events for which each column is a partial test. Composition is, via coarse-graining, that of matrices:
\begin{equation} \label{eq:partial-mat-comp}
(N \circ M )(\x,\z) = \bigovee_{\y \in  \Y} N(\y,\z) \circ M(\x,\y)
\end{equation}
We take as unit object $I := \ts{I}$ and define $\ts{A_\x}_{\x \in \X} \otimes \ts{B_\y}_{\y \in  \Y} := \ts{A_\x \otimes B_\y}_{(\x, \y) \in \X \times \Y}$, on morphisms being given by the Kronecker product
\begin{align} \label{eq:tensor-Kronecker}
[M \otimes N]((\x,\w),(\y,\z)) &:= M(\x,\y) \otimes N(\w,\z)
\end{align}
Finally on an object $A=\ts{A_\x}_{\x \in \X}$ we set $\discard{A} = \ts{\discard{A_\x}}_{\x \in \X}$.

For any basic operational theory $\Theta$ we define a category $\PTest(\Theta)$ in just the same way, but instead take objects to be only finite non-empty indexed copies $\ts{A}_{\x \in \X}$ of a fixed system $A$. One may instead denote such objects by a pair $(A,X)$, so that morphisms $M \colon (A,X) \to (B,Y)$ are again $\X$-indexed collections of partial tests, each now having the form $\ts{M(\x,\y) \colon A \to B }_{\y \in  \Y}$.
\end{definition}

\begin{lemma} Let $\Theta$ be a (basic or proper) operational theory. Then $\PTest(\Theta)$ is a well-defined symmetric monoidal category with discarding.
\end{lemma}
\begin{proof}
For any composable morphisms $M, N$, the coarse-graining \eqref{eq:partial-mat-comp} is well-defined since $\ts{N(\y,\z) \circ M(\x,\y)}_{\y \in \Y}$ is a partial test by (basic) control. Then $\circ$ and $\otimes$ are well-defined by Axioms \ref{assump:tests} and \ref{assump:cg}. Each object $A=\ts{A_\x}_{\x \in \X}$ has an identity morphism with $\id{A}(\x,\y)$ given by $\id{A_\x}$ if $\x=\y$, and $0$ otherwise.
\end{proof}

\subsection{From theories to categories}

The main features of any (basic or proper) theory $\Theta$ may all be described within the category $\catC := \PTest(\Theta)$. Firstly, systems and events may be viewed as objects $A := \ts{A}$ and morphisms $f \colon A \to B$ of $\catC$, respectively.  

Next, the impossible events extend to a family of zero arrows $0 \colon A \to B$ in $\catC$. In the case of a proper theory, the empty collection $0 := \ts{}$ now forms a \indef{zero object}\index{zero object} in $\catC$\label{not:initialobject}. This means that it is \indef{initial}\index{initial object}, with every object having a unique morphism $! \colon 0 \to A$, and \indef{terminal}\index{terminal object}\label{not:terminal} meaning there is a unique morphism $! \colon A \to 0$. Any such object always provides zero morphisms via
\[
0_{A,B} = (A \to 0 \to B)
\]

Interestingly, tests may also be captured internally. Firstly, note that we may now represent each outcome set as an object $\X := \ts{I}_{\x \in \X}$ of $\catC$. For each outcome $\x$ there is a corresponding state and effect 
\[
\scalebox{0.8}{\input{./figures/ptest5.tikz}} 
\quad
\text{ with }
\quad 
\scalebox{0.8}{\input{./figures/ptest8.tikz}}
\begin{cases}
\id{I} & \x = \y\\
0 & \x \neq \y
\end{cases}
\]
Each object of the form $\ts{A}_{\x \in \X}$ is then isomorphic to $A \otimes X$. For each $\x \in \X$ it comes with a morphism
\begin{equation} \label{eq:coprojections}
\scalebox{0.8}{\input{./figures/ptest9.tikz}}
\end{equation}
More generally, in the case of a proper theory each object $A = \ts{A_\x}_{\x \in \X}$ comes with a morphism $\coproj_\x \colon A_\x \to A$ corresponding to the test 
\begin{equation} \label{eq:coprod-test}
\ts{0, \dots, 0, \id{A_\x}, 0, \dots, 0}
\end{equation} 
for each $\x \in \X$. 

\paragraph{Coproducts and Copowers}
We can use these maps to characterise each object $A = \ts{A_\x}_{\x \in \X}$, as follows. Thanks to control, they have the property that for any collection of morphisms
\[
\begin{tikzcd}
A_\x \rar{f_\x} & B
\end{tikzcd}
\]
for $\x \in \X$, there is a unique morphism $f \colon A \to B$ with $f \circ \coproj_\x = f_\x$ for all $\x \in \X$. In categorical language, this states that $A$ forms a \indef{coproduct} \index{coproduct} of the objects $A_\x$ with \indef{coprojections} $\pcoproj_\x$\index{coprojection}. \label{not:coproduct}\label{not:coprojection}

A coproduct of $\ts{A_i}^n_{i=1}$ is often denoted by $A_1 + \dots + A_n$. In fact to have coproducts of all finite collections of objects is equivalent to the presence of an initial object and binary coproducts $A + B$ of all objects $A, B$. Explicitly, binary coproducts have the property that for all $f, g$ as below there is a unique morphism $[f,g]$ making the following diagram commute:
\begin{equation*} \label{eq:coprod} 
\begin{tikzcd}[row sep = large]
A \rar{\coproj_A} \drar[swap]{f} & A + B
\arrow[d,dashed, "{[f,g]}" description]
& B \lar[swap]{\coproj_B} \dlar{g} \\ 
& C &
\end{tikzcd} 
\end{equation*}
\label{not:cotuple}
When considering these we write $f_1 + f_2 \colon A_1 + A_2 \to B_1 + B_2$\label{not:diagonal} for the unique morphism with $(f_1 + f_2) \circ \coproj_i = \coproj_i \circ f_i$ for $i=1,2$.

Now in particular, each $\ts{A}_{\x \in \X}$ in $\PTest(\Theta)$ forms a coproduct of the form 
\[
\X \cdot A := \overbrace{A + \dots + A}^{|\X|}
\]
which is called an \indef{$\X$-ary copower} \index{copower}\label{not:copower} of the object $A$. As a special case each object $\X$ forms a copower $\X \cdot I$. We will also write $n \cdot A := \X \cdot A$ where $|\X| = n$.


The coprojections $\coproj_i$ described above are given by indexed collections of (total) tests, of the form \ref{eq:coprod-test}, rather than merely partial ones. This makes these coproducts and copowers \indef{causal}\index{coproduct!causal}, meaning that each coprojection $\coproj_i$ is causal. 

By our definition of the tensor $\otimes$ in \eqref{eq:tensor-Kronecker}, it is also respected by these coproducts as follows.
In a symmetric monoidal category $\catC$ we say that coproducts are \indef{distributive}\index{distributive!coproducts} when each morphism
\[
\begin{tikzcd}[column sep=7em]
A \otimes B + A \otimes C
\rar{[\id{A} \otimes \pcoproj_A,\id{B} \otimes \pcoproj_C]}
&
A \otimes (B + C)
\end{tikzcd}
\]
is an isomorphism. 
Similarly finite copowers are called \indef{distributive} when each canonical morphism $X \cdot (A \otimes B) \to A \otimes (X \cdot B)$ is an isomorphism.

Usefully, thanks to the presence of zero arrows we may define, for any finite coproduct (or copower), `projection' morphisms
\[ \label{not:coprodproj}
\begin{tikzcd}
A_1 + \dots + A_n \rar{\triangleright_i} & A_i
\end{tikzcd}
\quad 
\text{ by }
\quad 
\triangleright_i \circ \coproj_j = 
\begin{cases}
\id{} & i = j  \\ 
0 & i \neq j
\end{cases}
\]
for $i=1,\dots, n$. Note that each morphism $\triangleright_i$ is not typically causal. Distributivity in $\PTest(\Theta)$ ensures that each object $n \cdot A \simeq A \otimes n$, where $n := n \cdot I$, has
\[
\scalebox{0.8}{\input{./figures/ptest10ii.tikz}}
\]
We'll see that for each coproduct or copower the set of morphisms $\triangleright_i$ can be used to pick out the events of corresponding partial tests, and so they are \indef{jointly monic}\index{joint monicity}, meaning that for all morphisms $f, g \colon B \to A_1 + \dots + A_n$ with $\triangleright_i \circ f = \triangleright_i \circ g$ for all $i$, we have $f = g$. 


\paragraph{Operational Categories}
For a (basic or proper) operational theory $\Theta$ we can summarise the properties of $\PTest(\Theta)$ as follows.

\begin{definition} \label{def:opcategory-first}\footnote{In the original pre-print~\cite{mainpaper} we instead used the term `operational category' for what here we later call a `test category'.}
A \deff{(basic) operational category}\index{operational category}\index{operational category!basic} is a symmetric monoidal category with discarding $(\catC, \otimes, \discard{})$, with zero morphisms and finite causal distributive coproducts (resp.~non-empty copowers) such that:
\begin{enumerate}[label=\arabic*., ref=\arabic*]
\item \label{enum:JM}
For each coproduct (resp. copower) the morphisms $\triangleright_i$ are jointly monic;
\item \label{enum:opCatexists}
For every $f \colon A \to B$ there is some causal morphism $g$ of type $A \to B + I$ (resp.~$A \to B + B$) with $f = \triangleright_1 \circ g$.  
\end{enumerate}
\end{definition}

For the first condition it in fact suffices to have causal coproducts $A + B$ for which $\triangleright_1, \triangleright_2 \colon A + A \to A$ are jointly monic~\cite[Lemma 5]{EffectusIntro}. As remarked above in the case of coproducts the initial object $0$ is then in fact a zero object.

\begin{lemma} \label{lem:op-cat}
Let $\Theta$ be a (basic) operational theory. Then $\catC = \PTest(\Theta)$ is a (basic) operational category.
\end{lemma}
\begin{proof}
We have explained all but condition~\ref{enum:opCatexists}, which follow from Lemma~\ref{lem:get-causal-states}. 
\end{proof}


\omitthis{
The features of any basic theory may be fully described in terms of these objects. Partial tests $\ts{f_\x \colon A \to B}_{\x \in \X}$ now correspond in $\catC$ to morphisms
\[
\scalebox{0.8}{\input{./figures/ptest1.tikz}}
\]
or equivalently those of the form $f' \colon A \to X \cdot B$. Such a partial test is then a test precisely when this morphism is causal in $\catC$. The individual events are given by 
\[
\scalebox{0.8}{\input{./figures/ptest2i.tikz}}
\
=
\
\scalebox{0.8}{\input{./figures/ptest2ii.tikz}}
\
=
\
\left(
\begin{tikzcd}
A \rar{f'} & \X \cdot B \rar{\triangleright_\x} & B
\end{tikzcd}
\right)
\]
where in any category with zero morphisms we define `projection' morphisms $\triangleright_\x \colon X \cdot A \to A$ by $\triangleright_\x \circ \coproj_\x = \id{A}$ and $\triangleright_\x \circ \coproj_\y = 0$ for $\x \neq \y$. In distributive settings such as $\catC$ we then have
\[
\scalebox{0.8}{\input{./figures/ptest10.tikz}}
\]
Since a (partial) test is determined by its collection of events the morphisms $\triangleright_\x$ are then \indef{jointly monic}, meaning that for all morphisms $f, g$ with $\triangleright_\x \circ f = \triangleright_\x \circ g$ for all $\x \in \X$, we have $f = g$.  Finally we may also describe coarse-graining in terms of copowers or diagrams as  
\[
\bigovee_{\x \in \X}
\ 
\scalebox{0.8}{\input{./figures/ptest3new2.tikz}}
=
\scalebox{0.8}{\input{./figures/ptest3new1.tikz}}
=
\left(
\begin{tikzcd}
A \rar{f'} & \X \cdot B \rar{\triangledown} & B
\end{tikzcd}
\right)
\]
where $\triangledown$ is the unique morphism with $\triangledown \circ \pcoproj_\x = \id{B}$ for all $\x \in \X$.

Just as we were earlier able to describe the features of a basic operational theory in terms of copowers, we may do the same for proper operational theories using these more general coproducts.

}
\subsection{From categories to theories} \label{sec:CatsToTheories}

Let us now see in detail how the categorical properties of $\catC = \PTest(\Theta)$ may be used to describe the theory $\Theta$. Firstly, general partial tests $\ts{f_i \colon A \to B_i}^n_{i=1}$ in our theory correspond to morphisms 
\begin{equation} \label{eq:ptest-coprod2}
\begin{tikzcd}
A \rar{f}
&
B_1 + \dots + B_n
\end{tikzcd}
\end{equation}
with individual events 
$f_i = \triangleright_i \circ f$. 
Such a collection is a test whenever $f$ is causal.

In particular partial tests of the kind $\ts{f_i \colon A \to B}^n_{i=1}$ appearing in a basic operational theory correspond to morphisms
\begin{equation} \label{eq:ptest-copower2}
\begin{tikzcd}
A \rar{f}
& B + \dots + B = n \cdot B 
\end{tikzcd}
\end{equation}
with $f_i := \triangleright_i \circ f$ for all $i$, or equivalently as morphisms 
\begin{equation*} 
\scalebox{0.8}{\input{./figures/ptest1i.tikz}} 
\quad \text{ with }
\quad 
\scalebox{0.8}{\input{./figures/ptest2icomb.tikz}}
\end{equation*}
for all $i$. The coarse-graining of such a partial test may then be described in terms of copowers by 
\begin{equation} \label{eq:c-g-coprods}
\bigovee^n_{i=1} f_i
=
\begin{tikzcd}
A \rar{f} & n \cdot B
 \rar{\triangledown} & B
\end{tikzcd}
\end{equation}
where we define $\triangledown$ by $\triangledown \circ \coproj_i = \id{B}$ for all $i$, or in diagrams by simply discarding the outcomes:
\begin{equation*} 
\bigovee^n_{i = 1}
\ 
\scalebox{0.8}{\input{./figures/ptest3new2comb.tikz}}
\end{equation*}

In fact, these ideas allow us to define the full structure of a theory from any operational category.

\begin{theorem} \label{thm:OpCatToOpTheory}
Let $\catC$ be a (basic) operational category. Then $\catC$ forms the category of events of a (resp.~basic) operational theory denoted $\OTC(\catC)$ (resp.~$\OT(\catC)$) defined as follows. 
\begin{itemize}
\item 
 A collection $\ts{f_i}^n_{i=1}$ forms a test iff there is a causal morphism $f$ as in~\eqref{eq:ptest-coprod2} (resp.~\eqref{eq:ptest-copower2}) with $\triangleright_i \circ f = f_i$ for all $i$.
\item 
Whenever $f, g \colon A \to B$ are compatible there is then a unique $h \colon A \to B + B$ with $\triangleright_1 \circ h = f$ and $\triangleright_2 \circ h = g$, and we define $f \ovee g = \triangledown \circ h$. 
\end{itemize}
\end{theorem}
This definition of a partial addition comes from Jacobs et.~al~\cite{NewDirections2014aJacobs,Partial2015Cho}. 
\begin{proof}
The condition~\ref{enum:opCatexists} in Definition~\ref{def:opcategory-first} gives that every event belongs to a test. Distributivity ensures that tests are closed under $\otimes$, and control follows from the definition of a coproduct (resp.~copower) as above. Coarse-graining behaves as expected thanks to basic properties of these and distributivity. 

For zero morphisms, note that given any test $(f_i)^n_{i=1}$ corresponding to a morphism $f \colon A \to B$ where $B = B_1 + \dots + B_n$, we may compose it with the coprojection $B \to B + B_{n +1}$ to obtain the test $(f_1, \dots, f_n, 0)$, and the case of copowers is similar. Moreover we get $f \ovee 0 = f$ for all events $f \colon A \to B$ by considering $\coproj_1 \circ f \colon A \to B + B$. Causality is immediate from the definition. 
\end{proof}
 
\subsection{Representable theories}

The theories which arise from either kinds of operational category come with systems encoding the outcome types of tests, characterised as follows.

\begin{definition} \label{def:representable}
An operational theory $\Theta$ is \deff{representable} \index{operational theory!representable} when for every finite indexed collection of system $\ts{A_\x}_{\x \in \X}$ there is a system $A$ and test 
\begin{equation} \label{eq:rep}
\ts{\triangleright_\x \colon A \to A_\x}_{\x \in \X}
\end{equation}
such that for each partial test $\ts{f_\x \colon B \to A_\x}_{\x \in \X}$ there is a unique event $f \colon B \to A$ with $\triangleright_\x \circ f = f_\x$ for all $\x$. 

Similarly a basic operational theory is \indef{representable} when the same holds with respect to finite non-empty collections of the form $\ts{A}_{\x \in \X}$, now in terms of partial tests  $\ts{f_\x \colon B \to A}_{\x \in \X}$.
\end{definition}

\begin{lemma} \label{lem:characterising-rep}
A (basic) operational theory $\Theta$ is representable iff $\Events_\Theta$ has finite coproducts (resp.~non-empty copowers) for which the maps $\triangleright_\x $ are jointly monic and form a test. 
\end{lemma}
\begin{proof}
We prove the result for operational theories, the basic case being similar. Fix a collection $\ts{A_\x}_{\x \in \X}$. Suppose that $\Theta$ is representable, and let $A$ be as in~\eqref{eq:rep}.
Define $\pcoproj_\x \colon A_\x \to A$ to be the unique event with $\triangleright_\y \circ \pcoproj_\x = 0$ for $\x \neq \y$ and $\triangleright_\x \circ \pcoproj_\x = \id{A_\x}$. Then thanks to control the event $\bigovee_{\x \in \X} \coproj_\x \circ \triangleright_\x$ is well-defined and
\[
\triangleright_\y \circ (\bigovee_{\x \in \X} \coproj_\x \circ \triangleright_\x) = 
\bigovee_{\x \in \X} (\triangleright_\y \circ \coproj_\x \circ \triangleright_\x )=
\triangleright_\y
\]
so that by uniqueness it is equal to $\id{A}$. 
Then for any collection of events ${g_\x \colon A_\x \to B}$, for $\x \in \X$, if $g \colon A \to B$ has $g \circ \coproj_\x = g_\x$ for all $\x$ we have 
\[
g = g \circ (\bigovee_{\x \in \X} \coproj_\x \circ \triangleright_\x) = \bigovee_{\x \in \X} g_\x \circ \triangleright_\x 
\] 
Hence this defines the unique such $g$, making $A$ a coproduct. 

Conversely, if $\Events_\Theta$ has such coproducts they satisfy the properties of~\eqref{eq:rep}. Indeed for any partial test $\ts{f_\x \colon B \to A_\x}_{\x \in \X}$ the event $f = \bigovee_{\x \in \X} (\coproj_\x \circ f_\x)$ is well-defined by control, and satisfies $\triangleright_\x \circ f = f_\x$ for all $\x \in \X$, being unique by joint monicity.
\end{proof}

\begin{theorem} \label{thm:corresp-immediate}
There is a one-to-one correspondence between:
\begin{itemize}
\item (basic) operational categories $\catC$;
\item representable (basic) operational theories $\Theta$;
\end{itemize}
via the assignments $\catC \mapsto \OTC(\catC)$ (resp.~$\OT(\catC)$) and $\Theta \mapsto \Events_{\Theta}$. 
\end{theorem}

\begin{proof}
Again we give a proof for operational theories and the basic case is similar.

For any such $\catC$, the theory $\OTC(\catC)$ is representable by Lemma~\ref{lem:characterising-rep}. Conversely let $\Theta$ be a representable theory. By Lemma~\ref{lem:characterising-rep} again, $\Events_{\Theta}$ has finite coproducts with $\triangleright_i$ being jointly monic and forming a test. This ensures that the coprojections $\pcoproj_i$ are causal. Condition~\ref{enum:opCatexists} of an operational category follows since these coproducts have the property of Definition~\ref{def:representable}. 

We now check distributivity. Using control and that tests are closed under $\otimes$, one may verify that the event 
\[
\begin{tikzcd}[column sep =  8em]
A \otimes (B + C) \rar{(\id{A} \otimes \triangleright_B) \ovee (\id{A} \otimes \triangleright_C)} & A \otimes B + A \otimes C
\end{tikzcd}
\]
is well-defined, and thanks to the coarse-graining equations is inverse to the canonical morphism in the opposite direction. Hence $\Events_{\Theta}$ is an operational category. 

Finally we need to check that $\Theta = \OTC(\Events_{\Theta})$. By Lemma~\ref{lem:characterising-rep} the finite coproducts in $\Events_{\Theta}$ are such that partial tests $\ts{f_i}^n_{i=1}$ correspond to morphisms $f \colon A \to B_1 + \dots + B_n$. Moreover, for any compatible pair $f, g$ letting $h \colon A \to B + B$ with $\triangleright_1 \circ h = f $ and $\triangleright_2 \circ h = g$, we have 
\[
\triangledown \circ h 
=
\triangledown \circ (\coproj_1 \circ \triangleright_1 \ovee \coproj_2 \circ \triangleright_2) \circ h
=
f \ovee g
\]
and so coarse-graining in $\Theta$ also coincides with that in $\OTC(\Events_{\Theta})$.
\end{proof}

\noindent
In particular any (basic) theory $\Theta$ may thus be `completed' to a representable one
\[ \label{not:repcompletion}
\Theta^+ := \OTC(\PTest(\Theta))
\]
In fact if $\Theta$ is already representable, this leaves it unaltered, as we now show. 

By a \indef{morphism}\index{morphism!of operational theories} $\Theta \to \Theta'$ (resp.~\indef{equivalence}\index{equivalence!of operational theories} $\Theta \simeq \Theta'$) of theories we mean one $F \colon \Events_\Theta \to \Events_{\Theta'}$ of symmetric monoidal categories with discarding such that $\ts{F(f_\x)}_{\x \in \X}$ is a test if (resp.~if and only if) $\ts{f_\x}_{\x \in \X}$ is, and with $F(0) = 0$ and $F(f \ovee g) = F(f) \ovee F(g)$ for all events $f, g$.


\begin{lemma}
Let $\Theta$ be a (basic) operational theory. Then $\Theta$ is representable iff there is an equivalence of theories $\Theta \simeq \Theta^+$.
\end{lemma}
\begin{proof}
We prove the case of a proper operational theory, the basic case being similar. If $\Theta \simeq \Theta^+$ then since $\Theta^+$ is representable so is $\Theta$. Conversely, suppose that $\Theta$ is representable, and consider the assignment
\begin{align*}
\PTest(\Theta) \ \ \ & \to \ \ \   \Events_\Theta \\ 
\ts{A_i}^n_{i=1} \ \ \ &\mapsto \ \ \ A_1 + \dots + A_n \\ 
(M \colon\ts{A_i}^n_{i=1} \to \ts{B_j}^m_{j=1})  \ \ \ & \mapsto  \ \ \ M'
\end{align*}
where $M'$ is the unique event with $\triangleright_j \circ M' \circ \coproj_i = M(i,j)$ for all $i, j$. It is straightforward to check that this defines an equivalence of symmetric monoidal categories with discarding, preserving coproducts. Hence these are equivalent operational categories, and so Theorem~\ref{thm:corresp-immediate} gives an equivalence of theories $\Theta^+ \simeq \Theta$.
\end{proof}

\subsection{Examples}  \label{sec:OPCATEXAMPLES}

Most of our examples of theories $\Theta$ are already representable as a theory of either kind, and hence determined entirely by their category $\Events_{\Theta} \simeq \PTest(\Theta)$ which forms an operational category, as well as a basic operational category after excluding zero objects. 
\begin{exampleslist}
\item 
The theories $\ClassDet$, $\ClassProb$ and $\RelOT$ are representable. Hence $\PFun$, $\KlSD$ and $\Rel$ are operational categories, with coproducts in each given by disjoint union of sets. Similarly so is $\Rel(\catC)$ whenever $\catC$ is coherent. 
\item 
For any unital semi-ring $S$, $\MatOT_S$ is representable. Then $\Mat_{S^{\leq 1}}$ has finite coproducts given by addition $n + m$ of natural numbers which make it a (basic) operational category. Here every object $n$ is an $n$-ary copower $n \cdot I$.
\item 
$\CStarOT$ is presentable, making $\CStarSUop$ an operational category. Here coproducts are given by the direct sum $A \biprod B$ of C*-algebras. In particular copowers arise from the presence of classical systems $X \cdot I  = \mathbb{C}^{|X|}$. 
\item \label{ex:QuantOP}
In contrast $\QuantOT$ is not representable as a theory of either kind, with $\QuantSU$ containing no such classical systems or coproducts.

Its completion to a representable basic theory is equivalent to the sub-theory of $\CStarOT$ given by restricting to algebras which may be written as a tensor $B(\hilbH) \otimes \mathbb{C}^n$ of a (finite-dimensional) quantum and classical algebra, for some $n \geq 1$, via the correspondence $(\hilbH, \X) \mapsto B(\hilbH) \otimes \mathbb{C}^{|X|}$. 

Its completion instead to a representable proper theory is precisely the full sub-theory $\FinCStarOTC$ of $\CStarOT$ given by the finite-dimensional C*-algebras, via the assignment
\[
\ts{\hilbH_\x}_{\x \in \X} \mapsto \bigoplus_{\x \in \X} B(\hilbH_\x)
\]
Indeed it is well-known that every finite dimensional C*-algebra is of this form (see~\cite{bratteli1972inductive} and~\cite[Example~3.4]{EPTCS171.7}).
\end{exampleslist}

\subsection{Functoriality}

The correspondence between operational theories and categories can itself be made categorical, by considering maps between such categories and theories.

Let us write $\OTCat$\label{cat:OTCat} for the category of operational theories and their morphisms. There is a full subcategory $\OTRep$\label{cat:OTrep} given by the representable theories. Next we write $\OpCats$\label{cat:OpCats} for the category whose objects are operational categories and morphisms $F \colon \catC \to \catD$ are those of symmetric monoidal categories with discarding which preserve finite coproducts $(+,0)$.

\begin{theorem}
Theorem~\ref{thm:corresp-immediate} extends to an isomorphism of categories
\[
\begin{tikzcd}[column sep = large]
\OpCats
\rar[shift left  = 4]{\OTC(-)}[swap]{\simeq}
& 
\OTRep
\arrow[shift left = 4]{l}[yshift=-0.5ex, xshift=0.5ex]{\Events_{(-)}}
\end{tikzcd}
\]
\end{theorem}
\begin{proof}
Since the initial object is a zero object, any functor preserving this preserves zero morphisms and vice versa. In a representable theory tests and coproducts may each be defined in terms of each other using Definition~\ref{def:representable} and Lemma~\ref{lem:characterising-rep}, and so both notions of morphism may be seen to be identical.
\end{proof}

Representability can also be made functorial. We define a category $\OTRepStrict$\label{cat:OTRepStrict} just like $\OTRep$, but now consider theories for which each collection $(A_\x)_{\x \in \X}$ comes with a \emps{specified} representing object $A$ and test $\ts{\triangleright_\x \colon A \to A_\x}_{\x \in \X}$, and require morphisms $F$ to preserve these strictly. 

\begin{theorem} 
The assignment $\Theta \mapsto \Theta^+$ extends to an adjunction 
\[
\scalebox{1.0}{\input{./figures/Rep-adj-mod.tikz}}
\]
where $U$ is the forgetful functor.
\end{theorem}
\begin{proof}
For any theory $\Theta$, $\Theta^+$ has a specified representation of each indexed collection of objects $(\ts{A_\y}_{\y \in \Y_\x})_{\x \in \X}$ given by the object $\ts{A_\y}_{\x \in \X, \y \in \Y_\x}$. For any similar theory $\Theta'$, any morphism $F \colon \Theta \mapsto \Theta'$ may be seen to have a unique extension to one $\widehat{F} \colon \Theta^+ \to \Theta'$ in $\OTRepStrict$. 

In detail, we set $\widehat{F}(\ts{A_\x}_{\x \in \X})$ to be the representing system of the collection $\ts{F(A_\x)}_{\x \in \X}$ in $\Theta'$, and for each morphism $M \colon (A_\x)_{\x \in \X} \to (B_\y)_{\y \in \Y}$ define $\widehat{F}(M)$ to be unique with $\triangleright_\y \circ \widehat{F}(M) \circ \coproj_\x = M(\x,\y)$ for all $\x,\y$.
\end{proof}

A similar result can be given without requiring strictness, simply in terms of $\OTRep$ itself, using the language of \emps{2-categories}. However we will not pursue this here. The corresponding results for basic operational theories are functorial in just the same way.

\subsection{Interlude: theories as multicategories} \label{sec:multicats}

 There is another perspective on operational theories which sheds light on their relationship with categories. Let us draw a (partial) test $\ts{f_i \colon A \to B_i}^n_{i=1}$ as
\[
\scalebox{0.6}{\input{./figures/multicat1.tikz}}
\]
with its single input system $A$ and each of its $n$ outcomes corresponding to an output system $B_i$. Thanks to control we can  `plug in' any other (partial) test with input $B_k$, for some $k$,  to make a new (partial) test:
\[
\scalebox{0.6}{\input{./figures/multicat3.tikz}}
\quad
\mapsto
\quad
\scalebox{0.6}{\input{./figures/multicat2.tikz}}
\]
A general mathematical structure containing such composable `multi-arrows' is that of a \index{multicategory}\emps{multicategory}~\cite[Chapter~2]{leinster2004higher}. These are usually defined like categories, except with arrows allowing multiple inputs $\theta \colon A_1, \dots, A_n \to B$, with a common example being where the $\theta$ are the operations of a (multi-sorted) algebraic theory. To treat operational theories however it is natural to instead flip this picture and think of multi-arrows as having multiple outputs $f \colon A \to B_1, \dots, B_n$ as above. 

 Now our basic assumptions about (partial) tests mean that they form a special kind of multicategory. Firstly, we can always relabel our outcomes, making the multicategory \emps{symmetric}~\cite[p.~54]{leinster2004higher}, with swap maps
\[
\scalebox{0.6}{\input{./figures/multicat4.tikz}}
\mapsto
\scalebox{0.6}{\input{./figures/multicat4swap.tikz}}
\]
which allow us to perform any permutation on outputs. Next, by inserting impossible events $0 \colon A \to B_{n+1}$ we can always add extra redundant outputs: 
\[
\scalebox{0.6}{\input{./figures/multicat5.tikz}}
\mapsto
\scalebox{0.6}{\input{./figures/multicat5-insert.tikz}}
\]
and the operation of coarse-graining $\ovee$ allows us to merge any two outputs of the same type, which we may depict as:
\[
\scalebox{0.6}{\input{./figures/multicat6.tikz}}
\mapsto
\scalebox{0.6}{\input{./figures/multicat6-cg.tikz}}
\]
Together, these features make the multicategory of partial tests \emps{Cartesian}~\cite[4.1]{pisani2014sequential}. 
Hence an operational theory may be equivalently defined as a Cartesian multicategory with extra features, namely a `monoidal' structure $\otimes$ on multi-arrows, as well as discarding $\discard{A}$ and zero multi-arrows, satisfying certain properties.  

\paragraph{Representablility}

The correspondence between representable operational theories and operational categories can be readily understood in this context. 

 In general any monoidal category $(\catC, \boxtimes)$ defines a multicategory $\Multicat(\catC)$ whose multi-arrows $f \colon A \to B_1, \dots, B_n$ are morphisms $f \colon A \to B_1 \boxtimes \dots \boxtimes B_n$ in $\catC$~\cite[p.~36]{leinster2004higher}. Conversely, a multicategory $\mcat{M}$ arises in this way precisely when it is \emps{representable}, meaning that for every tuple $B_1, \dots, B_n$ it has an object $B$ and multi-arrow 
\[
\scalebox{0.6}{\input{./figures/multicat8.tikz}}
\]
such that for every multi-arrow $f \colon A \to B_1, \dots, B_n$ there is a unique $g \colon A \to B$ with 
\[
\scalebox{0.6}{\input{./figures/multicat7.tikz}}
=
\scalebox{0.6}{\input{./figures/multicat9.tikz}}
\]
and moreover that these multi-arrows $\star$ are closed under composition~\cite{hermida2000representable}. Then the category $\mcat{M}_0$ of multi-arrows in $\mcat{M}$ of the form $f \colon A \to B$ has a monoidal structure $\boxtimes$ and there is an equivalence $\mcat{M} \simeq \Multicat(\mcat{M}_0)$~\cite[Theorem 3.3.4]{leinster2004higher}. Moreover when $\mcat{M}$ is a Cartesian multicategory $\mcat{M}_0$ then has finite (co)products, and these provide its monoidal structure $A \boxtimes B =A + B$~\cite[4.9]{pisani2014sequential}. 

In fact by unravelling the definitions one sees that an operational theory is representable in our earlier sense precisely when its multicategory $\mcat{M}$ of partial tests is representable (in a way compatible with $\discard{A}$), with $\mcat{M}_0$ then being an operational category. 

\begin{remark} 
 Beyond multicategories, there has been much study of \emps{generalised multicategories} in which (co)domains $B_1, \dots, B_n$ of multi-arrows are replaced by more general structures~\cite[Chapter~4]{leinster2004higher}, and representability has been considered also in this setting~\cite{cruttwell2010unified}. 

These should allow one to include basic operational theories and their representability in the same picture, by taking multi-arrows to be of the form $A \to (B,n)$ for some object $B$ and $n \in \mathbb{N}$. More generally, one may hope to describe more complex notions of operational theory, for example those including tests with infinitely many outcomes $(f \colon A \to B_i)^\infty_{i=1}$, or outcomes as subsets of $\mathbb{R}$, modelling continuous measurements.
\end{remark}

\section{Further Axioms for Theories} \label{sec:FurtherAxioms}

There are several more basic assumptions which we may have expected to form a part of our definition of an operational theory, and which are often automatic in other frameworks such as~\cite{chiribella2010purification,EffectusIntro}. We first list several of these, before examining their categorical consequences.

\subsection{Positivity}

Our first new property reflects our interpretation of $\discard{}$ and coarse-graining.

\begin{definition}
We call a (basic) operational theory $\Theta$ \deff{positive}\index{operational theory!positive} when it satisfies
\begin{align*}
\discard{} \circ f = 0 & \implies f = 0\\ 
f \ovee g =0 & \implies f = g = 0
\end{align*}
for all events $f, g$.
\end{definition}

This is a natural assumption to make; intuitively, if either of `$f$ occurs and then the system is discarded' or `$f$ or $g$ occurs' are impossible, then so is $f$. 

\begin{lemma}
A (basic) operational theory $\Theta$ is positive iff in $\Theta^+$ we have that $\discard{} \circ f = 0 \implies f = 0$ for all events $f$.
\end{lemma}
\begin{proof}
From the definition of $\Theta^+$ this is equivalent to requiring that any partial test $\ts{f_\x}_{\x \in  \X}$ in $\Theta$ with $\bigovee_{\x \in  \X} \discard{} \circ f_\x = 0
$ 
has $f_\x = 0$ for all $\x \in \X$. Thanks to the properties of $\ovee$ this is equivalent to positivity of $\Theta$.
\end{proof}

\subsection{Complements} \label{subsec:complements}
The next property fits the interpretation of effects as outcomes of binary tests.

\begin{definition}\label{not:complements}
An operational theory is \deff{complemented}\index{operational theory!complemented} when for every effect $e$ there is a unique effect $e^\bot$ for which $\ts{e,e^\bot}$ is a test.\footnote{In~\cite{mainpaper} we originally only considered complemented operational theories, calling them `operational theories with control'.}
\end{definition}

We call the effect $e^\bot$ the \indef{complement}\index{complement} of $e$, thinking of it as simply stating that `$e$ did not occur'. In general such an effect $e^\bot$ exists but is not necessarily unique. Note that complementation in fact automatically ensures causality of a theory.

\begin{lemma}
Let $\Theta$ satisfy all the conditions of an operational theory aside from Axiom~\ref{assump:causality}, and be complemented in the above sense. Then $\Theta$ satisfies causality iff $\ts{\id{I}}$ and $\ts{\lambda_I}$ form tests.  
\end{lemma}
\begin{proof}
The conditions hold in any operational theory by Axiom~\ref{assump:causality} and Lemma~\ref{lem:coherence-isoms-causal}. Conversely, for any object $A$ define $\discard{A} = (0 \colon A \to I)^\bot$, so that $\discard{A}$ is the unique effect for which $\ts{\discard{A}}$ is a test. Since tests are closed under $\otimes$, by the above assumptions $(\Events_\Theta, \discard{})$ then forms a category with discarding.

Now by Axiom~\ref{assump:zeroes} any partial test $\ts{f_i}^{n}_{i=1}$ forms a test iff the unique effect $e$ for which $\ts{f_1, \dots, f_n, e}$ is a test has that $e=0$. But since
\[
e^\bot  = \big( \bigovee^n_{i=1} \discard{} \circ f_i \big)
\]
this holds iff the right-hand sum is equal to $\discard{}$, as in Axiom~\ref{assump:causality}. 
\end{proof}

\subsection{Algebraicity} \label{sec:algebraic}

We have seen two approaches to axiomatizing operational physical theories, based on allowing tests to have events with either varying or non-varying output systems. In fact in most examples the choice is inconsequential, thanks to the following properties which may hold in a theory of either form. 

\begin{definition}
A (basic) operational theory:
\begin{itemize}
\item has that \deff{observations determine tests}\index{observations determine tests} if any suitable collection of events $\ts{f_\x}_{\x \in \X}$ forms a partial test whenever $\ts{\discard{} \circ f_\x}_{\x \in \X}$ does;
\item 
is \deff{algebraic} if whenever $\ts{f \ovee g, h_1, \dots, h_n}$ is a partial test so is $\ts{f, g, h_1, \dots, h_n}$;
\item 
is \deff{strongly algebraic}\index{operational theory!(strongly) algebraic} when both hold.
\end{itemize}
\end{definition}

These may all be seen as `no restriction' properties, stating that any collection of events which might plausibly form a partial test in fact do. 

\begin{lemma}
A (basic) operational theory $\Theta$ is strongly algebraic precisely when observations determine tests in $\Theta^+$.
\end{lemma}
\begin{proof}
Suppose first that $\Theta$ is strongly algebraic, and consider a collection of events $\ts{f^1, \dots, f^n}$ in $\Theta^+$ for which $\ts{\discard{} \circ f^\oi}^n_{\oi=1}$ is a partial test in $\Theta^+$. Without loss of generality we may suppose that each event $f^\oi$ is a partial test $\ts{f^\oi_\x}_{\x \in \X_\oi}$ in $\Theta$. Then so is the following 
\[
\ts{\bigovee_{\x_1 \in \X_1} \discard{} \circ f^1_{\x_1}, \dots, \bigovee_{{\x_n} \in \X_n} \discard{} \circ f^n_{\x_n}}
\]
and so by algebraicity $\ts{\discard{} \circ f^\oi_\x}^{\oi=1,\dots,n}_{\x \in \X_\oi}$ is also a partial test in $\Theta$. Since observations determine tests $\ts{f^\oi_\x}^{\oi=1,\dots,n}_{\x \in \X_\oi}$ is then a partial test in $\Theta$, making $\ts{f^1, \dots, f^n}$ one in $\Theta^+$ as required.

Conversely, if observations determine tests in $\Theta^+$ then clearly they also do in $\Theta$. Now suppose that $(f \ovee g, h_1, \dots, h_n)$ is a partial test in $\Theta$, for some $f, g \colon A \to B$. Then in $\Theta^+$ the following is a partial test
\[
(\discard{} \circ k, \discard{} \circ h_1, \dots, \discard{} \circ h_n)
\]
where $k \colon A \to B + B$ is the unique morphism with $\triangleright_1 \circ k = f$ and $\triangleright_2 \circ k = g$. Hence in $\Theta^+$ so is $\ts{k, h_1, \dots, h_n}$. Composing with the morphisms $\triangleright_1, \triangleright_2$,  it follows that $\ts{f, g, h_1, \dots, h_n}$ is a partial test also. 
\end{proof}

\paragraph{PCMs}
In an algebraic theory of either form, coarse-graining $\ovee$ provides each collection of events $\Events_\Theta(A,B)$ with the following well-behaved structure. For two expressions $e_1$, $e_2$ referring to a partial operation we write $e_1 \kleene e_2$ to mean that $e_1$ is defined precisely when $e_2$ is, and that in this case both are equal.

\begin{definition} \label{def:PCM} \label{not:PCM}
A \indef{partial commutative monoid} (PCM)\index{partial commutative monoid}\index{PCM|see {partial commutative monoid}}~\cite{foulis1994effect} is a set $M$ together with a partial binary operation $\ovee$ and element $0$ satisfying
\[
a \ovee (b \ovee c) \kleene (a \ovee b) \ovee c
\qquad
a \ovee b \kleene b \ovee a
\qquad
a \ovee 0 \kleene a
\]
  for all $a, b, c \in M$. We often write $\bigovee^n_{i=1} a_i$ for the expression $a_1 \ovee (a_2 \ovee (\dots a_n))$.
\end{definition}

Indeed in any theory coarse-graining automatically satisfies all but the first condition of a PCM, which now follows from algebraicity. Since coarse-graining is respected by composition thanks to Axiom~\ref{assump:cg}, this makes each category $\Events_{\Theta}$ \emps{enriched} in partial commutative monoids.

In fact in the presence of (strong) algebraicity this PCM structure suffices to determine the rest of the theory, removing the need for much distinction between proper and basic such operational theories. In a strongly algebraic theory of either form we simply have that a suitable collection $(f_i \colon A \to B)^n_{i=1}$ or $(f_i \colon A \to B_i)^n_{i=1}$ forms a partial test precisely when the sum
\[
\bigovee^n_{i=1} \discard{} \circ f_i
\]
is defined, and a test when this is equal to $\discard{A}$.
Hence we may equivalently define a strongly algebraic theory as symmetric monoidal category with discarding $(\catC, \otimes, \discard{})$ which is enriched in PCMs and satisfies some mild conditions; we return to this and make it precise in Section~\ref{sec:Subcausalcats} of the next chapter.

\begin{remark}[\textbf{D-Test Spaces}] \label{rem:D-Test}
In~\cite{dvurevcenskij1994d}, Dvure{\v{c}}enskij and Pulmannov{\'a} introduced the notion of a \emps{D-test space}, generalising a similar concept due to Foulis and Randall~\cite{foulis1972operational}. Such a structure consists of a collection $T$ of (here finite) indexed sets $t=(x_i)^n_{i=1}$ called \emps{D-tests}, whose elements are called \emps{outcomes}, such that whenever $s, t \in T$ and $t$ extends $s$ then $s = t$.

It is easy to see that any system $A$ of a complemented, positive operational theory $\Theta$ defines a D-test space
\begin{align*}
T &:= \{ \text{ Tests  } \ts{e_i \colon A \to I}^n_{i=1} \mid \text{each $e_i$ is non-zero} \}
\end{align*}
as well as a broader one 
\[
S := \{ \text{ Tests  } \ts{f_i \colon A \to B_i}^n_{i=1} \mid \text{each $f_i$ is non-zero} \}
\]
ignoring size issues from the fact that $S$ may not strictly be a set. Whenever $\Theta$ is algebraic, each of these are then \emps{D-algebraic} in the sense of~\cite[5.1]{dvurevcenskij1994d}, and in fact such special D-Test spaces correspond to \emps{effect algebras}\index{effect algebra}, well-known structures from quantum logic; see~\cite[6.1]{dvurevcenskij1994d}, ~\cite{foulis1994effect} and~\cite{paulinyova2014d}. We thank a referee of~\cite{mainpaper} for suggesting this connection.
\end{remark}

\subsection{Examples} \label{axioms:examples}

The theories $\ClassDet$, $\ClassProb$, $\QuantOT$ and $\CStarOTC$ are all positive and complemented, with their operation $\ovee$ being \emps{cancellative} in that $f \ovee g = f \ovee h \implies g = h$ for all events $f, g, h$. The same holds for any causal probabilistic theory in the sense of~\cite{PhysRevA.84.012311InfoDerivQT}. Moreover:
\begin{exampleslist} \label{examples:notcompl}
\item 
Each theory $\RelOT(\catC)$ is positive, and in particular so is $\RelOT$. However it is not complemented, since here any system comes with tests $\ts{\discard{}, \discard{}}$ and $\ts{\discard{}, 0}$.
\item 
Each theory $\MatOT_R$ is positive whenever $(a + b = 0 \implies a = b = 0)$ in $R$, and complemented whenever $(a + b = 1 = a + c \implies b  = c)$ in $R$.
\end{exampleslist}
All of these examples are strongly algebraic; we leave open the problem of finding a theory which is not algebraic.

\section{Categories of Tests} \label{sec:CatsOfTests}

We have seen that an operational theory may be described, up to representability, by its (category of) partial tests. In fact any complemented theory has yet another presentation in terms of its tests alone, and which fits well into more traditional approaches from categorical logic.

\begin{definition}
For any operational theory $\Theta$ we define the category
\[ \label{not:testcat}
\Test(\Theta) := \PTest(\Theta)_{\causal}
\]
so that morphisms $M \colon \ts{A_\x}_{\x \in \X} \to \ts{B_\y}_{\y \in  \Y}$ here are $\X$-indexed collections of tests in $\Theta$, under matrix composition. 
\end{definition}
Now $\catB = \Test(\Theta)$ is symmetric monoidal with finite coproducts in just the same way as $\PTest(\Theta)$. Moreover, since all morphisms are causal every object $A$ here has a unique morphism $! \colon A \to I$, making $I = \ts{I}$ a terminal object, denoted $1$. These features are related by the following rule. Consider a test 
\[
\ts{f_1 \colon B \to A_1, \dots, f_n \colon B \to A_n, e \colon B \to I}
\]
in $\Theta$ corresponding to an arrow $g \colon B \to A + 1$ in $\catB$, where $A = \ts{A_i}^n_{i=1}$. When $\ts{f_i}^n_{i=1}$ is already a test it corresponds to a unique arrow $f \colon B \to A$ in $\catB$, with $g$ then equal to $\coproj_1 \circ f$. When $\Theta$ has complements this holds iff $e = \discard{}^\bot = 0$, or equivalently when the morphisms $(! + !) \circ f = \ts{ \bigovee^n_{i=1} \discard{} \circ f_i, e }$ and $\coproj_1 \circ ! = \ts{\discard{}, 0_{B,I}}$ are equal:
\[
\begin{tikzcd}[font=\normalsize]
B
\arrow[bend left]{drr}{!}
\arrow[bend right]{ddr}[swap]{g=\ts{f_1,\dots,f_n,e}}
\arrow[dotted]{dr}[description]{\exists \ ! \ \ts{f_i}^n_{i=1}} & & \\
& A
\arrow[swap]{d}{\coproj_1}
\arrow{r}{!}
& I
\arrow{d}{\coproj_1} \\
& A + I
\arrow{r}[swap]{!{} + !{}}
& I + I 
\end{tikzcd}
\]
Categorically this states that the lower-right square is a \indef{pullback}\index{pullback} in $\catB$~\cite[p.71]{mac1978categories}. We can summarise the properties of $\catB$ as follows.

\begin{definition} \label{def:totalform}
A  \deff{(plain) test category}\index{test category}\index{test category!plain} is a category $\catB$ with finite coproducts $(+, 0)$ and a terminal object $1$ such that:
\begin{enumerate}[label=\arabic*., ref=\arabic*]
\item \label{testcatjointmonic} 
The following pair of morphisms are jointly monic:
\[
\begin{tikzcd}[column sep = large]
(A + A) + 1 \arrow[r,"{[\triangleright_1, \coproj_2]}", shift left = 2.5]
\arrow[r,"{[\triangleright_2, \coproj_2]}", swap,shift left = -2.5]
 & A + 1 
\end{tikzcd}
\]
 where we define $\triangleright_1 = [\coproj_1, \coproj_2 \circ !]$ and $\triangleright_2 = [\coproj_2 \circ !, \coproj_1]$ of type $A + A \to A + 1$;
\item \label{testcatcatpullweak} Diagrams of the following form are pullbacks:
\[
\begin{tikzcd}
A \rar{!} \dar[swap, "\coproj_1"] & 1 \dar{\coproj_1} \\
A + 1 \rar[swap]{! + !} & 1 + 1
\end{tikzcd}
\]
\end{enumerate}
A \deff{monoidal}\index{test category!monoidal} test category is one which is symmetric monoidal $(\catB, \otimes)$ with $I = 1$ and for which finite coproducts are distributive. Unless otherwise indicated by use of the word `plain', by `test category' we always mean a monoidal one. 
\end{definition}

\begin{theorem} \label{thm:totalform}
Let $\Theta$ be a complemented operational theory. Then $\Test(\Theta)$ is a test category.
\end{theorem}
\begin{proof}
It only remains to verify condition~\ref{testcatjointmonic}, which we turn to shortly. 
\end{proof}

To complete this proof we will first need to see how the broader category $\catC=\PTest(\Theta)$ can be defined in terms of $\catB = \Test(\Theta)$. For this, note that in any complemented theory we may define events or more general partial tests $\ts{f_i \colon A \to B_i}^n_{i=1}$ as special kinds of tests. Indeed any such partial test corresponds uniquely to a test of the form 
\[
\ts{f_1 \colon A \to B_1, \dots, f_n \colon A \to B_n, e \colon A \to I}
\]
by taking $e = (\bigovee^n_{i=1} \discard{} \circ f_i)^{\bot}$. 

In this way arrows $A \to B$ in $\catC$ correspond to arrows $A \to B + 1$ in $\catB$. This situation of a `partial' category associated to a `total' one has been studied already by Cho~\cite{Partial2015Cho} and Jacobs et al.~\cite{EffectusIntro} and we borrow their approach here. 

\subsection{The category  $\Partial(\catB)$}
\label{not:partialarrow}
For any category $\catB$ with finite coproducts $(+,0)$ and a terminal object $1$, by a \indef{partial}\index{partial arrow} arrow $f \colon A \partialto B$ we mean an arrow $f \colon A \to B + 1$ in $\catB$. These partial arrows form a category $\Partial(\catB)$ under composition:

\begin{equation*} 
\big(
\begin{tikzcd}
A \tikzcdpartialright{f} & B \tikzcdpartialright{g} & C
\end{tikzcd}
\big)
=
\big(
\begin{tikzcd}
A \rar{f} & B + 1 \rar{[g,\coproj_2]} & C + 1
\end{tikzcd}
\big)
\end{equation*}
\noindent which we denote by $g \partialcomp f$, with $\id{A}$ given by the morphism $\coproj_1 \colon A \to A + 1$ in $\catB$. 
There is an identity-on-objects functor $\totaspar{-} \colon \catB \to \Partial(\catB)$ defined by
\begin{equation*} 
\big(
\begin{tikzcd}
A \tikzcdpartialright{ \totaspar{f} } & B 
\end{tikzcd}
\big)
:=
\big(
\begin{tikzcd}
 A \rar{f} & B \rar{\coproj_1} &  B + 1
\end{tikzcd}
\big)
\end{equation*}
Abstractly, $\Partial(\catB)$ is described as the Kleisli category of the \emps{lift monad} $(-) + 1$ on $\catB$~\cite{Partial2015Cho}. It inherits nice properties in general:

\begin{itemize}[leftmargin=*]
\item the initial object $0$ of $\catB$ forms a zero object in $\Partial(\catB)$, with $0_{A,B} \colon A \partialto B $ given by the arrow $\coproj_2 \circ ! \colon A \to B + 1$ of $\catB$;
\item  each coproduct $A + B$ in $\catB$ is again a coproduct in $\Partial(\catB)$, with coprojections $\totaspar{\coproj_1} \colon A \partialto A + B$ and $\totaspar{\coproj_2} \colon B \partialto A + B$, giving $\Partial(\catB)$ finite coproducts;

\item  when $\catB$ is symmetric monoidal with distributive coproducts so is $\Partial(\catB)$. Here $A \otimes B$ is the same as in $\catB$, satisfying $\totaspar{f} \otimes \totaspar{g} = \totaspar{f \otimes g}$ and with all coherence isomorphisms coming from $\catB$; 

\item when $I$ is also terminal in $\catB$, $\Partial(\catB)$ has discarding with $\discard{A} \colon A \partialto I$ given by $\coproj_1 \circ ! \colon A \to 1 + 1$ in $\catB$.
\end{itemize}  

We can now understand the property~\ref{testcatjointmonic} of a test category $\catB$: it simply asserts the joint monicity of the maps $\triangleright_1, \triangleright_2 \colon A + A \partialto A$ in the category $\catC = \Partial(\catB)$. When $\catB = \Test(\Theta)$ for a complemented theory we indeed have $\catC \simeq \PTest(\Theta)$ as expected, and we saw that this condition simply corresponded to partial tests being determined by their individual events. In fact, the other properties of an operational category also hold, along with the following. 

\begin{definition} \label{Def:comp-cat}
An operational category $\catC$ is \deff{complemented}\index{operational category!complemented} when every morphism $f \colon A \to B$ has $f = \triangleright_1 \circ g$ for a \emps{unique} causal morphism $g \colon A \to B + I$.
\end{definition}

\begin{theorem} \label{thm:Total-Partial}
There is a one-to-one correspondence between:
\begin{itemize}
\item 
test categories $\catB$;
\item 
complemented operational categories $\catC$;
\end{itemize}
given by $\catB \mapsto \Partial(\catB)$ and $\catC \mapsto \catC_{\causal}$.
\end{theorem}

\begin{proof}
For any test category $\catB$, as outlined above $\Partial(\catB)$ is an operational category. In particular, for condition \ref{enum:opCatexists}, note that any morphism $f \colon A \partialto B$ in $\Partial(\catB)$, given by some $f \colon A \to B + 1$ in $\catB$, has that $\totaspar{f}$ is causal in $\Partial(\catB)$ with $f = \triangleright_1 \partialcomp \totaspar{f}$.

 Next we claim that causal morphisms $A \partialto B$ in $\Partial(\catB)$ correspond precisely to morphisms $f \colon A\to B$ in $\catB$ via $f \mapsto \totaspar{f}$.
 Indeed, by the definition of $\Partial(\catB)$, a morphism $g \colon A \partialto B$ here is causal precisely when in $\catB$ the morphism $g \colon A \to B + 1$ makes the outer rectangle below commute. But then $g = \coproj_1 \circ f$ for some unique $f \colon A \to B$, since the lower square is a pullback. Equivalently $g = \totaspar{f}$ in $\Partial(\catB)$.
\[
\begin{tikzcd}[font=\normalsize]
A
\arrow[bend left]{drr}{!}
\arrow[bend right]{ddr}[swap]{g}
\arrow[dotted]{dr}[description]{\exists \ ! f} & & \\
& B
\arrow[swap]{d}{\coproj_1}
\arrow{r}{!}
& I
\arrow{d}{\coproj_1} \\
& B + I
\arrow{r}[swap]{!{} + !{}}
& I + I 
\end{tikzcd}
\]

For complementation, note that the definition of $\Partial(\catB)$ gives thats for any $f \colon A \partialto B$ in $\Partial(\catB)$, given by some $f \colon A \to B + 1$ in $\catB$, a causal morphism $g = \totaspar{h}$ has $f = \triangleright_1 \partialcomp g$ iff $\totaspar{f} = \totaspar{h}$. Equivalently, $f = h$ in $\catB$. Hence $g = \totaspar{f}$ is the unique such morphism. 

Conversely, for any complemented operational category $\catC$, the theory $\Theta = \OTC(\catC)$ is complemented and we have $\catC \simeq \PTest(\Theta)$. Hence $\catC_{\causal} \simeq \Test(\Theta)$, which we've seen is a test category.

For the correspondence, we have just shown above that each symmetric monoidal functor $\totaspar{-} \colon \catB \to \Partial(\catB)_{\causal}$ is full and faithful, and so an isomorphism of categories. By complementation each symmetric monoidal functor $\Partial{(\catC_{\causal})} \to \catC$ sending $(f \colon A \to B + I)$ to $(\triangleright_1 \circ f \colon A \to B)$ is an isomorphism also.
\end{proof}

In fact this assignment can be made functorial. Define a category $\TestCats$\label{cat:testcats} whose objects are test categories and morphisms $F \colon \catB \to \catB'$ are strong symmetric monoidal functors preserving finite coproducts $(+,0)$. Let us write $\OpCatsComp$\label{cat:opcatscomp} for the full subcategory of $\OpCats$ given by the complemented operational categories. 

\begin{theorem}
The above assignments extend to an equivalence of categories
\[
\begin{tikzcd}[column sep = large]
\TestCats
\rar[shift left  = 4]{\Partial(-)}[swap]{\simeq}
& 
\OpCatsComp
\arrow[shift left = 4]{l}[yshift=-0.5ex, xshift=0.5ex]{(-)_\causal}
\end{tikzcd}
\]
\end{theorem}
\begin{proof}
Any morphism of test categories $F \colon \catB \to \catB'$ preserves $+$ and $1$ and so is easily seen to extend to a morphism $\Partial(F) \colon \Partial(\catB) \to \Partial(\catB')$. Conversely any morphism $G \colon \catC \to \catD$ in $\OpCatsComp$ preserves discarding and hence restricts to a morphism $G_\causal \colon \catC_\causal \to \catD_\causal$. These functors form an equivalence just as in Theorem~\ref{thm:Total-Partial}.
\end{proof}

Hence we may study any complemented operational theory $\Theta$ equivalently in terms of the test category $\catB=\Test(\Theta)$ or its `partial form' $\catC= \PTest(\Theta)$. This second perspective is useful when working with test categories, as in the following.

\begin{lemma} \label{opcatlemma}
In a test category all coprojections $\coproj_1 \colon A \to A + B$ are monic and diagrams of the following forms are pullbacks:
\[
\begin{tikzcd}
A \rar{f} \arrow[d, swap, "\coproj_1"] & B \dar{\coproj_1} \\
A + C \rar[swap]{f + \id{}} & B + C
\end{tikzcd}
\qquad
\qquad
\qquad
\begin{tikzcd}
0 \arrow[r, "!"] \arrow[d, "!", swap] & A \arrow[d, "\coproj_1"] \\
B \arrow[r, "\coproj_2", swap] & A + B
\end{tikzcd}
\]
\end{lemma}

\begin{proof}
Each coprojection in $\catB$ is again a coprojection in the broader category $\Partial(\catB)$. But $\Partial(\catB)$ has zero morphisms, and so each coprojection $\coproj_i$ here is split monic via the morphism $\triangleright_i$ since $\triangleright_i \circ \coproj_i = \id{}$. This makes them monic in $\catB$ also. 

For the left-hand pullback, suppose that we have a commuting square 
\[
\begin{tikzcd}
D \rar{g} \arrow[d, swap, "h"] & B \dar{\coproj_1} \\
A + C \rar[swap]{f + \id{}} & B + C
\end{tikzcd}
\]
Letting $k = (\id{} + !) \circ h \colon D \to A + 1$, it's routine to check that $(! + !) \circ k = \coproj_1 \circ ! \colon D \to 1 + 1$, and so by the pullback in the definition of an operational category,  there is a unique $r \colon  D \to A$ such that $k = \coproj_1 \circ r$. Working in $\Partial(\catB)$ we then have
\begin{align*}
&\triangleright_1 \partialcomp h = \triangleright_1 \partialcomp k = \triangleright_1 \partialcomp \coproj_1 \partialcomp r = r
\\ 
&\triangleright_2 \partialcomp h = \triangleright_2 \partialcomp (f + \id{}) \partialcomp h = \triangleright_2 \partialcomp \coproj_1 \partialcomp g = 0 = \triangleright_2 \partialcomp \coproj_1 \partialcomp r
\end{align*}
Hence by joint monicity of the $\triangleright_i$ we have $h = \coproj_1 \circ r$ in $\catB$. Now in $\catB$ we have 
\[
\coproj_1 \circ g = (f + \id{}) \circ \coproj_1 \circ r = \coproj_1 \circ f \circ r
\]
Since $\coproj_1$ is monic, we then have $g = f \circ r$, and this $r$ is unique, as required. Finally, the right-hand pullback is in fact a special case of the left-hand one:
\[
\begin{tikzpicture}[scale = 2.0]
\node (tl) at (0,0) {$0$};
\node (bl) at (0,-1.5) {$0 + B$};
\node (bbl) at (-1.3,-2.1){$B$};
\node (tr) at (2,0) {$A$}; 
\node (br) at (2,-1.5) {$A + B$};
\path[commutative diagrams/.cd, every arrow, every label]
(tl) edge node{$!$} (tr)
(tl) edge[bend right] node[swap]{$!$} (bbl)
(tl) edge node[swap]{$\coproj_1$} (bl)
(bl) edge node{$! + \id{}$} (br)
(tr) edge node{$\coproj_1$} (br)
(bbl) edge node[pos=1.0,rotate=25]{$\sim$} (bl)
(bbl) edge[bend right, swap] node{$\coproj_2$} (br);
\end{tikzpicture}
\]
\end{proof}

\begin{examples}
Considering our examples of representable complemented operational theories $\Theta$, with $(\Events_{\Theta})_{\causal} \simeq \Test(\Theta)$, we have that:
\begin{exampleslist} \label{examples:compl}
\item 
$\Set$ is a test category, with partial form $\Partial(\Set)\simeq \PFun$. More generally any \emps{extensive} category forms a plain test category~\cite{carboni1993introduction};
\item
$\KlD := \KlDT_\causal$,\label{cat:KlD} also known as the Kleisli category of the \emps{distribution monad}, is a test category with partial form $\KlSD$~\cite{jacobs2011probabilities};
\item 
The (opposite of) the category of C*-algebras and completely positive unital maps $\CStarUop$\label{cat:cstaru} is a test category with partial form $\CStarSUop$. Similarly the subcategory $\vNUop$\label{cat:vnU} of unital maps in $\vNop$ is a test category, with partial form consisting of sub-unital such maps.
\end{exampleslist}
\end{examples}

\section{Effectuses}  \label{sec:EffectusTheory}

The categorical structures we have made use of in this chapter were first considered by Jacobs et al.~\cite{NewDirections2014aJacobs,StatesConv2015JWW,Partial2015Cho}  in an approach to the study of quantum computation based on categorical logic called \emps{effectus theory}\index{effectus!theory}. An introduction to this area is found in~\cite{EffectusIntro}, the central notion being the following.

\begin{definition}
A \deff{(monoidal) effectus}\index{effectus} is a plain (resp.~monoidal) test category for which diagrams of the following form are pullbacks:
\[
\begin{tikzcd}[row sep = 0.3em, column sep = 0.3em]
A \arrow[rr,"!"] \arrow[dd, swap, "\coproj_1"] 
& & 1 \arrow[dd, "\coproj_1"] \\ & \effpullbone &  \\ 
A + B \arrow[rr,swap,"! + !"] & & 1 + 1
\end{tikzcd}
\qquad
\qquad
\begin{tikzcd}[row sep = 0.3em, column sep = 0.3em]
A + B \arrow[rr,"! + \id{}"] \arrow[dd, swap, "\id{} + !"] 
& & 1 + B \arrow[dd, "! + !"] \\ & \effpullbtwo &  \\ 
A + 1 \arrow[rr,swap,"! + !"] & & 1 + 1
\end{tikzcd}
\]
Note that the pullback in the definition of a test category is a special case of \effpullbone. 
\end{definition}

The approach of this chapter can now provide us with an operational interpretation of the effectus axioms; in fact they correspond to the earlier properties we considered for operational theories.

Let us call a test category $\catB$ \indef{positive}\index{test category!positive} when it has that diagrams of the form \effpullbone{} are pullbacks. 

\begin{proposition} \label{prop:positivity}
Let $\Theta$ be a complemented operational theory. Then $\Theta$ is positive iff $\Test(\Theta)$ is positive. 
\end{proposition}


\begin{proof}
Interpreted in $\Test(\Theta)$, the pullback \effpullbone{} tells us that any test of the form $\ts{f_1, \dots, f_n, g_1, \dots ,g_m}$ making the outer rectangle below commute factors over $\coproj_1$. 
\[
\begin{tikzcd}[font=\normalsize]
(A)
\arrow[bend left=15]{drr}{!}
\arrow[bend right=15]{ddr}[swap]{\ts{f_1,\dots,f_n,g_1, \dots g_m}}
\arrow[dotted]{dr}[description]{\exists} & & \\
& (B_1, \dots, B_n)
\arrow[swap]{d}{\coproj_1}
\arrow{r}{!}
& (I)
\arrow{d}{\coproj_1} \\
& (B_1, \dots, B_n) + (C_1, \dots, C_m)
\arrow{r}[swap]{!{} + !{}}
& (I, I) 
\end{tikzcd}
\]
Explicitly this means that any such test for which 
\begin{equation} \label{eq:poscatproof}
\bigovee^m_{i=1} \discard{} \circ g_i = 0
\end{equation}
has $g_i = 0$ for all $i$. But this is equivalent to stating that any partial test $\ts{g_i}^m_{i=1}$ satisfying~\eqref{eq:poscatproof} has $g_i = 0$ for all $i$, which is equivalent to positivity of $\Theta$.
\end{proof}




\noindent
Combining positivity with complements gives some nice categorical features.

\begin{lemma} \label{lemma:PosOpCat} 
Let $\catB$ be a positive test category.  
\begin{enumerate} 
\item \label{anyisototal} Any isomorphism in $\Partial(\catB)$ is causal. 
\item \label{initobstrict} The initial object $0$ is strict in $\catB$. That is, any morphism $f \colon A \to 0$ is an isomorphism.
\item \label{posopcatpull} Diagrams of the following form are pullbacks in $\catB$: 
\begin{equation}
\label{eq:generalPosPullbacks}
\begin{tikzcd}
A \rar{f} \dar[swap]{\coproj_1} & B \dar{\coproj_1} \\
A + C \rar[swap]{f + g} & B + D
\end{tikzcd}
\end{equation}
\end{enumerate}
\end{lemma}
\begin{proof}
For the first two parts, we reason in the theory $\OTC(\Partial(\catB))$.

\ref{anyisototal}
Let the event $f \colon A \to B$ be an isomorphism, and $\ts{f, e}$ and $\ts{f^{-1}, d}$ be tests for (unique) effects $d, e$. Then by control $\ts{f^{-1} \circ f, d \circ f, e}$ is also a test. But since $f^{-1} \circ f = \id{A}$ is causal we have $e \ovee (d \circ f) = \discard{}^\bot = 0$. Hence $e = 0$ by positivity, and so $f$ is causal. 

\ref{initobstrict}
 If $f \colon A \to 0$ is causal then $\discard{A} = \discard{0} \circ f = 0$ and so $\id{A} = 0$ by positivity. It follows that $A \simeq 0$ and, since both objects are initial, that $f$ is an isomorphism.

\ref{posopcatpull}
Both the right-hand and outer rectangles in the diagram below are pullbacks.
\[
\begin{tikzcd}
A \rar{f} \dar[swap]{\coproj_1} & B \dar{\coproj_1} \rar{!} & 1 \dar{\coproj_1} \\
A + C \rar[swap]{f + g} & B + D \rar[swap]{! + !} & 1 + 1
\end{tikzcd}
\]
By the well-known `Pullback Lemma' this means that the left-hand square is also~\cite[Lemma 5.10]{awodey2010category}. 
\end{proof}

In this setting discarding morphisms are in fact uniquely determined, rather than having to be stated as extra structure.

\begin{lemma}
Let $\catC$ be a symmetric monoidal category. Then there is at most one choice of discarding $\discard{}$ making $\catC$ a positive and complemented operational category.
\end{lemma}

\begin{proof}
Let $\discard{}$ and $\discard{}'$ be two such choices of discarding on $\catC$. 
Since all isomorphisms in $\catC$ are causal by Lemma~\ref{lemma:PosOpCat}, and coproducts are always unique up to isomorphism, any coproduct $A + B$ has causal coprojections with respect to either choice. Hence each of the theories defined by $(\catC, \discard{})$ and $(\catC, \discard{}')$ have identical partial tests and coarse-graining. Consider the (unique) effect $a$ on $A$ such that 
\[
 \discard{A} \ovee a = \discard{A}'
\]
Then $\ts{\discard{A}, a}$ is a partial test and so in the theory defined by $(\catC, \discard{})$ extends to a test $\ts{\discard{A},a,b}$. Then $a \ovee b = \discard{A}^\bot = 0$ and so by positivity $a = 0$, giving $\discard{A} = \discard{A}'$. 
\end{proof}

Finally, the other effectus axiom corresponds to one of our earlier notions.

\begin{lemma} \label{lem:effectus-interp}
Let $\Theta$ be a complemented operational theory. Then in $\Test(\Theta)$ diagrams of the form \effpullbtwo{} are pullbacks iff $\Theta$ is strongly algebraic.
\end{lemma}

\begin{proof}
Interpreted in $\Test(\Theta)$, the pullback states that any pair of partial tests $\ts{f_i \colon A \to B_i}^n_{i=1}$ and $\ts{g_j \colon A \to C_j}^m_{j=1}$ for which
\[
\Big( 
\ 
\bigovee^n_{i=1} \discard{} \circ f_i \ , \ 
\bigovee^m_{j=1} \discard{} \circ g_j
\
\Big)
\] 
forms a test have that $\ts{f_1, \dots, f_n, g_1, \dots, g_m}$ does also. By appending an extra effect to the $\ts{g_j}^m_{j=1}$ and using complementation this implies the same when replacing `test' by `partial test'. We now show that the latter condition is equivalent to strong algebraicity. 

Firstly, for any such pair of partial tests, repeatedly applying algebraicity we see that 
$\ts{\discard{} \circ f_1, \dots, \discard{} \circ f_n, \discard{} \circ g_1, \dots, \discard{} \circ g_m}$ is a test, and so by strong algebraicity $\ts{f_1, \dots, f_n, g_1, \dots, g_m}$ is one also.

 Conversely, suppose that this property holds. To see that $\Theta$ is algebraic, suppose that $\ts{f \ovee g, h_1, \dots, h_n}$ is a partial test. Then so are $\ts{f, g}$ and $\ts{h_1, \dots, h_n}$, as well as 
\[
\Big( 
\ 
\discard{} \circ f \ovee \discard{} \circ g \ , \  
\bigovee^n_{j=1} \discard{} \circ h_i
\
\Big)
\] 
Hence by assumption $\ts{f, g, h_1, \dots, h_n}$ is also a partial test.
 
To verify strong algebraicity, now suppose that $\ts{\discard{} \circ f_i}^n_{i=1}$ is a partial test. Then so is $\ts{\discard{} \circ f_1, \discard{} \circ f_2}$ and hence by assumption so is $\ts{f_1, f_2}$. Similarly since then $\ts{\discard{} \circ (f_1 \ovee f_2), \discard{} \circ f_3}$ is a partial test so is $\ts{f_1, f_2, f_3}$. Continuing in this way we get that $\ts{f_i}^n_{i=1}$ is a partial test as required.
\end{proof}

\begin{corollary} \label{cor:Effiscetainotcs}
There is a correspondence (up to equivalence) between:
\begin{itemize}
\item monoidal effectuses $\catB$; 
\item operational theories $\Theta$ which are representable, complemented, positive and have observations determining tests;
\end{itemize}
given by $\catB \mapsto \OTC(\Partial(\catB))$ and $\Theta \mapsto (\Events_{\Theta})_\causal$.
\end{corollary}

In this way we can equate monoidal effectuses with particularly well-behaved operational theories. More broadly, noting the independence of the $\otimes$ and coproducts throughout this chapter, we may think of a non-monoidal effectus as the causal part of an `operational theory without a tensor'; this is spelled out in~\cite{mainpaper}.

 The partial form $\catC := \Partial(\catB)$ of an effectus $\catB$ has been axiomatized by Cho~\cite{Partial2015Cho}, who already noted that $\ovee$ makes each homset $\catC(A,B)$ a PCM, as we discussed in Section~\ref{sec:algebraic}. Moreover, thanks to complementation each set of effects $\catC(A,I)$ in fact forms an effect algebra~\cite{foulis1994effect}, as suggested by Remark~\ref{rem:D-Test}, this being the original motivation for the effectus axioms~\cite[Prop.~4.4]{NewDirections2014aJacobs}. 
 
 Beyond their original purpose of capturing classical deterministic, probabilistic and quantum computation~\cite{NewDirections2014aJacobs}, these results show that effectus theory may be seen as a logic for computation in very general physical theories.

\begin{examples}
Since their induced theories satisfy the above properties, the test categories $\Set$, $\KlD$, $\CStarUop$, as well as $\vNUop$, are all monoidal effectuses, being the motivating examples in~\cite{EffectusIntro}. Any extensive category forms an effectus~\cite{NewDirections2014aJacobs}. $(\Mat_\mathbb{Z})_\causal$ is a test category which is not an effectus, failing to be positive.
\end{examples}

%% file: figures/ptest2i.tikz
\begin{tikzpicture}
	\begin{pgfonlayer}{nodelayer}
		\node [style=label] (0) at (1.25, 2) {$B$};
		\node [style=map] (1) at (1.25, -0) {$f_x$};
		\node [style=label] (2) at (1.25, -2) {$A$};
		\node [style=none] (3) at (1.25, -1.5) {};
		\node [style=none] (4) at (1.25, 1.5) {};
	\end{pgfonlayer}
	\begin{pgfonlayer}{edgelayer}
		\draw (3.center) to (4.center);
	\end{pgfonlayer}
\end{tikzpicture}

%% file: figures/ptest3new2.tikz
\begin{tikzpicture}
	\begin{pgfonlayer}{nodelayer}
		\node [style=label] (0) at (1.25, 2) {$B$};
		\node [style=map] (1) at (1.25, -0) {$f_x$};
		\node [style=label] (2) at (1.25, -2) {$A$};
		\node [style=none] (3) at (1.25, -1.5) {};
		\node [style=none] (4) at (1.25, 1.5) {};
	\end{pgfonlayer}
	\begin{pgfonlayer}{edgelayer}
		\draw (3.center) to (4.center);
	\end{pgfonlayer}
\end{tikzpicture}

%% file: chapter3.tex
\chapter{From Sub-causal to Super-causal Processes} \label{chap:totalisation}

From basic assumptions, we have seen how any operational physical theory defines a category with a partially defined addition $\ovee$ on its morphisms, and that this often suffices to determine the theory entirely. 
In this category we saw that every morphism $f$ was \emps{sub-causal} in the sense that $\discard{} \circ f \ovee e = \discard{}$ for some effect $e$.

 In practice, however, it is more typical and simpler to instead work with a totally defined addition $f + g$ on morphisms, and thus consider more general ones which we may call \emps{super-causal}. 
For example this occurs whenever one uses positive real numbers $\Rpos$ as weightings in place of the probabilistic interval $[0,1]$, and indeed each of our main examples from Chapter~\ref{chap:OpCategories} were first introduced in Chapter~\ref{chap:CatsWDiscarding} as the sub-causal part of such a broader category. 

 In this chapter, we connect both perspectives, constructing for any suitable category $\catC$ with a partial addition a new one $\To(\catC)$ with a total addition, of which it forms the subcategory of sub-causal morphisms. By identifying the necessary conditions for such a broader category to exist, we thus provide a clear operational interpretation to the common usage of a total addition on processes. 

 Following this, we'll see that working in the broader super-causal setting allows us to consider some powerful well-known diagrammatic features on our category.

\section{Addition and Biproducts}

\begin{definition}\label{not:addition}
Let us say that a category $\catC$ has \deff{addition}\index{category!with addition} when it is enriched in commutative monoids. That is, it has zero morphisms and each homset $\catC(A,B)$ comes with a commutative operation $+$ satisfying
\begin{align*}
f \circ (g + h) &= f \circ g + f \circ h & 
(f + g) \circ h &= f \circ h + g \circ h
&
f + 0 = f
\end{align*}
for all morphisms $f, g, h$. When $\catC$ is symmetric monoidal we also require $f \otimes 0 = 0$ and $f \otimes (g + h) = f \otimes g + f \otimes h$ for all $f, g, h$. 
\end{definition}

In a category with discarding and addition $(\catC, \discard{}, +)$ we think of $f + g$ as the coarse-graining of the processes $f$ and $g$. Previously we have described this with a partial operation $\ovee$, which we will return to shortly, and which often arose from certain coproducts in our category.

 The corresponding way to add objects together in the presence of addition is as follows. Recall that in any category a \indef{product}\index{product} of objects $A, B$ is given by an object and morphisms $(A \leftarrow A \times B \rightarrow B)$ satisfying the dual conditions to those of a coproduct.

\begin{definition} \label{def:biproduct}
~\cite{mac1978categories}\label{not:biproduct}\label{not:projection}
In any category $\catC$ with zero morphisms, a \deff{biproduct}\index{biproduct} of a pair of objects $A, B$ is another object $A \biprod B$ together with morphisms
\begin{equation} \label{eq:biprod-maps}
\scalebox{1.0}{\input{./figures/biprod-mod.tikz}}
\end{equation}
for which $(\coproj_A, \coproj_B)$ and $(\pproj_A, \pproj_B)$ make $A \biprod B$ a coproduct and product, respectively, and which satisfy the equations  
\begin{align}
\pproj_A \circ \coproj_A &= \id{A} & \pproj_A \circ \coproj_B &= 0 \label{eq:biprod-eq-row1}
\\
\pproj_B \circ \coproj_A &= 0 & \pproj_B \circ \coproj_B  &= \id{B} \label{eq:biprod-eq-row2}
\end{align}
As for coproducts, in a category with discarding we call a biproduct \deff{causal}\index{biproduct!causal} when $\coproj_A$ and $\coproj_B$ are causal. 
\end{definition}
\noindent
Note that, like the morphisms $\triangleright_i$ earlier, $\pproj_A$ and $\pproj_B$ are typically not causal. 

More generally, we may define a (causal) biproduct $A_1 \biprod \dots \biprod A_n$ of any finite set of objects similarly. A category in fact has such finite biproducts precisely when it has binary ones and a zero object. It is well-known that in the presence of addition biproducts may also be described entirely equationally, as follows.

\begin{lemma} \label{lem:equational}
In any category with addition, a collection of morphisms as in~\eqref{eq:biprod-maps} forms a biproduct iff they satisfy~\eqref{eq:biprod-eq-row1}, \eqref{eq:biprod-eq-row2} and
\begin{equation} \label{eq:extra-biprod-eqn}
\coproj_A  \circ \pproj_A + \coproj_B  \circ \pproj_B = \id{A \biprod B}
\end{equation}
\end{lemma}
\begin{proof}
For any biproduct the morphism $\coproj_A  \circ \pproj_A + \coproj_B  \circ \pproj_B$ preserves each of the (co)projections and so is indeed equal to $\id{A \biprod B}$. Conversely if this holds then for any $f \colon A \to C, g \colon B \to C$ the morphism $h = \pproj_A \circ f + \pproj_B \circ g$ has $h \circ \coproj_A = f$ and $h \circ \coproj_B = g$. This makes $(\coproj_A, \coproj_B)$ a coproduct, and $(\pproj_A, \pproj_B)$ is a product dually. 
\end{proof}

The presence of biproducts provides a way to describe addition and matrix-like features internally to a category. 
Indeed any category with finite biproducts has a unique enrichment in commutative monoids, given for morphisms $f, g \colon A \to B$ by
\begin{equation} \label{eq:additionfrombiprod}
f + g := (
\begin{tikzcd}
A \rar{\Delta} & A \biprod A \rar{[f,g]} & B
\end{tikzcd}
)
\end{equation}
where $\Delta$ is defined by $\pproj_1 \circ \Delta = \id{A} = \pproj_2 \circ \Delta$.

In a monoidal category $\catC$ we call biproducts \indef{distributive}\index{distributive!biproducts} when they are distributive as coproducts. In this case the addition moreover makes the scalars $S = \catC(I,I)$ into a commutative semi-ring, and there is a full monoidal embedding $\MatS \hookrightarrow \catC$ sending each object $n$ to 
\[\label{not:copowerOfI}
n \cdot I:= \overbrace{I \biprod \dots \biprod I}^n
\]
In any category with discarding and addition we call a morphism $f$ \index{sub-causal}\indef{sub-causal} when it satisfies
\[
\discard{} \circ f + e = \discard{}
\] 
for some effect $e$, writing $\catC_{\subcausal}$\label{not:subcausalsubcat} for the subcategory of sub-causal morphisms. 
The ability to add arbitrary morphisms $f + g$ means that categories with addition typically contain not only sub-causal morphisms, as in Chapter~\ref{chap:OpCategories}, but more general ones which we will call \indef{super-causal}\index{super-causal}. 

\subsection{Examples} \label{sec:Addition-Examples}
All of our examples of biproducts are distributive.
\begin{exampleslist}
\item 
Each category $\MatS$ has finite causal biproducts given on objects by $n \biprod m = n + m$, with the induced addition being simply point-wise addition of matrices. 
\item 
$\KlDT$ has finite causal biproducts with $A \biprod B$ given by disjoint union of the sets $A, B$, with $f + g$ being the point-wise addition of functions. In particular so does $\Class$.
\item 
$\Rel$ also has finite causal biproducts given by disjoint union of sets. Here these induce the addition $R + S := R \vee S$. More generally $\Rel(\catC)$ has finite biproducts whenever $\catC$ is regular and coherent.
\item 
$\CStarop$, $\vNop$ and $\FCStar$ all have causal finite biproducts given by the direct sum $A \biprod B$ of algebras, inducing the usual addition $f + g$ of completely positive maps.
\item 
$\Hilb$ has finite biproducts given by the direct sum $\hilbH \biprod \hilbK$ of Hilbert spaces. In contrast to the above examples whose biproducts encode coarse-graining, here the addition operation $f + g$ on linear maps describes quantum \emps{superpositions}. 
\item 
$\Quant{}$ has addition, given by the usual addition of completely positive maps, but does not have biproducts. 

For any pair of objects $\hilbH, \hilbK$ we may consider their biproduct $\hilbH \biprod \hilbK$ in $\Hilb$, which induces morphisms 
$
B(\hilbH) \leftrightarrows B(\hilbH \biprod \hilbK) \leftrightarrows B(\hilbK)
$
in $\Quant{}$. 
However, this is no longer a biproduct in $\Quant{}$, where addition is the coarse-graining of completely positive maps, rather than superposition. At the level of Kraus maps these morphisms have further properties which we study in Chapter~\ref{chap:superpositions}.
\item 
To define addition as in~\eqref{eq:additionfrombiprod} it in fact suffices to have $n$-ary \indef{bipowers}\index{biproduct!bipower}, which are biproducts of the form $n \cdot A := A \biprod \dots \biprod A$. In~\cite{gogioso2017categorical} Gogioso and Scandolo define a notion of an \emps{$R$-probabilistic theory}, for a given commutative semi-ring $R$. Equivalently this is just a symmetric monoidal category $\catC$ with discarding and finite distributive causal bipowers, with $R$ then given by the scalars $R=\catC(I,I)$. Hence super-causal processes and the mild physical assumptions which induce them, which we discuss shortly, are implicit in this approach. 
\item
Any category with addition $\catC$ embeds universally into one with biproducts $\catC^\oplus$\label{not:biproductcompletion}, its \indef{biproduct completion}, defined as follows~\cite[Ex.~VIII.2.6]{mac1978categories}:
\begin{itemize}
\item 
objects are finite lists $(A_1, \dots, A_n)$ of objects of $\catC$;
\item 
morphisms $M \colon (A_i)^n_{i=1} \to (B_j)^m_{j=1}$ are matrices $(M_{i,j} \colon A_i \to B_j)^{n,m}_{i=1,j=1}$ of morphisms from $\catC$, under matrix composition.
\end{itemize}
Biproducts here are given by concatenation of lists. Moreover when $\catC$ is symmetric monoidal with discarding so is $\catC^\oplus$, with $I = (I)$, $\otimes$ being the Kronecker product of matrices, and $\discard{(A_1, \dots, A_n)} = (\discard{A_i})^n_{i=1}$.

This construction is analogous to our earlier $\PTest(-)$ construction for those with a partial addition $\ovee$, in fact again being a special case of representability for multicategories~\cite[4.16]{pisani2014sequential}. 

In particular for any commutative semi-ring we have $\Mat_S \simeq S^\oplus$. Moreover 
\[
\Quant{}^{\oplus} \simeq \FCStar 
\]
since, as remarked in Example~\ref{sec:OPCATEXAMPLES}~\ref{ex:QuantOP}, finite-dimensional C*-algebras all have the form $\bigoplus^n_{i=1} B(\hilbH_i)$ for some finite-dimensional $\hilbH_i$. In fact one may recover $\FCStar$ from $\Quant{}$ without mentioning addition, using a construction on its `idempotents', as the author explored with Coecke and Selby in~\cite{coecke2017two}. 
\end{exampleslist}

\section{Totalisation} \label{sec:totalisation}

The results of this section are in collaboration with Kenta Cho.

\subsection{Sub-causal categories}  \label{sec:Subcausalcats}

We now wish to understand how such a total addition on morphisms arises from basic operational assumptions. Earlier, under the mild assumptions of Section~\ref{sec:algebraic}, we saw that a physical theory may be fully described by a category instead with a partial addition on morphisms, satisfying the following. 

Recall that a partial commutative monoid (PCM) is a set $M$ together with a suitably associative and commutative partial binary operation $\ovee$ with a unit element $0$. We write $x \bot y$ \label{not:compatPCM} whenever $x \ovee y$ is defined.

\begin{definition} \label{def:PCM-with-discarding}
A \deff{sub-causal category}\index{sub-causal category} $\catC$ is a category with discarding for which:
\begin{enumerate}[label=\arabic*., ref=\arabic*]
\item $\catC$ is \emps{enriched in PCMs}, meaning that it has zero morphisms, that 
each homset $\catC(A,B)$ forms a PCM with unit $0$, and that whenever $f \bot g$ we have 
\label{not:PCMenriched}
\begin{align*}
h \circ (f \ovee g) = h \circ f \ovee h \circ g
\\ 
(f \ovee g) \circ k = f \circ k \ovee g \circ k
\end{align*}
and also $(f \ovee g) \otimes k = f \otimes k \ovee g \otimes k$ when $\catC$ is symmetric monoidal;
\item 
Every morphism is sub-causal; 
\item \label{enum:disc-compat}
For all $f, g \colon A \to B$ we have $\discard{B} \circ f \ \bot \ \discard{B} \circ g \ \implies \ f \ \bot \  g$.
\end{enumerate}
\end{definition}

Here by \indef{sub-causality}\index{sub-causal} of a morphism $f$ in terms of a partial operation $\ovee$ we as expected mean that $\discard{} \circ f \ovee e = \discard{}$ for some effect $e$. All of the categories of events of the operational theories we met in Chapter~\ref{chap:OpCategories} are sub-causal, with our terminology justified by the following.

\begin{lemma} \label{lem:sub-causal}
Let $\catC$ be a (symmetric monoidal) category with discarding and addition. Then $\catC_{\subcausal}$ is a (symmetric monoidal) sub-causal category. 
\end{lemma}
\begin{proof}
We first check that $\catC_{\subcausal}$ is indeed a monoidal subcategory of $\catC$. 
If $f, g$ are sub-causal via the effects $d, e$ respectively, then 
\[
\discard{} \circ (g \circ f) + (d \circ f + e) = (\discard{} \circ g + d) \circ f + e = \discard{} \circ f + e = \discard{}
\]
and so $g \circ f$ is sub-causal. Similarly if $f \colon A \to C$ and $g \colon B \to D$ are sub-causal via effects $d, e$ then $f \otimes g$ is sub-causal since we have
\[
\scalebox{0.8}{\input{./figures/subcausal-tensor.tikz}}
\]
By Lemma~\ref{lem:coherence-isoms-causal} all coherence isomorphisms in $\catC$ are causal and so restrict to $\catC_{\subcausal}$. Each effect $\discard{A}$ and all zero morphisms are clearly sub-causal. 

 Next, in $\catC_{\subcausal}$ we set $f \ovee g$ to be defined and equal to $f + g$ whenever this morphism is sub-causal. To see that this makes each homset a PCM, we just need to check that if $f, g, h$, $f + g$ and $(f + g) + h$ are all sub-causal, then so is $g + h$. But this is immediate by associativity. Finally condition~\ref{enum:disc-compat} holds in $\catC_{\subcausal}$ since if $\discard{} \circ f + \discard{} \circ g$ is sub-causal then clearly so is $f + g$. 
 \end{proof}

\subsection{The $\To(\catC)$ construction}

We now wish to provide a converse result, showing that every sub-causal category arises from one with a total addition.  

Our approach is based on the following construction due to Jacobs and Mandemaker, allowing one to extend any PCM to a commutative monoid~\cite{Jacobs2012a}. For any set $A$ we write $\multiset(A)$ for the free commutative monoid on $A$. Its elements are finite formal sums $\sum^n_{i=1} a_i = a_1 + \dots + a_n$ of elements of $A$. The monoid operation $+$ is formal addition of sums and $0_{\multiset(A)}$ is the empty sum. 

\begin{definition}[Totalisation]~\cite{Jacobs2012a}
Let $(M, \ovee, 0_M)$ be a PCM. We define a commutative monoid 
\[
\To(M) := \multiset(M) / \! \sim
\]
where $\sim$ is the smallest monoid congruence such that $0_M \sim 0_{\To(M)}$ and for all $x, y \in M$ we have $x + y \sim x \ovee y$ whenever $x \bot y$ in $M$.  

Now $M$ embeds faithfully into $\To(M)$ as $\{ [x] \mid x \in M \}$. This makes $M$ a \indef{downset}\index{downset} of $\To(M)$, meaning that if $a, b \in \To(M)$ with $a + b \in M$ then $a, b \in M$. 
\end{definition}

Totalisation is characterised by a universal property. Recall that a \indef{coreflection}\index{coreflection} is an adjunction $F \dashv G$ for which the left adjoint $F$ is full and faithful, or equivalently the unit $\eta \colon \id{} \to G \circ F$ is an isomorphism. Write $\DCM$\label{cat:DCM} for the category of commutative monoids with a specified downset and $\PCM$\label{cat:PCM} for the category of PCMs, with suitable morphisms in each case. The assignment $M \mapsto \To(M)$ is left adjoint to the functor $\DCM \to \PCM$ which takes downsets, and moreover this adjunction is a coreflection~\cite[Theorem 4.1]{Jacobs2012a}. 

The following fact will be useful.

\begin{lemma}~\cite[p.~93]{ExpMon2012Ja} \label{lem:totalisation-fact}
Let $M$ be a PCM. If $[a_1 + \dots + a_n]= [b]$ in $\To(M)$, for $a_i, b \in M$, then $\bigovee^n_{i=1} a_i$ is defined in $M$ and equal to $b$. 
\end{lemma}

A motivating example is the passage from probabilities to `unnormalised' ones.

\begin{example}
Let $[0,1]$ be the unit interval, considered as a PCM with $p \ovee q$ defined and equal to $p + q$ whenever this is $\leq 1$. Then $\To([0,1]) \simeq \mathbb{R}^+$, the monoid of positive real numbers under addition.
\end{example}

Let us now extend totalisation to the level of categories. 

\begin{definition}[Totalisation of a category] \label{def:tot-cat}
Let $\catC$ be a category enriched in PCMs. We define the category $\To(\catC)$ to have the same objects as $\catC$, with  \label{not:totalisationcat}
\[
\To(\catC)(A,B) = \To(\catC(A,B))
\]
That is, morphisms $A \to B$ are $\sim$-equivalence classes $[ \sum^n_{i=1} f_i]$ for morphisms $f_i \colon A \to B$ in $\catC$. Composition is given by  
\[
 [\sum^n_{i=1} g_i] \circ [\sum^n_{i=1} f_j] = [\sum^n_{i=1} g_i \circ f_j] 
 \qquad
\]
and we set $\id{A} = [\id{A}]$. Then $\To(\catC)$ is enriched in commutative monoids, with $+$ defined in $\To(\catC(A,B))$ as before. When $\catC$ is symmetric monoidal, we define a symmetric monoidal structure on $\To(\catC)$ by setting $A \otimes B$ to be as in $\catC$ and
\begin{equation*}
 [\sum^n_{i=1} f_i] \otimes [\sum^m_{j=1} g_j] = [\sum^{n}_{i=1} \sum^m_{j=1} f_i \otimes g_j]
\end{equation*}
with unit object $I$ and coherence isomorphisms inherited from $\catC$, i.e.~$\alpha_{A,B,C}=[\alpha_{A,B,C}]$, $\lambda_A = [\lambda_A]$, $\rho_A = [\rho_A]$ and $\sigma_{A,B} = [\sigma_{A,B}]$. When $\catC$ has discarding so does $\To(\catC)$ via $\discard{A} := [\discard{A}]$.
\end{definition}

By a \indef{morphism}\index{morphism!of sub-causal categories} of sub-causal categories, we mean one $F \colon \catC \to \catD$ of categories with discarding such that $F(0) = 0$ and whenever $f \bot g$ we have $F(f) \bot F(g)$ with $F(f \ovee g) = F(f) \ovee F(g)$.

\begin{theorem} \label{thm:Total-subcausal}
Let $\catC$ be a (symmetric monoidal) sub-causal category. Then $\To(\catC)$ is a well-defined (symmetric monoidal) category with discarding and addition, and there is a (symmetric monoidal) isomorphism of sub-causal categories $\catC \simeq \To(\catC)_\subcausal$.
\end{theorem}
\begin{proof}
One may verify directly that these definitions of $\circ, \otimes$ and $+$ make $\To(\catC)$ a well-defined (symmetric monoidal) category with addition. 

 Alternatively, this in fact holds for entirely abstract reasons. By results of Jacobs and Mandemaker~\cite[Theorems 3.2, 4.1]{Jacobs2012a} totalisation defines a strong monoidal functor $\To \colon \PCM \to \DCM$, which is easily seen to be symmetric monoidal also, and hence in particular defines such a functor from $\PCM$ to the category $\CMon$\label{cat:CMon} of commutative monoids. By the `change of base' for enriched categories~\cite{eilenberg1966closed}, this means that it sends categories (monoidally) enriched in $\PCM$ to categories (monoidally) enriched in $\CMon$~\cite[Theorem 5.7.1]{cruttwell2008normed}.

When $\catC$ has discarding it's immediate that this lifts to $\To(\catC)$ as above. From the definition we see that there is always a faithful (symmetric monoidal) identity-on-objects functor $\catC \hookrightarrow \To(\catC)$ given by $f \mapsto [f]$. By sub-causality in $\catC$, each morphism $[f]$ is sub-causal in $\To(\catC)$. 

Conversely, let $f = [f_1 + \dots + f_n] \colon A \to B$ be sub-causal in $\To(\catC)$, via some effect $e = [e_1 + \dots + e_m]$ on $A$. Then 
\[
[\discard{A}] = \discard{A} = \discard{B} \circ f + e = [\sum^n_{i=1} \discard{B} \circ f_i + \sum^m_{j=1} e_j]
\] 
and so by Lemma~\ref{lem:totalisation-fact} we have $(\bigovee^n_{i=1} f_i) \ovee (\bigovee^m_{j=1} e_j) = \discard{}$ in $\catC$. In particular $g = \bigovee^n_{i=1} f_i$ is defined in $\catC$ and so
\[
f = [\sum^n_{i=1} f_i] = [\bigovee^n_{i=1} f_i] = [g]
\]
so that $f \in \catC$. Hence the inclusion $\catC \hookrightarrow \To(\catC)_\subcausal$ is an isomorphism of categories, and it always preserves $\ovee$. Finally we need that if $f \bot g$ in $\To(\catC)_\subcausal$ then $f \bot g$ in $\catC$. But if $f + g$ is sub-causal in $\To(\catC)$, say $\discard{} \circ f + \discard{} \circ g + e = \discard{}$, then by Lemma~\ref{lem:totalisation-fact}  $\discard{} \circ f \bot \discard{} \circ g$ in $\catC$ and so $f \bot g$ in $\catC$ also.
\end{proof}

The universal property of $\To$ lifts to the level of categories.
Let $\PCMD$\label{cat:PCMD} be the category of sub-causal categories and morphisms between them. Let $\CMonD$\label{cat:CMonD} be the category of categories with addition and discarding, with morphisms being functors $F \colon \catC \to \catD$ which preserve discarding and satisfy $F(0)=0$ and $F(f+g) = F(f) + F(g)$ for all $f, g$. There is a functor $(-)_\subcausal \colon \CMonD \to \PCMD$ sending $\catC$ to $\catC_\subcausal$.

\begin{theorem} \label{thm:total-univ}
Totalisation defines a left adjoint to $(-)_\subcausal$, giving a coreflection
\[
\begin{tikzcd}[bend angle = 15, row sep = large]
 \PCMD \arrow[rr, bend left, "\To(-)"] & \bot &  \CMonD \arrow[ll, bend left, "(-)_{\subcausal}"]
 \end{tikzcd}
\]
\end{theorem}
\omitthis{
\noindent
Explicitly, this means that any morphism $F \colon \catC \to \catD_{\subcausal}$ in $\PCMD$ extends uniquely to some $\widehat{F} \colon \To(\catC) \to \catD$ in $\CMonD$ such that the following commutes:
\begin{center}
\begin{tabular}{c}
 \end{tabular}
    \begin{tikzcd}
    \catC \rar{F} \dar[swap]{\simeq} & \catD_{\subcausal} \\
    \To(\catC)_{\subcausal} \arrow[ur,swap, dotted,"\hat{F}_{\subcausal}"]
    \end{tikzcd}
    \begin{tabular}{c}
    \end{tabular}
\end{center}
}
\begin{proof}
Let $\catC$ and $\catD$ be objects of $\PCMD$ and $\CMonD$ respectively, and $F \colon \catC \to \catD_{\subcausal}$ a morphism in $\PCMD$. We need to show that $F$ extends to a unique morphism $\widehat{F} \colon \To(\catC) \to \catD$ in $\CMonD$. Now $F$ defines a family of PCM-homomorphisms 
\[
F_{A,B} \colon \catC(A,B) \to \catD_{\subcausal}(F(A),F(B))
\] 
with each $\catD_{\subcausal}(F(A),F(B))$ forming a downset of $\catD(F(A),F(B))$. By the universal property of $\To(-)$, these each have a unique extension to a monoid homomorphism $\hat{F}_{A,B} \colon \To(\catC(A,B)) \to \catG(F(A),F(B))$ given by $\widehat{F}_{A,B}([\sum^n_{i=1} f_i]) = \sum^n_{i=1} F_{A,B}(f_i)$. It's straightforward to check that this makes $\hat{F}$ a morphism in $\CMonD$.
For each sub-causal category $\catC$, the unit $\eta_\catC \colon \catC \to \To(\catC)_\subcausal$ is precisely the isomorphism of Theorem \ref{thm:Total-subcausal}, making this a coreflection. 
\end{proof}

\subsection{Examples}

Let us now see how each of the categories we met in Chapter~\ref{chap:OpCategories} form the sub-causal part of a category with addition.
\begin{center}
\begin{tabular}{c|c}
Sub-causal category $\catC$ & Category with addition  $\To(\catC)$ \\ 
\hline
$\PFun$ & $\KlMn$ \\ 
$\KlSD$ & $\KlDT$ \\ 
$\QuantSU$ & $\Quant{}$ \\ 
$\CStarSUop$ & $\CStarop$ \\ 
$\Rel$ & $\Rel$ \\ 
\end{tabular}
\end{center}
\begin{exampleslist}
\item 
By definition the totalisation of $\PFun$ is the Kleisli category $\KlMn$\label{cat:KlMn} of the \emps{finite multiset} monad. More precisely objects are sets and morphisms $f \colon A \to B$ are functions sending each $a \in A$ to a finite multiset of elements of $b \in B$, with $I=\{\star\}$ and $\discard{A}$ being simply that of $\PFun$. 
\item 
Each of our probabilistic examples have totalisations given by extending their scalars from $[0,1]$ to $\mathbb{R}$, as we prove shortly in Section~\ref{sec:divisscalars}. 
\item 
$\Rel$ has that every morphism is sub-causal with respect to its total addition $R \vee S$, as does $\Rel(\catC)$ whenever $\catC$ is coherent. Hence $\To(\Rel(\catC)) \simeq \Rel(\catC)$.
\item 
Each category $\Mat_{S^{\leq 1}}$ arises as the sub-causal morphisms of $\Mat_S$; however in general $\To(S^{\leq 1}) \neq S$ and so $\To(\Mat_{S \leq 1}) \not \simeq \Mat_S$.
\end{exampleslist}

\subsection{Totalisation for effectuses}

Let us now make the connection between biproducts and the kinds of coproducts we met in Chapter~\ref{chap:OpCategories}, such as those of an effectus, more precise. In~\cite{Partial2015Cho}, Cho defines a \emps{finite partially additive category (FinPAC)} to be a category enriched in PCMs with finite coproducts $(+,0)$ for which the maps
\[
\begin{tikzcd}
A + A 
\arrow[r,shift left=2.5,"\triangleright_1"]
\arrow[r,shift left=-2.5,"\triangleright_2",swap]
& A
\end{tikzcd}
\]
are jointly monic, and which induce each operation $\ovee$ just as in Section~\ref{sec:CatsToTheories}.

\begin{lemma} \label{lem:butterfliesAndBiprod}
If $\catC$ is a FinPAC then $\To(\catC)$ has finite biproducts. Conversely, if $\catD$ is a category with discarding and causal biproducts then $\catD_\subcausal$ is a FinPAC with finite causal coproducts.
\end{lemma}
\begin{proof}
For the first statement, we claim that each object $A + B$ in $\catC$ forms a biproduct in $\To(\catC)$. Indeed, each morphism
 \[
\begin{tikzcd}[column sep = large]
A + B \rar{\coproj_A + \coproj_B} & (A + B) + (A + B)
\end{tikzcd}
\ \ \  
\text{has}
\ \ \ 
\begin{cases}
\triangleright_1 \circ (\coproj_A + \coproj_B) = \coproj_A \circ \triangleright_A \\ 
\triangleright_2 \circ (\coproj_A + \coproj_B) = \coproj_B \circ \triangleright_B
\end{cases}
\]
and so the definition of $\ovee$ in terms of coproducts gives that 
\[
\id{A + B} = (\coproj_A \circ \triangleright_A) \ovee (\coproj_B \circ \triangleright_B)
\]
\noindent
Hence since the inclusion $\catC \hookrightarrow \To(\catC)$ preserves $\ovee$, the morphisms $\coproj_A, \coproj_B$, $\pproj_A := \triangleright_A$ and $\pproj_B := \triangleright_B$ satisfy~\eqref{eq:extra-biprod-eqn} and so form a biproduct.

For the second statement, note that by~\eqref{eq:extra-biprod-eqn} any causal biproduct $A \biprod B$ in $\catD$ has that $\triangleright_A = \pproj_A$ and $\triangleright_B = \pproj_B$ are sub-causal, and they will remain jointly monic in $\catD_\subcausal$. Moreover the $\coproj_A, \coproj_B$ again form a coproduct in $\catD_\subcausal$, since $[f,g]$ is sub-causal whenever $f$ and $g$ are.
\end{proof}

Our main examples of such categories arise from the `partial form' of an effectus (see Section~\ref{sec:EffectusTheory}), which may be defined as follows~\cite{EffectusIntro}.

\begin{definition} \label{def:EfInParForm}
An \deff{effectus in partial form} or \deff{FinPAC with effects}\index{effectus!in partial form} \index{FinPAC with effects|see {effectus in partial form}} is a sub-causal category $(\catC, \discard{})$ which is a FinPAC, whose coproducts are causal and which satisfies:
\begin{enumerate}[label=\arabic*., ref=\arabic*]
\item \label{enum:subcausal-cancel} 
$a \ovee b = \discard{} = a \ovee c \implies b = c$ for all effects $a, b, c$;
\item \label{enum:discard}
$\discard{} \circ f = 0 \implies f = 0$ for all morphisms $f$. 
\end{enumerate}
\end{definition}

There is also a `totalised' version of an effectus. In~\cite{EffectusIntro} a \indef{grounded biproduct category} is defined to be a category $\catD$ with discarding and finite causal biproducts satisfying the analogous first condition
\[
a + b = \discard{} = a + c \implies b = c
\]
as well as \ref{enum:discard} above. 
Immediately we have a result from~\cite{EffectusIntro}.

\begin{lemma} \label{lem:gBCtoEff}
Let $\catD$ be a grounded biproduct category. Then $\catD_{\subcausal}$ is an effectus in partial form. Hence $\catD_\causal$ is an effectus. 
\end{lemma}
\begin{proof}
Lemmas~\ref{lem:sub-causal} and \ref{lem:butterfliesAndBiprod}.
\end{proof}

We can now show that every effectus arises in this way. This connects effectus theory, which studies sub-causal morphisms, with categorical quantum mechanics~\cite{abramskycoecke:categoricalsemantics}\index{categorical quantum mechanics}, which studies super-causal ones.

\begin{corollary} \label{cor:GBiprod}
Let $\catC$ be an effectus in partial form. Then $\To(\catC)$ is a grounded biproduct category with $\catC \simeq \To(\catC)_\subcausal$. 
\end{corollary}
\begin{proof}
By Theorem~\ref{thm:Total-subcausal} and Lemma~\ref{lem:butterfliesAndBiprod}, $\To(\catC)$ has biproducts and the above isomorphism holds. Note that in any effectus in partial form we have 
\[
f \ovee g = 0 \implies f = g= 0
\] 
for all morphisms $f, g \colon A \to B$ as we saw in Proposition~\ref{prop:positivity}. Using this and Lemma~\ref{lem:totalisation-fact}, both properties~\ref{enum:subcausal-cancel} and \ref{enum:discard} immediately lift from $\catC$ to $\To(\catC)$.
\end{proof}

\begin{examples}
$\Set$, $\KlD$, and $\CStarUop$ are all effectuses, and we've seen that their partial forms have totalisations $\KlMn$, $\KlDT$, $\CStarop$ respectively. 
\end{examples}

\subsection{Totalisation with divisible scalars} \label{sec:divisscalars}

In settings such as $\Quant{}$ and $\KlDT$ it is more common to view a general morphism as a multiple $r \cdot f$ of a sub-causal one, rather than as a finite sum $\sum^n_{i=1} f_i$ of them. In these settings, there is no loss of information in working with either sub-causal morphisms or more general ones. 

These facts can be generalised to categories with the following feature.

\begin{definition} \label{def:divis}
We say that a sub-causal category $\catC$ has \deff{naturally divisible scalars}\index{naturally divisible scalars} when for every $n \in \mathbb{N}_{>0}$ there exists a scalar $\frac{1}{n}$ with $\bigovee^n_{i=1} \frac{1}{n} =  \id{I}$. 
\end{definition}

Let us call a category with discarding and addition $(\catD, \discard{}, +)$ \indef{causally generated}\index{causally generated} when every morphism may be written as a finite sum $f = \sum^n_{i=1} f_i$ for which each $f_i$ is sub-causal. By construction $\To(\catC)$ is causally generated when $\catC=\catC_\subcausal$. 

 We write $\PCMD_{\mathsf{n.d}}$ and $\CMonD^{\mathsf{c.g}}_{\mathsf{{n.d}}}$ for the full subcategories of $\PCMD$ and $\CMonD$, respectively, given by those categories with naturally divisible scalars in each case, and which in the latter case are causally generated.

\begin{theorem} \label{thm:Total-for-good-scalars}
Totalisation restricts to an equivalence of categories 
\[
\begin{tikzcd}[bend angle = 15, row sep = large]
 \PCMD_{\mathsf{n.d}} \arrow[rr, bend left, "\To(-)"] & \simeq &  \CMonD^{\mathsf{c.g}}_{\mathsf{{n.d}}} \arrow[ll, bend left, "(-)_{\subcausal}"]
 \end{tikzcd}
\]
Hence if $(\catD, \discard{}, +)$ has naturally divisible scalars and is causally generated then there is an isomorphism $\catD \simeq \To(\catD_\subcausal)$. 
\end{theorem}

\begin{proof}

It is clear that the coreflection of Theorem~\ref{thm:total-univ} restricts as above, and so it suffices to show that the counit $\varepsilon_\catD \colon \To(\catD_{\subcausal}) \to \catD$ given by $[\sum^n_{i=1} f_i] \mapsto \sum^n_{i=1} f_i$ is an isomorphism of categories. By definition $\varepsilon$ is surjective on objects, and it is full since $\catD$ is causally generated. 

We now show $\varepsilon$ is faithful. Suppose that $\sum^n_{i=1} f_i = \sum^m_{i=1} g_j$ with the $f_i$ and $g_j$ all sub-causal. Then $\frac{1}{n+m}$ is sub-causal and hence 
\[
\sum^n_{i=1} \frac{1}{n+m} \cdot f_i  = \sum^m_{j=1} \frac{1}{n+m} \cdot g_j
\]
is sub-causal also. This gives that
\[
[\sum^n_{i=1} f_i] 
= 
\sum^{n + m}_{k=1}[ \sum^n_{i=1} \frac{1}{n+m} \cdot f_i] 
=  
\sum^{n + m}_{k=1}[ \sum^m_{j=1} \frac{1}{n+m}\cdot g_j] 
= 
[\sum^m_{j=1} g_j]
\]
as required.
\end{proof}

As a result in this setting we may work with either sub-causal or more general morphisms, at no extra cost. We also have an alternative description of the $\To$ construction. 

\begin{theorem} \label{thm:Total-Alternative}
Let $\catC$ be a symmetric monoidal sub-causal category with naturally divisible scalars $M = \catC(I,I)$, and set $R := \To(M)$. Then $\To(\catC)$ is isomorphic to the category $\Ro(\catC)$ whose objects are the same as $\catC$ and morphisms $A \to B$ are equivalence classes of pairs $(f, r)$ for $f \colon A \to B$ in $\catC$ and $r \in R$, under 
\[
(f ,r) \sim (g,s)
\] 
whenever $a \cdot f = b \cdot g$ for some $a, b \in M$ such that $n \cdot a = r$ and $n \cdot b = s$ in $R$ for some $n \in \mathbb{N}$. 
Here we set 
\[
[(g, s)] \circ [(f, r)] := [(g \circ f, s \cdot r)] 
\qquad
\id{A} := [(\id{A},\id{I})]
\]
\end{theorem}

\begin{proof}
Define $F \colon \Ro(\catC) \to \To(\catC)$ by $A \mapsto A$ and $[(f,r)] \mapsto [r] \cdot [f]$. This is well-defined and faithful since 
\begin{align*}
(f, r) \sim (g, s) 
 &\iff 
a \cdot f = b \cdot g \\ 
 &\iff 
[a] \cdot [f] = [b] \cdot [g]
\iff 
[r] \cdot [f] = [s] \cdot [g]
\end{align*}
where for some $n \in \mathbb{N}$ we have $n \cdot a =r$ and $n \cdot b= s$ in $R$. 

Now given any morphism $f=\sum^n_{i=1} f_i$ in $\To(\catC)$ with each $f_i$ sub-causal, the morphism $g := \frac{1}{n} \cdot f$ has  
$g=\sum^n_{i=1} (\frac{1}{n} \cdot f_i) = \bigovee^n_{i=1} (\frac{1}{n} \cdot f_i)$
and so is sub-causal with $n \cdot g = f$ in $\To(\catC)$. Hence $F([(g,n)]) = f$, making $F$ full. Finally, $F$ respects composition since
 \begin{align*}
F([(g,s)]) \circ F([(f,r)])
&=
([s] \cdot [g]) \circ ([r] \cdot [f])
\\
&=
([s \circ r]) \cdot [g \circ f]
=
F([(g,s)] \circ [(f,r)])
 \end{align*}
It follows that $\Ro(\catC)$ is a well-defined category and $F$ is an isomorphism.
\end{proof}

\begin{examples} 
$\KlSD$ and $\CStarSUop$ both have naturally divisible scalars $[0,1]$ with $\To([0,1]) = \mathbb{R}^+$. In their totalisations $\KlDT$ and $\CStarop$ morphisms may thus be viewed as a multiples $r \cdot f$ of sub-causal (i.e.~sub-unital) ones, for some $r \in \mathbb{R}$, as is standard.
\end{examples}

\section{Compact and Dagger Categories}

\subsection{Compact categories}

Working with a category whose morphisms are super-causal processes, rather than merely sub-causal ones, allows us to make use of some powerful extra categorical features. In particular, the field of categorical quantum mechanics\index{categorical quantum mechanics} has emphasised the study of categories with the following diagrammatic property~\cite{abramskycoecke:categoricalsemantics}.

\label{not:dualobject}
Let $A$ be any object in a monoidal category. We say that an object $A^*$ is (right) \indef{dual}\index{dual objects} to $A$ when there exists a state $\eta \colon I \to A^* \otimes A$ and effect $\varepsilon \colon  A \otimes A^* \to I$ satisfying the \indef{snake equations}\index{snake equations}:
\[
\scalebox{0.8}{\input{./figures/snake.tikz}}
\]
We may have similarly considered left duals for objects; however in a symmetric monoidal category any left dual is a right dual and vice versa, and from now one we will ignore either prefix. Dual objects are unique up to unique isomorphism when they exist, and so we speak of `the' dual $A^*$ of an object $A$.

\begin{definition}~\cite{kelly1972many,kelly1980coherence}
A symmetric monoidal category $(\catC, \otimes)$ is \deff{compact closed} or \deff{compact}\index{compact category} when every object in $\catC$ has a dual. 
\end{definition}

There is a helpful graphical notation for compact categories~\cite{selinger2011survey}. Firstly, we distinguish between an object $A$ and its dual $A^*$ by drawing their identity morphisms as upward and downward directed wires, respectively:
\[
\scalebox{0.8}{\input{./figures/id-A-and-dual.tikz}}
\]
and depict the $\eta$ as a `cup' and $\varepsilon$ as a `cap': \label{not:cupcap}
\[
\scalebox{0.8}{\input{./figures/cup.tikz}}
\]
Then the snake equations become simply `yanking wires':
\[
\scalebox{0.8}{\input{./figures/yank.tikz}}
\]
In this way compactness can be seen as a relaxation on our graphical rules, by allowing us to `bend wires' and so exchange inputs and outputs in our diagrams. 
The following shows that it may generally only be considered outside of the sub-causal setting. 



\begin{lemma}
Let $\catC$ be a monoidal effectus in partial form which is compact closed. Then $\catC$ is trivial, i.e.~satisfies $A \simeq 0$ for all objects $A$.
\end{lemma}
\begin{proof}
By a result of Houston, any compact closed category with finite coproducts has finite biproducts~\cite{houston2008finite}. It follows that every coproduct in $\catC$ is a biproduct, or equivalently that $\ovee$ is in fact total. But then $\id{I} \ovee \id{I}$ is defined, and so $\id{I} = {\id{I}}^\bot = 0$, giving $\id{A} = \id{A} \cdot \id{I} = 0$ for all objects $A$.
\end{proof}

To give an operational interpretation to compactness, we should relate it to a condition on sub-causal processes.

\begin{theorem} \label{thm:CompactInterp}
Let $\catC$ be a symmetric monoidal sub-causal category. The following are equivalent:
\begin{enumerate}[label=\arabic*., ref=\arabic*]
\item \label{enum:total-compact}
$\To(\catC)$ is compact;
\item \label{enum:some-compact}
$\catC \simeq \catD_\subcausal$ for some compact, causally generated category with discarding and addition $\catD$;
\item \label{enum:summed-snake}
For every object $A$ there exists an object $A^*$ and collections of states $(\eta_i)^n_{i=1}$ of $A^* \otimes A$ and effects $(\varepsilon_j)^m_{j=1}$ on $A^* \otimes A$ satisfying
\begin{equation} \label{eq:summed-snake}
\mathlarger{\bigovee}^{n,m}_{i,j=1}
\ 
\scalebox{0.8}{\input{./figures/snake-sum1.tikz}}
\qquad
\text{and}
\qquad
\mathlarger{\bigovee}^{n,m}_{i,j=1}
\ 
\scalebox{0.8}{\input{./figures/snake-sum2.tikz}}
\end{equation}
\end{enumerate}
Moreover, if $\catC$ has naturally divisible scalars these hold iff for every object $A$ there exists an object $A^*$,  a state $\eta$ of $A^* \otimes A$ and an effect $\varepsilon$ on $A \otimes A^*$ satisfying
\begin{equation} \label{eq:scaled-snake}
\scalebox{0.8}{\input{./figures/scaled-snake-1.tikz}}
\qquad
\text{and}
\qquad
\scalebox{0.8}{\input{./figures/scaled-snake-2.tikz}}
\end{equation}
for some $n \in \mathbb{N}$.
\end{theorem}
\begin{proof}
\ref{enum:total-compact} $\implies$ \ref{enum:some-compact}: Always we have that $\catC \simeq \To(\catC)_\subcausal$ and $\To(\catC)$ is causally generated. 

\ref{enum:some-compact} $\implies$ \ref{enum:summed-snake}: 
Since $\catD$ is compact, every object has a dual $A^*$ via some state $\eta$ and effect $\varepsilon$ satisfying the snake equations. Now since $\catD$ is causally generated we have $\eta = \sum^n_{i=1} \eta_i$ and $\varepsilon = \sum^m_{j=1} \varepsilon_j$ for some collections $(\eta_i)^n_{i=1}$ and $(\varepsilon_j)^m_{j=1}$ as above. The snake equations then amounts to 
\[
\sum^{n,m}_{i,j=1}
\ 
\scalebox{0.8}{\input{./figures/snake-sum1.tikz}}
\qquad
\text{and}
\qquad
\sum^{n,m}_{i,j=1}
\ 
\scalebox{0.8}{\input{./figures/snake-sum2.tikz}}
\]
Since $\id{A}$ and $\id{A^*}$ are sub-causal, so are all of the terms in the above sums. Hence each sum restricts to one in terms of $\ovee$ in $\catD_\subcausal$, and so we are done since $\catC \simeq \catD_\subcausal$.

\ref{enum:summed-snake} $\implies$ \ref{enum:total-compact}: In $\To(\catC)$ for each object $A$ the state and effect
\[
\scalebox{0.8}{\input{./figures/cup-sum1.tikz}}
\
:=
\
\sum^n_{i=1} 
\
\scalebox{0.8}{\input{./figures/cup-sum2.tikz}}
\qquad
\qquad
\scalebox{0.8}{\input{./figures/cup-sum3.tikz}}
\
:=
\
\sum^m_{j=1} 
\
\scalebox{0.8}{\input{./figures/cup-sum4.tikz}}
\]
satisfy the snake equations thanks to~\eqref{eq:summed-snake}.

Now suppose that $\catC$ has naturally divisible scalars. If~\eqref{eq:summed-snake} holds define
\[
\scalebox{0.8}{\input{./figures/cup-sum1.tikz}}
:=
\bigovee^n_{i=1} 
\
\scalebox{0.8}{\input{./figures/cup-sum2i.tikz}}
\qquad
\qquad
\scalebox{0.8}{\input{./figures/cup-sum3.tikz}}
:=
\bigovee^m_{j=1} 
\
\scalebox{0.8}{\input{./figures/cup-sum4i.tikz}}
\]
Then $\eta$ and $\varepsilon$ satisfy~\eqref{eq:scaled-snake} after replacing $n$ by $n \cdot m$. Conversely if~\eqref{eq:scaled-snake} holds then~\eqref{eq:summed-snake} is satisfied by setting $\eta_i = \eta$ for $i = 1, \dots, n$ and $m=1$ with $\varepsilon_1 = \varepsilon$. 
\end{proof}

\begin{remark}
Each equation in~\eqref{eq:scaled-snake} can be seen as a \emps{probabilistic teleportation} protocol. For example, in the left-hand equation, Alice and Bob share an entangled state $\eta$. With probability $\frac{1}{n}$, Alice can measure its corresponding effect $\varepsilon$ and thus transmit her system to Bob. Similarly, as is well-known, the snake equations can be seen to describe \emps{superselected} teleportation~\cite{abramskycoecke:categoricalsemantics}. 
\end{remark}

\subsection{Dagger categories} \label{subsec:daggercats}

Working beyond merely sub-causal processes also allows us to consider the presence of an extra structure  which lets us `reverse' any morphism in our category.

\begin{definition}~\cite{Selinger2007139} \label{not:daggercat}
A \deff{dagger category}\index{dagger category} $(\catC, \dagger)$ is a category $\catC$ together with an identity-on-objects contravariant involutive endofunctor $(-)^\dagger$. Explicitly, for every morphism $f \colon A \to B$ in $\catC$ there is a morphism $f^\dagger \colon B \to A$ such that 
\[
f^{\dagger \dagger} = f
\qquad
(g \circ f)^\dagger = f^\dagger \circ g^\dagger
\qquad
(\id{A})^\dagger = \id{A}
\qquad
\]
for all morphisms $f, g$ and objects $A$.
\end{definition}

Dagger categories come with their own graphical calculus~\cite{selinger2011survey}. When working in a dagger category, we depict morphisms $f \colon A \to B$ with pointed boxes:
\[
\scalebox{0.8}{\input{./figures/dagger-morphism.tikz}}
\]
and the dagger is represented by turning pictures upside-down:
\begin{equation*} 
\scalebox{0.8}{\input{./figures/dagger-i.tikz}}
\end{equation*}
In this setting, monoidal or compact structure should respect the dagger as follows. In a dagger category, a \indef{unitary}\index{unitary} is an isomorphism $U$ with $U^\dagger = U^{-1}$.

\begin{definition}~\cite{Selinger2007139} \label{def:daggers-of-everything}
A \deff{dagger (symmetric) monoidal category}\index{dagger monoidal category} $(\catC, \otimes, \dagger)$ is a dagger category with a (symmetric) monoidal structure satisfying 
\[
(f \otimes g)^{\dagger} = f^\dagger \otimes g^\dagger
\] 
for all morphisms $f, g$, and for which all coherence isomorphisms are unitary.

A \indef{dagger compact}\index{dagger compact category} category is a dagger symmetric monoidal category for which every object $A$ has a \indef{dagger dual}\index{dagger dual objects}, i.e.~a dual object for which
\[
\scalebox{0.8}{\input{./figures/dagdual.tikz}}
\]
By a \deff{dagger compact category with discarding}~\cite{coecke2010environment}\index{dagger compact category!with discarding} $(\catC, \otimes, \dagger, \discard{})$ we mean one with a choice of discarding such that for all objects $A$ we have
\begin{equation} \label{eq:discard-compact}
\scalebox{0.8}{\input{./figures/discard-compact.tikz}}
\end{equation}
Explicitly, on each object $A$ the state $\discardflip{A}$\label{not:upsidediscard} above denotes $\discard{A}^\dagger$, as standard for dagger notation. The rule \eqref{eq:discard-compact} thus relates the discarding effect on $A$ with that on $A^*$.
\end{definition}
Dagger compactness further relaxes our approach to diagrams, allowing us to now both bend wires and flip pictures upside-down.
In particular any morphism $f \colon A \to B$ now induces a morphism
\[
\scalebox{0.8}{\input{./figures/conjugate.tikz}}
\]
\paragraph{Notions in dagger categories} When working in a dagger category we typically adapt all categorical notions to be compatible with the dagger. For example we are usually interested in unitaries rather than mere isomorphisms, and in the following kinds of monics or biproducts. In any dagger category:
\begin{itemize}
\item 
an \indef{isometry}\index{isometry} is a morphism $i \colon A \to B$ with $i^\dagger \circ i = \id{A}$;
\item 
a \indef{dagger biproduct}\index{biproduct!dagger} is a biproduct $A \biprod B$ with $\pproj_A = {\coproj_A}^\dagger$ and $\pproj_B = {\coproj_B}^\dagger$. 
\end{itemize}
Whenever we say a dagger category has addition we mean one satisfying 
$(f + g)^\dagger = f^\dagger + g^\dagger$
for all $f, g$. As we would expect such an addition is provided by finite dagger biproducts. In a dagger category zero morphisms automatically satisfy $0^\dagger = 0$. 

A \indef{dagger functor}\index{functor!dagger} $F \colon \catC \to \catD$ between dagger categories is one satisfying $F(f^\dagger) = F(f)^\dagger$ for all morphisms $f$. A \indef{dagger (symmetric) monoidal functor}\index{functor!dagger monoidal} also has that all of its structure isomorphisms are unitary. A dagger functor is an \indef{equivalence}\index{equivalence!dagger} $\catC \simeq \catD$ when it is full, faithful and for every object $B$ of $\catD$ there is a unitary $B \simeq F(A)$ for some object $A$ of $\catC$.

\subsection{Examples}

Let us now meet some examples of compact, dagger and dagger compact categories.

\begin{exampleslist}
\item 
For any field $\fieldk$, let $\Veck$\label{cat:veck} be the symmetric monoidal category whose objects are vector spaces $V$ over $\fieldk$ and morphisms are $k$-linear map $f \colon V \to W$, with $I=\fieldk$ and $\otimes$ being the usual tensor product of vector spaces. Here an object $V$ has a dual precisely when it is finite-dimensional as a vector space, in this case being given by its dual space
\[
V^* := \{f \colon V \to \fieldk \mid f \text{ is linear }\}
\]
Choosing any basis $\{\ket{i}\}^n_{i=1}$ for $V$, let $\bra{i} \in V^*$ be the unique functional sending each basic vector $\ket{j}$ to $\delta_{i,j}$. Then $V^*$ is indeed a dual object to $V$ via 
\begin{equation} \label{eq:VecSp-cups-caps}
\scalebox{0.8}{\input{./figures/veccup.tikz}} :: 1 \mapsto \sum^n_{i=1} \bra{i} \otimes \ket{i} 
\qquad
\qquad
\scalebox{0.8}{\input{./figures/veccap.tikz}} :: \bra{i} \otimes \ket{j} \mapsto 
\begin{cases}
1 & i = j \\ 
0 & i \neq j
\end{cases}
\end{equation}
In fact both maps are independent of our choice of basis. This makes the full subcategory $\FVeck$\label{cat:FVeck}, whose objects are the finite-dimensional spaces, compact. 
\item 
$\Hilb$ is dagger symmetric monoidal, with $f^{\dagger} \colon \hilbK \to \hilbH$ being the adjoint of the linear map $f \colon \hilbH \to \hilbK$, i.e.~the unique map satisfying 
$\langle f(v), w \rangle = \langle v, f^\dagger(w) \rangle$
for all $v \in \hilbH, w \in \hilbK$. 

A Hilbert space $\hilbH$ has a dual in $\Hilb$ again precisely when it is finite-dimensional, then being given by its dual space $\hilbH^*$ via the morphisms~\eqref{eq:VecSp-cups-caps}, which are often referred to as the (unnormalised) \emps{Bell state} and \emps{Bell effect} on $\hilbH$. In this way $\FHilb$ is dagger compact, and similarly so is $\FHilbP$. 
\item 
Each category $\MatS$ is compact closed. Here each object $n$ is self-dual with
\[
\scalebox{0.8}{\input{./figures/MatS-dual.tikz}}
\
=
\
\sum^n_{i=1} \ 
\scalebox{0.8}{\input{./figures/MatS-dual2.tikz}}
\qquad
\qquad
\scalebox{0.8}{\input{./figures/MatS-dual3.tikz}}
\
=
\
\sum^n_{i=1} \ 
\scalebox{0.8}{\input{./figures/MatS-dual4.tikz}}
\]
where above we label by $i$ the respective column and row vectors with a value $1$ at position $i$ and $0$ elsewhere.

Whenever $S$ is \indef{involutive}\index{involutive semi-ring} \label{not:involsring}, meaning that it comes with an automorphism $s \mapsto s^\dagger$ with $s^{\dagger \dagger} = s$ for all $s \in S$, this makes $\MatS$ dagger compact with 
\begin{equation} \label{eq:dagger-of-matrix}
(M^\dagger)_{i,j} := M_{j,i}^\dagger
\end{equation}
In particular $\Class \simeq \Mat_{\mathbb{R}^+}$ is dagger compact with discarding.
\item 
More generally when $\catC$ is a dagger category with addition then $\catC^\oplus$ has a dagger as in~\eqref{eq:dagger-of-matrix} giving it dagger biproducts. Similarly when $\catC$ is compact or dagger compact then so is $\catC^\oplus$; for any object $A = (A_i)^n_{i=1}$ we take $A^* := (A_i^*)^n_{i=1}$ with 
\[
(\scalebox{0.8}{\input{./figures/smallcup.tikz}})_{i,j}
= 
\begin{cases} 
\scalebox{0.8}{\input{./figures/smallercup.tikz}}
& \text{if $i = j$} 
\\  \ \ 0 & \text{ otherwise } 
\end{cases}
\qquad
\qquad
(\scalebox{0.8}{\input{./figures/smallcap.tikz}})_{i,j}
= 
\begin{cases} 
\scalebox{0.8}{\input{./figures/smallercap.tikz}}
& \text{if $i = j$} 
\\  \ \ 0 & \text{ otherwise } 
\end{cases}
\]
\item 
$\Quant{}$ and $\FCStar$ are both dagger compact categories with discarding. 

Indeed $\FCStar$ forms a dagger subcategory of $\FHilb$ since each finite-dimensional C*-algebra is in particular a Hilbert space, with $f^\dagger$ again given by the adjoint of each (completely positive) map $f$, and similarly so does $\Quant{}$. Then $\Quant{}$ inherits compactness from its subcategory $\FHilbP$, so that $\FCStar \simeq \Quant{}^\oplus$ is dagger compact also.

\item 
In contrast, the infinite dimensional settings $\CStarop$ and $\vNop$ lack daggers or compactness. 

\item 
$\Rel$ is a dagger compact category with discarding. For any relation we define $R^{\dagger} \colon B \to A$ by relational converse $R^\dagger(b,a) \iff R(a,b)$. Here every object $A$ is self-dual via the relations 
\begin{equation} \label{eq:cups-in-Rel}
\scalebox{0.8}{\input{./figures/relcup.tikz}} :: \star \mapsto (a,a) \ \ \forall a \in A 
\qquad
\qquad
\scalebox{0.8}{\input{./figures/relcap.tikz}} :: (a,b) \mapsto \star \ \text{     if $a = b$}
\end{equation}
More generally so is $\Rel(\catC)$ for any regular category, or indeed any bicategory of relations in the sense of Carboni and Walters~\cite{carboni1987cartesian}.
\item 
$\Spek$ and $\MSpek$ are both dagger compact subcategories of $\Rel$. Indeed by definition both are closed under the dagger, and $\Spek$ contains the cups from~\eqref{eq:cups-in-Rel} on each of its objects, which are built from its generators $\tinycomult[whitedot]$, $\tinyunit[whitedot]$ by tensors of the state
\[
\scalebox{0.8}{\input{./figures/spek-cup.tikz}}
\]
\item 
A \emps{groupoid} is a category in which every morphism is an isomorphism~\cite{mac1978categories}, so that a group is a one-object groupoid. Any groupoid forms a dagger category by setting $f^\dagger = f^{-1}$, so that every morphism is unitary. 
\item 
For any group $G$, we may define a category $\Rep(G)$\label{cat:repg} whose objects are (finite-dimensional) unitary representations $\phi \colon G \to \Aut(V)$ of $G$, and morphisms $f \colon (V,\phi) \to (W, \psi)$ are \emps{intertwiners}, i.e.~linear maps $f \colon V \to W$ satisfying 
$f(\phi(g(v)) = \psi(g)(f(v))$ for all $g \in G, v \in V$. One may verify that $\Rep(G)$ is dagger compact, inheriting this structure from $\FHilb$.

Verdon and Vicary have used $\Rep(G)$ to study reference frame-independent quantum protocols, by taking $G$ to be a group of transformations of such frames~\cite{verdon2016tight}. 
\end{exampleslist}

Note that, like compactness, the presence of a dagger indeed usually requires morphisms which are not sub-causal. For example, for any object $\hilbH$ with orthonormal basis $\{\ket{i}\}^n_{i=1}$ in $\Quant{}$ we have 
\[
\scalebox{0.8}{\input{./figures/disc-disc.tikz}}
\sum^n_{i=1}
\scalebox{0.8}{\input{./figures/disc-disc2.tikz}}
 \ n
\]
and so the state $\discardflip{}$ is not sub-causal.

\subsection{The $\CPM$ Construction} \label{sec:CPM}

The notion of dagger compactness provides a new way to generate examples of categories with discarding, first introduced by Selinger~\cite{Selinger2007139}, based on a description of $\Quant{}$ in terms of $\FHilb$.

\begin{definition} \label{def:CPM}
~\cite{Selinger2007139,coecke2008axiomatic}\index{CPM construction} \label{not:CPM}
Let $\catA$ be a dagger compact category. The category $\CPM(\catA)$ is defined as having the same objects as $\catA$, with morphisms $A \to B$ being those morphisms in $\catA$ of the form
\[
\scalebox{0.8}{\input{./figures/CPM-map.tikz}}
\]
for some $f \colon A \to C \otimes B$. Using the graphical rules for dagger compact categories, it is straightforward to check that this category is again dagger compact, and has discarding given by 
\[
\scalebox{0.8}{\input{./figures/CPM-discarding.tikz}}
\]
\end{definition}
\noindent
There is then a dagger monoidal functor $\Dbl{(-)} \colon \catA \to \CPM(\catA)$\label{not:doublingfunctor} given by `doubling':
\[
\scalebox{0.8}{\input{./figures/DBL.tikz}}
\]
which generalises sending any linear map in $\FHilb$ to its Kraus map. This construction also often comes with a notion of coarse-graining.

\begin{proposition}~\cite[Cor.~5.3]{Selinger2007139} \label{prop:addition-in-CPM}
Let $\catA$ be a dagger compact category with finite dagger biproducts. Then $\CPM(\catA)$ has addition defined by 
\begin{equation} \label{eq:addition-in-CPM}
\scalebox{0.8}{\input{./figures/CPM-map-smaller.tikz}}
+
\scalebox{0.8}{\input{./figures/CPM-map-smaller2.tikz}}
:=
\scalebox{0.8}{\input{./figures/CPM-map-smaller3.tikz}}
\end{equation}
where $\langle f, g \rangle$ is the unique morphism with 
\[
\scalebox{0.8}{\input{./figures/CPM-proj1.tikz}}
\]
Equivalently, this is just the addition in $\catA$ induced by its dagger biproducts. 
\end{proposition}

\begin{example}   
Motivation for this construction comes from the fact that we have dagger monoidal equivalences
\[
\Quant{} \simeq \CPM(\FHilb) \simeq \CPM(\FHilbP)
\]
To see this, let us first expand the definition of $\CPM(\FHilb)$. Morphisms $\hilbH \to \hilbK$ all take the form 
\begin{equation} \label{eq:CPM-explicit}
\scalebox{0.8}{\input{./figures/CPM-kraus1i.tikz}}
=
\sum^n_{i=1} \ 
\scalebox{0.8}{\input{./figures/CPM-kraus1iii.tikz}}
\quad
\text{ where }
\quad
\scalebox{0.8}{\input{./figures/CPM-krausa.tikz}}
\end{equation}
for any orthonormal basis $\{\ket{i}\}^n_{i=1}$ of $\hilbL$. In particular, states of $\Hilb$ in this category may be identified with density matrices $\rho = \sum^n_{i=1} \ket{\psi_i} \bra{\psi_i}$ via the correspondence
\[
\sum^n_{i=1}
\scalebox{0.8}{\input{./figures/CPM-kraus2i.tikz}}
\ \ 
=
\ \ 
\sum^n_{i=1}
\scalebox{0.8}{\input{./figures/CPM-kraus2ii.tikz}}
\ \
=
\ \ 
\scalebox{0.8}{\input{./figures/CPM-kraus2iii.tikz}}
\]
It is well-known that completely positive maps $B(\hilbH) \to B(\hilbK)$ are then precisely maps of the form $\rho \mapsto \sum^n_{i=1} f_i \circ \rho \circ f_i^*$ as in~\eqref{eq:CPM-explicit}, with this being known as the \emps{Kraus decomposition} of a completely positive map. This provides the equivalence $\Quant{} \simeq \CPM(\FHilb)$. Since any such map is invariant under multiplying each $f_i$ by a global phase, it follows that $\CPM(\FHilb) \simeq \CPM(\FHilbP)$ also. The addition in $\Quant{}$ induced as in~\eqref{eq:addition-in-CPM} by the dagger biproducts in $\FHilb$ is precisely the usual one of completely positive maps.
 
Replacing $\FHilb$ by other dagger compact categories allows us to consider varied quantum theories; for example we may define a dagger theory $\Quant{}^G := \CPM(\Rep(G))$ modelling quantum processes up to some group of symmetries $G$. 
\end{example}

The $\CPM$ construction will be useful to us in the next chapter as an abstract treatment of the quantum setting, and we will use it to define and study further such generalised quantum theories in Chapter~\ref{chap:recons}.

%% file: figures/disc-disc.tikz
\begin{tikzpicture}
	\begin{pgfonlayer}{nodelayer}
		\node [style=upground] (0) at (-1.75, 0.75) {};
		\node [style=downground] (1) at (-1.75, -0.75) {};
		\node [style=none] (2) at (0, -0) {$=$};
		\node [style=label] (3) at (-2.25, -0) {$\hilbH$};
	\end{pgfonlayer}
	\begin{pgfonlayer}{edgelayer}
		\draw [style=none] (1) to (0);
	\end{pgfonlayer}
\end{tikzpicture}

%% file: figures/disc-disc2.tikz
\begin{tikzpicture}
	\begin{pgfonlayer}{nodelayer}
		\node [style=upground] (0) at (-1.75, 0.75) {};
		\node [style=dagpoint] (1) at (-1.75, -0.75) {$i$};
		\node [style=none] (2) at (0, -0) {$=$};
	\end{pgfonlayer}
	\begin{pgfonlayer}{edgelayer}
		\draw [style=none] (1) to (0);
	\end{pgfonlayer}
\end{tikzpicture}

%% file: chapter4.tex
\chapter{Principles for Operational Theories} \label{chap:principles}

A common topic of research in the foundations of physics lies in singling out the consequences of various physical or operational principles which a theory may satisfy (see for instance~\cite{Barrett2007InfoGPTs,Barnum2007NoBroadcast,pawlowski2009information,barnum2012teleportation}). For example, quantum theory is known to have major operational advantages over the classical world~\cite{deutsch1992rapid,shor1999polynomial}, and many have sought to characterise precisely which of its properties lie at the source of these benefits~\cite{d2010probabilistic,howard2014contextuality}. In the strongest case, combinations of such principles have been used to reconstruct quantum theory itself from among all finite-dimensional probabilistic theories (see Chapter~\ref{chap:recons}).

Particular principles have been introduced and studied in a variety of frameworks for general physical theories, along with probabilistic theories~\cite{PhysRevA.84.012311InfoDerivQT}, including categorical quantum mechanics~\cite{coecke2008axiomatic,selby2017leaks,cunningham2017purity} and effectus theory~\cite{QC2015ChoJaWW,EffectusIntro}. The categorical approach provides a new perspective on many of these principles, whilst also suggesting natural new ones to consider. 

In this chapter we survey a range of principles for operational theories, unifying features which have arisen independently in each of these frameworks, and showing that they may be studied in a very lightweight diagrammatic setting.

To address theories of a quantum-like nature, we will pay particular attention to the principle of \emps{purification}~\cite{chiribella2010purification} and associated notions of \emps{purity} of morphisms. Later in Chapter~\ref{chap:recons}, we will use principles from this chapter, such as purification, to provide our own categorical reconstruction of quantum theory.

\subsection*{Setup}

Through this chapter, we will typically work in a basic categorical framework, capable of accommodating either the sub-causal or super-causal setting. By a \indef{theory}\index{theory} we will simply a mean a symmetric monoidal category with discarding $(\catC, \otimes, \discard{})$ with zero morphisms, satisfying our earlier rule
\begin{equation} \label{eq:pos-eq}
\scalebox{0.8}{\input{./figures/zero-moni.tikz}} \ 0 
\quad
\implies 
\quad
\scalebox{0.8}{\input{./figures/zero-monii.tikz}} \ 0 
\end{equation}
 for all morphisms $f$. 

 At times we will also consider when our theory has extra structure. We call a theory \indef{ordered} \index{theory!ordered} when it is monoidally enriched in partially ordered sets, with each zero morphism as a bottom element. That is, each homset $\catC(A,B)$ has a partial order $\leq$\label{not:ordertheory} on morphisms which is respected by composition, with each function $f \circ (-)$, $(-) \circ f$, and $f \otimes (-)$ being monotone, and satisfying 
 \begin{equation} \label{eq:lesszero}
f \leq 0 \implies f = 0
 \end{equation}
 for all morphisms $f$.
 We will sometimes consider when $\catC$ has a partial $\ovee$ or total $+$ addition operation on morphisms as in Chapter~\ref{chap:totalisation}, which in many cases induces this ordering. We say that an ordered theory is \indef{ordered by} a partial (or total) addition $\ovee$ when its order is given by
\[
f \leq g \iff g = f \ovee h \text{ for some $h$}
\]
For example the addition in each of the theories $\PFun$, $\KlDT$, $\Quant{}$, $\vNop$, $\Rel$ induces an order on them in this way. Finally we will also at times consider when $\catC$ has compact, dagger or dagger compact structure. In the latter cases we require the rule $f \leq g \implies f^\dagger \leq g^\dagger$ for all morphisms $f,g$ and call $\catC$ a \index{dagger theory}\indef{dagger theory}.

\section{Minimal Dilations}

One of the most elementary notions in a theory is that of a dilation of a morphism. Our first principle requires each morphism to have a canonical `smallest' dilation.

\begin{definition}\label{def:min-dilation}
\label{not:mindil}
A \deff{minimal dilation}\index{minimal dilation} of a morphism $f$ is a dilation
\[
\scalebox{0.8}{\input{./figures/min-dila.tikz}}
\]
such that for all morphisms $g$ we have 
\[
\scalebox{0.8}{\input{./figures/min-dil.tikz}}
\]
for some $h \colon \MDil{f} \to C$, and for all morphisms $k, l$ we have
\begin{equation} \label{eq:min-dilunique}
\scalebox{0.8}{\input{./figures/min-dilb.tikz}}
\end{equation}
\end{definition}
In particular, a minimal dilation of an effect $e \colon A \to I$ is an epimorphism $\mdil{e} \colon A \to \MDil{e}$ with $\discard{} \circ \mdil{e} = e$ such that for all $f \colon A \to B$ with $\discard{} \circ f = e$ we have
\begin{equation} \label{eq:extensions}
\begin{tikzcd}
A \rar{\mdil{e}} \drar[swap]{f} & \MDil{e} \dar[dashed]{\exists ! g} \\
 & B
\end{tikzcd}
\end{equation}
In a compact theory, by bending wires, one may see that for minimal dilations to exist it suffices to have them for effects. 

\begin{lemma}
For any two minimal dilations $\mdil{f}$ and $\mdil{f}'$ of the same morphism $f$ there is a unique causal isomorphism $g$ with 
\begin{equation} \label{eq:min-dil-uniqueIsom}
\scalebox{0.8}{\input{./figures/min-dilc.tikz}}
\end{equation}
\end{lemma}
\begin{proof} 
Since both morphisms have marginal $f$, there is a unique morphism $g$ as above, and dually we obtain a unique morphism $h \colon  \MDil{f} \to \MDil{f}'$. It follows from~\eqref{eq:min-dilunique} that both $g$ and $h$ are causal and inverse to each other.
\end{proof}

\subsection{Dilations in ordered theories}

In the setting of ordered theories minimal dilations typically take on an extra property. In such a theory let us call a minimal dilation $\mdil{f}$ an \indef{order dilation}\index{order dilation} when it satisfies 
\begin{equation} \label{eq:order-dilation} 
\scalebox{0.8}{\input{./figures/min-dilorder.tikz}}
\end{equation}
for some (unique) $h \colon \MDil{f} \to C$. 

Next, let us say that a theory has \indef{disjoint embeddings}\index{disjoint embedding} when for all objects $A, B$ there is an object $C$ and morphisms
\[
\begin{tikzcd}[row sep = large]
A \rar[shift left = 2.5]{\coproj_A}  & C  \lar[shift left = 2.5]{\pproj_A} \rar[shift right = 2.5,swap]{\pproj_B} & B \lar[shift right = 2.5,swap]{\coproj_B}  
\end{tikzcd}
\]
satisfying $\pproj_A \circ \coproj_A = \id{A}$, $\pproj_B \circ \coproj_B = \id{B}$, $\pproj_A \circ \coproj_B = 0$ and  $\pproj_B \circ \coproj_A = 0$, and with $\coproj_A$ and $\coproj_B$ causal.

\begin{proposition} \label{prop:derive-order-embed}
Let $\catC$ be a theory with disjoint embeddings which is ordered by either a partial addition $\ovee$ making it a sub-causal category, or a total addition $+$. Then in $\catC$ any minimal dilation is an order dilation.
\end{proposition}
\begin{proof}
We prove the result for a partial addition $\ovee$, the total case is simpler. Let $\mdil{f}$ be a minimal dilation of $f \colon A \to B$ and let $g \colon A \to B \otimes C$ satisfy the left-hand side of \eqref{eq:order-dilation}. Then letting $g_B \colon A \to B$ be its marginal we have $g_B \ovee h = f$ for some $h \colon A \to B$. Now consider a disjoint embedding
\[
\begin{tikzcd}[row sep = large]
C \rar[shift left = 2.5]{\pcoproj_C}  & D  \lar[shift left = 2.5]{\pproj_C} \rar[shift right = 2.5,swap]{\pproj_I} & I \lar[shift right = 2.5,swap]{\pcoproj_I}  
\end{tikzcd}
\]
and define 
\[
\scalebox{0.8}{\input{./figures/min-dilorderproof.tikz}}
\]
Note that $k$ is well-defined thanks to condition \ref{enum:disc-compat} of Definition \ref{def:PCM-with-discarding}.
Then $k$ is a dilation of $f$ and so factors over $\mdil{f}$ by some unique causal morphism $r \colon \MDil{f} \to D$. Finally then $g = (\id{B} \otimes (\pproj_C \circ r)) \circ \mdil{f}$. 
\end{proof}

In the non-monoidal setting with discarding we may define minimal and order dilations for effects $e \colon A \to I$ just as above, and the same proof holds.

\begin{remark}[Quotients]
Order dilations of effects were already defined under a different name in the context of effectus theory~\cite{QC2015ChoJaWW,EffectusIntro}, where they are called \emps{quotient} maps $\xi_e$ and defined via the complement $e^{\bot}$ of an effect $e$ by the property:
\[
\forall f \text{ s.t. } \discard{} \circ f \leq e^{\bot}
\qquad
\begin{tikzcd}
A \rar{\xi_e} \arrow[dr,"f",swap ]& A   /  e \dar[dashed]{\exists ! g} \\ 
& B 
\end{tikzcd}
\]
Hence quotients coincide with order dilations via the correspondence 
\[
\Ext{A}{e} = A / e^\bot
\qquad
\mdil{e} = \xi_{e^\bot} 
\]
Definition~\ref{def:min-dilation} allows us to extend this notion to settings where effects lack complements or any ordering. 
\end{remark}

\begin{corollary} \label{cor:effectusMinDil}
An effectus in partial form $\catC$ has minimal dilations for effects iff it has quotients.
\end{corollary}
\begin{proof}
The coproducts in $\catC$ give it disjoint embeddings. Hence by Proposition~\ref{prop:derive-order-embed} it has minimal dilations for effects iff it has order dilations, and these coincide with quotients by the above remark. 
\end{proof}

\subsection{Examples}  \label{examples:min-dilations}

\begin{exampleslist}
\item 
$\PFun$ has minimal dilations. For $f \colon A \to B$ we define
\[
\MDil{f} := \{a \mid f(a) \text{ is defined }\} \subseteq A
\]
and $\mdil{f} \colon A \to B \times \MDil{f}$ by $a \mapsto (f(a),a)$ for all $a \in A$.

Indeed any dilation $g$ of $f$ via some object $C$ factors over $\mdil{f}$ by the unique $h \colon \MDil{f} \to C$ with $g(a) = (b,h(a))$ for some $b \in B$.

In particular for each effect $e \colon A \to I$ defined on $E \subseteq A$ we may set $\Ext{A}{e} = E$ with $\mdil{e} \colon A \to E$ given by $a \mapsto a$ whenever $a \in E$. 

Hence since $\PFun$ is an effectus in partial form this map it has quotients, i.e.~order dilations, as shown in~\cite{QC2015ChoJaWW}.

\item 
$\KlDT$ and $\KlSD$ have minimal dilations. For each $f \colon A \to B$ we set 
\[
\MDil{f} := \{(a, b) \mid f(a)(b) > 0 \} \subseteq A \times B
\]
and define $\mdil{f} \colon A \to B \times \MDil{f}$ by 
\[
\mdil{f}(a)(b,(a',b')) 
= 
\begin{cases} 
f(a)(b) & a = a' \text{ and } b = b' \\ 
0 & \text{otherwise}
\end{cases}
\]
Indeed every other dilation $g \colon A \to B \otimes C$ factors over $\mdil{f}$ via the mediating map $h \colon \MDil{f} \to C$ with $h((a,b))(c) = g(a)(b,c)$.

In particular for an effect $e \colon A \to I$ we set $\Ext{A}{e} = \{a \in A \mid e(a) \neq 0\}$ and define $\ext{e} \colon A \to \Ext{A}{e}$ to be the map sending $a \mapsto p \cdot a$ where $p = e(a)$. 

Both theories are ordered by their respective (partial, total) coarse-graining operations, making each $\mdil{f}$ an order dilation by Proposition~\ref{prop:derive-order-embed}, as considered for $\KlSD$ in~\cite{QC2015ChoJaWW}. Similarly $\Class$ has order dilations.

\item 
$\vNSUop$ is an effectus in partial form with quotients, as shown in~\cite{QC2015ChoJaWW}, and so has order dilations for all effects. Here an effect $A \to I$ corresponds to a positive element $e \in A$ and we have
\[
\Ext{A}{e} 
= 
p \cdot A \cdot p
:=
\{ p \cdot a \cdot p \mid a \in A \}
\]
where $p = \lceil e \rceil$ is a projection in $A$, satisfying $p = p^* \cdot p$, and is the least such with $e = p \cdot e = e \cdot p$, often referred to as the \emps{support projection} of $e$. Then $\ext{e}$ is defined as the completely positive map (in the opposite direction) $\Ext{A}{e} \to A$ given by $a \mapsto \sqrt{e} \cdot a \cdot \sqrt{e}$.
The proof that this defines a minimal dilation is non-trivial, see~\cite{westerbaan2016universal},~\cite{QC2015ChoJaWW} and~\cite[Example 82.4]{EffectusIntro}. $\vNop$ has order dilations similarly.
\item \label{example:QuantMinDil}
$\Quant{}$ and $\FCStar$ each have order dilations. Indeed they have them for effects by restricting the previous example to finite dimensions, and by compactness these extend to arbitrary morphisms. The minimal dilation of a completely positive map $f \colon B(\hilbH) \to B(\hilbK)$ is given by a Kraus map, namely its minimal \emps{Stinespring dilation}~\cite{stinespring1955positive,Paschke2016}. $\QuantSU$ has order dilations in the same way.
\item
$\Rel$ has order dilations. For any $R \colon A \to B$ we set 
\[
\MDil{R} := R \subseteq A \times B
\] 
and $\mdil{R} \colon A \to B \times R$ to relate $a$ to $(b,(a,b))$ whenever $(a,b) \in R$. Then any dilation $S \colon A \to B \times C$ of $R$ is equal to $\mdil{R}$ up to the unique relation $h \colon R \to C$ with $(a,b) \sim c$ whenever $g$ relates $a$ with $(b,c)$, making this a minimal dilation. Since $\Rel$ satisfies the requirements of Proposition~\ref{prop:derive-order-embed} these are then order dilations.

$\Rel(\catC)$ lacks addition or zero morphisms for a general regular category $\catC$, unless $\catC$ is coherent (so is not strictly a theory in our sense). However it is still ordered under the usual ordering $R \leq S$ of subobjects and has order dilations in the same sense, defined just as above in $\Rel$. 

\item 
Let $\catC$ be an ordered theory satisfying the `sub-causality' rule $e \leq \discard{}$ for all effects $e$. Then $\catC$ may be extended to a new such theory with order dilations for effects $\Eff(\catC)$: 
\begin{itemize}
\item objects are pairs $(A, e)$ consisting of an object $A$ of $\catC$ and effect $e \colon A \to I$; 
\item morphisms $f \colon (A, e) \to (B, d)$ are those $f \colon A \to B$ in $\catC$ with $d \circ f \leq e$. 
\end{itemize}
We set $(A,e) \otimes (B,d) = (A \otimes B, e \otimes d)$ with unit $(I, \id{I})$ and define $f \otimes g$, $\leq$ and $\discard{}$ all as in $\catC$. Here any effect $d$ on an object $(A, e)$ must have $d \leq e$ and so has a minimal dilation given by $\id{A} \colon (A,e) \to (A,d)$.

There is a forgetful functor $U \colon \Eff(\catC) \to \catC$ which a has full and faithful left adjoint sending each object $A$ to $(A, \discard{A})$. Then $\catC$ has order dilations iff this functor in turn has a left adjoint. An alternative universal characterisation of quotients in effectuses was first given in~\cite{QC2015ChoJaWW}.
\end{exampleslist}

\section{Kernels}

 Our next principle has appeared explicitly in categorical studies of quantum and classical physics~\cite{heunen2010quantum,EffectusIntro}, and implicitly in the reconstruction~\cite{PhysRevA.84.012311InfoDerivQT}. 
Motivation comes from the fact that, for any effect $e$ in either theory, the collection of states $\rho$ for which $e$ never occurs, with $e \circ \rho = 0$, forms a new system. A standard categorical notion extends this idea to arbitrary morphisms~\cite[p.~191]{mac1978categories}. 

\begin{definition} \label{def:kerAndCoker}
\label{not:kernel}
In any category with zero morphisms, a \deff{kernel}\index{kernel} of a morphism $f \colon A \to B$ is a morphism $\ker(f)$ with $f \circ \ker(f) = 0$ such that every morphism $g \colon C \to A$ with $f \circ g = 0$ has $g = \ker(f) \circ h$ for a unique morphism $h$. 
\[
\begin{tikzcd}
\Ker(f) \rar{\ker(f)} & A \rar[shift left=2.5]{f} \rar[shift right = 2.5,swap]{0} & B \\ 
 C \arrow[ur,swap, "g"] \arrow[u, dashed, "\exists ! h"] 
\end{tikzcd}
\]
In other words, $\ker(f)$ is an \emps{equaliser} of the morphisms $f, 0\colon A \rightrightarrows B$. 

Dually, a \deff{cokernel}\index{cokernel}\label{not:cokernel} of $f$ is a \emps{coequaliser} of this pair of morphisms. That is, it is a morphism $\coker(f)$ with $\coker(f) \circ f = 0$ and for which every morphism $g$ with $g \circ f = 0$ has $g = h \circ \coker(f)$ for some unique $h$. 
\[
\begin{tikzcd}[column sep = large]
A \rar[shift left=2.5]{f} \rar[shift right = 2.5,swap]{0} & B  \rar{\coker(f)} &  \Coker(f)
\end{tikzcd}
\]
We say that a theory $(\catC, \discard{})$ \deff{has (co)kernels}\index{theory with (co)kernels} when every morphism has a causal kernel and a cokernel. 
\end{definition}

In categories with discarding we always consider kernels which are causal; however cokernels generally will not be. Any two (causal) kernels $k, k'$ of the same morphism have $k' = k \circ U$ for a unique (causal) isomorphism $U$, and so we may speak of `the' kernel of a morphism, and `the' cokernel dually.

The presence of kernels and cokernels introduces another very useful notion. We define the \indef{image}\index{image}\label{not:image} of a morphism $f \colon A \to B$ by 
\[
\img(f) := \ker(\coker(f))
\]
Then $f$ factors uniquely as
\[
\begin{tikzcd}[row sep = small]
A  \arrow[rr,"f"] \arrow[dr,"e",swap] & & B  \\
& \Img(f) \arrow[ur,"\img(f)",swap]
\end{tikzcd}
\]
where the morphism $e$ is \indef{zero-epic}\index{zero-epic}, meaning that $g \circ e = 0 \implies g = 0$. Dually we define the \indef{coimage}\index{coimage}\label{not:coimage} of $f$ by 
\[
\coim(f)
:=
\big(
\begin{tikzcd}[column sep = 5em]
A \rar{\coker(\ker(f))} & \CoIm(f)
\end{tikzcd}
\big)
\]
and have $f = m \circ \coim(f)$ where $m$ is a \indef{zero-monic}\index{zero-mono}, satisfying $m \circ x = 0 \implies x = 0$. Using these we can identify intrinsically when a morphism is a kernel.


\begin{lemma} \label{lem:kernels-image-rule}
In a category with kernels and cokernels, a morphism $k$ is a kernel iff it satisfies $k = \img(k)$.
\end{lemma}
\begin{proof}
Let $k = \ker(f)$ for some $f \colon A \to B$. Then since $f \circ k =0$ we have $f = h \circ \coker(k)$ for some morphism $h$. Now let $g \colon A \to C$ satisfy $\coker(k) \circ g = 0$. Then $f \circ g = h \circ \coker(k) \circ g = 0$ and so $g$ factors over $k$ as required.
\end{proof}

It is natural to require kernels to interact well with monoidal structure. In a monoidal category let us say that kernels are \indef{$\otimes$-compatible} \index{$\otimes$-compatible kernels} when they satisfy
\begin{equation} \label{eq:kernels}
\scalebox{0.8}{\input{./figures/kernel-compact-intro1.tikz}} 0 \ \implies \exists ! h  \text{ s.t. } \ 
\scalebox{0.8}{\input{./figures/kernel-compact-intro2.tikz}}
\end{equation}
In the compact setting this is in fact automatic.

\begin{proposition} \label{prop:kernels-in-compact}
Let $\catC$ be a compact category with zero morphisms and kernels. 
\begin{enumerate}
\item \label{enum:gives-cokernels}
$\catC$ has cokernels of all morphisms.

\item  \label{enum:kernels-new-rule}
Kernels are $\otimes$-compatible. 
\item \label{enum:images-under-tensor}
For all morphisms $f, g$ we have $\img(f \otimes g) = \img(f) \otimes \img(g)
$.

\item \label{enum:kernels-compatible}
If $k$ and $l$ are kernels then so is $k \otimes l$. 
\end{enumerate}
\end{proposition}

\begin{proof}
\ref{enum:gives-cokernels} Thanks to compactness, $\catC$ is equivalent to its opposite category $\catC^\op$, the category whose arrows $A \to B$ are given by arrows $B \to A$ in $\catC$. Hence $\catC^\op$ also has kernels, and so $\catC$ has cokernels.

\ref{enum:kernels-new-rule} Bending wires we obtain a unique morphism $h$ as below:
\[
\scalebox{0.8}{\input{./figures/kernel-compact-intro1a.tikz}} \ 0 
\iff 
\scalebox{0.8}{\input{./figures/kernel-compact-intro1b.tikz}}
\iff
\scalebox{0.8}{\input{./figures/kernel-compact-intro2.tikz}}
\]

\ref{enum:images-under-tensor}
Note that $\img(f) \otimes \img(g)$ is monic. Indeed, by~\eqref{eq:kernels} each morphism $\img(f) \otimes \id{}$ is monic as $\img(f)$ is a kernel, and similarly so is $\id{} \otimes \img(g)$. 

Write $f = \img(f) \circ d$ and $g = \img(g) \circ e$ where $d$ and $e$ are zero epic. Then by bending wires as in the previous part one sees that $d \otimes e$ is also zero epic. But then since $\coker(f \otimes g) \circ (f \otimes g) = 0$ this gives that 
\[
\coker(f \otimes g) \circ (\img(f) \otimes \img(g)) = 0
\]
Hence we have $\img(f) \otimes \img(g) = \img(f \otimes g) \circ u$ for some $u \colon \Img(f) \otimes \Img(g) \to \Img(f \otimes g)$. Conversely we have implications:
\[
\scalebox{0.8}{\input{./figures/kernel-elem-proof0.tikz}} \ 0
\implies
\scalebox{0.8}{\input{./figures/kernel-elem-proof1.tikz}} \ 0 
\implies
\scalebox{0.8}{\input{./figures/kernel-elem-proof2.tikz}}
\]
for some unique morphism $h$, using~\eqref{eq:kernels} in the last step. Similarly $\img(f \otimes g)$ also factors over $\id{} \otimes \img(g)$ and hence factors over $\img(f) \otimes \img(g)$. By uniqueness it follows that $u$ is an isomorphism, as required. 

\ref{enum:kernels-compatible}
A morphism $k$ is a kernel iff $k = \img(k)$ by Lemma~\ref{lem:kernels-image-rule}. But then by the previous part whenever $k$ and $l$ are kernels so is $k \otimes l$ since 
\[
k \otimes l = \img(k) \otimes \img(l) = \img(k \otimes l)
\]
\end{proof}

\subsection{Dagger kernels}

In dagger categories we expect kernels to interact well with the dagger as follows.



\begin{definition} \label{def:dag-kern}
In a dagger category, a \deff{dagger kernel}\index{kernel!dagger} $k$ is a kernel which is an isometry, i.e.~satisfies $k^{\dagger} \circ k = \id{}$.
\end{definition}

Such a compatible dagger structure makes kernels especially well behaved, and in the context of a dagger category by `kernel' we will always mean `dagger kernel'. Dagger (co)kernels are always unique up to unitary isomorphism. The presence of dagger kernels provides a canonical choice of cokernel 
\[
\coker(f) = \ker(f^{\dagger})^{\dagger}
\]
and a zero object given by $0 = \Ker(\id{A})$ for any object $A$.

Dagger kernels were first studied in detail by Heunen and Jacobs~\cite{heunen2010quantum}, where they were shown to have a surprisingly rich structure, resembling the subspaces of a Hilbert space and studied extensively in quantum logic~\cite{birkhoff1975logic,piron1976foundations}. 

Recall that a lattice $(L, \leq, 0, 1)$ is \indef{orthomodular}\index{orthomodular lattice} when for every element $a$ there is an element $a^{\bot}$, its \indef{(ortho)complement}, satisfying
\[
a \vee a^{\bot} = 1 
\qquad 
a \wedge a^{\bot} = 0
\qquad
a^{\bot\bot} = a
\qquad
a \leq b \implies b^{\bot} \leq a^{\bot}
\]
 as well as the \indef{orthomodular law}:
\[
a \leq b \implies b = a \vee (b \wedge a^\bot)
\]
In~\cite{heunen2010quantum}, it is shown that in any dagger category with dagger kernels, the collection  $\DKer(A)$ \label{not:dker} of (unitary isomorphism classes of) dagger kernels $k \colon K \to A$ on any object $A$ form an orthomodular lattice 
under the ordering
\[
k \leq l \iff 
\begin{tikzcd}
K \arrow[r,"k"] \arrow[d,dashed,swap,"\exists m"] & A \\ 
L \arrow[ur,"l",swap]
\end{tikzcd}
\]
where we define the complement of a kernel $k$ by \label{not:orthocomp} \index{complement}
\begin{equation} \label{eq:orthocomp}
k^{\bot} := \coker(k)^{\dagger}
\end{equation}
For any dagger kernel $l$ we will write $l = k^\bot$ whenever $l$ belongs to the unitary isomorphism class of \eqref{eq:orthocomp}. 

Now, in the dagger compact setting we also obtain the following. In a monoidal category with zero morphisms we say that \indef{zero-cancellativity}\index{zero-cancellative} holds when 
\[
f \otimes g = 0 
\quad
\implies 
\quad 
f = 0 
\
\text{   or   } 
\
g = 0
\]
for all morphisms $f, g$.
\begin{lemma}
 \label{lem:dagzero}
In any dagger compact category with dagger kernels:
\begin{enumerate}[label=\arabic*., ref=\arabic*]
\item \label{enum:dag-law} 
$f^{\dagger} \circ f = 0 \implies f = 0$ for all morphisms $f$.
\item \label{enum:joins}
For all dagger kernels $k, l$ on $A$ and $m$ on $B$ we have 
\begin{align*}
(k \wedge l) \otimes m &= k \otimes m \wedge l \otimes m 
\end{align*}
 in $\DKer(A \otimes B)$.
\item \label{enum:zero-mult}
Suppose that every non-zero object $A$ has a state $\psi \colon I \to A$ which is an isometry. Then zero-cancellativity holds.
\end{enumerate}
\end{lemma}
\begin{proof}
\ref{enum:dag-law}.~\cite[Lemma 2.4]{vicary2011categorical}
Write $f = \img(f) \circ m$ where $\coker(m) = 0$ and $\img(f)$ is an isometry. Then $f^{\dagger} \circ f = 0 \implies m^{\dagger} \circ m =0$ and so $m^{\dagger} = 0$ and then $m = 0$.

\ref{enum:joins}.
Note that $(k \wedge l) \otimes m$ is indeed an isometry, and is a kernel by Proposition \ref{prop:kernels-in-compact}. Clearly it factors over $ k \otimes m \wedge l \otimes m$. Conversely, let $f \colon C \to A \otimes B$ have $\img(f) \leq k \otimes m$ and $\img(f) \leq l \otimes m$. Then defining 
\[
\scalebox{0.8}{\input{./figures/fbend.tikz}}
\]
we have $\coker(k) \circ g = 0$ and $\coker(l) \circ g = 0$. Hence $\img(g) \leq \img(k) \wedge \img(l) = k \wedge l$, so that $g$ factors over $k \wedge l$. Hence $f$ factors over $(k \wedge l) \otimes \id{B}$, say by some morphism $h$. Again $\img(h) \leq \id{} \otimes m$, and so $h$ factors over $\id{} \otimes m$. Hence $f$ factors over $(k \wedge l) \otimes m$ as required. 

\ref{enum:zero-mult}. Using compactness we have
\[
\scalebox{0.8}{\input{./figures/kernel-bend1.tikz}} \ 0
\iff 
\scalebox{0.8}{\input{./figures/kernel-bend2.tikz}} \ 0
\]
Let the effect $e$ be given by bending $g$'s output to an input as above. If $\Img(e) \simeq 0$ then since $e$ factors over $\Img(e)$ we have $e = 0$. Otherwise, let $\psi$ be an isometric state of $\Img(e)$. Then $\img(e) \circ \psi$ is an isometry $I \to I$ and hence unitary since scalars are commutative. So $\coker(e) = 0$ and then by the above we have $f = 0$. 
\end{proof}

\subsection{Examples}  \label{examples:kernels}

Let us now meet some examples of kernels, cokernels and dagger kernels. Note that in any theory thanks to~\eqref{eq:pos-eq} we have
\[
\ker(f) = \ker(\discard{} \circ f)
\] 
for all morphisms $f$ and so it suffices to consider kernels of effects. 

\begin{exampleslist}
\item 
$\PFun$ has (co)kernels. Each partial function $f \colon A \to B$ has a causal kernel given by the inclusion 
\[
\{a \in A \mid f(a) \text{ is not defined }\}  \hookrightarrow A
\]
Indeed, any partial function $g \colon C \to A$ with $f \circ g = 0$ has that $f(g(c))$ is undefined for all $c \in C$ and so factors over this inclusion. The cokernel of $f$ is given by the partially defined projection to the subset 
\[
\{b \in B \mid \nexists a \in A  \ \  f(a) = b \} \subseteq B
\]
so that $\img(f)$ is the inclusion $\{f(a) \mid a \in A\} \hookrightarrow B$.
\item 
$\KlSD$ and $\KlDT$ have (co)kernels given for $f \colon A \to B$ by
\[
\Ker(f) = \{ a \in A \mid f(a) = 0 \}
\qquad
\Coker(f) = \{b \in B \mid f(a)(b) = 0 \ \forall a \in A  \}
\]
More precisely, the causal map $\ker(f)$ is given by $\ker(f)(a)(a') = 1$ if $a = a'$ and is $0$ otherwise. Similarly $\coker(f)(b)(b') = 1$ if $b = b'$ and is otherwise $0$. Hence $\img(f)$ is given by the inclusion 
\[
\{b \in B \mid \exists a \in A \text{ s.t. } f(a)(b) > 0\} \hookrightarrow B
\] 
The dagger-compact sub-theory $\Class$ has causal dagger kernels in the same way.

\item 
$\Hilb$ has dagger kernels, being the motivating example in~\cite{heunen2010quantum}. Here each map $f \colon \hilbH \to \hilbK$ has a dagger kernel given by the inclusion of the subspace
\[
\{\psi \in \hilbH \mid f(\psi) = 0 \} \hookrightarrow \hilbH
\] 
In particular $\img(f) \hookrightarrow \hilbK$ is then the inclusion of the closure of $f(\hilbH) \subseteq \hilbK$~\cite[Ex.~23]{heunen2010quantum}. $\FHilb$ has dagger kernels in the same way, as do $\HilbP$ and $\FHilbP$.

\item \label{example:QuantKernels}
$\Quant{}$ has causal dagger kernels inherited from $\FHilb$; we prove this fact abstractly in Example~\ref{ex:CPMD-kernels} ahead. As remarked above it suffices to consider kernels for effects. Here any effect on an object $\hilbH$ is of the form $\rho^{\dagger}$ for some unnormalised density matrix $\rho \in B(\hilbH)$. Let us write $\supp(\rho) \subseteq \hilbH$ for the support of $\rho$ as a linear map on $\hilbH$. Then its kernel is given by the orthogonal complement
\[
\Ker(\rho^\dagger)
=
\supp(\rho)^\bot 
 \subseteq \hilbH
\]
and $\ker(\rho^\dagger)$ is the Kraus map $\Dbl{i}$ induced by the inclusion $i \colon \supp(\rho)^\bot \hookrightarrow \hilbH$.

Dually each state $\rho$ has cokernel given by the projection from $\hilbH$ to $\supp(\rho)^\bot$, and we have 
\[
\Img(\rho) = \supp(\rho)
\]
with $\img(\rho)$ given by the (Kraus map of) the inclusion $\supp(\rho) \hookrightarrow \hilbH$. 
\item 
More broadly $\vNop$ and $\vNSUop$ have (co)kernels, as shown in detail in~\cite[77.4]{EffectusIntro}. Here we sketch the result briefly. Let $f \colon A \to B$ be a morphism in $\vNop$, corresponding to a map $g \colon B \to A$ between von Neumann algebras in the opposite direction. Then $\discard{} \circ f$ is given by a unique element $e = g(1) \in A$. 

As before let $\lceil e \rceil \in A$ be the support projection of $e$, and now let $p = \img(g) \in B$ be the least projection in $B$ with $g(p) = g(1)$. We then have 
\[
\Ker(f) = \lceil e \rceil^\bot \cdot A \cdot \lceil e \rceil^\bot
\qquad
\qquad
\Coker(f) =  p^\bot \cdot B \cdot p^\bot 
\]
Here $\ker(f) \colon \Ker(f) \to A$ is given by the completely positive map in the opposite direction sending $a \in A$ to $\lceil e^\bot \rceil \cdot a \cdot \lceil e^\bot \rceil$, and $\coker(f)$ is given by the completely positive inclusion $\lceil p^\bot \rceil B \lceil p^\bot \rceil \hookrightarrow B$. In particular we have 
\[
\Img(f) = p \cdot A \cdot p
\]
Similarly, extending $\Quant{}$, the sub-theory $\FCStar$ has causal dagger kernels in the same way as we show soon in Example~\ref{example:biprod}.

\item
$\Rel$ has causal dagger kernels. For any relation $R \colon A \to B$ its kernel is given by the inclusion
\begin{align*}
\Ker(R) &= \{a \in A \mid \nexists b \in B \ \  R(a,b)\} \hookrightarrow A
\end{align*}
It follows that $\img(R) = \{b \in B \mid \exists a \in A \ R(a,b)\} \hookrightarrow B$. More generally, for any regular category $\catC$ which is \emps{Boolean}, $\Rel(\catC)$ will have dagger kernels~\cite{johnstone2002sketches}.

\item \label{ex:CPMD-kernels}
Let $\catA$ be a dagger compact category with zero morphisms and dagger kernels. Then $\CPM(\catA)$ is a dagger theory with causal dagger kernels.
\begin{proof}
The zero object and morphisms are easily seen to lift from $\catA$ to $\CPM(\catA)$. In order for~\eqref{eq:pos-eq} to hold in $\CPM(\catA)$ we require in $\catA$ that 
\begin{equation} \label{eq:midproof}
\scalebox{0.8}{\input{./figures/CPM-map-smaller-k.tikz}} 
= \ 0 
\quad
\implies 
\quad
\scalebox{0.8}{\input{./figures/CPM-map-smaller-ki.tikz}}
= \ 
0 
\end{equation}
which after bending wires is precisely Lemma~\ref{lem:dagzero}~\ref{enum:dag-law}. We claim that a general morphism as on the left-hand side above has dagger kernel $\Dbl{\ker(f)}$ in $\CPM(\catA)$. Indeed using~\eqref{eq:midproof} along with~\eqref{eq:kernels} we obtain
\begin{align*}
\scalebox{0.8}{\input{./figures/CPM-map-smallerg.tikz}} = \ 0 
\ \ 
&\implies 
\scalebox{0.8}{\input{./figures/CPM-map-smallerg2.tikz}} = \ 0 
\ \ 
\implies
\scalebox{0.8}{\input{./figures/CPM-map-smallerg3i.tikz}}
=
\scalebox{0.8}{\input{./figures/CPM-map-smallerg3ii.tikz}}
\\ 
&\implies
\scalebox{0.8}{\input{./figures/CPM-map-smallerg4i.tikz}}
=
\scalebox{0.8}{\input{./figures/CPM-map-smallerg4ii.tikz}}
\end{align*}
for some morphism $h$, as required.
\end{proof}
In particular as we've seen $\Quant{} \simeq \CPM(\FHilb)$ has dagger kernels. 
\item \label{example:biprod}
If $\catD$ is a (dagger) theory with causal (dagger) kernels and addition satisfying $f + g = 0 \implies f = g = 0$ for all morphisms $f, g$ then so is $\catD^\oplus$. 

Hence $\FCStar \simeq \Quant{}^\oplus$ has causal dagger kernels.
\begin{proof}
By assumption $\discard{} \circ f + \discard{} \circ g = 0 \implies f = g = 0$ in $\catD$, which ensures~\eqref{eq:pos-eq} in $\catD^\oplus$. It suffices to show that each effect $e = (e_i)^n_{i=1}$ on some $A=(A_i)^n_{i=1}$ has a causal (dagger) kernel. But any such effect has kernel given by the block diagonal matrix $(\Ker(e_i))^n_{i=1} \to A$ whose $i$-th diagonal entry is $\ker(e_i)$. 
\end{proof}
\end{exampleslist}

\begin{remark}[Comprehensions]
Like order dilations, kernels appear under another name in effectus theory~\cite{QC2015ChoJaWW,EffectusIntro}. Here an effect $e \colon A \to I$ is said to have a \emps{comprehension} map when there is a morphism $\pi_e$ satisfying
\[
\forall f \text{ s.t. } e^\bot \circ f = 0 \ \ 
\begin{tikzcd}
B \arrow[dr,"f",swap] \rar[dashed, "\exists ! g"] & \{A \mid e\}  \dar["\pi_e"] \\ 
& A
\end{tikzcd}
\]
Hence an effectus in partial form has kernels iff it has comprehensions, via 
\[
\Ker(e) = \{A \mid e^\bot \}
\qquad
\ker(e) = \pi_{e^\bot}
\]
The fact that comprehensions are kernels is noted in~\cite[Lemma~79]{EffectusIntro}. Similarly, cokernels exist in effectuses with quotients and `image predicates'~\cite[Lemma~83]{EffectusIntro}.
\end{remark}

\section{Combining Minimal Dilations and Kernels}

\subsection{Compatible dilations and kernels}

In theories containing both minimal dilations and (co)kernels it is natural to expect these features to be related in some way. Indeed for any effect $e \colon A \to I$ in such a theory we have
\[
\discard{} \circ \ext{e} \circ \ker(e) = e \circ \ker(e)  = 0
\]
so that $\ext{e} \circ \ker(e) = 0$. Hence there is a unique morphism $\med_e$ making the following diagram commute:
\begin{equation} \label{eq:compatible}
\begin{tikzcd}[row sep = large, column sep = large]
A \arrow[r,"\mdil{e}"] \arrow[d,"\coim(e)",swap]& \MDil{e}\\ 
\CoIm(e) \arrow[dashed,ur,swap,"\med_e"] & 
\end{tikzcd}
\end{equation}

In theories of a finite-dimensional nature $\med_e$ can typically be inverted up to a scalar in the following sense. Let us call a morphism $f \colon A \to B$ a \indef{p-isomorphism} when there is a morphism $g \colon B \to A$ and non-zero scalar $p$ such that
\[
g \circ f = p \cdot \id{A} \qquad f \circ g = p \cdot \id{B}
\]
%
When all non-zero scalars are invertible a p-isomorphism is simply an isomorphism.

\begin{definition}
In any theory $\catC$ having (co)kernels and minimal dilations for effects, we say that they are \deff{compatible} \index{compatible kernels and dilations}  when for each effect $e$ the morphism $\med_e$ is a p-isomorphism and \deff{strongly compatible}\index{compatible kernels and dilations!strongly} when it is an isomorphism.
\end{definition}

Strong compatibility generally requires us to work outside of a sub-causal setting, with suitable non-zero scalars being invertible, by the following. We will often call an effect $e$ \indef{internal}\index{internal effect} when it is zero-monic, i.e.~has $\ker(e) = 0$.

\begin{lemma} \label{lem:internal-ext-isom}
In any theory $\catC$ with strongly compatible (co)kernels and minimal dilations:
\begin{enumerate}
\item \label{enum:Internal-If-Isomorphism}
An effect $e$ is internal iff $\mdil{e}$ is an isomorphism;
\item \label{enum:everyscalarr}
Suppose that every non-zero object has a causal state. Then every zero-monic scalar $r$ is an isomorphism.
\end{enumerate}
\end{lemma}
\begin{proof}
\ref{enum:Internal-If-Isomorphism}
Since $\mdil{e} = \med_e \circ \coim(e)$ for an isomorphism $\med_e$, we have that $\coim(e)$ is an isomorphism iff $\mdil{e}$ is. But $\coim(e)$ is an isomorphism iff $\ker(e) = 0$.

\ref{enum:everyscalarr}
Since $r$ is zero-monic, we may take $\coim(r) = \id{I}$. Then by the first part $\med_r \colon I \to \MDil{r}$ is an isomorphism, and it dilates $r$. Now let $\sigma \colon I \to \MDil{r}$ be any causal state. Then $r$ is invertible since $
\id{I} 
= 
\discard{} \circ \sigma 
=
\discard{} \circ \med_r \circ {\med_r}^{-1} \circ \sigma 
= 
r \circ {\med_r}^{-1} \circ \sigma$.
\end{proof}

\subsection{The internal isomorphism property}

Compatibility of dilations and kernels can be derived from another principle studied by Alex Wilce~\cite{wilceRoyal} and found to be characteristic of quantum, classical and related theories. We state it in two forms, relevant to either `sub-causal' or `super-causal' theories.



\begin{definition}
A theory $\catC$ satisfies the \deff{internal (p-)isomorphism property}\index{internal isomorphism property} when every internal effect $e \colon A \to I$ has a dilation of the form $f \colon A \to A$ which is a {(p-)isomorphism}.
\end{definition}

In general this map $f \colon A \to A$ is not unique, since an object $A$ may have many causal isomorphisms $A \simeq A$. In a compact category by bending wires this property is equivalent to the similar principle for states $\rho$:
\begin{equation} \label{eq:int-states}
\left(
\scalebox{0.8}{\input{./figures/internal-isom-statesnew1.tikz}} \ 0 \implies
\scalebox{0.8}{\input{./figures/internal-isom-statesnew2.tikz}} \ 0
\right)
\implies
\exists \text{(p-)isomorphism }
\scalebox{0.8}{\input{./figures/internal-isom-states1.tikz}}
\text{ s.t. }
\scalebox{0.8}{\input{./figures/internal-isom-states2.tikz}}
\end{equation}

In quantum or classical theory a state $\rho$ of a system $A$ is zero-epic as above precisely when it lies in the interior of the positive cone of states of $A$, and indeed any two such states are related (up to a factor) by a reversible physical process. This fact is discussed in depth by Wilce in~\cite{wilceRoyal} who, drawing on a result due to Koecher~\cite{koecher1957positivitatsbereiche} and Vinberg~\cite{vinberg1960homogeneous}, uses it to reconstruct the Jordan algebra structure of quantum and classical physics. 

\begin{proposition} \label{prop:Int-Isom-Deduce}
Let $\catC$ have (co)kernels. The following are equivalent:
\begin{enumerate}[label=\arabic*., ref=\arabic*]
\item \label{enum:int-isom}
$\catC$ satisfies the internal isomorphism property;
\item \label{enum:compat-eff-ext}
$\catC$ has strongly compatible minimal dilations for effects and moreover every effect $e$ 
has a coimage coming with a causal isomorphism 
\begin{equation} \label{eq:non-canoi}
\Ext{A}{e} \simeq \CoIm(e)
\end{equation}
\end{enumerate}
\end{proposition}
\begin{proof}
\ref{enum:int-isom} $\implies$ \ref{enum:compat-eff-ext}:
 Let $e \colon A \to I$ be any effect and $c=\coim(e) \colon A \to C$, so that $e = d \circ c$ for a unique zero-monic effect $d$. By assumption $d$ has a dilation $g \colon C \to C$ which is an isomorphism. 

 Then we claim that $f := g \circ c$ is a minimal dilation for $e$. Since $c$ is a coequaliser it is an epimorphism, as is the isomorphism $g$, and hence $f$ is epic also. By construction $f$ indeed dilates $e$. Moreover any other dilation $h \colon A \to B$ has that
 \[
\discard{} \circ h \circ \ker(e) = e \circ \ker(e) = 0
 \]
so that $h \circ \ker(e) = 0$, giving $h = k \circ c$ for some unique $k$. But then 
\[
h 
=
k \circ c
=
k \circ g^{-1} \circ f
\]
and so $h$ factors over $f$. By construction $\med_e = g$ is a p-isomorphism, and $\id{C}$ provides a causal isomorphism $\Ext{A}{e} \simeq C$. 

\ref{enum:compat-eff-ext} $\implies$ \ref{enum:int-isom}:
Let $e \colon A \to I$ be an internal effect. Then we may take $\CoIm(e) = A$ with $\coim(e) = \id{A}$. By strong compatibility $\med_e = \ext{e}$ is an isomorphism. Hence there is a causal isomorphism $k \colon A \to \Ext{A}{e}$. Then $l =k^{-1} \circ \ext{e}$ is an isomorphism with $\discard{} \circ l = e$ as desired. 
\end{proof}

Non-canonical isomorphisms of the form~\eqref{eq:non-canoi} have also been considered in the context of so-called `quotient-comprehension chains' in effectus theory~\cite{QC2015ChoJaWW}.

\begin{examples} \label{examples:InternalIsom-StrongCompat}
Our examples of a `finite-dimensional' character satisfy the internal isomorphism property, providing strongly compatible minimal dilations. 
\begin{exampleslist}
\item 
In $\PFun$ or $\Rel$ an effect $e$ is internal precisely when $e = \discard{}$, and hence the internal isomorphism property holds trivially.
\item 
In $\KlDT$ an effect $e \colon A \to I$ is internal precisely when $e(a) > 0 $ for all $a \in A$. In this case the dilation $f \colon A \to A$ with $f(a)(a')$ equal to $e(a)$ when $a = a'$ and $0$ otherwise is an isomorphism. Hence $\KlDT$ and $\Class$ satisfy the internal isomorphism property, and $\KlSD$ the p-isomorphism property similarly.
\item 
$\Quant{}$ satisfies the internal isomorphism property.
Indeed here any internal effect on an object $\hilbH$ is a $\mathbb{R}^+$-weighted combination
\[
e = \sum^n_{i=1} p_i \cdot \Dbl{\bra{i}}
\]
for some orthonormal basis $\{\ket{i}\}^n_{i=1}$ of $\hilbH$. Then any completely positive isomorphism $f$ on $B(\hilbH)$ with $\Dbl{\bra{i}} \circ f = p_i \cdot \Dbl{\bra{i}}$ for each $i$ is a dilation $e$.

$\QuantSU$ satisfies the internal p-isomorphism property in the same way, and in fact this property lifts to $\FCStar$ also. Indeed this follows from the fact that $\FCStar \simeq \Quant{}^\biprod$, and that any internal effect on a biproduct is simply given by an internal effect on each component. More generally~\eqref{eq:int-states} is studied for finite-dimensional Euclidean Jordan algebras in~\cite{wilceRoyal}.

\item 
In contrast compatibility of minimal dilations and kernels fails in the infinite-dimensional setting $\vNop$. 

\begin{proof}
Let $\hilbH = l^2(\mathbb{N})$. For each $i \in \mathbb{N}$ let $\ket{i} \in \hilbH$ be the sequence with value $1$ at entry $i$ and zero elsewhere. It induces a state $\psi_i$ on the algebra $A = B(\hilbH)$ via $a \mapsto \bra{i} a \ket{i} $. 
Now since the $\ket{i}$ span $\hilbH$, the effect
\[
e =
\sum^{\infty}
_{i=1} \frac{1}{2^i}\ket{i}\bra{i} \in A
\]
is internal. Suppose that $f \colon B \to A$ is a completely positive isomorphism that dilates $e$ as a morphism $A \to B$ in $\vNop$, so that $e=f(1)$. Then since any state $\omega$ of a C*-algebra satisfies $\|\omega\|=\omega(1)$, we have 
\[
\|f^{-1}\|
\geq
\frac{\|\psi_i \circ f \circ f^{-1}\|}{\|\psi_i \circ f\|}
=
\frac{\|\psi_i\|}{\|\psi_i \circ f\|}
=
\frac{\psi_i(1)}{(\psi_i \circ f)(1)}
=
\frac{1}{\psi_i(e)}
=
2^i
\]
making $f^{-1}$ unbounded, a contradiction.
\end{proof}
\item \label{ex:int-isom-in-MatS}
Generalising the case of $\mathbb{R}^+$ in $\Class$, let $S$ be a commutative semi-field satisfying $r + s = 0 \implies r = s = 0$ for all $r, s$. Then $\MatS$ is a theory with the internal isomorphism property. Indeed here an effect $e = (e_i)^n_{i=1} \colon n \to I$ is internal iff $e_i \neq 0$ for each $i$, and it has a dilation $n \to n$ given by the invertible diagonal matrix with entries $e_i$.
\end{exampleslist}
\end{examples}

\section{Perfect Distinguishability and Ideal Compressions}

The categorical principles we have met so far in fact closely relate to notions from the study of generalised probabilistic theories, and which appear as two major principles in the CDP quantum reconstruction~\cite{PhysRevA.84.012311InfoDerivQT}. In order to treat these, we will need to consider theories of a more probabilistic-like nature. 

We say that a theory has \indef{zero-cancellative scalars}\index{zero-cancellative!scalars} when for all scalars $r$ and morphisms $f$ we have $r \cdot f = 0 \implies r = 0$ or $f = 0$. In any ordered theory $(\catC, \discard{}, \leq)$ with such scalars we may define a pre-order on morphisms by 
\[  \label{not:facepreorder}
\scalebox{0.8}{\input{./figures/facepre.tikz}} \text{ for some non-zero $r$}
\]
which we call the \indef{face pre-order}\index{face pre-order} on each homset $\catC(A,B)$. In this section let us call a theory $(\catC, \discard{})$ \indef{suitable} \index{theory!suitable} when it is ordered, has zero-cancellative scalars, and has $e \preceq_F \discard{A}$ for all effects $e \colon A \to I$.

The above relation is often considered in probabilistic theories with (partial or total) addition, where states satisfy $\rho \preceq_F \sigma$ whenever $\sigma$ may be given by mixing $\rho$ with some other state. The face pre-order appears for example repeatedly (implicitly) in~\cite{PhysRevA.84.012311InfoDerivQT}. In fact in many settings it coincides with another naturally defined pre-order.
For any effects $d, e$ on the same object let us write  $d \preceq_K e$ when \label{not:ker-preorder}
\begin{equation} \label{eq:ker-preorder-effects}
\scalebox{0.8}{\input{./figures/eff-preorderiia.tikz}} \ 0 
\ \ 
\implies
\ \ 
\scalebox{0.8}{\input{./figures/eff-preorderiib.tikz}} \ 0 
\end{equation}
for all morphisms $f$.

\begin{lemma} \label{lem:internal-cmixed-preorderscollapse}
Let $\catC$ be a suitable theory with (co)kernels and order dilations.  
The following are equivalent:
\begin{enumerate}[label=\arabic*., ref=\arabic*]
\item \label{enum:compatt}
Minimal dilations and kernels are compatible;
\item \label{enum:coincidee}
For all effects $d, e$ we have $d \preceq_F e$ iff $d \preceq_K e$. 
\end{enumerate}
\end{lemma}
\begin{proof}
\ref{enum:compatt} $\implies$ \ref{enum:coincidee}:
For any effect $e$ we first show that $\discard{} \circ \coim(e) \preceq_F e$. 
Since $e$ factors over $\coim(e)$ by an internal effect, it suffices to assume that $e$ is internal and show that $\discard{A} \preceq_F e$. 

In this case $\coim(e) = \id{A}$ and so by compatibility $\ext{e}$ is a p-isomorphism. Let $f$ be a morphism and $r$ a non-zero scalar with $f \circ \ext{e} = r \cdot \id{A}$. Then 
\[
r \cdot \discard{A}
=
\discard{} \circ f \circ \ext{e}
\preceq_F
\discard{} \circ \ext{e}
=
e
\]
and so $\discard{A} \preceq_F e$ as desired.

Now for any effects $d, e$ we always have $d \preceq_F e \implies d \preceq_K e$, thanks to suitability and the rule \eqref{eq:lesszero}. Conversely if $d \preceq_K e$ then $d = c \circ \coim(e)$ for some effect $c$. But then
\[
d = c \circ \coim(e) \preceq_F \discard{} \circ \coim(e) \preceq_F e
\]
\ref{enum:coincidee} $\implies$ \ref{enum:compatt}:
Let $e$ be any effect, with $e = d \circ \coim(e)$ where $\ker(d) = 0$. Then $\discard{} \preceq_K d$ and so $\discard{} \preceq_F d$ giving $\discard{} \circ \coim(e) \preceq_F e$. Then there is some non-zero scalar $r$ and morphism $f$ with $r \cdot \coim(e) = f \circ \mdil{e}$. But then since $\coim(e)$ is an epimorphism and $\mdil{e} = \med_e \circ \coim(e)$ we have $r \cdot \id{} = f \circ \med_e$. Finally since $\med_e$ is epic we obtain $\med_e \circ f = r \cdot \id{}$ also, making $\med_e$ a p-isomorphism. 
\end{proof}

Let us now meet these two principles from the CDP quantum reconstruction.

\subsection{Perfect distinguishability} 

Say that a state $\omega$ is \indef{completely mixed}\index{completely mixed state} when every state $\rho$ of the same object has $\rho \preceq_F \omega$. The following is a slight adaptation of~\cite[Axiom 2]{PhysRevA.84.012311InfoDerivQT}.

\begin{definition}
A pair of states $\rho, \sigma \colon I \to A$ are said to be \deff{perfectly distinguishable} when there is a pair of effects $d, e $ on $A$ satisfying
\[
\scalebox{0.8}{\input{./figures/p-d.tikz}}
\qquad \ \ \ 
\scalebox{0.8}{\input{./figures/p-d-2a.tikz}}
\ 0 \ \ 
\scalebox{0.8}{\input{./figures/p-d-2b.tikz}}
\]
A suitable theory $(\catC, \discard{})$ satisfies \deff{perfect distinguishability} \index{perfect distinguishability} when every state which is not completely mixed is perfectly distinguishable from some non-zero state. 
\end{definition}

We view the effects $d, e$ above as a procedure which determines with maximal certainty which of the states $\rho, \sigma$ the system has been prepared in. 

\begin{lemma} \label{lem:PDLem}
Let $\catC$ be a suitable dagger theory with causal dagger kernels and let $\rho$ be any state. Then any states $\sigma, \tau$ with $ \sigma \preceq_F \rho$ and $\tau^\dagger \circ \rho  = 0$ are perfectly distinguishable via 
\[
\scalebox{0.8}{\input{./figures/pd-lemma.tikz}}
\]
\end{lemma}

\begin{proof}
Since $\sigma \preceq_F \rho$ we have 
$\coker(\rho) \circ \sigma \preceq_F \coker(\rho) \circ \rho = 0$.
Hence by suitability $\sigma = \img(\rho) \circ a$ for some state $a$, and by assumption $\tau = \img(\rho)^{\bot} \circ b$ for some $b$. Then
\begin{align*}
d \circ \sigma 
&= 
\discard{} \circ \img(\rho)^{\dagger} \circ \img(\rho) \circ a
=
\discard{} \circ a
=
\discard{} \circ \img(\rho) \circ a
=
\discard{} \circ \sigma
\\
e \circ \sigma 
&= 
\discard{} \circ \coker(\rho) \circ \img(\rho) \circ a
=
0
\end{align*}
and in just the same way $d \circ \tau = 0$ and $e \circ \tau = \discard{} \circ \tau$.
\end{proof}

\subsection{Ideal compressions}

The next principle allows us to identify, for any state $\rho$, the collection of states $\sigma$ satisfying $\sigma \preceq_F \rho$ with a particular system in our theory. 

\begin{definition} \label{def:ideal-compression-schemes}~\cite[Axiom 3]{PhysRevA.84.012311InfoDerivQT}\label{not:idealcomp}
A suitable theory $\catC$ is said to have \deff{ideal compressions}\index{ideal compression} when for every state $\rho \colon I \to A$ there is an object $\idob_\rho$, and a causal morphism $\dec_\rho \colon \idob_\rho \to A$ with a left inverse $\enc_\rho$, i.e.~with 
\[
\enc_\rho \circ \dec_\rho = \id{\idob_\rho}
\]
and such that every $f \colon B \to A$ with $f \circ \sigma \preceq_F \rho$ for all states $\sigma$ factors over $\dec_\rho$:
\[
\begin{tikzcd}
B \arrow[dr,swap,"f"]  \arrow[r,dotted, "\exists ! g"]& \idob_\rho \dar[shift right = 2.5,swap]{\dec_\rho}
\\ & A \arrow[u,"\enc_\rho", swap,shift right = 2.5]
\end{tikzcd}
\]
\end{definition}

The morphisms $\dec_\rho$ and $\enc_\rho$ are called `decoding' and `encoding' maps for $\rho$, respectively; note that $\enc_\rho$ is not unique in general. The original formulation by CDP is given in a slightly different form shown to be equivalent to the above in their context~\cite[Lemma 2]{PhysRevA.84.012311InfoDerivQT}. We can now relate both principles to our earlier categorical features.

\begin{theorem} \label{thm:mainthmonPDIComp}
Let $\catC$ be a suitable dagger theory with causal dagger kernels and order dilations. 
The following are equivalent:
\begin{enumerate}[label=\arabic*., ref=\arabic*]
\item \label{enum:commpat}
Minimal dilations and kernels are compatible;
\item \label{enum:pdd}
Perfect distinguishability holds in $\catC$.
\end{enumerate}
Moreover in this case $\catC$ has ideal compressions.
\end{theorem}
\begin{proof}
\ref{enum:compatt} $\implies$ \ref{enum:pdd}:
Applying the dagger we see that perfect distinguishability is equivalent to requiring that any effect $e \colon A \to I$ with $\ker(e) = 0$ has that $e^\dagger$ is completely mixed. But by suitability the latter holds iff $\discard{} \preceq_F e$. Now if $\ker(e) = 0$ then $\coim(e) = \id{A}$ and so by compatibility $e = \discard{} \circ \med_e$ for the p-isomorphism $\med_e$. Then letting $f \circ \med_e = r \cdot \id{A}$ for some morphism $f$, and non-zero scalar $r$, we have that 
\[
\discard{} \preceq_F r \cdot \discard{} = \discard{} \circ f \circ \med_e \preceq_F \discard{} \circ \med_e = e
\] 
\ref{enum:pdd} $\implies$ \ref{enum:commpat}:
For any effect $e$, by construction $\discard{} \circ \med_e$ is internal and so by the (dagger of) perfect distinguishability we have $\discard{} \preceq_F \discard{} \circ \med_e$, so that $\discard{} \circ \coim(e) \preceq_F e$. Then just as in Lemma~\ref{lem:internal-cmixed-preorderscollapse} this ensures that $\preceq_K$ and $\preceq_F$ coincide, ensuring compatibility.

Now when these hold we claim that any state $\rho$ has ideal compression scheme 
\begin{equation*} 
\idob_\rho := \Img(\rho) \quad \dec_\rho := \img(\rho) \quad \enc_\rho := \img(\rho)^\dagger
\end{equation*}
By Lemma~\ref{lem:internal-cmixed-preorderscollapse} it suffices to verify the definition of an ideal compression replacing $\preceq_F$ by $\preceq_K$. Note that each object has a completely mixed state $\discardflip{A} := \discard{A}^\dagger$.

Firstly, we have $\img(\rho) \circ \discardflip{} \preceq_K \rho$ since $f \circ \rho = 0 \implies f \circ \img(\rho) = 0$. Now suppose that $f \circ \discardflip{} \preceq_K \rho$. Then since $\coker(\rho) \circ \rho = 0$ we have $\coker(\rho) \circ f \circ \discardflip{} = 0$ and so $\coker(\rho) \circ f = 0$. Hence $f$ factors over $\img(\rho)$ as desired. 
\end{proof}

The fact that the ideal compressions arise from kernels, and the behaviour of the (dagger) idempotents $A \to A$ induced by $\img(\rho)$ and $\img(\rho)^\bot$ as picking out those states in the face of, and perfectly distinguishable from $\rho$, respectively, forms a major part of the CDP reconstruction~\cite[Section 11]{PhysRevA.84.012311InfoDerivQT}. Complementary projections of this form associated with effects are also prominent in Alfsen and Shultz's axiomatisation of state spaces of C*-algebras~\cite[Chapters 7,8]{alfsen2012geometry}.

We may see the use of the maps $\img(\rho)$ as a reformulation of ideal compression applicable to infinite dimensions where the conditions of Lemma~\ref{lem:internal-cmixed-preorderscollapse} typically fail. In Section~\ref{sec:Pure-excl} we will meet another related principle to perfect distinguishability.

\begin{remark} \label{remark:IdealCompForEffects}
Inspecting Theorem~\ref{thm:mainthmonPDIComp} we see that, without requiring daggers, any suitable theory with compatible (co)kernels and order dilations satisfies a dual form of ideal compression. That is, each effect $e \colon A \to I$ 
has a universal morphism $\coim(e) \colon A \to F_e$ over which all morphisms with $\discard{} \circ f \preceq_F e$ factor.
\end{remark}

\begin{examples} \label{examples:Ideal-Comp}
$\KlDT$, $\Rel$, $\Quant{}$ and $\FCStar$ are all suitable, and as we've seen satisfy (strong) compatibility, making $\preceq_F$ and $\preceq_K$ coincide, and have ideal compressions given by their (co)image maps as above.

$\PFun$ lacks a dagger or completely mixed states on arbitrary objects, with every state being empty or singleton. Nonetheless here $\preceq_F, \preceq_K$ and $\leq$ all coincide, and perfect distinguishability and ideal compression are satisfied trivially.
\end{examples}

\section{Purification}

The principles we have examined so far are equally true of quantum, classical and more general physical theories. The remainder of this chapter will focus on a principle characteristic of quantum theory itself.

 A major aspect of the quantum world is that every process may be seen to arise, due to ignorance of certain degrees of freedom, from one of maximal knowledge or sharpness. 
In our framework we can characterise such processes as follows. 

\begin{definition} \label{def:tensor-pure}
In any theory we say that a morphism $f \colon A \to B$ is \deff{pure}\index{pure}, or \deff{$\otimes$-pure}\index{pure!$\otimes$-pure|see {pure}}, when either $f=0$ or $f$ satisfies 
\begin{equation} \label{eq:wholedef}
\scalebox{0.8}{\input{./figures/whole-def.tikz}}
\ \ \ 
\text{  for some causal $\rho$}
\end{equation}
\end{definition}

This characterisation of purity was put forward by Giulio Chiribella~\cite{chiribella2014distinguishability}, and we discuss more standard accounts of purity shortly in Section~\ref{sec:alt-purity}. In quantum theory every process arises from such a pure one in the following manner.

\begin{definition} \label{def:purification} 
We say that a theory $(\catC, \discard{})$ \deff{has dilations with respect to} a class of morphisms $\catC_\prepure$ when every morphism has a dilation in $\catC_\prepure$\label{not:prepuresubcat}:
\[
(\forall f)
\ \ 
(\exists g \in \catC_\prepure)
\ \ 
\text{ s.t. }
\quad
\scalebox{0.8}{\input{./figures/purif-def.tikz}}
\]
and that these dilations are \deff{essentially unique} when for every pair of morphisms $f, g \colon A \to B \otimes C$ in $\catC_\prepure$ we have
\begin{equation} \label{eq:EU}
\scalebox{0.8}{\input{./figures/EU-purif.tikz}}
\end{equation}
for some causal isomorphism $U \colon C \to C$ with $U \in \catC_\prepure$. 

When $\catC$ has dilations with respect to the class $\catC_\pure$\label{not:puresubcat} of $\otimes$-pure morphisms we say that $\catC$ satisfies \deff{purification}\index{purification}\index{purification!essentially unique}. We say $\catC$ has \deff{essentially unique purification} when these dilations are essentially unique. In either case we call any $\otimes$-pure dilation of a morphism $f$ a \deff{purification} of $f$.
\end{definition}

Probabilistic theories with a similar form of essentially unique purification are studied by CDP in~\cite{chiribella2010purification}, being shown to share many features of quantum theory, and this forms the central principle of~\cite{PhysRevA.84.012311InfoDerivQT}. Purification is also the basis for numerous constructions in categorical quantum mechanics~\cite{coecke2008axiomatic,CKbook}. In this context we are usually interested 
in the case when $\catC_\pure$ is a monoidal subcategory of $\catC$, being closed under $\circ$, $\otimes$ and containing all identity morphisms, as holds in quantum theory and appears as an extra axiom in~\cite{PhysRevA.84.012311InfoDerivQT}.

\subsection{Reversible dilations}

Several known consequences of essential uniqueness for probabilistic theories can be immediately extended to our basic setting. Firstly, essential uniqueness extends to morphisms with different types, as in~\cite[Lemma 21]{chiribella2010purification}. 

\begin{lemma} \label{lemma:purif_env_connect} 
Let $\catC$ be a theory with essentially unique purification and $\otimes$-pure morphisms closed under $\otimes$. Let $B, C$ be objects each possessing a causal pure state. Then for all pure morphisms $f \colon A \to B \otimes C$ and $g \colon A \to D \otimes C$ we have
\begin{equation} \label{eq:EU-channels}
\scalebox{0.8}{\input{./figures/channel4i.tikz}}
\quad
\text{ where }
\quad
\scalebox{0.8}{\input{./figures/channel3.tikz}}
\end{equation}
for some isomorphism $U$ on $B \otimes C$ and state $\phi$ of $C$ which are causal and $\otimes$-pure.
\end{lemma}
\begin{proof}
Let $\psi, \phi$ be causal $\otimes$-pure states of $B, C$, respectively. Then 
\[
\scalebox{0.8}{\input{./figures/channel1.tikz}}
\quad
\text{ and so }
\quad 
\scalebox{0.8}{\input{./figures/channel2.tikz}}
\]
for some causal $\otimes$-pure isomorphism $U$ on $C \otimes D$. Applying $\discard{C}$ yields the result.
\end{proof}

Purification can be seen to encode the idea of an underlying (pure) physics which is ultimately reversible, in that any causal process arises via ignorance from some larger reversible one. More precisely, following~\cite{chiribella2010purification} let us say that a morphism $f \colon A \to B$ has a \indef{reversible dilation} when it has a dilation of the form 
\[
\scalebox{0.8}{\input{./figures/reverse-dilation2ii.tikz}}
\]
for some causal $\otimes$-pure state $\phi \colon I \to D$ and isomorphism $U \colon A \otimes D \to B \otimes C$. Then as in~\cite[Thm.~15]{chiribella2010purification} we have the following.

\begin{corollary} \label{cor:rev-dilation}
Let $\catC$ be a theory with essentially unique purifications such that every non-zero object has a causal pure state. Then every causal morphism has a reversible dilation.
\end{corollary}
\begin{proof}
Let $f \colon A \to B$ be a causal morphism with some purification $g \colon A \to B \otimes C$. Then since $\discard{} \circ g = \discard{A}$, by Lemma~\ref{lemma:purif_env_connect} we have 
\[
\scalebox{0.8}{\input{./figures/rev-EU-proof1.tikz}}
\]
for some causal pure state $\phi$ of $B \otimes C$ and causal isomorphism $U$ on $A \otimes B \otimes C$. But then 
\[
\scalebox{0.8}{\input{./figures/rev-EU-proof2.tikz}}
\]
providing $f$ with a reversible dilation. 
\end{proof}
Whenever $\catC_\pure$ is closed under composition, reversible dilations are indeed $\otimes$-pure and hence satisfy the essential uniqueness properties of~\eqref{eq:EU} and~\eqref{eq:EU-channels}.

\subsection{Alternative notions of purity} \label{sec:alt-purity}

This notion of $\otimes$-purity differs at first sight from the typical concept of purity in probabilistic theories, used for example in~\cite{chiribella2010purification}, based on coarse-graining. In theories with extra structure we may consider purity in this sense, as follows.

\begin{definition}
We call a morphism $f \colon A \to B$ in a theory with 
\begin{itemize}[leftmargin=2em]
\item 
addition \deff{$+$-pure} \index{pure!$+$-pure} when 
 $f = g + h \implies g = r \cdot f$ for some scalar $r$;
\item 
an order \deff{$\leq$-pure} \index{pure!$\leq$-pure} when $g \leq f \implies g = r \cdot f$ for some scalar $r$.
\end{itemize}
\end{definition}
In a theory ordered by addition both of these notions coincide. In fact they typically coincide with $\otimes$-purity, as the following suggests.
\begin{proposition} \label{prop:wholetoatomic}
Let $\catC$ be a theory with addition containing a pair of perfectly distinguishable causal states. Then in $\catC$ any $\otimes$-pure morphism is $+$-pure. 
\end{proposition}
\begin{proof}
Let $f \colon A \to B$ be $\otimes$-pure and suppose that $f = g + h$ for some $g, h \colon A \to B$. 
Let $C$ be any object with a pair of states $\ket{0}$, $\ket{1}$ which are perfectly distinguishable via some effects $e_0$ and $e_1$. Define
\[
\scalebox{0.8}{\input{./figures/whole_variation0.tikz}}
\]
Then
\[
\scalebox{0.8}{\input{./figures/whole_variation.tikz}}
\quad
\text{and so}
\quad
\scalebox{0.8}{\input{./figures/whole_variation2.tikz}}
\]
for some causal state $\rho$. But then
\[
\scalebox{0.8}{\input{./figures/whole_variation3.tikz}}
\]
and so $g$ is a scalar multiple of $f$ as required. 
\end{proof}

Aside from these, several further categorical definition of purity have appeared in the literature. Our earlier notion of purity coincides with that due to Coecke and Selby~\cite{selby2017leaks} whenever all identity morphisms are $\otimes$-pure, so that:
\begin{equation} \label{eq:noleaks}
\scalebox{0.8}{\input{./figures/noleaks.tikz}}
\end{equation}
This is called having `no leaks' in~\cite{selby2017leaks}. A categorical definition of purity has also been introduced by A.~and B.~Westerbaan in the context of effectus theory~\cite{Paschke2016,westerbaan2018dagger}.

Elsewhere, Cunningham and Heunen have introduced the following notion of purity which arises in a very general setting and is categorically well-behaved~\cite{cunningham2017purity}. A morphism  $f \in \catC_\prepure$ is called \indef{copure}\index{pure!copure} when it satisfies
\begin{align}
\scalebox{0.8}{\input{./figures/copure.tikz}} \label{eq:copure1} \\ 
 \text{ for some $k$ with } \quad   \scalebox{0.8}{\input{./figures/copure2.tikz}} \nonumber
\end{align}

\begin{lemma} \label{lem:copure}
Let $\catC$ be a theory with essentially unique dilations with respect to a monoidal subcategory $\catC_\prepure$, and suppose that every non-zero object has a causal state in $\catC_\prepure$. Then any morphism $f \in \catC_\prepure$ is copure.
\end{lemma}
\begin{proof}
Let $f \in \catC_\prepure$ and suppose the left hand of~\eqref{eq:copure1} is satisfied. To establish the right-hand side it suffices to consider the case when $h \in \catC_\prepure$.
Let $l \colon B \to D \otimes F$ be a dilation of $g$ with $l \in \catC_\prepure$.
If $E$ or $F$ are zero objects then $h = 0$ and $g=0$ making the result trivial. Otherwise let $\psi$ and $\phi$ be causal states of $E, F$ respectively belonging to $\catC_\prepure$. Then we have
\[
\scalebox{0.8}{\input{./figures/copurearg-new.tikz}}
\]
and so by essential uniqueness there is some causal isomorphism $U$ with  
\[
\scalebox{0.8}{\input{./figures/copurearg-new2.tikz}}
\]
Then the morphism above $f$ on the right-hand side dilates $g$ as required. 
\end{proof}

Morphisms satisfying~\eqref{eq:copure1} are automatically closed under composition and $\otimes$, and contain all isomorphisms, and they will be $\otimes$-pure whenever the `no leaks' condition~\eqref{eq:noleaks} is satisfied. 

Now in fact if we wish to assume a kind of essentially unique purification, for any notion of purity satisfying some basic features, then the notion of $\otimes$-purity is forced upon us, as we now show.
Let us say that a class of morphisms $\catC_\prepure$ is \indef{$\otimes$-complete}\index{$\otimes$-complete} when it contains all zero morphisms and for all morphisms $f$ and causal states $\sigma$ we have 
\[
\scalebox{0.8}{\input{./figures/closed-under-statesi.tikz}} \in \catC_\prepure 
\implies
\scalebox{0.8}{\input{./figures/closed-under-statesii.tikz}} \in \catC_\prepure 
\]

\begin{proposition} \label{prop:pure-lem}
Let $\catC$ be a theory with essentially unique dilations with respect to a class of morphisms $\catC_\prepure$ which is closed under $\otimes$ and such that every non-zero object has a causal state in $\catC_\prepure$. Suppose further that $\catC_\prepure$ is $\otimes$-complete. Then a morphism  belongs to $\catC_\prepure$ iff it is $\otimes$-pure.
In particular $\catC$ has purification.
\end{proposition}
\begin{proof}
Let $f \colon A \to B$ be non-zero and belong to $\catC_{\prepure}$, and suppose that $f$ has a dilation $g \colon A \to B \otimes C$. Dilating $g$ if necessary, we may assume that $g \in \catC_{\prepure}$. Let $\psi$ be any causal state of $C$ belonging to $\catC_\prepure$. 
Then 
\[
\scalebox{0.8}{\input{./figures/whole_arg3.tikz}}
\quad
\text{ and so }
\quad
\scalebox{0.8}{\input{./figures/whole_arg4.tikz}}
\]
for some causal isomorphism $U$, with $U \circ \psi$ then being causal as desired.

Conversely, suppose that $f \colon A \to B$ is $\otimes$-pure and non-zero. Let $g \colon A \to B \otimes C$ be a dilation of $f$ with $g \in \catC_{\prepure}$. Then we have
\[
\scalebox{0.8}{\input{./figures/whole_arg_new1.tikz}}
\]
for some causal state $\sigma$ of $C$. Hence by assumption $f$ belongs to $\catC_\prepure$.
\end{proof}

In particular asking for such essentially unique dilations with respect to any of the other classes of morphisms we've considered is equivalent to that in terms of $\otimes$-purity, as the next result shows.

\begin{lemma} \label{lem:caustensclosedexamples}
In any theory in which they may be defined, the classes of 
\begin{inlinelist}
\item \label{enum:wholecomp}
$\otimes$-pure 
\item \label{enum:leqcomp}
$\leq$-pure
\item \label{enum:pluscomp}
$+$-pure 
\item \label{enum:copure}
copure
\end{inlinelist} 
morphisms are all $\otimes$-complete.
\end{lemma}
\begin{proof}
Let $f \colon A \to B$ be non-zero and $\sigma \colon I \to C$ any causal state and suppose that $f \otimes \sigma$ belongs to each class in question.

\ref{enum:wholecomp}
If $g \colon A \to B \otimes D$ dilates $f$ then $g \otimes \sigma$ dilates $f \otimes \sigma$ and so for some state $\rho$ 
\[
\scalebox{0.8}{\input{./figures/closed-under-statesarg.tikz}}
\quad
\text{ giving }
\quad
\scalebox{0.8}{\input{./figures/closed-under-statesargi.tikz}}
\]

\ref{enum:leqcomp}
Suppose that $g \leq f$. Then $g \otimes \sigma \leq f \otimes \sigma$ so that $g \otimes \sigma = r \cdot f \otimes \sigma$ for some scalar $r$. Then taking marginals gives $g = r \cdot f$.

\ref{enum:pluscomp}
Suppose that $f = g + h$. Then $f \otimes \sigma = g \otimes \sigma + h \otimes \sigma$, so that $g \otimes \sigma = r \cdot f \otimes \sigma$ for some scalar $r$, again giving $g = r \cdot f$.

\ref{enum:copure} 
Suppose the left hand side of~\eqref{eq:copure1} is satisfied, replacing the label $C$ there by $F$, for some $g \colon B \otimes F \to D$. Then placing the state $\sigma$ to the left of this equation yields that 
\[
\scalebox{0.8}{\input{./figures/copureextra1.tikz}}
\quad
\text{ so that }
\quad
\scalebox{0.8}{\input{./figures/copureextra2.tikz}}
\]
for some $k$ dilating $\id{C} \otimes g$. But then the morphism above $f$ on the right-hand side is a dilation of $g$ as required.
\end{proof}

\begin{remark}
To extend the notion outside settings without zero morphisms, we may instead define a morphism to be $\otimes$-pure whenever it satisfies~\eqref{eq:wholedef} with the state $\rho$ being only required to be \emps{locally causal} in that 
\[
\scalebox{0.8}{\input{./figures/locally_causal.tikz}}
\]
If $C$ has a causal state then any such $\rho$ is in fact causal. However we will not pursue this here. 
\end{remark}

\subsection{Examples} \label{examples:Purif}

\begin{exampleslist}
\item \label{ex:purif-Quant}
$\Quant{}$ has essentially unique purification. It is well known that a completely positive map $B(\hilbH) \to B(\hilbK)$ is $+$-pure here precisely when it is a Kraus map $\Dbl{f}$ for some linear map $f \colon \hilbH \to \hilbK$. Every completely positive map may be dilated to such a map via its Stinespring dilation~\cite{stinespring1955positive,Paschke2016}, and these are essentially unique as discussed in depth in~\cite{chiribella2010purification}. Hence by Proposition~\ref{prop:wholetoatomic} essentially unique purification indeed holds with $\otimes$-purity and $+$-purity coinciding, and in fact they also coincide with copurity~\cite{cunningham2017purity}. Such pure morphisms are closed under composition and form the dagger compact subcategory 
\[
(\Quant{})_\pure \simeq \FHilbP
\]
as remarked in Chapter~\ref{chap:CatsWDiscarding}.

\item \label{examples:purif-in-CPM}
Generalising the previous example, each theory of the form $\CPM(\catA)$ has dilations with respect to the dagger-compact subcategory of all morphisms of the form 
\[
\scalebox{0.8}{\input{./figures/CPDbli.tikz}}
\]
for some morphism $f \colon A \to B$ in $\catA$, or equivalently the image $\Dbl{\catA}$\label{not:doublesubcat} of the functor $\Dbl{(-)} \colon \catA \to \CPM(\catA)$. We meet some sufficient conditions for this to give $\CPM(\catA)$ essentially unique purification in Chapter~\ref{chap:recons}.
\item \label{example:pure-in-MSpek}
$\MSpek$ has essentially unique purification, with a morphism being $\otimes$-pure iff it is $\leq$-pure iff it belongs to $\Spek$.
\begin{proof}
First we show that a morphism is $\leq$-pure iff it belongs to $\Spek$. Firstly, suppose that $f \colon A \to B$ is $\leq$-pure. Then it has some dilation $g \colon A \to B \otimes C$ belonging to $\Spek$. Now from the inductive definition of $\Spek$ it follows that there is some effect $\psi$ for which 
\[
\scalebox{0.8}{\input{./figures/Spek-dil2.tikz}}
\]
is non-zero. But then 
\[
\scalebox{0.8}{\input{./figures/Spek-dil3.tikz}}
\quad
\text{ and so }
\quad 
\scalebox{0.8}{\input{./figures/Spek-dil4.tikz}}
\]
since $f$ is $\leq$-pure, giving $f \in \Spek$.
For the converse, by dagger compactness it suffices to check that each state of $\Spek$ is $\leq$-pure. But by~\cite[Theorem 5.14, 5.29]{coecke2012spekkens} every non-zero state $\rho$ of $\IV^n$ in $\MSpek$ has $|\rho| \geq 2^n$, while those in $\Spek$ have $|\rho| = 2^n$. Hence whenever $\psi \leq \rho$ we must have $\psi = 0$ or $\psi = \rho$. 

Next, let us turn to essential uniqueness. 
For this we use that states in $\MSpek$ can be equivalently represented by their \emps{stabilizer groups}~\cite{pusey2012stabilizer}. It is known that a state of $\IV^n$ belongs to $\Spek$ iff its stabilizer group is composed of the minimum possible number of independent generators $n$ (see~\cite{backens2016complete}, particularly $\S$4.3). Disilvestro and Markham have shown that every state in $\MSpek$ has an essentially unique dilation to a state with this property~\cite[Theorem~2]{disilvestro2017quantum}. Using compactness, the fact that $\Spek$ satisfies the conditions of Proposition~\ref{prop:pure-lem} now makes the class of such states coincide with the class of $\otimes$-pure ones, providing essentially unique purification.
\end{proof}

\item 
Purification in our sense fails in $\Rel$ and $\Class$, and in each theory $+$-purity, $\leq$-purity and $\otimes$-purity all coincide.

In $\Rel$ a morphism is $\otimes$-pure iff it it is a singleton relation $R= \{(a,b)\}$. Similarly in $\Class$ a morphism $f \colon A \to B$ is $\otimes$-pure iff there are unique $a \in A$ and $b \in B$ for which $f(a)(b)$ is non-zero.

In each of these theories copurity is more well behaved, providing them with an alternative notion of purification~\cite{cunningham2017purity}.
\end{exampleslist}

\subsection{Deriving purification}

Purification can in fact to be seen to arise from little more than the categorical principles from earlier in this chapter. Let us say that a theory $\catC$ is \indef{$\otimes$-pure} \index{theory!$\otimes$-pure} when all of its identity morphisms are, as in~\eqref{eq:noleaks}.
 
\begin{lemma} \label{lem:whatswhole} \label{lem:helpful-tensPureKer}
In any theory which is $\otimes$-pure so is any:
\begin{inlinelist} 
\item \label{enum:ext-whole}
minimal dilation 
\item  \label{enum:coker-whole}
cokernel 
\item \label{enum:ker-whole}
kernel $k$ satisfying~\eqref{eq:kernels}. Moreover, for any such $k$ and morphism $g$, if $k \circ g$ is $\otimes$-pure then so is $g$.
\end{inlinelist}
\end{lemma}
\begin{proof}
\ref{enum:ext-whole}
Let $\mdil{f} \colon A \to B \otimes C$ be the minimal dilation of $f \colon A \to B$, and let $g$ be any dilation of $\mdil{f}$ via some object $D$. Then we have implications
\[
\scalebox{0.8}{\input{./figures/ext-whole.tikz}}
\]
for some causal morphism $h$. But then by the definition of $\mdil{f}$ we have
\[
\scalebox{0.8}{\input{./figures/ext-whole2.tikz}}
\quad 
\text{ and so }
\quad
\scalebox{0.8}{\input{./figures/ext-whole3.tikz}}
\]
for some causal state $\rho$, or $h = 0$. Hence $g$ splits as desired.

\ref{enum:coker-whole}
Let $g$ be a dilation of $c = \coker(f) \colon B \to C$ for some $f \colon A \to B$. Then 
\[
\scalebox{0.8}{\input{./figures/cokernelsarewholei.tikz}} \ 0 
\ \implies \ 
\scalebox{0.8}{\input{./figures/cokernelsarewholeii.tikz}} \ 0 
\ \implies \ 
\scalebox{0.8}{\input{./figures/cokernelsarewholeiii.tikz}}
\]
Then since $c$ is epic just as in the previous part $h$ is zero or a dilation of $\id{C}$ and so $g$ splits as desired.

\ref{enum:ker-whole}
Let $g$ be a dilation of $k = \ker(f) \colon K \to A$ for some $f \colon A \to B$. Then we have implications:
\[
\scalebox{0.8}{\input{./figures/kernelsarewholei.tikz}} \ 0 
\ \implies \ 
\scalebox{0.8}{\input{./figures/kernelsarewholeii.tikz}} \ 0 
\ \implies \ 
\scalebox{0.8}{\input{./figures/kernelsarewholeiii.tikz}}
\]
for some unique morphism $h$, 
since $\ker(f \otimes \id{C}) = \ker(f) \otimes \id{C}$. 
Then by uniqueness $h$ is a dilation of $\id{K}$ and as in the previous parts this yields a splitting for $g$.

For the final statement let $f = k \circ g \colon A \to B$ be $\otimes$-pure and let $h \colon A \to K \otimes C$ be a dilation of $g$. Then $(k \otimes \id{C}) \circ h$ is a dilation of $f$ and so for some causal state $\rho$ we have 
\[
\scalebox{0.8}{\input{./figures/tens-pureker.tikz}}
\]
By assumption $k \otimes \id{C}$ is monic, and so $h$ splits as desired.
\end{proof}

\begin{corollary} \label{cor:purifFromMinDil}
Let $\catC$ be any $\otimes$-pure theory with minimal dilations. Then $\catC$ has purification.
\end{corollary}

In such a theory each minimal dilation $\mdil{f}$ forms a purification for $f$. Such `minimal purifications' are considered for states in~\cite[Theorem 4]{PhysRevA.84.012311InfoDerivQT}. Hence we may then view the presence of minimal dilations as a generalisation of (minimal) purifications which holds classically. Another extension of purification to this setting is found in~\cite{selby2017leaks,selby2018reconstructing}. 

This result also gives another means of deriving a form of purification. Say that a theory has \indef{effect purification} when every effect has a $\otimes$-pure dilation.

\begin{corollary} \label{cor:Eff-pUrif}
Let $\catC$ be any $\otimes$-pure theory with (co)kernels and satisfying the internal isomorphism property. Then $\catC$ has effect purification.
\end{corollary}

\begin{proof}
By Proposition~\ref{prop:Int-Isom-Deduce} and Corollary~\ref{cor:purifFromMinDil}. 
\end{proof}

We can also consider when minimal dilations satisfy the other notions of purity from Section~\ref{sec:alt-purity}. Let us say that an ordered theory $\catC$ is \indef{$\leq$-pure}\index{theory!$\leq$-pure} when every identity morphism is $\leq$-pure. 

We call a kernel $k$ \indef{split}\index{kernel!split} when it is split monic, i.e.~there is some $f$ with $f \circ k = \id{}$. Dually a cokernel $c$ is split when $c \circ g = \id{}$ for some morphism $g$. Any dagger (co)kernel is split by definition.

\begin{lemma} \label{lem:whatsatomic}
Let $\catC$ be an ordered theory which is $\leq$-pure. Then so is any split kernel, split cokernel, or morphism  $\dec_\rho$ of an ideal compression scheme.
\end{lemma}
\begin{proof}
Let $k = \ker(f) \colon K \to A$ having a splitting $l$, for some $f \colon A \to B$. Then if $g \colon K \to A$ has $g \leq k$ then $f \circ g = 0$ and so $g = k \circ h$ for some $h$ as below.
\[
\begin{tikzcd}
K \rar[shift left = 2.5]{k} & A \lar[shift left = 2.5]{l} \rar{f} & B \\ 
K \arrow[ur,swap,"g"] \uar{h}
\end{tikzcd}
\]
But then $h = l \circ g \leq l \circ k = \id{K}$. 
Hence for some scalar $r$ we have $h = r \cdot \id{K}$ so that $g = r \cdot k$ as required. 
The result for cokernels follows dually. 

Finally if $\dec_\rho$ is as above and $g \leq \dec_\rho$ then $g \circ \sigma \preceq_F \rho$ for all states $\sigma$ and so $g$ factors over $\dec_\rho$. Since $\dec_\rho$ is split by definition it follows again that $g = r \cdot \rho$ for some scalar $r$.
\end{proof}

In many cases minimal dilations are pure in the other senses we have considered. 
\begin{theorem} \label{thm:deriveleqpure}
Let $\catC$ be an ordered theory which is $\leq$-pure. 
\begin{enumerate}
\item \label{enum:derive-atomic}
If $\catC$ is ordered by a total addition $+$ then any order dilation is $\leq$-pure. 
\item  
\label{enum:derive-atomic15}
If $\catC$ is ordered by some $\ovee$ making it a sub-causal category (Chapter~\ref{chap:totalisation}) and which is cancellative on effects, then any order dilation of an effect is $\leq$-pure.
\item \label{enum:derive-atomic2}
If $\catC$ has strongly compatible split (co)kernels and minimal dilations, any minimal effect dilation is $\leq$-pure.
\end{enumerate}
\end{theorem}
\begin{proof}
\ref{enum:derive-atomic} Let $\mdil{f} \colon A \to B \otimes C$ be an order dilation of $f \colon A \to B$ and suppose that $g \leq \mdil{f}$. Then we have $g + h = \mdil{f}$ for some morphism $h$. 
But now 
\[
\scalebox{0.8}{\input{./figures/order-dil-proof1.tikz}}
\]
and so there are unique morphisms $l, m \colon C \to C$ with 
\[
\scalebox{0.8}{\input{./figures/order-dil-proof2.tikz}}
\]
But then by uniqueness property of minimal dilations $l + m = \id{C}$, 
so that $l \leq \id{C}$. Hence for some scalar $r$ we have $l = r \cdot \id{C}$ and then $g = r \cdot \mdil{f}$.

\ref{enum:derive-atomic15} 
Note that every effect $e$ now has a unique $e^\bot$ with $e \ovee e^\bot = \discard{}$.
The proof is similar to the previous part: for any effect $e$, if $g \leq \mdil{e}$ we have $g \ovee h = \mdil{e}$ for some  $h$, giving unique morphisms $l, m$ with 
$g = l \circ \mdil{e}$ and $h = m \circ \mdil{e}$. Now we have
\begin{align*}
(\discard{} \circ l \circ \mdil{e}) \ovee (\discard{} \circ m \circ \mdil{e})
&=
\discard{} \circ (g \ovee h)
\\
&=
\discard{} \circ \mdil{e}
\\
&=
(\discard{} \circ l \circ \mdil{e}) \ovee ((\discard{} \circ l)^\bot \circ \mdil{e})
\end{align*}
Hence by assumption and epicness of $\mdil{e}$ we have $(\discard{} \circ l)^\bot = \discard{} \circ m$ so that $l \ovee m$ is defined. Again by epicness we have $l \ovee m = \id{C}$ so that $l = r \cdot \id{C}$ and $g = r \cdot \mdil{e}$ for some non-zero scalar $r$.

\ref{enum:derive-atomic2}
By strong compatibility for any effect $e$ we have $\mdil{e} = \med_e \circ \coim(e)$ for some isomorphism $\med_e$. By Lemma~\ref{lem:whatsatomic} $\coim(e)$ is $\leq$-pure and then it easily follows that $\mdil{e}$ is also. 
\end{proof}

\begin{examples}
We've seen that $\Quant{}$ and $\QuantSU$ are both $\otimes$-pure and $\leq$-pure. Hence any kernel, cokernel or minimal dilations in either is pure. In fact in Example~\ref{examples:min-dilations}~\ref{example:QuantMinDil} we already saw that minimal dilations in $\Quant{}$ are given by minimal Stinespring dilations and in Example~\ref{examples:kernels}~\ref{example:QuantKernels} that kernels here are induced from $\FHilb$.
\end{examples}

\section{Pure Exclusion} \label{sec:Pure-excl}

In our main theories of interest the pure states of any suitable system always satisfy an extra property, namely that they may be \emps{excluded} by some experimental test. This provides us with a natural further principle to consider. 

Let us call an object $A$ \indef{trivial} \index{trivial object} when  $\discard{A} \colon A \to I$ is an isomorphism, or $A$ is a zero object, and a theory \indef{trivial} \index{theory!trivial} when every object is trivial. 

\begin{definition} \label{def:Pure-Exclusion}
A theory $(\catC, \discard{})$ satisfies \deff{pure exclusion} \index{pure exclusion} when for every $\otimes$-pure state $\psi$ of a non-trivial object $A$ there is a non-zero effect $e$ with 
\[
\scalebox{0.8}{\input{./figures/pure-excl-effect.tikz}} \ 0
\]
\end{definition}

Equivalently, no such state $\psi$ is zero-epic. For any such pure state $\psi$ we think of $e$ as a potentially observable effect which tells us that the system is not currently in state $\psi$. In probabilistic theories of the form of~\cite{PhysRevA.84.012311InfoDerivQT} it may be seen as a weaker form of perfect distinguishability. Indeed that principle tells us that any zero-epic pure state $\psi$ is completely mixed, in this context ensuring triviality of $A$.

Pure exclusion is particularly natural to consider in theories with kernels and cokernels, where it corresponds to yet another characterisation of purity for states.

\begin{definition} \label{def:kernelpure}
In a theory with (co)kernels we call a state $\psi$ \deff{kernel-pure}\index{pure!kernel-pure} when $\Img(\psi)$ is trivial. Equivalently if $\psi$ is non-zero we have 
\begin{equation} \label{eq:kpurepic}
\scalebox{0.8}{\input{./figures/k-pdefi.tikz}}
\end{equation}
for some zero-epic scalar $r$. In a compact theory we may more generally call a morphism $f \colon A \to B$ \deff{kernel-pure} when the state
\begin{equation} \label{eq:state-bend}
\scalebox{0.8}{\input{./figures/state-bend-noncompact.tikz}}
\end{equation}
is kernel-pure. 
\end{definition}

\begin{lemma} \label{lem:PureExclusion}
Let $\catC$ be a theory with $\otimes$-compatible (co)kernels. Then $\catC$ satisfies pure exclusion iff every $\otimes$-pure state is kernel-pure. 
\end{lemma}
\begin{proof}
 Let $\psi \colon I \to A$ be any non-zero $\otimes$-pure state. Now we can write $\psi = \img(\psi) \circ \phi$ for some state $\phi$ with $\coker(\phi) = 0$. But by Lemma~\ref{lem:helpful-tensPureKer} $\phi$ is also $\otimes$-pure and so by pure exclusion $\Img(\psi)$ is trivial, making $\psi$ kernel-pure. 

 Conversely suppose the condition holds, and let $\psi \colon I \to A$ be a zero-epic $\otimes$-pure state. Then let $\psi = \img(\psi) \circ r$ as in~\eqref{eq:kpurepic}. Since $\psi$ is zero-epic we have a causal isomorphism of causal kernels $\img(\psi) = \ker(0) = \id{A}$. Hence $\img(\psi)$ is a causal isomorphism, making $A$ trivial. 
\end{proof}

In particular pure exclusion tells us that every causal $\otimes$-pure state is a kernel. 

We now collect some facts about pure exclusion and kernel-purity. Let us say that a theory has \indef{normalisation}\index{normalisation} when every non-zero state $\rho \colon I \to A$ is of the form $\rho = \sigma \circ r$ for some causal state $\sigma$ and scalar $r$. For example this certainly holds when the scalars are $\mathbb{R}^+$ or the Booleans $\mathbb{B}$.

\begin{proposition} \label{prop:kerpure}
Let $\catC$ be a theory with $\otimes$-compatible (co)kernels. 
\begin{enumerate}[label=\arabic*., ref=\arabic*]
\item \label{enum:kerpure}
If $\catC$ has normalisation, any kernel-pure state is $\otimes$-pure. 
\item \label{enum:kerpure2}
If $\catC$ has purification and pure exclusion it has normalisation. 
\item \label{enum:ker-pure-2-out-3}
If $\psi$ and $\phi$ are kernel-pure states then so is $\psi \otimes \phi$. Conversely if $\psi \otimes \phi$ is non-zero and kernel-pure then so are $\psi$ and $\phi$. 
\end{enumerate}
\end{proposition}
\begin{proof}
\ref{enum:kerpure}. 
Let $\Psi \colon I \to A \otimes B$ be a dilation of a kernel-pure state $\psi \colon I \to A$. Then
\[
\scalebox{0.8}{\input{./figures/basic-kp-arg.tikz}} \ 0 
\quad
\implies 
\quad
\scalebox{0.8}{\input{./figures/basic-kp-arg2.tikz}} \ 0
\]
and so $\Psi$ factors over $\img(\psi) \otimes \id{B}$. But since $\psi$ is kernel-pure we have $\Img(\psi) = I$, so that for some state $\rho$ we have 
\[
\scalebox{0.8}{\input{./figures/basic-kp-arg3.tikz}}
\]
Now note that $\psi = \img(\psi) \circ r$ for the scalar $r = \discard{A} \circ \psi$. Applying $\discard{B}$ we see that $\discard{B} \circ \rho = r$ also. Then by normalisation $\rho = \sigma \circ r$ for some state $\sigma$. Finally then $\Psi$ is given by $\psi \otimes \sigma$ as desired. 

\ref{enum:kerpure2}. 
Thanks to purification, it suffices to be able to normalise any non-zero $\otimes$-pure state $\psi \colon I \to A$. But any such state $\psi$ is kernel-pure and so we may take $\Img(\psi) = I$. Then $\psi = \img(\psi) \circ r$ for the scalar $r$ and causal state $\img(\psi)$.

\ref{enum:ker-pure-2-out-3}. 
 Let $A = \Img(\psi)$ and $B = \Img(\phi)$. Then as in Proposition~\ref{prop:kernels-in-compact} we have a causal isomorphism $\Img(\psi \otimes \phi) \simeq A \otimes B$. Now suppose that $\psi$ and $\phi$ are kernel-pure. If $\psi$ or $\phi$ is the zero state then so is $\psi \otimes \phi$. Otherwise we have $\Img(\psi \otimes \phi) \simeq I \otimes I \simeq I$, making $\psi \otimes \phi$ kernel-pure.
Conversely if $\psi \otimes \phi$ is non-zero and kernel-pure then $A \otimes B$ is trivial, with some causal state $\rho$ which is an isomorphism, and we have
\[
\scalebox{0.8}{\input{./figures/kparg-1.tikz}}
\quad
\text{ and so }
\quad
\scalebox{0.8}{\input{./figures/kparg.tikz}}
\]
making $A$ trivial also, and hence $\psi$ is kernel-pure. 
\end{proof}

Hence in any compact theory with (co)kernels, kernel-pure morphisms are closed under $\otimes$ and are $\otimes$-complete. Finally, we note that in the presence of kernels and the purification we considered earlier, pure exclusion has another simple form.

\begin{lemma} \label{lem:Pure-Excl-Is-Easy}
Let $\catC$ be a compact theory with $\otimes$-compatible (co)kernels, zero-cancellative scalars, and purification satisfying the properties of Proposition~\ref{prop:pure-lem}. The following are equivalent for $\catC$:
\begin{enumerate}[label=\arabic*., ref=\arabic*]
\item  \label{enum:PE}
Pure exclusion holds;
\item  \label{enum:NormAndExtra}
Normalisation holds, and every causal $\otimes$-pure state is a kernel.
\end{enumerate}
\end{lemma}
\begin{proof}
\ref{enum:PE} $\implies$ \ref{enum:NormAndExtra}: Lemma~\ref{lem:PureExclusion} and Proposition~\ref{prop:kerpure}.

\ref{enum:NormAndExtra} $\implies$ \ref{enum:PE}: 
Let $\psi \colon I \to A$ be a non-zero $\otimes$-pure state, with $\psi = \phi \circ r$ for some causal state $\phi$ and scalar $r$. Now any purification $\sigma \colon I \to A \otimes  B$ of $\phi$ satisfies 
\[
\scalebox{0.8}{\input{./figures/imgarg-1.tikz}}
\]
for any causal $\otimes$-pure state $\mu$ of $B$. Normalisation implies that every scalar is $\otimes$-pure, and so by essential uniqueness we have
\[
\scalebox{0.8}{\input{./figures/imgarg-new.tikz}}
\]
for some causal $\otimes$-pure isomorphism $U$. Letting $\eta$ be the causal $\otimes$-pure state $U \circ \mu$ we have implications
\[
\scalebox{0.8}{\input{./figures/imgarg-2.tikz}} \ 0 \  \ 
\implies
\scalebox{0.8}{\input{./figures/imgarg-2a.tikz}} \ 0 \  \ 
\implies
\scalebox{0.8}{\input{./figures/imgarg-2b.tikz}}
\]
for some state $\rho$, since kernels are $\otimes$-compatible and by pure exclusion $\eta = \img(\eta)$ is a kernel. Then applying $\discard{B}$ gives $\rho = \phi$.  
Hence $\phi \otimes \eta$ is $\otimes$-pure, and then so is $\phi$ by Lemma~\ref{lem:caustensclosedexamples}. Then by assumption $\phi$ is a kernel and $r$ is zero-epic so that $\img(\psi)=\img(\phi)=\phi$, making $\psi$ kernel-pure.
\end{proof}

\begin{examples}
Pure exclusion is satisfied in the following theories.
\begin{exampleslist}
\item 
$\Quant{}$ satisfies pure exclusion. Here any non-trivial $\hilbH$ has dimension $\geq 2$. Then for any (causal) pure state $\Dbl{\psi}$ induced by some $\psi \in \hilbH$, any unit vector $\phi$ orthogonal to $\psi$ induces a causal pure state $\Dbl{\phi}$ with $\Dbl{\phi}^\dagger \circ \Dbl{\psi} = 0$. Similarly so does $\FCStar$, as may be seen thanks to its equivalence with $\Quant{}^\biprod$.

\item 
More generally let $\catA$ be dagger-compact with dagger kernels and such that the $\otimes$-pure morphisms in $\CPM(\catA)$ are precisely those belonging to $\Dbl{\catA}$. Suppose also that in $\catA$ every non-zero state $\psi$ is of the form $k \circ r$ for a dagger kernel state $k \colon I \to A$ and zero-epic scalar $r$. Then by Example~\ref{examples:kernels}~\ref{ex:CPMD-kernels}, for any such state, $\Dbl{\psi}$ is kernel-pure so that $\CPM(\catA)$ satisfies pure exclusion. 

\item 
$\MSpek$ satisfies pure exclusion. Here an object $\IV^n$ is non-trivial iff $n \geq 1$. By Example~\ref{examples:Purif}~\ref{example:pure-in-MSpek}, every $\otimes$-pure state in $\MSpek$ belongs to $\Spek$, with any non-zero pair of such states related by a unitary. Hence it suffices to show that for $n \geq 1$ some such state has non-trivial cokernel. But we always have 
\[
\scalebox{0.8}{\input{./figures/manydots2.tikz}} \ 0 
\]
where $\tinyunit[greydot] = \{2, 4\}$. 

\item 
$\KlDT$ and $\Rel$ are easily seen to satisfy pure exclusion. For instance in $\Rel$ any non-trivial non-zero object $A$ has $|A| \geq 2$. Any pure state is then given by a singleton $a \in A$, so that any effect given by $b \in A$ with $b \neq a$ has $b \circ a = 0$. 
\end{exampleslist}
\end{examples}

\subsection{Kernel-purity, daggers and orthomodular lattices}

In a dagger theory with dagger kernels, several facts relating to pure exclusion can surprisingly be re-stated in terms of the orthomodular lattices $\DKer(A)$ of dagger kernels on any fixed object $A$.

\paragraph{Atomicity}
Firstly, we've often required non-zero systems to have causal pure states, which thanks to pure exclusion we then expect to be \indef{kernel states}\index{kernel!state}, i.e.~dagger kernels $k \colon I \to A$. This translates to the following lattice-theoretic property.

In a lattice, an \indef{atom} is a minimal non-zero element, and an orthomodular lattice is \indef{atomistic}\index{atomistic lattice} if for all $a$ there is an atom $b$ with $b \leq a$.

\begin{proposition} \label{pop:kerstatesAtomistic}
Let $\catC$ be a dagger monoidal category with dagger kernels. Then every non-zero object has a kernel state iff every lattice $\DKer(A)$ is atomistic with atoms being precisely kernel states.
\end{proposition}

Atomicity of $\DKer(A)$ is also studied in~\cite[Section~8]{heunen2010quantum}.   

\begin{proof}
Suppose that each $\DKer(A)$ is atomistic with atoms of this form. Then in particular whenever $A$ is a non-zero object, $\DKer(A)$ is non-zero and so contains a non-zero atom. Hence $A$ has a kernel state. 

We now establish the converse. For any non-zero kernel $k \colon K \to A$, by assumption $K$ possess a kernel state $\phi \colon I \to K$. Then $\psi = k \circ \phi$ is a kernel below $k$ in $\DKer(A)$. Hence any atom must be given by a kernel state. 

Finally we claim that any kernel $\psi \colon I \to A$ is indeed an atom in $\DKer(A)$. Suppose that $l \colon L \to A$ is a kernel with $l \leq \psi$, so that $l = \psi \circ i$ for some isometry $i \colon L \to I$. Then if $l=0$ we are done, otherwise let $\eta \colon I \to L$ be any kernel state. Then $i \circ \eta \colon I \to I$ is an isometry also, and since scalars are commutative is then unitary, so that $i$ is also. Hence $l$ and $\psi$ are equal as dagger kernels on $A$.
\end{proof}

Next let us turn to the notion of kernel-purity. It will be helpful to slightly abuse our earlier terminology, and in any dagger-compact category (without requiring discarding) call a state $\psi$ \indef{\dkernel-pure} \index{pure!kernel-pure} when there is a \emps{unitary} $\Img(\psi) \simeq I$, or $\psi=0$, and a morphism $f \colon A \to B$ \indef{\dkernel-pure} when its induced state on $A^* \otimes B$ as in~\eqref{eq:state-bend} is. This coincides with Definition~\ref{def:kernelpure} in theories of interest, and more generally when $\discard{}$ and $\dagger$ interact well; see Section~\ref{sec:CP}.

 Now, we've seen often that pure morphisms satisfy the natural requirement of being closed under composition, though this is not immediate from their definition. As such it is natural to ask when \dkernel-pure morphisms have this property. In fact this corresponds to the following feature of a lattice. 

\paragraph{The Covering Law} In any lattice we say that an element $b$ \indef{covers} an element $a$ if $a \leq b$ and $a \leq c \leq b \implies c = a$ or $c = b$. A lattice satisfies the \indef{covering law}\index{covering law} if for every atom $a$ and element $b$, either $a \leq b$ or $a \vee b$ covers $b$. It can be show that an atomistic orthomodular lattice satisfies the covering law iff for every atom $a$ and element $b$ we have that 
\[
b \wedge (a \vee b^\bot)
\] 
is either an atom or zero~\cite{piron1976foundations,sep-qt-quantlog}. To apply this fact the following will be useful. 

\begin{lemma} \label{lem:dag-ker-image-lemma}
Let $\catC$ be a dagger category with dagger kernels. Then for all dagger kernels $k, l$ on the same object we have
\[
\img(k \circ k^\dagger \circ l) = k \wedge (l \vee k^\bot)
\]
\end{lemma}

\begin{proof}
In~\cite{heunen2010quantum} it is shown that each $\DKer(A)$ has intersections given by $k \wedge l := k \circ \ker(\coker(l) \circ k)$.
Hence for all such $k, l$ we have 
\begin{align*}
\coker(k^\bot \vee l) 
&=
((k^\bot \vee l)^\bot)^\dagger & (\text{def. }\bot)
\\ 
&= 
(k \wedge l^\bot)^\dagger & (\bot \text{ orthocomp.})
\\ 
&=
(k \circ \ker(\coker(l^\bot) \circ k))^\dagger  & (\text{def. }\wedge)
\\  
&=
(k \circ \ker(l^\dagger \circ k))^\dagger & (\bot \text{ orthocomp.})
\\ &=
\ker(l^\dagger \circ k)^\dagger \circ k^\dagger & (\dagger \text{ a functor})
\\ 
&=
\coker(k^\dagger \circ l) \circ k^\dagger & (\dagger \text{ a functor})
\\ 
\\
k \wedge (k^\bot \vee l) 
&=
k \circ \ker(\coker(k^\bot \vee l) \circ k)  &  (\text{def. }\wedge)
\\ 
&=
k \circ \ker(\coker(k^\dagger \circ l)) &  (k \text{ isometry})
\\ 
& =
k \circ \img(k^\dagger \circ l) 
=\img(k \circ k^\dagger \circ l) & (k \text{ kernel})
\end{align*}
\end{proof}

We can now characterise the covering law as follows.

\begin{theorem} \label{thm:closedUnderComp-Equiv}
Let $\catC$ be a dagger compact category with dagger kernels for which every non-zero object has a kernel state and all identity morphisms are \dkernel-pure. The following are equivalent:
\begin{enumerate}[label=\arabic*., ref=\arabic*]
\item \label{en:klnew}
For every kernel state $\psi$ and cokernel $c$, the state $c \circ \psi$ is \dkernel-pure;
\item \label{en:kl3}
For every \dkernel-pure state $\psi$ and cokernel $c$, $c \circ \psi$ is \dkernel-pure;
\item \label{en:kl4}
The collection of \dkernel-pure morphisms is closed under composition.
\item \label{en:kl5}
Each lattice $\DKer(A)$ satisfies the covering law. 
\end{enumerate}
\end{theorem}

\begin{proof}
Throughout we use that scalars are zero-cancellative by Lemma~\ref{lem:dagzero}.

\ref{en:klnew} $\implies$ \ref{en:kl3}:
Let $\psi$ be a \dkernel-pure state, say with $\psi = k \circ r$ for some scalar $r$ and kernel state $k$, and let $\phi = c^\dagger \circ \psi$. Then since scalars are zero-cancellative we have $\img(\phi) = \img(c^\dagger \circ k \circ r) = \img(c^\dagger \circ k)$ which is either zero or trivial. Hence $\phi$ is kernel-pure. 

\ref{en:kl3} $\implies$ \ref{en:klnew}: Any kernel state is \dkernel-pure. 

\ref{en:kl3} $\implies$ \ref{en:kl4}:
Let $f \colon A \to B, g \colon B \to C$ be non-zero \dkernel-pure morphisms. Since $\id{B}$ is \dkernel-pure we have 
\[
\scalebox{0.8}{\input{./figures/kpurecomp1.tikz}}
\]
for some dagger cokernel $c \colon B^* \otimes B \to I$ and non-zero scalar $r$. 
Then since \dkernel-pure states and cokernels are closed under $\otimes$, and since scalars are cancellative we have $\img(\psi) = \img(r \cdot \psi)$ for all non-zero scalars $r$, the morphisms 
\[
\scalebox{0.8}{\input{./figures/kpurecomp3a.tikz}}
\qquad
\scalebox{0.8}{\input{./figures/kpurecomp2.tikz}}
\]
are a \dkernel-pure state and cokernel on $A^* \otimes B \otimes B^* \otimes C$, respectively. Then
\[
\scalebox{0.8}{\input{./figures/kpurecomp4.tikz}}
\]
Hence by assumption the right-hand state is \dkernel-pure, so that $g \circ f$ is also.

\ref{en:kl4} $\implies$ \ref{en:kl3}: 
We claim that all dagger cokernels are \dkernel-pure, so that this is a special case. Let $c \colon A \to C$ be a non-zero dagger cokernel. Then we have that 
\begin{equation} \label{eq:coker-bend}
\scalebox{0.8}{\input{./figures/cokernel-bend.tikz}}
\end{equation}
The morphism above $\tinycup_C$ on the right-hand side above is a dagger kernel. 
Now in general for any kernel $k$ and morphism $f$, letting $f = \img(f) \circ e$ for some zero-epi $e$, we have that 
\[
\Img(k \circ f) = \Img(k \circ \img(f) \circ e)
= \Img(k \circ \img(f))
= \Img(f)
\]
Hence in particular the state~\eqref{eq:coker-bend} has image given by $\Img(\tinycup_C)$. But this is simply $I$, since $\id{C}$ is \dkernel-pure by assumption. 

\ref{en:kl5} $\iff$ \ref{en:kl3}: By Proposition~\ref{pop:kerstatesAtomistic} each lattice $\DKer(A)$ is atomistic with atoms being the kernel states. Now for any kernel state $\psi$ and kernel $k$ on the same object we have 
\[
\img(k \circ k^{\dagger} \circ \psi) = k \wedge (\psi \vee k^\bot)
\]
by Lemma~\ref{lem:dag-ker-image-lemma}. The covering law then states precisely that this is either zero or an atom, i.e.~that $k \circ k^{\dagger} \circ \psi$ is \dkernel-pure. But since $\img(k^\dagger \circ \psi) = k^\dagger \circ \img(k \circ k^\dagger \circ \psi)$ and cokernels are precisely the morphisms $k^\dagger$, this is equivalent to~\ref{en:kl3}.
\end{proof}

The preservation of atoms by projections appears as one of the requirements in Alfsen and Schult'z reconstruction of Jordan algebra state spaces from among lattices, and in that context is shown to be equivalent to the covering law~\cite[Proposition~9.7]{alfsen2012geometry}.

\begin{examples}
All of the conditions of Theorem~\ref{thm:closedUnderComp-Equiv} are satisfied in the following categories, providing each $\DKer(A)$ with atomicity and the covering law.
\begin{exampleslist}
\item 
In $\FHilb$, every morphism is kernel-pure. Here $\DKer(\hilbH)$ is the lattice of subspaces of the Hilbert space $\hilbH$, which is indeed atomistic via the states $\psi \colon I \to \hilbH$ and satisfies the covering law~\cite{piron1976foundations}. 
\item 
The same goes for $\Quant{}$, where all kernels are induced from $\FHilb$. More broadly, in $\FCStar$ each lattice $\DKer(A)$ on $A=\bigoplus^n_{i=1} B(\hilbH_i)$ inherits these properties from each $B(\hilbH_i)$, as do those of $\Class$ similarly. 
\item 
In $\Rel$ each $\DKer(A)$ is the Boolean lattice of subsets of $A$, which satisfies these properties, and indeed atoms here are kernel states, i.e.~singletons $a \in A$. 
\end{exampleslist}
\end{examples}

\section{Purification and Daggers} \label{sec:CP}

In theories containing both purification and a dagger on their morphisms it is natural to expect these features to interact well. A notion of purification using the dagger which applies to both quantum and classical theory is considered in~\cite{selby2018reconstructing}. Here we will focus on the behaviour of purifications in $\Quant{}$ with respect to the dagger, which are captured by the following notion due to Coecke.

\begin{definition}~\cite{coecke2008axiomatic}
Let $\catC$ be a dagger compact category with discarding. An \deff{environment structure} \index{environment structure}\label{not:envstruc} on $\catC$ is a dagger compact subcategory $\catC_\prepure$ within which every morphism of $\catC$ has a dilation, and such that all morphisms $f \colon A \to B$,  $g \colon A \to C$ in $\catC_\prepure$ satisfy the \deff{CP axiom}\index{CP axiom}:
\begin{equation} \label{eq:CP axiom} \tag{CP}
\scalebox{0.8}{\input{./figures/CP-axiom.tikz}}
\end{equation}
\end{definition}

\begin{examples}
$\Quant{}$ has an environment structure given by its pure subcategory $\FHilbP$. More generally any category of the form $\CPM(\catA)$ has an environment structure given by its subcategory $\Dbl{\catA}$, as in Example~\ref{examples:Purif}~\ref{examples:purif-in-CPM}, see~\cite{coecke2008axiomatic}.
\end{examples}

Crucially, the converse of this example holds; this notion of purification in fact captures precisely those categories arising from the $\CPM$ construction~\cite[Theorem~5.1]{coecke2008axiomatic}.

\begin{proposition}
Let $\catC$ be a dagger compact category with discarding and an environment structure $\catC_\prepure$. Then there is an equivalence of dagger monoidal categories with discarding $\CPM(\catC_\prepure) \simeq \catC$ given by
\[
\scalebox{0.8}{\input{./figures/CPM-map-smalleri.tikz}}
\mapsto
\scalebox{0.8}{\input{./figures/env_struc_map.tikz}}
\]
\end{proposition}


 Hence the CP axiom is a powerful and useful one for singling out quantum theory, with purification ensuring that a dagger theory has the quantum-like form $\CPM(\catA)$.



\subsection{Deriving the CP axiom}
At first glance the rule~\eqref{eq:CP axiom} appears rather ad hoc. Given its usefulness,  it would be desirable to understand how this axiom arises from more natural principles.

Firstly, note that it tells us that any causal isomorphism in $\catC_\prepure$ is unitary. In fact this ensures half of the axiom under some familiar conditions. Let us say that a subcategory $\catC_\prepure$ of $\catC$ \indef{has causal states} when every non-zero object of $\catC$ has a causal isometric state in $\catC_\prepure$.

\begin{proposition} \label{prop:halfOfCPM} 
Let $\catC$ be a compact dagger theory with essentially unique dilations with respect to a dagger compact subcategory $\catC_\prepure$ which has causal states, and suppose that all causal isomorphisms in $\catC_\prepure$ are unitary.
Then $\catC_\prepure$ satisfies the direction $\impliedby$ of the CP axiom.
\end{proposition}
\begin{proof}
Let $f \colon A \to B$ and $g \colon A \to C$  belong to $\catC_\prepure$ with $\discard{B} \circ f = \discard{C} \circ g$. If $B$ or $C$ are zero objects then $\discard{B} \circ f = 0$ and so $f = 0$ and $g = 0$ similarly, yielding the result. Otherwise let $\psi \colon I \to B$ and $\phi \colon I \to C$ be causal isometries in $\catC_\prepure$. Then 
\[
\scalebox{0.8}{\input{./figures/CPMhalf1.tikz}}
\quad
\text{ and so }
\quad
\scalebox{0.8}{\input{./figures/CPMhalf2.tikz}}
\]
for some unitary $U \in \catC_\prepure$. Since $U$, $\psi$ and $\phi$ are all isometries,  composing each morphism above with its dagger gives $f^\dagger \circ f = g^\dagger \circ g$.
\end{proof}

\paragraph{Homogeneous Kernels}
In fact we can deduce the presence of both essentially unique dilations and the CP axiom from some of our earlier principles. Beyond these we will merely need the following weakening of essential uniqueness applying only to kernels and which holds even classically.

\begin{definition} \label{def:ker-homog}
We say that a theory $\catC$ has \deff{homogeneous} kernels \index{homogeneous kernels} when it has causal kernels and for any pair of causal kernels of the same type $k, l \colon A \to B$ there exists a causal isomorphism $U$ for which the following commutes:
\[
\begin{tikzcd}
A \rar{k} \arrow[dr, "l",swap] & B \dar{U}
\\ & B
\end{tikzcd}
\]
\end{definition}

\begin{examples} 
As the name suggests, homogeneity of kernels requires objects to be suitably `uniform'.
\begin{exampleslist}
\item 
$\Quant{}$ has homogeneous kernels. Here we've seen that kernels are $\otimes$-pure and so are homogeneous by essential uniqueness. 
\item 
$\KlDT$ and $\Rel$ have homogeneous kernels.
In either case, any pair of causal kernels $k, l \colon K \to A$ may be seen as injections of the set $K$ into the set $A$. By the axiom of choice there then is an isomorphism $U$ on $A$ exchanging $k$ and $l$, which induces such a causal isomorphism in either case. Similarly $\Class$ has homogeneous kernels in the same way. 
\item 
Kernels in the quantum-classical theory $\FCStar$ are not homogenous. For example, consider the biproduct $A = \mathbb{C} \biprod B(\mathbb{C}^2)$ and let $\psi$ be a causal pure state of $B(\mathbb{C}^2)$. The states $\coproj_1$ and $\pcoproj_2 \circ \psi$ are both kernels by pure exclusion. But any isomorphism $U$ with $U \circ \coproj_2 \circ \psi = \pcoproj_1$ must send elements of $B(\mathbb{C}^2)$ to a mixture of those from both sectors, which is easily seen to be a contradiction since $B(\mathbb{C}^2)$ is simple.
\end{exampleslist}
\end{examples}

\begin{theorem} \label{thm:DeduceCPM}
Let $\catC$ be a compact dagger theory with dagger kernels which are all causal, and a dagger compact subcategory $\catC_\prepure$ containing all isomorphisms and kernels. Suppose that the internal isomorphism property holds and that every causal morphism in $\catC_\prepure$ is an isometry. 
\begin{enumerate}
\item \label{enum:alt-pure-descrip}
A morphism $f$ belongs to $\catC_\prepure$ iff in the commutative diagram 
\begin{equation} \label{eq:capture-pure}
\begin{tikzcd}
A \dar[swap]{\coim(f)} \rar{f} & B \\
\CoIm(f) \rar[dashed,swap]{\bar{f}} & \Img(f) \uar[swap]{\img(f)} 
\end{tikzcd} 
\end{equation}
the unique morphism $\bar{f}$ is an isomorphism.
\item \label{enum:EUP}
Suppose that $\catC$ has homogeneous kernels and that $\catC_\prepure$ has causal states. Then $\catC_\prepure$ forms an environment structure on $\catC$ satisfying essential uniqueness.
\end{enumerate}
\end{theorem}
\begin{proof}
\ref{enum:alt-pure-descrip}
By assumption any such morphism belongs to $\catC_\prepure$. Conversely, a morphism $f$ belongs to $\catC_\prepure$ precisely when $\bar{f}$ does, since either morphism may be obtained from the other by composing with (co)kernels. Hence it suffices to show that any $f \in \catC_\prepure$ which is both a zero-epi and zero-mono is an isomorphism. 

Now in this case $\discard{} \circ f$ is again zero-mono and so by the internal isomorphism property we have $\discard{} \circ f = \discard{} \circ g$ for some automorphism $g$. Then $f \circ g^{-1} \in \catC_\prepure$  is causal, and hence an isometry. Then $f^\dagger \circ f = g^\dagger \circ g$, making $f$ split monic. Dually $f^{\dagger}$ is also split monic, making $f$ an isomorphism.   

\ref{enum:EUP}
By the internal isomorphism property every effect has a dilation $f \circ \coim(e)$, for some automorphism $f$, which belongs to $\catC_\prepure$. Hence by compactness every morphism has a dilation in $\catC_\prepure$. We now verify essential uniqueness.

By compactness it suffices to consider $f, g \colon A \to B$ in $\catC_\prepure$ with $\discard{} \circ f = \discard{} \circ g$. In this case $f \circ h = 0 \iff g \circ h =0$ for all morphisms $h$ and so we may take $c := \coim(f) = \coim(g) $. Then writing $f = \img(f) \circ \bar{f} \circ c$ as above, and $g = \img(g) \circ \bar{g} \circ c$ similarly, we have 
\[
\discard{} \circ \bar{f} \circ c
=
\discard{} \circ \img(f) \circ \bar{f} \circ c 
=
\discard{} \circ f
=
\discard{} \circ g
=
\discard{} \circ \bar{g} \circ c
\]
and so $\discard{} \circ \bar{f} = \discard{} \circ \bar{g}$ since $c$ is epic. By the first part $\bar{f}$ is an isomorphism. Then $U = \bar{g} \circ \bar{f}^{-1}\colon \Img(f) \to \Img(g)$ is a causal isomorphism, and so unitary. Then $\img(g) \circ U$ and $\img(f)$ are both dagger kernels of type $\Img(f) \to B$, so by homogeneity and our assumptions there is some unitary $V$ with $V \circ \img(f) = \img(g) \circ U$. Then as desired we have $V \circ f = g$ since the following diagram commutes: 
\[
\begin{tikzcd}[row sep = tiny]
& & \Img(f) \arrow[r,"\img(f)",swap] \arrow[dd,"U"]  & B \arrow[dd,"V"]\\ 
A \arrow[rrru,bend left=30,"f"] \arrow[rrrd,bend right=30,"g"]
\rar{c} & C \arrow[ur,"\bar{f}"]  \arrow[dr,swap,"\bar{g}"] & \\ 
& & \Img(g) \arrow[r,"\img(g)"] & B
\end{tikzcd}
\]
Next let us establish~\eqref{eq:CP axiom} for all $f \colon A \to B$ and $g \colon A \to C$ in $\catC_\prepure$. By assumption all causal isomorphisms are unitary and so if $\discard{} \circ f = \discard{} \circ g$ then $f^\dagger \circ f = g^\dagger \circ g$ by Proposition~\ref{prop:halfOfCPM}. Conversely if this holds then using Lemma~\ref{lem:dagzero} 
\begin{align*}
f \circ h = 0 &\iff h^{\dagger} \circ f^{\dagger} \circ f \circ h = 0
\\ &\iff h^{\dagger} \circ g^{\dagger} \circ g \circ h = 0
\iff g \circ h = 0
\end{align*}
and so again we may take $c := \coim(f) = \coim(g)$ and write $f$ and $g$ in terms of $\bar{f}$ and $\bar{g}$ as before. It follows immediately that  
$\bar{f}^{\dagger} \circ \bar{f} 
= 
\bar{g}^{\dagger} \circ \bar{g}$
. This makes $U: = \bar{g} \circ \bar{f}^{-1}$ unitary and so a dagger kernel and hence causal, giving 
\begin{align*}
\discard{} \circ g 
&=
\discard{} \circ \img(g) \circ \bar{g} \circ c
\\ &=
\discard{} \circ \bar{g} \circ c 
\\ &=
\discard{} \circ U \circ \bar{f} \circ c 
\\
&=
\discard{} \circ \bar{f} \circ c 
=
\discard{} \circ f
\end{align*}
\end{proof}

\begin{remark}
The first part of this result tells us that in such a theory there is essentially one notion of purity closed under composition and containing all kernels, provided by~\eqref{eq:capture-pure}. This bares similarities to Westerbaan's notion of purity in effectuses~\cite[\S 3.4]{westerbaan2018dagger}. 
\end{remark}

In particular if we consider when $\catC_\pure$ is the collection of $\otimes$-pure morphisms in a $\otimes$-pure theory $\catC$ the above result yields essentially unique purification in $\catC$. If we instead assume this principle then we may deduce~\eqref{eq:CP axiom} by simply requiring causal isomorphisms to respect the dagger, as follows.

\begin{corollary}
Let $\catC$ be a compact dagger theory with dagger kernels. Suppose that $\catC$ satisfies the internal isomorphism property, essentially unique purification, and that $\catC_\pure$ forms a monoidal subcategory. Then $\catC_\pure$ forms an environment structure on $\catC$ iff all causal isomorphisms in $\catC$ are unitary and all dagger kernels are causal. 
\end{corollary}
\begin{proof}
By assumption the theory $\catC$ is $\otimes$-pure and hence so are all isomorphisms, kernels. Moreover $\catC_\pure$ is straightforwardly seen to be closed under the dagger and bending wires, making it a dagger-compact subcategory of $\catC$. 

Now the latter conditions are necessary by~\eqref{eq:CP axiom}, and homogeneity is implied by essential uniqueness. Conversely suppose that they are satisfied. Then by Proposition~\ref{prop:halfOfCPM} the direction $\impliedby$ of~\eqref{eq:CP axiom} is satisfied, making every causal $\otimes$-pure morphism an isometry. Hence the other direction is satisfied by Theorem~\ref{thm:DeduceCPM}.
\end{proof}

In closing we observe another quantum-like property of the theory $\MSpek$. 

\begin{example}
$\MSpek$ has $\Spek$ as an environment structure, so that
\[
\MSpek \simeq \CPM(\Spek)
\]
\begin{proof}
We saw in Example~\ref{examples:Purif}~\ref{example:pure-in-MSpek} that $\MSpek$ has essentially unique purification with $\otimes$-pure morphisms being those in $\Spek$. By construction in $\Spek$ every non-zero object has a non-zero state, which is then an isometry. In this theory, or $\Rel$ more broadly, any isomorphism is both causal and unitary. Hence Proposition~\ref{prop:halfOfCPM} ensures the direction $\impliedby$ of the CP axiom. 

Conversely, the direction $\implies$ in fact holds for arbitrary morphisms $R \colon A \to B$ in $\Rel$, since we have  
\begin{align*}
\scalebox{0.8}{\input{./figures/reldiscard.tikz}}
&=
\{a \in A \mid \exists b \text{ s.t. } R(a,b)\}
\\
&=
\{a \in A \mid \exists b, c \text{ s.t. } R(a,b) \wedge R(b,c)\} 
\\
&=
\scalebox{0.8}{\input{./figures/reldiscard3.tikz}}
\end{align*}
\end{proof}
\end{example}

%% file: figures/internal-isom-statesnew2.tikz
\begin{tikzpicture}
	\begin{pgfonlayer}{nodelayer}
		\node [style=none] (0) at (0.75, 0.75) {};
		\node [style=kpointadj] (1) at (0.75, 1) {$e$};
		\node [style=none] (2) at (0.75, -0.5) {};
		\node [style=none] (3) at (1.75, -0) {$=$};
	\end{pgfonlayer}
	\begin{pgfonlayer}{edgelayer}
		\draw [style=none] (2.center) to (0.center);
	\end{pgfonlayer}
\end{tikzpicture}

%% file: figures/p-d-2a.tikz
\begin{tikzpicture}
	\begin{pgfonlayer}{nodelayer}
		\node [style=kpointadj] (0) at (11, 0.75) {$e$};
		\node [style=none] (1) at (12.5, 0) {$=$};
		\node [style=kpoint] (2) at (11, -0.75) {$\rho$};
	\end{pgfonlayer}
	\begin{pgfonlayer}{edgelayer}
		\draw [style=none] (2) to (0);
	\end{pgfonlayer}
\end{tikzpicture}

%% file: figures/p-d-2b.tikz
\begin{tikzpicture}
	\begin{pgfonlayer}{nodelayer}
		\node [style=none] (4) at (15, 0) {$=$};
		\node [style=kpoint] (5) at (16.25, -0.75) {$\sigma$};
		\node [style=kpointadj] (6) at (16.25, 0.75) {$d$};
	\end{pgfonlayer}
	\begin{pgfonlayer}{edgelayer}
		\draw [style=none] (5) to (6);
	\end{pgfonlayer}
\end{tikzpicture}

%% file: figures/closed-under-statesii.tikz
\begin{tikzpicture}
	\begin{pgfonlayer}{nodelayer}
		\node [style=map] (0) at (0, -0) {$f$};
		\node [style=none] (1) at (0, -1) {};
		\node [style=none] (2) at (0, 1) {};
		\node [style=label] (3) at (0, -1.5) {$A$};
		\node [style=label] (4) at (0, 1.5) {$B$};
	\end{pgfonlayer}
	\begin{pgfonlayer}{edgelayer}
		\draw [style=none] (1.center) to (2.center);
	\end{pgfonlayer}
\end{tikzpicture}

%% file: figures/pure-excl-effect.tikz
\begin{tikzpicture}
	\begin{pgfonlayer}{nodelayer}
		\node [style=kpointadj] (0) at (0, 1) {$e$};
		\node [style=kpoint] (1) at (0, -0.75) {$\psi$};
		\node [style=none] (2) at (1.5, -0) {$=$};
	\end{pgfonlayer}
	\begin{pgfonlayer}{edgelayer}
		\draw [style=none] (1) to (0);
	\end{pgfonlayer}
\end{tikzpicture}

%% file: chapter5.tex
\chapter{Superpositions and Phases} \label{chap:superpositions}

A central feature of the quantum world is the ability to form \emps{superpositions}\index{superposition} of pure states and processes. If we wish to characterise quantum theory itself from among more general operational theories, it will be useful to be able to describe these within our framework.

 In fact in Chapter~\ref{chap:totalisation} we already seemingly met a categorical description of superpositions; they are given by an addition operation $f + g$ on morphisms in the category $\Hilb$, coming from the presence of biproducts 
\[
\hilbH \biprod \hilbK
\]
More generally, we saw that biproducts always induce such an addition operation, and as a result they have long been used to describe superpositions~\cite{abramskycoecke:categoricalsemantics,Selinger2007139}. 

Operationally, however, there is a problem. While $\Hilb$ has biproducts, its quotient $\HilbP$ after identifying global phases does not, and only the latter category directly models pure quantum processes. As such, a characterisation of the object $\hilbH \oplus \hilbK$ in the new setting $\HilbP$ is needed. 

In this chapter we provide such an account of superpositions using our new notion of a \emps{phased biproduct} or more general \emps{phased coproduct} $A \pcoprod B$. Roughly, these are coproducts coming with extra structure-preserving \emps{phase} isomorphisms 
\[
\begin{tikzcd}
A \pcoprod B
\rar{U}
&
A \pcoprod B
\end{tikzcd}
\]
In fact these features arise in a very general setting. Suppose we have a category $\cta$ with finite coproducts and a collection of `trivial' isomorphisms on each object. Well-known examples arise from global phases in quantum theory, and from \emps{projective geometry}~\cite{coxeter2003projective}. Then its quotient $\cta \!  /  \! \! \sim$ after identifying such maps has phased coproducts. Conversely, for any suitable category $\ctb$ with phased coproducts we will construct a new one $\plusI{\ctb}$ with coproducts from which it arises as such a quotient.

In particular this allows us to describe the more well-behaved category $\Hilb$ in terms of the operationally motivated one $\HilbP$ via
\[
\Hilb \simeq \plusI{\Hilb_\sim}
\]
which will be central to reconstructing quantum theory in Chapter~\ref{chap:recons}.

\section{Phased Coproducts} \label{sec:ph-coprod}

Our central definition in this chapter will be the following.

\begin{definition} \label{def:ph_coprod} \label{not:phcoprod} \label{not:phcoproj}
In any category, a \deff{phased coproduct} \index{phased coproduct} of objects $A, B$ is an object $A \pcoprod B$ together with a pair of morphisms  $\pcoproj_A \colon A \to A \pcoprod B$ and $\pcoproj_B \colon {B \to A \pcoprod B}$ satisfying the following. Firstly, for any pair of morphisms $f \colon A \to C$, $g \colon B \to C$, there exists $h \colon {A \pcoprod B \to C}$ making the following commute:
\[
\begin{tikzcd}
A \rar{\coproj_A} \drar[swap]{f} & A \pcoprod B 
\dar[dashed]{h}
& B \lar[swap]{\coproj_B} \dlar{g} \\ 
& C &
\end{tikzcd} 
\]
Secondly, any pair of such morphisms $h,h'$ have that $h' = h \circ U$
\[
\begin{tikzcd}
 \arrow[loop left, "U",distance = 2em] A \pcoprod B \rar[shift right = 2, swap]{h'} \rar[shift left = 2]{h}   & C
\end{tikzcd}
\]
for some endomorphism $U$ of $A \pcoprod B$ which satisfies
\begin{align} \label{eq:phase-eq}
U \circ \pcoproj_A = \pcoproj_A
\qquad
U \circ \pcoproj_B = \pcoproj_B
\end{align}
We call any endomorphism $U$ of $A \pcoprod B$ satisfying \eqref{eq:phase-eq} a \deff{phase} for $A \pcoprod B$, and the morphisms $\pcoproj_A$, $\pcoproj_B$ \deff{coprojections}\index{coprojection}.
\end{definition}

A coproduct is then a phased coproduct whose only phase is the identity. Straightforwardly extending the above, a \indef{phased coproduct} of any collection of objects $(A_i)_{i \in I}$ is defined as an object $A$ together with morphisms 
$(\coproj_i \colon A_i \to A)_{i \in I}$ satisfying the following\index{phased coproduct}.
 Firstly, for any collection of morphisms 
\[
\begin{tikzcd}
A_i \rar{f_i} & B
\end{tikzcd}
\]
 there exists $f \colon A \to B$ with $f \circ \coproj_i = f_i$ for all $i$. Furthermore, any such $f, f' \colon A \to B$ have $f' = f \circ U$ for some $U \colon A \to A$ satisfying $U \circ \pcoproj_i = \pcoproj_i$ for all $i$, which we call a \indef{phase}\index{phase}. 
A phased coproduct of finitely many $A_1, \dots, A_n$ is denoted $A_1 \pcoprod \dots \pcoprod A_n$. \label{not:phcoprod2}

Despite their generality, phased coproducts are surprisingly well-behaved, in particular being unique up to (non-unique) isomorphism.

\begin{lemma} \label{lem:isoms}
Let $A$ and $B$ be phased coproducts of objects $(A_i)_{i \in I}$ with respective coprojections $\coproj_i \colon A_i \to A$ and $\mu_i \colon A_i \to B$ for $i\in I$. Then any morphism $f$ for which each diagram
\[
\begin{tikzcd}[row sep = small, column sep = small]
& A_i \arrow[dl, "\coproj_i", swap] \arrow[dr, "\mu_i"] & \\ 
A \arrow[rr,"f",swap] & & B 
\end{tikzcd}
\]
commutes is an isomorphism.
Conversely, any object $C$ with an isomorphism $g \colon {A \isomto C}$ forms a phased coproduct of the $A_i$ with coprojections $\nu_i := g \circ \coproj_i$.
\end{lemma}

\begin{proof}
For the first statement, let $g \colon B \to A$ with $g \circ \mu_i = \coproj_i$ for all $i$. Then $f \circ g$ preserves the $\mu_i$ and so there is some phase $U$ on $B$ with $f \circ g \circ U = \id{B}$. But then $g \circ U \circ f$ preserves the $\pcoproj_i$ and so there is a phase $V$ on $A$ with 
\[
g \circ U \circ f \circ V = \id{A}
\]
Hence $g \circ U$ has left and right inverses, making it and hence $f$ both isomorphisms.

For the second statement, given any tuple $(f_i \colon A_i \to D)^n_{i=1}$, let $f \colon A \to D$ satisfy $f \circ \pcoproj_i = f_i$ for all $i$. Then $f \circ g^{-1} \circ \nu_i = f_i$ for all $i$. Moreover, if $h \circ \nu_i = k \circ \pcoproj'_i$ for all $i$ then 
\[
h \circ g^{-1} \circ \pcoproj_i = k \circ g^{-1} \circ \pcoproj_i
\] 
for all $i$ and so for some phase $U$ on $A$ we have that $h = k \circ V$ where $V = g^{-1} \circ U \circ g$. Finally, $V$ is easily seen to preserve the $\nu_i$.
\end{proof}

\begin{corollary} \label{phase_iso}
Any phase of a phased coproduct is an isomorphism. 
\end{corollary}

Next we observe that phased coproducts are associative in a suitable sense.

\begin{proposition}[Associativity]\label{prop:assoc} 
For any phased coproduct $A \pcoprod B$, any phased coproduct ${(A \pcoprod B) \pcoprod C}$ forms a phased coproduct of $A, B, C$ with coprojections:
\[
\begin{tikzcd}[row sep = tiny]
A \arrow[dr,"\pcoproj_A",pos=0.2]& & \\ 
B \arrow[r, "\pcoproj_B"] & A \pcoprod B \arrow[r,"\pcoproj_{A\pcoprod B}"]& (A \pcoprod B) \pcoprod C\\ 
C \arrow[urr,"\pcoproj_C",swap] & &
\end{tikzcd}
\]
More generally $((A_1 \pcoprod A_2) \pcoprod \dots  ) \pcoprod A_n$ forms a phased coproduct $A_1 \pcoprod \dots \pcoprod A_n$.
\end{proposition}
\begin{proof}
We prove the first case, with the $n$-ary case being similar. 

For any morphisms $f, g, h$ from $A, B, C$ to $D$ respectively, let $k \colon A \pcoprod B \to D$ satisfy $k \circ \pcoproj_A = f$ and $k \circ \pcoproj_B = g$. Then any morphism
\[
\begin{tikzcd}
(A \pcoprod B) \pcoprod C
\rar{t}
& D 
\end{tikzcd}
\]
with $t \circ \pcoproj_{A \pcoprod B} = k$ and  $t \circ \pcoproj_C = h$ composes with the morphisms above to give $f, g, h$ respectively. For uniqueness, suppose that $t'$ is another such morphism. Then there is a phase $U$ on $A \pcoprod B$ with $t' \circ \pcoproj_{A \pcoprod B} = t \circ \pcoproj_{A \pcoprod B} \circ U$. Now let $V$ be an endomorphism of $(A \pcoprod B) \pcoprod C$ with 
\[
V \circ \pcoproj_{A \pcoprod B} = \pcoproj_{A \pcoprod B} \circ U \qquad V \circ \pcoproj_C = \pcoproj_C
\]
Then immediately we have $t \circ V \circ \pcoproj_{A \pcoprod B} = t' \circ \pcoproj_{A \pcoprod B}$ and $t \circ V \circ \pcoproj_C = t' \circ \pcoproj_C$. So there is some $W$ preserving $\pcoproj_{A \pcoprod B}$ and $\pcoproj_C$ with $t = (t \circ V) \circ W$. Finally $V \circ W$ preserves each of the proposed coprojections as required.
\end{proof}

Let us now consider a phased coproduct of an empty collection of objects, which by definition is precisely the following. In any category, a \indef{phased initial object}\index{phased initial object} is an object $\pinit$\label{not:phinitial} for which every object $A$ has a morphism $\pinit \to A$, and such that for any pair of morphisms $a, b \colon \pinit \to A$ there is an endomorphism $U$ of $\pinit$ with $b = a \circ U$. In fact this notion typically coincides with a familiar one. 

\begin{proposition} \label{prop:phase_init_unit}
In a category with binary phased coproducts, any phased initial object $0$ is an initial object and each coprojection $\pcoproj_A \colon A \to A \pcoprod \pinit$ is an isomorphism.
\end{proposition}
\begin{proof}
We first show that $\pcoproj_A$ is an isomorphism. 
Let 
\[
\begin{tikzcd}
A \pcoprod \pinit \rar{f} & A
\end{tikzcd}
\]
with $f \circ \pcoproj_A = \id{A}$ and $f \circ \pcoproj_{\pinit}$ being any morphism $\pinit \to A$. Then $f$ makes $\pcoproj_A$ split monic. 
Because $\pinit$ is phased initial, it has an endomorphism $z$ with  
$\pcoproj_A \circ f \circ \pcoproj_0 = \pcoproj_0 \circ z$, which is an isomorphism by Lemma~\ref{lem:isoms}. 
Next let $g$ be an endomorphism of $A \pcoprod \pinit$ with $g \circ \pcoproj_A = \pcoproj_A$ and $g \circ \pcoproj_{\pinit} = \pcoproj_{\pinit} \circ z^{-1}$. Then it may be readily verified that, by construction, $U := \pcoproj_A \circ f \circ g$ preserves $\pcoproj_0$ and $\pcoproj_A$. Hence $U$ is a phase and so an isomorphism, making $\pcoproj_A$ split epic and hence an isomorphism also. 

We now show that $\pinit$ is initial. Given $a, b \colon \pinit \to A$ let $g, h \colon A \pcoprod \pinit \to A$ with
\[
g \circ \pcoproj_\pinit = a \qquad h \circ \pcoproj_\pinit = b \qquad g \circ \pcoproj_A = \id{A} = h \circ \pcoproj_A
\]
Then $g = \pcoproj_A^{-1} = h$ and so $a = g \circ \pcoproj_\pinit = h \circ \pcoproj_\pinit = b$. 
\end{proof}

\begin{corollary} \label{all_phase_coprod} 
A category has phased coproducts of all finite collections of objects iff it has binary phased coproducts and an initial object.
\end{corollary}

Thanks to this we will often only need to refer to binary phased coproducts from now on.

\begin{remark}[Phased Limits] \label{rem:phasedproducts}
We may have defined phased products $A \pprod B$\label{not:phproduct} and phased terminal objects by dualising the above definitions, but coproducts will be more natural for our familiar monoidal setting. 

 Products and coproducts are special cases of the notion of a (co)limit of a diagram $D \colon \catJ \to \ctb$~\cite{mac1978categories}. More generally we may say that such a diagram $D$ has a \emps{phased (co)limit} if the category of (co)cones over $D$ has a phased terminal (resp.~initial) object. However we won't consider  general phased limits here. 
\end{remark}

\subsection{Examples}

Our motivating example is the following. 

\begin{example}
Recall that $\Hilb$ has finite coproducts given for a pair $\hilbH, \hilbK$ by the direct sum $\hilbH \oplus \hilbK$ of Hilbert spaces, along with the inclusions $\coproj_1 \colon \hilbH \to \hilbH \biprod \hilbK$ and $\coproj_2 \colon \hilbK \to \hilbH \biprod \hilbK$. 

Then $\HilbP$ has finite phased coproducts, given again by $\hilbH \biprod \hilbK$ along with the equivalence classes $[\coproj_1]$ and $[\coproj_2]$ of these maps. Phases on this object are precisely equivalence classes of unitary operators 
\begin{equation*} 
U = 
  \begin{pmatrix}
    \id{\hilbH} & 0 \\
    0 & e^{i \theta} \cdot \id{\hilbK} 
  \end{pmatrix}
\end{equation*}
for some $\theta \in [0, 2\pi]$.
Indeed, for any pair of morphisms $[f] \colon \hilbH \to \hilbL$ and $[g] \colon \hilbK \to \hilbL$ in $\HilbP$, another $[h] \colon \hilbH \biprod \hilbK \to \hilbL$ will satisfy $[h] \circ [\coproj_1] = [f]$ and $[h] \circ [\coproj_2] = [g]$ precisely when in $\Hilb$ we have $h \circ \coproj_1 = e^{i \theta} \cdot f$ and $h \circ \coproj_2 = e^{i \theta'} \cdot g$ for some such $\theta, \theta'$. It is simple to check that any such $h, h'$ have $[h] = [h'] \circ [U]$ for some $U$ as above. 

In particular, let us consider the qubit $\mathbb{C}^2$. Any pair of orthonormal states $\ket{0}, \ket{1}$ form coprojections making $\mathbb{C}^2$ a coproduct in $\Hilb$, or phased coproduct in $\HilbP$. Effects $\psi$ on $\mathbb{C}^2$ in the latter category correspond in $\Hilb$ to weighted superpositions 
\[
r \cdot \bra{0} + s \cdot e^{i \cdot \theta} \cdot \bra{1} 
\]
where $r, s$ are positive reals given by 
$r = \psi \circ \ket{0}$ and $s = \psi \circ \ket{1}$ in $\HilbP$. The term $e^{i \cdot \theta}$ makes such superpositions unique only up to unitaries $U$ as above.
\end{example}

Now we can extend this example considerably. 

\begin{definition}
By a choice of \deff{trivial isomorphisms}\index{trivial isomorphisms} on a category $\cta$ we mean a choice, for each object $A$, of a subgroup $\TrI_A$\label{not:trivisom} of the group of isomorphisms $A \isomto A$ such that for all $f \colon A \to B$ and $p_B \in \TrI_B$ there exists $p_A \in \TrI_A$ with 
\begin{equation} \label{eq:triv-isom}
p_B \circ f = f \circ p_A
\end{equation}
We call a choice of trivial isomorphisms \indef{transitive} \index{transitive!trivial isomorphisms} when, conversely, for all such morphisms $f$ and every $p_A \in \TrI_A$ we have $f \circ p_A = p_B \circ f$ for some $p_B \in \TrI_B$. With or without transitivity, such a choice defines a congruence $\sim$ on $\cta$ given on morphisms $f, g \colon A \to B$ by 
\[
f \simP g \quad \text{ if } \quad  f = g \circ p \text{ for some $p \in \TrI_A$}
\]
In fact this congruence suffices to recover $\TrI_A$ as $\{f \colon A \to A \mid f \sim \id{A}\}$, and so we often equate a choice of trivial isomorphisms with its congruence. 

We write $\quot{\cta}{\sim}$\label{not:quotcat} for the category whose morphisms are equivalence classes $[f]_\tc$\label{not:equivclass} of morphisms $f$ in $\cta$ under $\sim$. There is a wide full functor $[-]_\tc \colon \cta \to \quot{\cta}{\sim}$ given by taking equivalence classes.
\end{definition}

\begin{lemma} \label{lem:generalrecipeforphcoprod}
Let $\cta$ be a category with finite coproducts and a choice of trivial isomorphisms. Then $\quot{\cta}{\sim}$ has finite phased coproducts. Moreover $[-]_\tc$ sends coproducts in $\cta$ to phased coproducts in $\quot{\cta}{\sim}$.
\end{lemma} 

\begin{proof}
Any initial object in $\cta$ is initial in $\quot{\cta}{\sim}$. For any $[f]_\tc \colon A \to C, [g]_\tc \colon B \to C$, the morphism $h \colon A + B \to C$ with $h \circ \coproj_A = f$ and $h \circ \coproj_C =g$ certainly has $[h]_\tc \circ [\pcoproj_A]_\tc = [f]_\tc$ and $[h]_\tc \circ [\pcoproj_B]_\tc = [g]_\tc$. Given any other such $[h']_\tc$, we have $h' \circ \pcoproj_A = f \circ p$ and $h \circ \pcoproj_B = g \circ q$ for some $p \in \TrI_A$ and $q \in \TrI_B$. Then $h = h' \circ U$ where $U \circ \pcoproj_A = \pcoproj_A \circ p$ and $U \circ \pcoproj_B = \pcoproj_B \circ q$, with $[U]_\tc$ preserving the $[\pcoproj_A]_\tc$ and $[\pcoproj_B]_\tc$ in $\quot{\cta}{\sim}$. 
\end{proof}

\begin{examples} The following choices of trivial isomorphisms provide examples of categories with phased coproducts.
\begin{exampleslist}
\item 
In $\Hilb$ choose as trivial isomorphisms on $\hilbH$ all maps of the form $e^{i \cdot \theta} \cdot \id{\hilbH}$ for $\theta \in [0, 2\pi)$. Then the induced congruence is 
\begin{equation} \label{eq:HilbRelation}
f \sim g \quad \text{   if   } \quad  f = e^{i \theta} \cdot g
\end{equation}
and so the category $\HilbP \simeq \quot{\Hilb}{\sim}$ has finite phased coproducts as we have seen. Similarly so does $\FHilbP$.  
\item 
Extending the above, in $\VecC$\label{cat:vecc} take as trivial isomorphisms on $V$ all linear maps $e^{i \theta} \cdot \id{V}$ for $\theta \in [0,2\pi)$. Again $\VecC$ has coproducts given by the direct sum of vector spaces, and so these become phased coproducts in $\VecP := \quot{\VecC}{\sim}$\label{cat:vecP}.
\item 
For any field $\fieldk$, in $\FVeck$ choose as trivial isomorphisms on $V$ all maps $\lambda \cdot \id{V}$ for $\lambda \neq 0$, and let $\VecProj := \quot{\FVeck}{\sim}$\label{cat:vecproj}. Morphisms here are linear maps up to an overall scalar $\lambda$. Identifying vectors $\psi \colon \fieldk \to V$ with the same span in this way leads to \emps{projective geometry}~\cite{coxeter2003projective}. Note however that $\VecProj$ differs from usual projective geometry by including zeroes and non-injective maps. 
\item 
For an abelian group $G$, let $\GSet$\label{cat:GSet} be the category of sets $A$ equipped with a group action $a \mapsto g \cdot a$, with morphisms being maps $f \colon A \to B$ which are equivariant, i.e.~with $f(g \cdot a) = g \cdot f(a)$ $\forall g, a$. Choose as trivial isomorphisms on $A$ the maps $g \cdot (-) \colon A \to A$ for some $g \in G$. Then $\quot{\GSet}{\sim}$ identifies maps $f, f'$ whenever there is some $g \in G$ with $f(a) = f'(g \cdot a)$ for all $a \in A$. It has finite phased coproducts given by the coproducts in $\GSet$, i.e. disjoint union of sets.
\end{exampleslist}
\end{examples}

Each of these examples of trivial isomorphisms are transitive, giving their induced phased coproducts a property which will be useful in what follows. First, let us say that a morphism $f \colon A \pcoprod B \to C \pcoprod D$ is \indef{diagonal}\index{diagonal morphism} when $f \circ \pcoproj_A = \pcoproj_C \circ g$ and $f \circ \pcoproj_B = \pcoproj_D \circ h$ for some $g, h$. 

\begin{definition} \label{def:trans_phases}
A category with phased coproducts has \deff{transitive phases} \index{transitive!phases} when every diagonal morphism $f \colon A \pcoprod B \to C \pcoprod D$ and phase $U$ of $A \pcoprod B$ has 
\[
f \circ U = V \circ f
\]
for some phase $V$ of $C \pcoprod D$.
\end{definition}

\section{From Phased Coproducts to Coproducts} \label{sec:phToCoprod}

We now wish to find a converse construction to Lemma~\ref{lem:generalrecipeforphcoprod}, allowing us to exhibit any suitable category with phased coproducts as a quotient of one with coproducts.

\begin{definition} \label{def:justplusI}
Let $\ctb$ be a category with finite phased coproducts and a distinguished object $\Gen$. The category $\plusI{\ctb}$ is defined as follows:
\begin{itemize}
\item 
objects are phased coproducts of the form $\obb{A} = A \pcoprod \Gen$ in $\ctb$ (each including as data the objects $A, \Gen$ and morphisms $\coproj_A$, $\coproj_\Gen$);
\item 
morphisms $f \colon \obb{A} \to \obb{B}$ are diagonal morphisms in $\ctb$ with $f \circ \pcoproj_\Gen = \pcoproj_\Gen$.
\[
\begin{tikzcd}
\obb{A} \rar{f} & \obb{B} \\ 
A \uar{\pcoproj_A} \rar[dashed,swap]{\exists} & B \uar[swap]{\pcoproj_B}
\end{tikzcd}
\qquad
\begin{tikzcd}
\obb{A} \arrow[rr, "f"] & & \obb{B} \\ 
& I \arrow[ul,"\pcoproj_\Gen"] \arrow[ur,swap,"\pcoproj_\Gen"]& 
\end{tikzcd}
\]
\end{itemize}
\end{definition}
Such diagonal morphisms are straightforwardly checked to be closed under composition, making  $\mathsf{GP}(\ctb)$ a well-defined category with composition and identity morphisms being the same as in $\ctb$. Our notation $\mathsf{GP}$ stands for `global phases', based on our motivating example $\Hilb$ and which we consider more abstractly in the next section.  

Now a sufficient condition on $I$ for $\plusI{\ctb}$ to have coproducts is the following. Call a morphism $f \colon A \pcoprod B \to C$ \indef{phase monic}\index{phase monic} when $f \circ U = f \circ V \implies U = V$ for all phases $U, V$. Similarly a morphism $g \colon C \to A \pcoprod B$ is \indef{phase epic}\index{phase epic} when $U \circ g = V \circ g \implies U = V$ for phases $U, V$.

\begin{definition} \label{def:phase-gen} \label{not:GP}
Let $\ctb$ be a category with finite phased coproducts. We say an object $\Gen$ is a \deff{phase generator}\index{phase generator} when:
\begin{itemize}
\item 
any $\triangledown \colon I \pcoprod I \to I$ with $\triangledown \circ \pcoproj_1 = \id{I} = \triangledown \circ \pcoproj_2$ is phase monic;
\item 
any diagonal monomorphism $m \colon I \pcoprod I \to A \pcoprod B$ is phase epic.

\end{itemize}
\end{definition}

Let us say that phased coproducts or coproducts in a category are \indef{monic} \index{coproduct!monic} \index{phased coproduct!monic} whenever all coprojections are monic.
In this case we write $[-] \colon \plusI{\ctb} \to \ctb$ for the functor sending $\obb{A} \mapsto A$ and $f \colon \obb{A} \to \obb{B}$ to the unique $[f] \colon A \to B$ with $f \circ \pcoproj_A = \pcoproj_B \circ [f]$.

\begin{theorem} \label{thm:localToGlobal}
Let $\ctb$ be a category with finite monic phased coproducts with transitive phases and a phase generator $\Gen$. Then $\plusI{\ctb}$ has monic finite coproducts. Moreover, it has a choice of trivial isomorphisms 
\[
\TrI_\obb{A} := \{ U \colon \obb{A} \to \obb{A} \mid U \text{ is a phase}\}
\]
whose congruence $\sim$ induces an equivalence of categories
\[
\ctb \simeq \quot{\plusI{\ctb}}{\sim}
\]
\end{theorem}

\begin{proof}
Note that any initial object $0$ in $\ctb$ forms an initial object $\obb{0} = 0 \pcoprod \Gen$ in $\plusI{\ctb}$. Indeed any morphism $f \colon \obb{0} \to \obb{A}$ preserves the $\pcoproj_\Gen$, but by Proposition~\ref{prop:phase_init_unit} $\pcoproj_\Gen \colon \Gen \to \obb{0}$ is an isomorphism, making $f$ unique. 

Now for any pair of objects $\obb{A} = A \pcoprod \Gen$, $\obb{B} = B \pcoprod \Gen$ in $\plusI{\ctb}$ we claim that any phased coproduct $A \pcoprod B$ and object 
\[
\obb{A} + \obb{B} := (A \pcoprod B) \pcoprod \Gen
\]
and morphisms $\pcoproj_{A,\Gen} \colon \obb{A} \to \obb{A} + \obb{B}$ and $\pcoproj_{B,\Gen} \colon \obb{B} \to \obb{A} + \obb{B}$ with $[\pcoproj_{A,\Gen}] = \pcoproj_A$ and $[\pcoproj_{B,\Gen}] = \pcoproj_B$ forms their coproduct in $\plusI{\ctb}$. These morphisms are special kinds of coprojections by associativity (Proposition~\ref{prop:assoc}) and so in particular are monic. We need to show that for all morphisms $f, g$ belonging to $\plusI{\ctb}$ that in $\ctb$ there is a unique $h$ making the following commute:
\[
\begin{tikzcd}
& (A \pcoprod B) \pcoprod \Gen \arrow["h", d, dotted] & \\
 A \pcoprod \Gen \arrow[swap,"f", r] \arrow[ur, "\pcoproj_{A,\Gen}"]  & C \pcoprod \Gen & \Gen \pcoprod B \arrow[ul, swap,"\pcoproj_{B,\Gen}"] \arrow["g", l] 
\end{tikzcd}
\]
We start with the existence property. By Proposition~\ref{prop:assoc} $(A \pcoprod B) \pcoprod \Gen$ also forms a phased coproduct of $A \pcoprod \Gen$ and $B$ via $\pcoproj_{A,\Gen}$ and $\pcoproj_B = \pcoproj_{A, \Gen} \circ \pcoproj_B$. So there exists $k \colon (A \pcoprod B) \pcoprod \Gen \to C \pcoprod I$ with $k \circ \pcoproj_{A,\Gen} = f$ and $k \circ \pcoproj_B = g \circ \pcoproj_B$. Then
\[
k \circ \pcoproj_\Gen = k \circ \pcoproj_{A,\Gen} \circ \pcoproj_\Gen = f \circ \pcoproj_\Gen = g \circ \pcoproj_\Gen
\]
also, and so $k \circ \pcoproj_{B,\Gen} = g \circ U$ for some phase $U$ on $B \pcoprod \Gen$. By transitivity there then is a phase $V$ with respect to $\pcoproj_{A, \Gen}$, $\pcoproj_B$ for which $\pcoproj_{B,\Gen} \circ U = V \circ \pcoproj_{B,\Gen}$. Then $h= k \circ V^{-1}$ is easily seen to have the desired properties.

We next verify uniqueness. Suppose that there exists $f, g$ with $f \circ \pcoproj_{A,\Gen} = g \circ \pcoproj_{A,\Gen}$ and $f \circ \pcoproj_{B,\Gen} = g \circ \pcoproj_{B,\Gen}$. Consider morphisms $h, j$ as in the diagram 
\[
\begin{tikzcd}
\Gen \pcoprod \Gen \dar[swap]{\triangledown} \arrow[loop left, "V", distance=2em] \rar{j} & (A \pcoprod \Gen) \pcoprod (B \pcoprod \Gen) \arrow[loop right, "U", distance=4em] \dar{h} & \\ 
\Gen \rar[swap]{\pcoproj_\Gen} & (A \pcoprod B) \pcoprod \Gen \arrow[r, shift left=2.5, "f"] \arrow[r, shift right=2.5, "g", swap] & C \pcoprod \Gen
\end{tikzcd}
\]
with
\begin{align*}
&h \circ \pcoproj_{A \pcoprod \Gen} = \pcoproj_{A,\Gen}
&j \circ \pcoproj_1 = \pcoproj_{A,\Gen} \circ \pcoproj_\Gen
\\
&h \circ \pcoproj_{B \pcoprod \Gen} = \pcoproj_{B,\Gen}
&j \circ \pcoproj_2 = \pcoproj_{B,\Gen} \circ \pcoproj_\Gen
\end{align*}
Then $g \circ h = f \circ h \circ U$ for some phase $U$ on $(A \pcoprod \Gen) \pcoprod (\Gen \pcoprod B)$. Now $h$ is split epic, since $h \circ k$ is a phase whenever $k$ is a morphism in the opposite direction defined via any of the obvious inclusions of $A, B$ and $I$ into each object. Hence it suffices to prove that $U=\id{}$.

Since $j$ is diagonal we have $U \circ j = j \circ V$ for some phase $V$ as above. We first show that $V = \id{\Gen \pcoprod \Gen}$. Composing with coprojections shows that $h \circ j  = \coproj_\Gen \circ \triangledown$ for some $\triangledown$ with $\triangledown \circ \pcoproj_1 = \triangledown \circ \pcoproj_2 = \id{\Gen}$. Then we have
\begin{align*}
\pcoproj^{C \pcoprod \Gen}_\Gen \circ \triangledown 
&= 
g \circ \pcoproj_\Gen \circ \triangledown
\\
&=
g \circ h \circ j
\\ 
&=
f \circ h \circ U \circ j
\\
&=
f \circ h \circ j \circ V
=
\pcoproj^{C \pcoprod \Gen}_\Gen \circ \triangledown \circ V
\end{align*}
and so $\triangledown  = \triangledown \circ V$. Then since $I$ is a phase generator $V = \id{\Gen \pcoprod \Gen}$, so that $U \circ j = j$. Now again by associativity of phased coproducts $j$ is a coprojection and so is monic, and then since it is diagonal and $I$ is a phase generator we have $U = \id{}$.

For the second statement, note that these $\TrI_\obb{A}$ are a valid choice of trivial isomorphisms, satisfying \eqref{eq:triv-isom} since all morphisms in $\plusI{\ctb}$ are diagonal in $\ctb$. Moreover we indeed have $[f]_\tc = [g]_\tc$ whenever $[f] = [g]$ for the functor $[-] \colon \plusI{\ctb} \to \ctb$. Hence $[-]$ restricts along $[-]_\tc$ to an equivalence $\quot{\plusI{\ctb}}{\sim} \simeq \ctb$.
\end{proof}

\section{Phases in Monoidal Categories} \label{sec:monoidal}

Our treatment of phased coproducts so far has been more general than needed for our main examples, which additionally come with a compatible monoidal structure which we will see makes the $\mathsf{GP}$ construction a natural one.

First, say that a functor $F \colon \ctb \to \ctb'$ \indef{strongly preserves} phased coproducts if for every phased coproduct $A \pcoprod B$ with coprojections $\pcoproj_A, \pcoproj_B$ in $\ctb$, $F(A \pcoprod B)$ is a phased coproduct with coprojections $F(\pcoproj_A)$, $F(\pcoproj_B)$ and moreover has that every phase is of the form $F(U)$ for some phase $U$ of $A \pcoprod B$.

\begin{definition} \label{def:distirb}
We say that phased coproducts in a monoidal category are \deff{distributive} \index{distributive!phased coproducts} when they are strongly preserved by the functors $A \otimes (-)$ and $(-) \otimes A$, for all objects $A$.
\end{definition}

Thanks to Lemma~\ref{lem:isoms} the requirement on $A \otimes (-)$ is equivalent to requiring that some (and hence any) morphism 
\begin{equation} \label{eq:distrib-isomorphism}
 \begin{tikzcd}
A \otimes B \pcoprod A \otimes C \rar{f} & A \otimes (B \pcoprod C)
\end{tikzcd}
\qquad
\text{with}
\qquad
\begin{tabular}{cc}
$f \circ \pcoproj_{A \otimes B} = \id{A} \otimes \pcoproj_B$
\\ 
$f \circ \pcoproj_{A \otimes C} = \id{A} \otimes \pcoproj_C$
\end{tabular}
\end{equation}
 is an isomorphism, and moreover has that every phase on its domain is of the form $f^{-1} \circ (\id{A} \otimes U) \circ f$ for some phase $U$ of $B \pcoprod C$. In the case of actual coproducts, this specialises to the usual notion of distributivity we met in Chapter~\ref{chap:OpCategories}, with the phase condition redundant.  

\begin{remark}
Our definition of distributivity, requiring from strong preservation that \emps{every} phase of $A \otimes B \pcoprod A \otimes C$ arises from one of $B \pcoprod C$, may indeed appear rather strong. However we will find it to hold in very general quotient categories, and in Section \ref{sec:compact_cats} to be automatic in any compact category.  
\end{remark}

Now the trivial isomorphisms in our main examples may be defined naturally using their monoidal structure as follows. In any monoidal category let us call a scalar $s$ \indef{central} \index{central scalar} when we have 
\[
\scalebox{0.8}{\input{./figures/central-scalar-sup.tikz}}
\]
for all morphisms $f$. In a symmetric monoidal category every scalar is central.

\begin{definition}
 By a choice of \deff{global phases} \index{global phases} in a monoidal category $\cta$ we mean a collection $\mathbb{P}$ \label{not:globalphases} of invertible, central scalars closed under composition and inverses. 
\end{definition}

Any such global phase group $\mathbb{P}$ determines a choice of trivial isomorphisms on $\cta$ by setting $\TrI_A := \{ p \cdot \id{A} \mid p \in \mathbb{P}\}$. Then $\TrI_I = \mathbb{P}$, the induced congruence is   
\begin{equation} \label{eq:quotient_rule_general}
\scalebox{0.8}{\input{./figures/mon-congruence-1.tikz}}
\quad
\text{ if } 
\quad
\scalebox{0.8}{\input{./figures/mon-congruence-2.tikz}}
\quad
\text{ for some $u \in \mathbb{P}$}
\end{equation}
and we write $\cta_\quotP := \quot{\cta}{\sim}$\label{not:quotglobphase}.

\begin{lemma} \label{lem:globaltolocalmonoidal}
Let $\cta$ be a monoidal category with distributive finite coproducts and a choice of global phases $\mathbb{P}$. Then $\cta_\quotP$ is a monoidal category with distributive finite phased coproducts with transitive phases.
\end{lemma}
\begin{proof}
Since the $p \in \mathbb{P}$ are central we have $f \sim h$, $g \sim k$ $\implies$ $f \otimes g \sim h \otimes k$. Hence $\otimes$ restricts from $\cta$ to $\cta_\quotP$, making the latter category monoidal. By Lemma~\ref{lem:generalrecipeforphcoprod} coproducts in $\cta$ become phased coproducts in $\cta_\quotP$. Distributivity is inherited from $\cta$, and transitivity from the fact that $(p \cdot \id{}) \circ f = p \cdot f = f \circ (p \cdot \id{})$ for all morphisms $f$ and scalars $p$ in a monoidal category. 
\end{proof}

\begin{examples} \label{example:global-phases}
$\VecC$ and $\Hilb$ are monoidal with distributive finite coproducts, and our earlier choice of trivial isomorphisms correspond to the global phase group $\mathbb{P} = \{ e^{i \theta} \mid \theta \in [0,2\pi)\}$ in both cases. Similarly $\FVeck$ is monoidal and its choice of trivial isomorphisms comes from the global phase group $\mathbb{P} = \{ \lambda \in k \mid  \lambda \neq 0\}$. 
\end{examples}

We now wish to give a converse to this result, showing that $\plusI{\ctb}$ is a monoidal category with a canonical choice of global phases. When $\ctb$ is monoidal we'll always take as chosen object $\Gen$ its monoidal unit. To prove monoidality of $\plusI{\ctb}$ we will use the following general result from~\cite[Prop.~2.6, Lemma~2.7]{kock2008elementary}. 

\begin{lemma} \label{lem:mon_cat_helpful}
A monoidal structure on a category is equivalent to specifying:
\begin{itemize}
\item  a bifunctor $\otimes$ and natural isomorphism $\alpha$ satisfying the pentagon equation;
\item 
an object $I$ such that every morphism $A \otimes I \to B \otimes I$ and $I \otimes A \to I \otimes B$ is of the form $g \otimes \id{I}$, $\id{I} \otimes g$ respectively, for some unique $g \colon A \to B$;
\item 
an isomorphism $\beta \colon I \otimes I \isomto I$.
\end{itemize}
\end{lemma}

We will also repeatedly use the following elementary observation.

\begin{lemma} \label{lem:useful}
Suppose that we have morphisms
\[
\begin{tikzcd}
A \pcoprod B \rar{f} & E & C \pcoprod D \lar[swap]{g}
\end{tikzcd}
\]
with $f \circ \pcoproj_A = g \circ \pcoproj_C \circ h$ and $f \circ \pcoproj_B = g \circ \pcoproj_D \circ k$ for some $h, k$. Then $f = g \circ l$ for some diagonal morphism $l$.
\end{lemma}
\begin{proof}
Let $m \colon A \pcoprod B \to C \pcoprod D$ have $m \circ \coproj_A = \coproj_C \circ h$ and $m \circ \coproj_B = \coproj_D \circ k$. Then $f = g \circ m \circ U$ for some phase $U$, giving $l = m \circ U$ as the desired morphism. 
\end{proof}

\begin{theorem} \label{thm:constr_is_monoidal}
Let $\ctb$ be a monoidal category with distributive monic finite phased coproducts. 
Then $\plusI{\ctb}$ is a monoidal category, and $[-] \colon \plusI{\ctb} \to \ctb$ is a strict monoidal functor. 
\end{theorem}

\begin{proof}
We define a monoidal product $\tens$ on $\plusI{\ctb}$ as follows. For each pair of objects $\obb{A}, \obb{B}$ choose some object $\obb{A} \tens \obb{B} = A \otimes B \pcoprod I$ and $c_{\obb{A},\obb{B}} \colon \obb{A} \tens \obb{B} \to \obb{A} \otimes \obb{B}$ satisfying
\begin{equation} \label{eq:corner}
c_{\obb{A},\obb{B}} \circ \pcoproj_{A \otimes B} = \pcoproj_A \otimes \pcoproj_B 
\qquad 
c_{\obb{A},\obb{B}} \circ \pcoproj_I = (\pcoproj_I \otimes \pcoproj_I) \circ \rho_{I}^{-1}
\end{equation}
which we depict as 
\[
\scalebox{0.8}{\input{./figures/sup3.tikz}}
\]
Using distributivity, associativity (Proposition~\ref{prop:assoc}), and $\rho_I$, we have isomorphisms
\begin{align*}
\obb{A} \otimes \obb{B} 
&\simeq
(A \otimes B \pcoprod A \otimes I) \pcoprod (I \otimes B \pcoprod I \otimes I) \\ 
&\simeq 
(\obb{A} \tens \obb{B})\pcoprod (A \otimes I \pcoprod I \otimes B)
\end{align*}
making any such morphism $c_{\obb{A},\obb{B}}$ a coprojection, and hence monic. Then for morphisms $f \colon \obb{A} \to \obb{C}$ and $g \colon \obb{B} \to \obb{D}$ in $\plusI{\ctb}$ we define $f \tens g$ to be the unique morphism in $\ctb$ such that 
\[
\scalebox{0.8}{\input{./figures/sup2.tikz}}
\]
Indeed such a map exists and belongs to $\plusI{\ctb}$ by Lemma~\ref{lem:useful} since we have
\begin{align*}
(f \otimes g) \circ c_{\obb{A},\obb{B}} \circ \pcoproj_{A \otimes B} &= c_{\obb{C}, \obb{D}} \circ \pcoproj_{C \otimes D} \circ ([f] \otimes [g]) \\ 
(f \otimes g) \circ c_{\obb{A},\obb{B}} \circ \pcoproj_{I \otimes I} &= c_{\obb{C},\obb{D}} \circ \pcoproj_I
\end{align*}
Uniqueness follows from monicity of $c_{\obb{C},\obb{D}}$ and ensures that $\tens$ is functorial. We define $\aalpha_{\obb{A},\obb{B},\obb{C}} \colon (\obb{A} \tens \obb{B}) \tens \obb{C} \to \obb{A} \tens (\obb{B} \tens \obb{C})$ to be the unique morphism such that 
\begin{equation} \label{eq:defaaalpha}
\scalebox{0.8}{\input{./figures/sup5.tikz}}
\end{equation}
Existence again follows from Lemma~\ref{lem:useful}. For uniqueness, distributivity tells us that each morphism $\id{A} \otimes c_{\obb{B}, \obb{C}}$ is again a coprojection since $c_{\obb{B}, \obb{C}}$ is, and hence is monic. By symmetry there is some $\aalpha'_{\obb{A}, \obb{B}, \obb{C}}$ satisfying the horizontally reflected version of~\eqref{eq:defaaalpha}, and then thanks to uniqueness this is an inverse to $\aalpha_{\obb{A}, \obb{B}, \obb{C}}$. 

Again using monicity of the $\id{\obb{A}} \otimes c_{\obb{B},\obb{C}}$ we verify that $\aalpha$ is natural:
\[
\scalebox{0.8}{\input{./figures/sup5il.tikz}}
\]
and that it satisfies the pentagon law:
\begingroup
\allowdisplaybreaks
\begin{align*}
\scalebox{0.8}{\input{./figures/sup6il-1.tikz}} \\ 
\scalebox{0.8}{\input{./figures/sup6il-2.tikz}} 
\end{align*}
\endgroup
For the unit object in $\plusI{\ctb}$ choose any $\obb{I} = I \pcoprod I$. Then any morphism $\beta \colon \obb{I} \tens \obb{I} \to \obb{I}$ with $\beta \circ \pcoproj_1 = \pcoproj_1 \circ \rho_I$ and $\beta \circ \pcoproj_2 = \pcoproj_2$ is an isomorphism belonging to $\plusI{\ctb}$. 

We now show that in $\plusI{\ctb}$ every morphism $f \colon \obb{A} \tens \obb{I} \to \obb{B} \tens \obb{I}$ is of the form $g \tens \id{\obb{I}}$ for a unique $g \colon \obb{A} \to \obb{B}$. Choose any $r_A$, $r_B$ in $\plusI{\ctb}$ with $[r_A] = \rho_A$ and $[r_B] = \rho_B$ in $\ctb$, setting
\[
\scalebox{0.8}{\input{./figures/sup7.tikz}}
\]
Then the statement is equivalent to requiring that for every diagonal $f \colon \obb{A} \to \obb{B}$ there is a unique diagonal $g \colon \obb{A} \to \obb{B}$ with
\[
\scalebox{0.8}{\input{./figures/sup8.tikz}}
\]
Now let $\tinycounit \colon \obb{I} \to I$ in $\ctb$ with $\tinycounit \circ \pcoproj_1 = \id{I} = \tinycounit \circ \pcoproj_2$. Applying coprojections we have 
\[
\scalebox{0.8}{\input{./figures/sup9i.tikz}} 
\]
for some phase $U$, which is in particular invertible. This makes $g$ unique. We now show that $g$ exists. 
Applying coprojections again one may see that
\[
\scalebox{0.8}{\input{./figures/sup13ii.tikz}}
\]
for some phases $V$ and $W$. But then
\[
\scalebox{0.8}{\input{./figures/sup12i.tikz}}
\]
yielding the result with $g = f \circ V \circ U$. The statement about morphisms $\obb{I} \tens \obb{A} \to \obb{I} \tens \obb{B}$ follows similarly. Hence by Lemma~\ref{lem:mon_cat_helpful}, $(\tens,\aalpha, \obb{I}, \beta)$ extends to a monoidal structure on $\plusI{\ctb}$. Finally from their definitions we quickly see that $[f \tens g]=[f] \otimes [g]$, $[\aalpha]= \alpha$, and $[\beta]=\rho_I$, and hence by~\cite[Proposition~3.5]{kock2008elementary} the functor $[-]$ is strict monoidal.
\end{proof}

\begin{lemma} 
In the situation of Theorem~\ref{thm:constr_is_monoidal}, if $\ctb$ is symmetric monoidal then so are $\plusI{\ctb}$ and the functor $[-]$.
\end{lemma}
\begin{proof}
Define $\ssigma_{\obb{A}, \obb{B}} \colon \obb{A} \tens \obb{B} \to \obb{B} \tens \obb{A}$ to be the unique map such that 
\[
\sigma_{\obb{A}, \obb{B}} \circ c_{\obb{A},\obb{B}} = c_{\obb{B}, \obb{A}} \circ \ssigma_{\obb{A}, \obb{B}}
\]
again establishing existence with Lemma~\ref{lem:useful}. Since $\sigma_{\obb{A}, \obb{B}}$ is an isomorphism with inverse $\sigma_{\obb{B}, \obb{A}}$, uniqueness forces $\ssigma_{\obb{A}, \obb{B}}$ to be the same. Naturality of $\ssigma$ is easily verified using monicity of the $c_{\obb{A}, \obb{B}}$ and the definition of $\tens$.
We now check the first hexagon equation, with the second being shown dually. 
\[
\scalebox{0.8}{\input{./figures/sup29l.tikz}}
\]
\end{proof}

To show next that $\plusI{\ctb}$ has coproducts, we use the following. 

\begin{lemma} \label{lem:deleter_cancellation_modified}
Let $\ctb$ be monoidal with distributive finite phased coproducts. Then $I$ is a phase generator. 
\end{lemma}

\begin{proof}
Let $\obb{I} = I \pcoprod I$, $\tinycounit \colon \obb{I} \to I$ with $\tinycounit \circ \pcoproj_1 = \tinycounit \circ \pcoproj_2 = \id{I}$ and $U$ be a phase on $\obb{I}$ with $\tinycounit \circ U = \tinycounit$. We need to show that $U = \id{\obb{I}}$. Let $\tinymultflip \colon \obb{I} \to \obb{I} \otimes \obb{I}$ with $\tinymultflip \circ \pcoproj_i =  (\pcoproj_i \otimes \pcoproj_i) \circ \rho_I^{-1}$ for $i = 1, 2$. Applying the $\pcoproj_i$ we see that there are phases $Q, V$ and $W$ with 
\[
\scalebox{0.8}{\input{./figures/sup15.tikz}}
\]
But then
\[
\scalebox{0.8}{\input{./figures/sup16.tikz}}
\]
and so $W = \id{}$. Hence $U \circ Q = Q \circ W = Q$ and so $U = \id{}$. 

For the next property, let $m \colon \obb{I} \to A \pcoprod B$ be a diagonal monomorphism and $U$ a phase on $A \pcoprod B$ with $U \circ m = m$. We need to show that $U = \id{A \pcoprod B}$. Let $\tinymultflip[whitedot] \colon A \pcoprod B \to (A \pcoprod B) \otimes \obb{I}$ with $\tinymultflip[whitedot] \circ \pcoproj_A = (\pcoproj_A \circ \pcoproj_1) \circ {\rho_A}^{-1}$ and $\tinymultflip[whitedot] \circ \pcoproj_B = (\pcoproj_B \circ \pcoproj_2) \circ {\rho_B}^{-1}$. Applying coprojections and using distributivity we see that there are phases $V$ and $W$ on $\obb{I}$ with
\[
\scalebox{0.8}{\input{./figures/sup30.tikz}}
\]
Then we have
\[
\scalebox{0.8}{\input{./figures/sup31.tikz}}
\]
and so composing with $\tinycounit$ and using monicity of $m$ we obtain
\[
\scalebox{0.8}{\input{./figures/sup32.tikz}}
\]
But now $(\tinycounit \otimes \id{\obb{I}}) \circ \tinymultflip$ is a phase and so is epic. Hence by the first part we have $V = \id{\obb{I}}$. Similarly $(\id{A \pcoprod B} \otimes \tinycounit) \circ \tinymultflip[whitedot] = Q$ for some phase $Q$ on $A \pcoprod B$, giving $Q \circ U = Q$ and so $U = \id{}$. 
\end{proof}

\begin{theorem} \label{thm:getmoncoprod}
Let $\ctb$ be a monoidal category with distributive monic finite phased coproducts with transitive phases. Then $\plusI{\ctb}$ has distributive, monic finite coproducts.  
\end{theorem}
\begin{proof}
The monoidal unit $I$ is a phase generator by Lemma~\ref{lem:deleter_cancellation_modified}. Hence by Theorem~\ref{thm:localToGlobal} $\plusI{\ctb}$ has finite coproducts $\obb{A} + \obb{B}$ and these are sent by $[-]$ to phased coproducts in $\ctb$. 
For distributivity consider the unique $f \colon \obb{A} \tens \obb{C} + \obb{B} \tens \obb{C} \to (\obb{A} + \obb{B}) \tens \obb{C}$ in $\plusI{\ctb}$ with $f \circ \pcoproj_1 = \coproj_{\obb{A}} \tens \id{\obb{C}}$ and $f \circ \pcoproj_2 =  \coproj_{\obb{B}} \tens \id{\obb{C}}$. We have
\[
[f] \circ \pcoproj_{A \otimes C} 
=
[f] \circ [\coproj_{\obb{A} \tens \obb{C}}]
=
[\coproj_{\obb{A}} \tens \id{\obb{C}}]
=
[\coproj_{\obb{A}}] \otimes [\id{\obb{C}}]
=
\pcoproj_A \otimes \id{C}
\]
and $[f] \circ \pcoproj_{B \otimes C} = \pcoproj_B \otimes \id{C}$ also. By distributivity in $\ctb$, $[f]$ is then an isomorphism. But since phases are invertible, $[-]$ reflects isomorphisms, so $f$ is invertible. 
\end{proof}
 
To equip $\plusI{\ctb}$ with a choice of global phases we will use the following.

\begin{lemma} \label{lem:central-helpful}
In any monoidal category with distributive monic finite coproducts a scalar $s$ is central iff for every object $A$ there is a scalar $t$ with $s \cdot \id{A} = \id{A} \cdot t$.
\end{lemma}
\begin{proof}
Let $A$ be any object. Suppose that $s \cdot \id{A + I} = \id{A + I} \cdot t$ for some scalar $t$. Then $\coproj_I \circ s = \coproj_I \circ t$ and so by monicity of $\pcoproj_I$ we have $s = t$. But then $\coproj_A \circ (s \cdot \id{A}) = \coproj_A \circ (\id{A} \cdot s)$ and so by monicity again $s \cdot \id{A} = \id{A} \cdot s$. 
\end{proof}

\begin{lemma} \label{lem:phasesAreScalars}
Let $\ctb$ be a monoidal category with distributive monic phased coproducts with transitive phases. Then $\plusI{\ctb}$ has a canonical choice of global phases
\[
\mathbb{P} := \{u \colon \obb{I} \to \obb{I} \mid u \text{ is a phase on $\obb{I}$ in $\ctb$} \}
\]
where $\obb{I} = I \pcoprod I$ is its monoidal unit.
Moreover, phases $U$ on $\obb{A} = A \pcoprod I$ in $\ctb$ are precisely morphisms in $\plusI{\ctb}$ of the form $u \cdot \id{\obb{A}}$ for some $u \in \mathbb{P}$.
\end{lemma}

\begin{proof}
We begin with the second statement. 
An endomorphism $U$ of $\obb{A}$ in $\plusI{\ctb}$ is a phase on $\obb{A}$ in $\ctb$ iff $[U] = \id{A}$.
For any $u$ as above, since $[-]$ is strict monoidal we indeed have $[u \cdot \id{\obb{A}}] = [u] \cdot [\id{\obb{A}}] = \id{A}$, and so $u \cdot \id{\obb{A}}$ is a phase. 

Conversely, for any phase $U$ on $\obb{A}$, consider it instead as an automorphism $V$ of $\obb{A} \tens \obb{I}$ in $\plusI{\ctb}$. Then $[V] = \id{A}$, and so $V$ is a phase on $\obb{A} \tens \obb{I}$. 
Now in $\ctb$, by distributivity, $\obb{A} \otimes \obb{I}$ forms $A \otimes \obb{I} \pcoprod I \otimes \obb{I}$ with every phase being of the form $\id{\obb{A}} \otimes u$ for some phase $u$ on $\obb{I}$. Moreover, $c:= c_{\obb{A}, \obb{I}} \colon \obb{A} \tens \obb{I} \to \obb{A} \otimes \obb{I}$ is then diagonal as a morphism from $\obb{A} \tens \obb{I}$ into this phased coproduct. Hence by transitivity $c \circ V = (\id{\obb{A}} \otimes u) \circ c$ for some phase $u$ of $\obb{I}$. But this states precisely that in $\plusI{\ctb}$ we have $V = \id{\obb{A}} \tens u$ or equivalently $U = u \cdot \id{\obb{A}}$.

Dually, every phase is of the form $\id{\obb{A}} \cdot v$ for some $v \in \mathbb{P}$. In particular for each $u \in \mathbb{P}$ so is $u \cdot \id{\obb{A}}$. Hence by Lemma~\ref{lem:central-helpful} every $u \in \mathbb{P}$ is central, making $\mathbb{P}$ a valid choice of global phases.
\end{proof} 

\begin{corollary} \label{cor:phcoprodcorrespon}
There is a one-to-one correspondence, up to monoidal equivalence, between monoidal categories 
\begin{itemize}
\item $\cta$ with distributive, monic finite coproducts and choice of global phases $\mathbb{P}$;
\item $\ctb$ with distributive, monic finite phased coproducts with transitive phases;
\end{itemize}
given by $\cta \mapsto \cta_\quotP$ and $\ctb \mapsto \plusI{\ctb}$.
\end{corollary}
\begin{proof}
The assignments are well-defined by Theorems~\ref{thm:constr_is_monoidal} and~\ref{thm:getmoncoprod} and Lemmas~\ref{lem:globaltolocalmonoidal} and~\ref{lem:phasesAreScalars}. Now by Theorem~\ref{thm:localToGlobal}, $[-]$ induces an equivalence $\ctb \simeq \quot{\plusI{\ctb}}{\sim}$ where $f \sim g$ when $f = g \circ U$ for some phase $U$ in $\ctb$. But now this is strict monoidal since $[-]$ is, and by Lemma~\ref{lem:phasesAreScalars} in $\plusI{\ctb}$ every such $U$ is of the form $\id{} \cdot u$ for some $u \in \mathbb{P}$. Hence $\quot{\plusI{\ctb}}{\sim} = \plusI{\ctb}_\quotP$.

Conversely, we must check that $\cta \simeq \plusI{\cta_\quotP}$ for such a category $\cta$. Define a functor $F \colon \cta \to \plusI{\cta_\quotP}$ on objects by $F(A) = A + I$ and for $f \colon A \to B$ by setting $F(f) = [f + \id{I}]_\tp \colon A + I \to B + I$, where $[-]_\tp$ denotes equivalence classes under~\eqref{eq:quotient_rule_general}. By Lemma~\ref{lem:globaltolocalmonoidal} the phased coproducts in $\cta_\quotP$ are precisely the coproducts in $\cta$, making $F$ well-defined. Now every $[g]_\tp \colon F(A) \to F(B)$ in $\plusI{\cta_\quotP}$ has $g = h + u$ for a unique $h \colon A \to B$ and $(u \colon I \to I) \in \mathbb{P}$. Then $[g]_\tp = F(f)$ iff 
\[
(f + \id{I}) = v \cdot (h + u) = (v \cdot h + v \cdot u)
\]
for some $v \in \mathbb{P}$. So $[g]_\tp = F(f)$ for the unique morphism $f = u^{-1} \cdot h$, making $F$ full and faithful. It is essentially surjective on objects by Lemma~\ref{lem:isoms}, and distributivity in $\cta$ ensures that $F$ is strong monoidal. Clearly $F$ also restricts to an isomorphism of global phase groups. 
\end{proof}

\begin{examples} \label{ex:ofconstruction}
We've seen that $\VecC$, $\Hilb$ and $\FVeck$ satisfy the above properties of $\cta$ and so they may be reconstructed from their quotients as
\[
\VecC \simeq \plusI{\VecP} 
\qquad
\Hilb \simeq \plusI{\HilbP}
\qquad
\FVeck \simeq \plusI{\VecProj}
\]
\end{examples}

\section{Phased Biproducts} \label{sec:phbiprod}

The phased coproducts in $\HilbP$ come with extra properties which we capture as follows. As in Remark~\ref{rem:phasedproducts} we define a \indef{phased product} \index{phased product} to be an object $A \pprod B$ with projections $\pproj_A \colon A \pprod B \to A$ and $\pproj_B \colon A \pprod B \to B$ satisfying the dual conditions to those of a phased coproduct.

\begin{definition} \label{def:phased_biprod} \label{not:phbiprod}
In a category with zero morphisms, a \deff{phased biproduct} \index{phased biproduct} of objects $A, B$ is an object $A \pbiprod B$ together with morphisms
\[
\scalebox{1.0}{\input{./figures/pbiprod-mod.tikz}}
\]
satisfying the equations 
\begin{align*}
\pproj_A \circ \coproj_A &= \id{A} & \pproj_A \circ \coproj_B &= 0 
\\
\pproj_B \circ \coproj_A &= 0 & \pproj_B \circ \coproj_B  &= \id{B} 
\end{align*} 
and for which $(\pcoproj_A, \pcoproj_B)$ and $(\pproj_A, \pproj_B)$ make $A \pbiprod B$ a phased coproduct and product, respectively, such that each have the same phases $U \colon A \pbiprod B \to A \pbiprod B$.
\end{definition}

We may straightforwardly define a phased biproduct $A_1 \pbiprod \dots \pbiprod A_n$ of any finite collection of objects similarly, with an empty phased biproduct being simply a zero object $0$. A biproduct is then a phased biproduct whose only phase is the identity. 

\begin{lemma} \label{lem:ph_biprod}
Let $\ctb$ be a category with a zero object and binary phased biproducts. 
\begin{enumerate}
\item \label{enum:hasfinite}
$\ctb$ has finite phased biproducts.
\item \label{enum:unique}
Any phased coproduct $A_1 \pcoprod \dots \pcoprod A_n$ has a unique phased biproduct structure.
\item \label{enum:transitive}
All phases are transitive.
\end{enumerate}
\end{lemma}

\begin{proof}

\ref{enum:hasfinite}
We will show that any object $(A \pbiprod B) \pbiprod C$ forms a phased biproduct of $A$, $B$ and $C$, with the general case of $((A_1 \pbiprod A_2) \pbiprod \cdots ) A_n$ being similar. By Proposition~\ref{prop:assoc} and its dual any such object forms a phased coproduct and product with coprojections $\pcoproj_{A \biprod B} \circ \pcoproj_A$,  
$\pcoproj_{A \biprod B} \circ \pcoproj_B$, and $\pcoproj_C$, and 
projections $\pproj_A \circ \pproj_{A \pbiprod B}$, $\pproj_B \circ \pproj_{A \pbiprod B}$ and $\pproj_C$. It's routine to check that these satisfy the necessary equations.

It remains to check that any endomorphism $U$ of $(A \pbiprod B) \pbiprod C$ preserving these coprojections then preserves the projections, with the converse statement then being dual. In this case we have 
\[
\pproj_{A \pbiprod B} \circ U \circ \pcoproj_{C} = 0 \qquad U \circ \pcoproj_{A \pbiprod B} = \pcoproj_{A \pbiprod B} \circ V\] 
for some phase $V$ on $A \pbiprod B$. But then $\pi_{A \pbiprod B} \circ U$ and $V \circ \pproj_{A \pbiprod B}$ have equal composites with $\pcoproj_{A \pbiprod B}$ and $\pcoproj_C$ and so
\[
 \pi_{A \pbiprod B} \circ U = V \circ \pproj_{A \pbiprod B} \circ W 
 \] 
 for some phase $W$. But then $\pi_{A \pbiprod B} \circ U = V \circ \pproj_{A \pbiprod B}$, ensuring that $U$ preserves the above projections. 

\ref{enum:unique}
We show the result for binary phased coproducts $A \pcoprod B$, with the $n$-ary case being similar. By Lemma~\ref{lem:isoms} any coprojection preserving morphism
\[
\begin{tikzcd}
A \pcoprod B \rar{g} & A \pbiprod B
\end{tikzcd}
\]
  is an isomorphism, and one may then check that $\pproj_A \circ g$ and $\pproj_B \circ g$ form projections making $A \pcoprod B$ a phased biproduct.

 For uniqueness note that for any phased biproduct, any $p_A \colon A \pbiprod B \to A$ with $p_A \circ \coproj_A = \id{A}$ and $p_A \circ \pcoproj_B = 0$ has $p_A = \pi_A \circ U$ for some phase $U$. But $\pi_A \circ U = \pi_A$ and so $p_A = \pi_A$ is unique.

 \ref{enum:transitive}
For any diagonal morphism 
 \[
\begin{tikzcd}
 A \pcoprod B \rar{f} &  C \pcoprod D
\end{tikzcd}
 \]
with $f \circ \pcoproj_A = \pcoproj_C \circ g$, by composing with the coprojections, we see that the unique projections $\pproj_C$ and $\pproj_A$ have $\pproj_C \circ f = g \circ \pproj_A$. Then for any phase $U$ we have
\[
\pproj_C \circ f \circ U 
= 
g \circ \pproj_A \circ U 
=
g \circ \pproj_A
= 
\pproj_C \circ f 
\] 
and $\pproj_D \circ f \circ U = \pproj_D \circ f$ also. Hence $f \circ U = V \circ f$ for some phase $V$ on $C \pcoprod D$. 
\end{proof}

In a category with phased biproducts, by a \indef{phase generator} \index{phase generator} let us now mean an object satisfying the properties of Definition~\ref{def:phase-gen} along with the dual statements about phased products. 

\begin{lemma} \label{lem:getBiprod}
Let $\ctb$ be a category with finite phased biproducts with a phase generator $I$. Then $\plusI{\ctb}$ has finite biproducts. Conversely, if $\cta$ is a category with finite biproducts and a transitive choice of trivial isomorphisms then $\quot{\cta}{\sim}$ has finite phased biproducts. 
\end{lemma}
\begin{proof}
Since $\ctb$ has phased biproducts, any phased coproduct $\obb{A} = A \pcoprod I$ has a unique phased biproduct structure $\pproj_A \colon \obb{A} \to A, \pproj_I \colon \obb{A} \to I$ in $\ctb$, and so we may equivalently view the objects of $\plusI{\ctb}$ as such phased biproducts.  Then $\plusI{\ctb}$ has zero morphisms
\[
0_{\obb{A}, \obb{B}} :=
\begin{tikzcd}
\obb{A} \rar{\pproj_I} & I \rar{\pcoproj_I} & \obb{B}
\end{tikzcd}
\]
and in particular the initial object $\obb{0} = 0 \pcoprod I$ has $\id{\obb{0}} = 0$ and so is a zero object.

Now by Theorem~\ref{thm:localToGlobal} for any objects $\obb{A}, \obb{B} \in \plusI{\ctb}$, any object and morphisms 
\[
\begin{tikzcd}[row sep = large]
\obb{A} 
\rar[shift left = 2.5]{\coproj_{\obb{A}}}  
& 
\obb{A \biprod B}  
\lar[shift left = 2.5]{\pi_{\obb{A}}} 
\rar[shift right = 2.5,swap]{\pi_\obb{B}} 
& \obb{B} 
\lar[shift right = 2.5,swap]{\coproj_\obb{B}}  
\end{tikzcd}
\]
which are sent by $[-]$ to a phased biproduct structure on $A, B$ in $\ctb$ have that $\coproj_{\obb{A}}$ and $\coproj_{\obb{B}}$ form a coproduct of $\obb{A}, \obb{B}$ in $\plusI{\ctb}$, and dually $\pproj_{\obb{A}}$ and $\pproj_{\obb{B}}$ form a product. 
Then since $[-]$ reflects zeroes and $[\pi_{\obb{B}} \circ \coproj_{\obb{A}}] = 0$ we have $\pi_{\obb{B}} \circ \coproj_{\obb{A}} = 0_{\obb{A}, \obb{B}}$, and $\pi_{\obb{A}} \circ \coproj_{\obb{B}} = 0_{\obb{B}, \obb{A}}$ similarly. By applying $[-]$ we also see that $\pi_{\obb{A}} \circ \coproj_{\obb{A}} = U$ and $\pproj_{\obb{B}} \circ \coproj_{\obb{B}} = V$ for some phases $U$ on $\obb{A}$ and $V$ on $\obb{B}$. Then finally $\coproj_{\obb{A}}$, $\coproj_{\obb{B}}$, $U^{-1} \circ \pi_\obb{A}$ and $V^{-1} \circ \pi_{\obb{B}}$ make $\obb{A \biprod B}$ a biproduct in $\plusI{\ctb}$.

For the converse statement, we know that biproducts in $\cta$ form distributive phased coproducts in $\quot{\cta}{\sim}$, and dually they form phased products also. The zero arrows in $\cta$ form zero arrows in $\quot{\cta}{\sim}$ with $[f]_\tc = 0 \implies f = 0$. Hence $[-]_\tc$ preserves the phased biproduct equations. Now, endomorphisms on $A \pbiprod B$ in $\quot{\cta}{\sim}$ preserving the coprojections are (equivalence classes) of endomorphisms $U$ of $A \biprod B$ in $\cta$ of the form $U = s + t$ for some $s \in \mathbb{T}_A$ and $t \in \mathbb{T}_B$. But equivalently $U = s \times t$ and so they preserve the projections in $\quot{\cta}{\sim}$.
\end{proof}

In a monoidal category we say that phased biproducts are \indef{distributive} \index{distributive!phased biproducts} when they are distributive as phased coproducts. 

\begin{corollary} \label{cor:biproducts}
The assignments $\cta \mapsto \cta_\quotP$ and $\ctb \mapsto \plusI{\ctb}$ give a one-to-one correspondence, up to monoidal equivalence, between monoidal categories 
\begin{itemize}
\item $\cta$ with finite distributive biproducts and a choice of global phases $\mathbb{P}$;
\item $\ctb$ with finite distributive phased biproducts;
\end{itemize}
\end{corollary}
\begin{proof}
For such a category $\ctb$, $I$ is a phase generator for phased coproducts by Lemma~\ref{lem:deleter_cancellation_modified} and hence also one for phased biproducts dually. Hence by Lemma~\ref{lem:getBiprod} $\plusI{\ctb}$ has finite biproducts. The assignment is then well-defined by Lemma~\ref{lem:getBiprod} and Corollary~\ref{cor:phcoprodcorrespon}.
\end{proof}

\begin{examples}
Since $\VecC$, $\Hilb$ and $\FVeck$ all have biproducts these become phased biproducts in $\VecP$, $\HilbP$ and $\VecProj$. 
\end{examples}

\section{Phases in Compact Categories} \label{sec:compact_cats}

In the setting of a compact category, such as our examples $\FVecP$ and $\FHilbP$, phased coproducts get several nice properties for free.

\begin{lemma} \label{lem:comp-closed}
Let $\ctb$ be a compact category with finite phased coproducts.  
\begin{enumerate}[label=\arabic*., ref=\arabic*]
\item \label{enum:pinit-pterm}
Any initial object in $\ctb$ is a zero object.
\item \label{enum:distmonic}
Phased coproducts are distributive and monic in $\ctb$.
\item \label{enum:phccompclosed}
$\plusI{\ctb}$ is compact closed.
\item \label{enum:pharescalars} 
Every phase $U$ on $\obb{A}$ is of the form $U = u \cdot \id{\obb{A}}$ in $\plusI{\ctb}$, for some global phase $u$.
\end{enumerate}
\end{lemma}
\begin{proof}
\ref{enum:pinit-pterm}. This is well-known; since $\ctb$ is self-dual it has a terminal object $1$, but since $0 \otimes (-)$ preserves products $1 \simeq 0 \otimes 1$, and also $0 \otimes 1 \simeq 0$ dually.

\ref{enum:distmonic}. The presence of zero arrows makes all coprojections split monic. Now for any phased coproduct $B \pcoprod C$, one may use the bijection on morphisms
\[
\scalebox{0.8}{\input{./figures/sup17i.tikz}}
\leftrightarrow
\
\scalebox{0.8}{\input{./figures/sup17ii.tikz}}
\]
to see that $A \otimes (B \pcoprod C)$ forms a phased coproduct of $A \otimes B$ and $A \otimes C$ with every phase of the form $\id{A} \otimes U$ for a phase $U$ of $B \pcoprod C$, as required. 

\ref{enum:phccompclosed}. By Theorem~\ref{thm:constr_is_monoidal} $\plusI{\ctb}$ is now a monoidal category and the functor $[-]$ is strict monoidal. 
Let $\obb{A} = A \pcoprod I$ be an object of $\plusI{\ctb}$, and $A^*$ be dual to $A$ in $\ctb$ via the state $\tinycup$ and effect $\tinycap$. For any object $\obb{A^*} = A^* \pcoprod I$ and morphisms $\mathbf{\eta}$, $\mathbf{\epsilon}$ in $\plusI{\ctb}$ with $[\mathbf{\eta}] = \tinycup$ and $[\mathbf{\epsilon}] = \tinycap$ we have 
\[
\left[
\scalebox{0.8}{\input{./figures/sup18.tikz}}
\right]
=
\scalebox{0.8}{\input{./figures/sup19.tikz}}
=
\scalebox{0.8}{\input{./figures/sup20.tikz}}
\]
where the diagram inside $[-]$ is in $\plusI{\ctb}$. Similarly the other snake equation also holds. Then in $\plusI{\ctb}$ we have 
\[
\scalebox{0.8}{\input{./figures/sup21.tikz}}
\]
for some phases $U,V$ in $\ctb$. Since $U$ and $V$ are invertible, setting  
\[
\scalebox{0.8}{\input{./figures/sup22.tikz}}
\]
one may check that $\eta$ and $\epsilon'$ form a dual pair in $\plusI{\ctb}$. 

\ref{enum:pharescalars}. Let $U \colon \obb{A} \to \obb{A}$ be a phase in $\ctb$. In $\plusI{\ctb}$ we have 
\[
\left[
\scalebox{0.8}{\input{./figures/sup23ia.tikz}}
\right]
=
\left[
\scalebox{0.8}{\input{./figures/sup23ib.tikz}}
\right]
\qquad
\text{ so that }
\qquad
\scalebox{0.8}{\input{./figures/sup23ii.tikz}}
\]
for some global phase $u$. But then $U = \id{\obb{A}} \cdot u$ since 
\[
\scalebox{0.8}{\input{./figures/sup24.tikz}}
\]
\end{proof}

\begin{corollary}
Let $\ctb$ be a compact closed category with finite phased coproducts with transitive phases. Then $\ctb$ has finite phased biproducts.
\end{corollary}

\begin{proof}
By Theorem~\ref{thm:getmoncoprod} and Lemma~\ref{lem:comp-closed}, $\plusI{\ctb}$ is compact closed with distributive coproducts. But any compact closed category with finite coproducts has biproducts~\cite{houston2008finite}. Hence so does $\ctb \simeq \plusI{\ctb}_\quotP$ by Corollary~\ref{cor:biproducts}. 
\end{proof}

We leave open the question of whether compact closure automatically ensures that phases are transitive.

\section{Phases in Dagger Categories}

Our motivating examples $\Hilb$ and $\HilbP$ come with the extra structure of a dagger (see  Section~\ref{subsec:daggercats}). The dagger in $\Hilb$ usefully allows us to identify global phases intrinsically, as those scalars $z \in \mathbb{C}$ with $z^\dagger \cdot z = 1$. The resulting phased biproducts in $\HilbP$ interact with the dagger as follows.

\begin{definition} \label{def:daggerphbiprod}
In any dagger category with zero morphisms, a \deff{phased dagger biproduct} \index{phased dagger biproduct} is a phased biproduct $A_1 \pbiprod \dots \pbiprod A_n$ with $\pproj_i = \coproj_i^{\dagger}$ for all $i$.
\end{definition}

A dagger biproduct (see Section \ref{subsec:daggercats}) is then simply a phased dagger biproduct whose only phase is the identity. More general ones are equivalently captured as follows. In a dagger category, we call morphisms $f \colon A \to B$ and $g \colon C \to B$ \indef{orthogonal} \index{orthogonal} when $g^\dagger \circ f = 0$~\cite{heunen2010quantum}. 

\begin{lemma} \label{lem:ph-d-biprod-help}
In any dagger category with zero morphisms, a phased dagger biproduct $A \pbiprod B$ is equivalently a phased coproduct for which:
\begin{itemize}
\item 
$\pcoproj_A$ and $\pcoproj_B$ are orthogonal isometries;
\item 
whenever $U$ is a phase so is $U^{\dagger}$.
\end{itemize}
\end{lemma}
\begin{proof}
The dagger sends phased coproducts to phased products and vice versa. 
The first point is a restatement of the equations of a biproduct, while the second is equivalent to the projections and coprojections then having the same phases.
\end{proof}

\begin{lemma} 
A dagger category has finite phased dagger biproducts iff it has a zero object and binary phased dagger biproducts.
\end{lemma}
\begin{proof}
We have seen that $(A \pbiprod B) \pbiprod C$ forms a phased biproduct of $A, B, C$ with coprojections $\pcoproj_{A \pbiprod B} \circ \pcoproj_A$, $\pcoproj_{A \pbiprod B} \circ \pcoproj_B$ and $\pcoproj_C$. But these are isometries whenever all of the $\pcoproj$ are. 
Similarly, we obtain phased dagger biproducts $A_1 \pbiprod \dots \pbiprod A_n$.
\end{proof}

Our motivating source of examples is the following.

\begin{lemma} \label{lem:dagger_ph_biprod_main}
Let $\cta$ be a dagger category with dagger biproducts and a choice of trivial isomorphisms which is transitive and closed under the dagger. Then $\quot{\cta}{\sim}$ is a dagger category with finite phased dagger biproducts.
\end{lemma}
\begin{proof}
This follows easily from Lemma~\ref{lem:getBiprod}, noting that thanks to our assumptions whenever $f \sim g$ then $f^\dagger \sim g^\dagger$ also, so that $\quot{\cta}{\sim}$ is indeed a dagger category.
\end{proof}

\begin{example}
$\Hilb$ has finite dagger biproducts, and so by the above $\HilbP$ has finite phased dagger biproducts. 
\end{example}

We now desire versions of our results on the $\mathsf{GP}$ construction for dagger categories. However, a problem arises from the fact that the canonical (non-unique) isomorphisms from Lemma~\ref{lem:isoms} or distributivity~\eqref{eq:distrib-isomorphism} need not be unitary as canonical isomorphisms in a dagger category should be.

\begin{example} \label{ex:wocoherentdagphases}
For each commutative involutive semi-ring $S$ we've seen that $\MatS$ has distributive dagger biproducts $n \biprod m := n + m$. Choose as global phases $\mathbb{P}$ all scalars $u \in S$ which are unitary, with $u^{\dagger} \cdot u = 1$, and suppose that $S$ has a unitary element of the form $s^{\dagger} \cdot s \neq 1$ for some $s \in S$; for example we may take $S = \mathbb{C}$ but with trivial involution $z^{\dagger} := z$ for all $z \in \mathbb{C}$, and choose $s = -1 = i^{\dagger} \cdot i$.

 Then the morphism $(1,0) \colon 1 \to 2$ in $\MatS$, together with either $(0,1)$ or $(0,s)$ makes the object $2$ a phased dagger biproduct $1 \pbiprod 1$ in $(\MatS)_\quotP$. But the endomorphism of $2$ in $\MatS$ with matrix 
\[
\begin{pmatrix}
1 & 0 \\ 0 & s
\end{pmatrix}
\]
which relates these is not unitary, and nor is its induced morphism in $(\MatS)_\quotP$.
\end{example}

We can remedy this with an extra assumption about phased dagger biproducts. In a dagger category a morphism $f \colon A \to A$ is called \indef{positive} \index{positive morphism} when $f = g^{\dagger} \circ g$ for some $g \colon A \to B$.

\begin{definition} \label{def:posfree}
We say that a dagger category with finite phased dagger biproducts has \deff{positive-free phases} \index{positive-free!phases} when any phase $U$ on $A \pbiprod B$ which is positive has $U= \id{A \pbiprod B}$. 
\end{definition}

Equivalently, any morphism $f \colon A \pbiprod B \to C$ for which $f \circ \pcoproj_A$ and $f \circ \pcoproj_B$ are orthogonal isometries is itself an isometry. In particular this makes all phases and canonical and distributivity isomorphisms between (finite) phased dagger biproducts unitary. It also follows that positive phases of any finite phased dagger biproduct $A_1 \pbiprod \dots \pbiprod A_n$ are trivial. 

\begin{definition}
Let $\ctb$ be a dagger category with phased dagger biproducts and a distinguished object $I$. We define the category $\plusIdag{\ctb}$ \label{not:GPdag} just like $\plusI{\ctb}$ but with objects being phased dagger biproducts $\obb{A} = A \pbiprod I$. 
\end{definition}

\begin{lemma} \label{lem:DagBiprodPlusI}
Let $\ctb$ be a category with finite phased dagger biproducts with positive-free phases and a phase generator $I$. Then $\plusIdag{\ctb}$ is a dagger category with finite dagger biproducts, and $[-] \colon \plusIdag{\ctb} \to \ctb$ preserves daggers. 
\end{lemma}
\begin{proof}
Let $\ctb$ be as above. One may check that any diagonal morphism $f \colon \obb{A} \to \obb{B}$ between phased dagger biproducts with $f \circ \pcoproj_i = \pcoproj_i \circ f_i$ has that $f^{\dagger} \colon \obb{B} \to \obb{A}$ is also diagonal with $f^{\dagger} \circ \pcoproj_i = \pcoproj_i \circ {f_i}^{\dagger}$. Hence $\plusIdag{\ctb}$ is a dagger category with the same dagger as $\ctb$, and $[-] \colon \ctb \to \plusIdag{\ctb}$ preserves daggers. 

Now any lifting $(\obb{A \pbiprod B}, \pcoproj_{\obb{A}}, \pcoproj_{\obb{B}})$ of a phased dagger biproduct in $\ctb$ is a biproduct in $\plusIdag{\ctb}$, just as in Lemma~\ref{lem:getBiprod}. Moreover each coprojection has that $[\pcoproj_{\obb{A}}^{\dagger} \circ \pcoproj_{\obb{A}}] = [\pcoproj_{\obb{A}}]^{\dagger} \circ [\pcoproj_{\obb{A}}] = \id{}$ and so  
$\pcoproj_{\obb{A}}^{\dagger} \circ \pcoproj_{\obb{A}}$ is a phase in $\ctb$, and hence by positive-freeness is the identity, making this a dagger biproduct.
\end{proof}

When $\ctb$ is a dagger monoidal category, in $\plusIdag{\ctb}$ we again set $\mathbb{P}$ to be the morphisms $\obb{I} \to \obb{I}$ in $\plusIdag{\ctb}$ which are phases in $\ctb$. We call a choice of global phases $\mathbb{P}$ on a dagger monoidal category \indef{positive-free} \index{positive-free!global phases} if whenever $p \cdot \id{A}$ is positive then it is equal to $\id{A}$, for any $p \in \mathbb{P}$ and object $A$.

\begin{corollary} \label{cor:daggerbiproducts}
There is a one-to-one correspondence, up to dagger monoidal equivalence, between dagger monoidal categories
\begin{itemize}
\item $\cta$ with distributive finite dagger biproducts and a positive-free choice of unitary global phases $\mathbb{P}$;
\item $\ctb$ with distributive finite phased dagger biproducts with positive-free phases;
\end{itemize}
given by $\cta \mapsto \cta_{\quotP}$ and  $\ctb \mapsto \plusIdag{\ctb}$.
\end{corollary}
\begin{proof}
$\cta_\quotP$ has phased dagger biproducts by Lemma~\ref{lem:dagger_ph_biprod_main}, and from the description of phases in this category we see that they are positive-free iff $\mathbb{P}$ is positive-free in $\cta$. Conversely, for $\ctb$ as above apply Lemma~\ref{lem:DagBiprodPlusI} and Corollary~\ref{cor:biproducts}. Thanks to positive-freeness, every phase is a unitary and hence so are all elements of $\mathbb{P}$.

We define the monoidal structure on $\plusIdag{\ctb}$ just as on $\plusI{\ctb}$. By positive-freeness the morphisms $c_{\obb{A}, \obb{B}}$ are isometries, and this in turn ensures that $\plusIdag{\ctb}$ is dagger monoidal. To show this, we will use the observation that in any dagger category, if $i$ and $j$ are isometries and  the following commutes
\[
\begin{tikzcd}
A \rar{f} \dar[swap]{i} & B \dar{j} \rar{h} & A \dar{i} \\ 
C \rar[swap]{g} & D \rar[swap]{g^\dagger} & C  
\end{tikzcd}
\]
then $h = f^{\dagger}$, and whenever $g$ is unitary so is $f$. Applying this to the situation 
\[
\begin{tikzcd}
\obb{A} \tens \obb{B} \rar{f \tens g} \dar[swap]{c_{\obb{A}, \obb{B}}} & 
\obb{C} \tens \obb{D} \rar{f^{\dagger} \tens g^{\dagger}} \dar[swap]{c_{\obb{C}, \obb{D}}} & \obb{A} \tens \obb{B} \dar{c_{\obb{A}, \obb{B}}} \\ 
\obb{A} \otimes \obb{B} \rar[swap]{f \otimes g} & \obb{C} \otimes \obb{D} \rar[swap]{f^{\dagger} \otimes g^{\dagger}} & \obb{A} \otimes \obb{B} 
\end{tikzcd}
\]
using that $ f^{\dagger} \otimes g^{\dagger} = (f \otimes g)^{\dagger}$ we see that $ f^{\dagger} \tens g^{\dagger} = (f \tens g)^{\dagger}$ also. Similarly, applying this observation to the definition of $\aalpha$ shows that it is unitary. 

Now any morphism $\beta \colon \obb{I} \tens \obb{I} \to \obb{I}$ as in the proof of Theorem~\ref{thm:constr_is_monoidal} is unitary thanks to positive-freeness. 
The natural isomorphisms $\rrho$ in $\plusIdag{\ctb}$ satisfy $\rrho_{\obb{A}} \tens \id{\obb{I}} = (\id{\obb{A}} \tens \beta) \circ \aalpha_{\obb{A}, \obb{I}, \obb{I}}$. Since the latter is unitary, the dagger respects $\tens$, and the assignment $f \mapsto f \tens \id{\obb{I}}$ is injective, it follows that $\rrho_{\obb{A}}$ is unitary. Similarly, so is $\llambda_{\obb{A}}$.

Now since $[-]$ is dagger monoidal so is the equivalence $\ctb \simeq \plusIdag{\ctb}_\quotP$. Conversely, the equivalence $F \colon \cta \to \plusIdag{\cta_\quotP}$ preserves daggers by definition and is such that every object in $\plusIdag{\cta_\quotP}$ is unitarily isomorphic to $F(A)$, for some $A$, making it a dagger equivalence.
\end{proof}

It is also easy to see that whenever either of $\cta$ or $\ctb$ is symmetric dagger monoidal, so is the other and each of the above functors.

\begin{example} \label{example:Hilb-dag-construction}
The global phases $e^{i \theta}$ in $\Hilb$ are positive-free, and so we have dagger monoidal equivalences 
\[
\Hilb \simeq \plusIdag{\HilbP}
\qquad
\FHilb \simeq \plusIdag{\FHilbP}
\]
\end{example}

It follows from our next result that the phased biproducts in $\HilbP$ satisfy the following condition, strengthening positive-freeness, which will be useful to us later. Let us say that phased dagger biproducts have \indef{positive cancellation} \index{positive cancellation} when any positive diagonal endomorphisms $p, q$ of $A \pbiprod B$ with $p = q \circ U$ for some phase $U$ have $p = q$. 

\begin{lemma} \label{lem:pos-cancellation}
Let $\ctb$ be a dagger monoidal category with distributive finite phased dagger biproducts with positive-free phases. Then  positive cancellation holds in $\ctb$ iff in $\plusIdag{\ctb}$ we have
\begin{equation} \label{eq:pos-condition}
[p] = [q] \implies p = q
\end{equation}
for all positive morphisms $p, q$.
\end{lemma}

\begin{proof}
Let $p, q$ be positive in $\plusIdag{\ctb}$ with $[p] = [q]$. Then $p = q \circ U$ for some phase $U$, and so when positive cancellation holds we have $p = q$.

 Conversely, suppose $\plusIdag{\ctb}$ satisfies the above and that $p, q$ are positive diagonal endomorphisms of $A \pbiprod B$ in $\plusIdag{\ctb}$ with $[p] = [q] \circ [U]$ for some phase $[U]$ in $\ctb$. Then in $\plusIdag{\ctb}$ we have $[\pproj_A \circ p \circ \pcoproj_A] = [\pproj_A \circ q \circ \pcoproj_A]$ and so $\pproj_A \circ p \circ \pcoproj_A = \pproj_A \circ q \circ \pcoproj_A$, and similarly for $B$, giving $p = q$. Hence $\ctb$ has positive cancellation.
\end{proof}

\begin{example}
$\Hilb$ satisfies the condition \eqref{eq:pos-condition}. Indeed let $p, q$ be positive linear maps with $p = e^{i \cdot \theta} \circ q$. Then since $p = p^\dagger$, subtracting $p$ gives that either $p=q=0$ or $e^{i \cdot \theta} = \pm 1$. But any positive maps with $p + q = 0$ have $p = q = 0$ also. 
\end{example}

\subsection{Phases in dagger compact categories}

Let us now consider the case where $\ctb$ is dagger compact. Although we've seen that compactness of $\ctb$ ensures that $\plusIdag{\ctb}$ is compact, to establish dagger compactness we make an extra assumption; it is an open question whether this is necessary. In any dagger monoidal category, let us call a state $\psi \colon I \to A$ a \indef{local isometry} \index{local isometry} when 
\begin{equation} \label{eq:local-isometric}
\scalebox{0.8}{\input{./figures/sup25.tikz}}
\end{equation}
For example, any isometric state is a local isometry, as is the state of a zero object. 

\begin{proposition} \label{prop:constr-compact}
Let $\ctb$ be dagger compact with phased dagger biproducts with positive-free phases. Suppose that in $\ctb$ every object $A$ has a state $\psi$ which is a local isometry. Then $\plusIdag{\ctb}$ is dagger compact.
\end{proposition}

\begin{proof}
In $\ctb$, let $A$ and $A^*$ be dagger dual objects via the state $\tinycup$.  Let $\psi \colon I \to A$ be as above, and let $\phi \colon \obb{I} \to \obb{A}$
and $\eta \colon \obb{I} \to \obb{A^*} \tens \obb{A}$ in $\plusIdag{\ctb}$ with $[\phi] = \psi$ and $[\eta] = \tinycup$. Then applying $[-]$ we see that in $\plusIdag{\ctb}$ we have 
\[
\scalebox{0.8}{\input{./figures/sup27.tikz}}
\]
for some $u \in \mathbb{P}$ and 
\[
\left[
\scalebox{0.8}{\input{./figures/sup26ia.tikz}}
\right]
=
\left[
\scalebox{0.8}{\input{./figures/sup26ib.tikz}}
\right]
\quad 
\text{    so that     }
\quad
\scalebox{0.8}{\input{./figures/sup26ii.tikz}}
\]
by positive-freeness.
But then 
\[
\scalebox{0.8}{\input{./figures/sup28.tikz}}
\]
Then by positive-freeness in $\plusIdag{\ctb}$ we have $\id{A} \cdot u = \id{A}$ , so that $\eta$ satisfies the first equation of a dagger dual. The second equation is shown identically. 
\end{proof}

\begin{example}
$\FHilbP$ satisfies the above conditions, and indeed $\FHilb$ is dagger compact also. 
\end{example}

\subsection{Phases and kernels}

In Chapter~\ref{chap:principles} we met another major feature of $\Hilb$ and $\HilbP$, the existence of dagger kernels. In the presence of these to have phased dagger biproducts it  suffices to have those of a special form, a fact that will be useful to us later.

\begin{lemma} \label{lem:dag-ker-coproj}
Let $\ctb$ be a dagger category with a phased dagger biproduct $A \pbiprod B$. Then $\pcoproj_A$ and $\pcoproj_B$ are dagger kernels with $\pcoproj_A = {\pcoproj_B}^{\bot}$.
\end{lemma}
\begin{proof}
By definition both coprojections are isometries. Let us show that $\pcoproj_A = \ker(\pproj_B)$ for $\pproj_B = \coproj_B^\dagger$. Suppose that $f \colon C \to A \pbiprod B$ has $\pproj_B \circ f = 0$, and let $g = \coproj_A \circ \pproj_A \circ f$. Then $\pproj_A \circ g = \pproj_A \circ f$ and $\pproj_B \circ f = 0 = \pproj_B \circ f$. Hence for some phase $U$ we have 
\[
f 
= 
U \circ g 
=
U \circ \coproj_A \circ \pproj_A \circ f
=
 \coproj_A \circ \pproj_A \circ f
\]
and so $f$ factors over $\coproj_A$, as required. 
\end{proof}

\begin{proposition} \label{prop:kernelsGiveBiproducts}
Let $\ctb$ be a dagger category with dagger kernels and phased dagger biproducts $A \pbiprod A$ for all objects $A$. Then $\ctb$ has finite phased dagger biproducts iff for every pair of objects $A, B$ there is an object $C$ and orthogonal kernels $k \colon A \to C$ and $l \colon B \to C$. 
\end{proposition}

\begin{proof}
The condition is necessary by Lemma~\ref{lem:dag-ker-coproj}. Conversely, let $k \colon A \to C$ and $l \colon B \to C$ be orthogonal kernels. Let $f$ be any endomorphism of $C \pbiprod C$ with $f \circ \pcoproj_1 = \pcoproj_1 \circ k \circ k^{\dagger}$ and $f \circ \pcoproj_2 = \pcoproj_2 \circ l \circ l^{\dagger}$ and let $i := \img(f)$. Then since 
\[
\coker(f) \circ \pcoproj_1 \circ k
=
\coker(f) \circ f \circ \pcoproj_1 \circ k 
= 
0
\]
and similarly for $\pcoproj_2$ and $l$, there are unique $\coproj_A, \coproj_B$ making the following commute:
\[
\begin{tikzcd}
A \rar{k} \arrow[drr,swap, dashed, "\pcoproj_A"]& C \rar{\pcoproj_1} & C \pbiprod C 
& C \lar[swap]{\pcoproj_2} & B \lar[swap]{l} \arrow[dll, dashed,"\pcoproj_B"] 
\\
& & \Img(f) \uar[swap]{i} & & 
\end{tikzcd}
\]
We claim that $\pcoproj_A$ and $\pcoproj_B$ make $\Img(f)$ a phased biproduct $A \pbiprod B$. To see the existence property, given $g \colon A \to D$ and $h \colon B \to D$, let $j \colon C \pbiprod C \to D$ with $j \circ \pcoproj_1 = g \circ k^{\dagger}$ and $j \circ \pcoproj_2 = h \circ l^{\dagger}$. Then $k := j \circ i$ has $k \circ \pcoproj_A = g$ and $k \circ \pcoproj_B = h$. 

We now show the uniqueness property. First, it is straightforward to show that $f^{\dagger}$ has the same composites with $\pcoproj_1$ and $\pcoproj_2$ as $f$. Then $f^{\dagger} \circ \pcoproj_1 \circ k^{\bot} = 0$ and so since $i = \img(f)$ we have $i^{\dagger} \circ \pcoproj_1 \circ k^{\bot} = 0$. Then since $k = \img(k)$ we have
\begin{equation*} 
i^{\dagger} \circ \pcoproj_1 
=
i^{\dagger} \circ \pcoproj_1 \circ k \circ k^{\dagger} 
=
i^{\dagger} \circ i \circ \pcoproj_A \circ k^{\dagger} 
=
\pcoproj_A \circ k^{\dagger} 
\end{equation*}

Now suppose that $m, p \colon \Img(f) \to D$ each have the same composites with $\pcoproj_A$ and $\pcoproj_B$. Let $q = m \circ i^{\dagger}$ and $r = p \circ i^{\dagger}$. Then 
\begin{align*}
q \circ \pcoproj_1 
 = 
m \circ i^{\dagger} \circ \pcoproj_1
 =
m \circ \pcoproj_A \circ k^{\dagger}
=
p \circ \pcoproj_A \circ k^{\dagger}
=
r \circ \pcoproj_1 
\end{align*}
and similarly for $\pcoproj_2$. So there is a phase $U$ on $C \pbiprod C$ with $q = r \circ U$. Then $m = p \circ u$ where 
\begin{equation} \label{eq:phonIm}
u = i^{\dagger} \circ U \circ i
\end{equation}
One may verify from the definitions of $\pcoproj_A$ and $\pcoproj_B$ that any such endomorphism $u$ preserves them, establishing the uniqueness property. Moreover, running the above argument with $p = \id{}$ shows that any phase on $\Img(f)$ is of the form~\eqref{eq:phonIm}. In particular, since phases on $C \pbiprod C$ are closed under the dagger, so are those on $\Img(f)$.
\end{proof}

The combination of phased biproducts and kernels will provide us with an axiomatization of $\HilbP$ in the next chapter.


%% file: figures/sup20.tikz
\begin{tikzpicture}
	\begin{pgfonlayer}{nodelayer}
		\node [style=none] (0) at (-0.75, -0.75) {};
		\node [style=label] (1) at (-0.75, -1.25) {$A$};
		\node [style=none] (2) at (-0.75, 1) {};
		\node [style=label] (3) at (-0.75, 1.5) {$A$};
	\end{pgfonlayer}
	\begin{pgfonlayer}{edgelayer}
		\draw [style=none] (0.center) to (2.center);
	\end{pgfonlayer}
\end{tikzpicture}

%% file: figures/sup26ib.tikz
\begin{tikzpicture}
	\begin{pgfonlayer}{nodelayer}
		\node [style=none] (0) at (2.25, 1.25) {};
		\node [style=none] (1) at (2.25, -1.25) {};
		\node [style=label] (2) at (2.25, 1.75) {$\obb{A}$};
		\node [style=label] (3) at (2.25, -1.75) {$\obb{A}$};
	\end{pgfonlayer}
	\begin{pgfonlayer}{edgelayer}
		\draw (1.center) to (0.center);
	\end{pgfonlayer}
\end{tikzpicture}

%% file: chapter6.tex
\chapter{Reconstructing Quantum Theories} \label{chap:recons}

Quantum theory itself has long been the main motivation for the study of operational theories of physics. Over the years the far from clear physical interpretation of the Hilbert space formalism has led numerous physicists to an instrumentalist reading of the theory, and also to ask whether it could instead be derived from more operational principles. A major goal has thus (implicitly) been to answer this question: what conditions ensure that a given category is equivalent to $\Quant{}$?

Following the work of Hardy~\cite{Hardy2001QTFrom5}, the first fully operational reconstruction of finite-dimensional quantum theory was provided by by Chiribella, D'Ariano and Perinotti (CDP)~\cite{PhysRevA.84.012311InfoDerivQT,d2017quantum}, and since then many more have been presented (see~\cite{Hardy2011a,wilceRoyal,selby2018reconstructing} and Refs.~in the introduction). However, all of these results rely on the technical assumptions typical to probabilistic theories; that scalars are given by probabilities, and that finite tomography holds, making the vector space generated by each collection of processes finite-dimensional. 

The approach of this thesis makes it natural to ask whether a reconstruction of a purely process-theoretic nature, without these assumptions, might instead be possible. Indeed in 2011 Coecke and Lal stated the need for a reconstruction in the categorical framework, and suggested drawing on the CDP reconstruction~\cite{coecke2011categorical}. Now in Chapter~\ref{chap:principles} we already saw how their principles could be treated in a basic categorical setting, via the (approximate) correspondence:
\[
\begin{tabular}{|c|c|}
\hline
\textbf{CDP Axioms}             & \textbf{Categorical Features} \\ 
\hline
Causality                  & Discarding $\discard{}$                                \\
\hline
Atomicity of composition   & \multirow{2}{*}{Environment structure}                 \\ \cline{1-1}
Purification               &                                                        \\ \hline
Perfect distinguishability &  Kernels + \\ \cline{1-1}
Ideal compressions         &  pure exclusion                                                      \\ \hline
\multicolumn{2}{|c|}{Essential uniqueness}                                  \\ \hline
\end{tabular}
\]
Motivated by these relations, in this chapter we provide such a categorical reconstruction of quantum theory. 

We show that any dagger compact category with discarding $(\catC, \discard{})$ with suitable forms of the above features, along with a mild scalar condition, is equivalent to one of the form $\Quant{S}$ for a suitable ring $S$, generalising the case of $\Quant{}$ where $S = \mathbb{C}$. A further scalar condition makes $S$ resemble either $\mathbb{R}$ or $\mathbb{C}$, so that specialising to probabilistic theories we immediately obtain either $\Quant{}$ or more unusually the quantum theory $\Quant{\mathbb{R}}$ over \emps{real Hilbert spaces}~\cite{stueckelberg1960quantum}. Following this, the results of Chapter~\ref{chap:principles} allow us to deduce several further reconstruction theorems.

Beyond the above principles, our result is in fact based on a very general approach to reconstructions, drawing on our treatment of superpositions in Chapter~\ref{chap:superpositions}, which we describe first.

\paragraph{Setup}
Throughout this chapter, by a \indef{(compact) dagger theory} $(\catC, \otimes, \discard{}, \dagger)$ \index{dagger theory} \index{dagger theory!compact} we will simply mean a dagger symmetric monoidal (resp.~compact) category with discarding and zero morphisms. Note that unlike Chapter~\ref{chap:principles} we no longer require the rule~\eqref{eq:pos-eq}, though we will derive it in our main examples.
By an \indef{embedding} \index{embedding!of dagger theories} or \indef{equivalence} \index{equivalence!of dagger theories} of dagger theories we mean one of dagger symmetric monoidal categories with discarding and which also preserves zero morphisms.

\section{A Recipe for Reconstructions}

Beyond quantum theory itself, our results so far in fact provide us with an approach to reconstructing a whole class of quantum-like theories. In Section~\ref{sec:CPM} we saw how to generalise $\Quant{}$ using Selinger's construction $\CPM(\catA)$ for a dagger compact category $\catA$, motivated by the example
\begin{equation} \label{eq:all-equivalences}
\Quant{} \simeq \CPM(\FHilb) \simeq \CPM(\FHilbP)
\end{equation}
Noting the equivalence $\FHilb \simeq \Mat_{\mathbb{C}}$ suggests a generalisation. 

\begin{definition} \label{def:QuantS}
For each commutative involutive semi-ring $S$ we define a dagger theory \label{cat:QuantS}
\[
\Quant{S} := \CPM(\Mat_S)
\]
Explicitly, objects in this theory are natural numbers $n$ and morphisms $n \to m$ are $S$-valued matrices of the form ${\sum^k_{i=1} M^i_* \otimes M^i}$, where each $M^i$ is an $n \times m$ matrix over $S$, and $(M^i_*)_{j,k} := (M^i_{j,k})^\dagger$.
\end{definition}

\begin{examples}
Standard quantum theory is $\Quant{} \simeq \Quant{\mathbb{C}}$. Another physically interesting example is provided by the quantum theory $\Quant{\mathbb{R}}$\label{cat:QuantR} on real Hilbert spaces~\cite{stueckelberg1960quantum,Hardy2012Holism}; for more on generalised quantum theories see~\cite{fantastic}. Computational complexity in quantum theories over general semi-rings $S$ has been studied by de Beaudrap~\cite{de2014computation}.
\end{examples}

In Section~\ref{sec:CP} we saw that dagger theories $(\catC, \discard{})$ arising from the $\CPM$ construction were precisely those coming with an environment structure $\catC_\prepure$, generalising the purifications provided by the subcategory $\FHilbP$ in $\Quant{}$, with any such theory satisfying $\catC \simeq \CPM(\catC_\prepure)$. 

Now when $\catC_\prepure$ has the features of Chapter~\ref{chap:superpositions} we can say much more. Let us say that a dagger compact category $\catB$ has the \indef{superposition properties} \index{superposition properties} when it has finite phased dagger biproducts satisfying positive cancellation, and every object $A$ has a state $\psi \colon I \to A$ which is a local isometry, satisfying~\eqref{eq:local-isometric}. 

Firstly, we obtain the following generalisation of~\eqref{eq:all-equivalences}.

\begin{lemma} \label{lem:CPMsCoincide}
Let $\catB$ be a dagger compact category with the superposition properties. Then $\plusIdag{\catB}$ is dagger compact and the functor $[-] \colon \plusIdag{\catB} \to \catB$ extends to an equivalence of dagger theories
\[
 \CPM(\plusIdag{\catB})
 \simeq
\CPM(\catB)
\]
\end{lemma}
\begin{proof}
$\plusIdag{\catB}$ is dagger compact by Proposition~\ref{prop:constr-compact}. Since $[-]$ is a wide full dagger symmetric monoidal functor and is surjective on objects up to unitary it lifts to such a functor $\CPM(\plusIdag{\catB}) \to \CPM(\catB)$. For faithfulness we require that
\[
\scalebox{0.8}{\input{./figures/CPM-mapsmallr.tikz}}
\sim
\scalebox{0.8}{\input{./figures/CPM-mapsmall2r.tikz}}
\implies
\scalebox{0.8}{\input{./figures/CPM-mapsmallr.tikz}}
=
\scalebox{0.8}{\input{./figures/CPM-mapsmall2r.tikz}}
\]
After bending wires this states precisely that for all positive morphisms $p, q \in \plusIdag{\catB}$ we have $p \sim q \implies p = q$. But this follows from positive cancellation by Lemma~\ref{lem:pos-cancellation}.
\end{proof}

This provides a general result for use in reconstructions, telling us when a given theory contains a copy of a quantum-like one. For any involutive monoid $(S,\dagger)$, as in a dagger category we call an element \indef{positive} \index{positive element} when it is of the form $s^\dagger \cdot s$ for some $s \in S$, denoting their collection by $S^\pos$\label{not:poselements}.

\begin{corollary} \label{cor:recipe}
Let $\catC$ be a dagger theory with an environment structure $\catC_{\prepure}$ which has the superposition properties. Then there is an embedding of dagger theories 
\begin{equation*} 
\Quant{S} \hookrightarrow \catC
\end{equation*}
for some commutative involutive semi-ring $S$ with $\catC_{\prepure}(I,I) \simeq S^{\pos}$ as monoids.
\end{corollary}
\begin{proof}

Since $\plusIdag{\catC_{\prepure}}$ is dagger compact, its biproducts are automatically distributive and so its scalars $S$ form a commutative involutive semi-ring, giving an embedding $\Mat_{S} \hookrightarrow \plusIdag{\catC_{\prepure}}$. Hence we obtain another embedding 
\[
\begin{tikzcd}
\CPM(\Mat_{S})
\rar[hook,swap]{}
&
\CPM(\plusIdag{\catC_{\prepure}}) 
\rar{\sim}[swap]{\ref{lem:CPMsCoincide}}
&
\CPM(\catC_{\prepure})
\rar{\sim}
&
\catC
\end{tikzcd}
\]
Finally, let $R=\catC_{\prepure}(I,I)$. By the CP axiom the map $R \to R^{\pos}$ sending $r \mapsto r^\dagger \circ r$ is a monoid isomorphism, and by Lemma~\ref{lem:pos-cancellation} so is the map $[-] \colon S^{\pos} \to R^{\pos}$.
\end{proof}

\begin{example}
Let $S$ be a commutative involutive semi-ring in which every non-zero element is invertible and for all positive elements $p$ we have $p^2 = 1 \implies p =1$. 
For example we may take $\mathbb{C}, \mathbb{R}, \Rpos$ or $\mathbb{B} = \{0,1\}$.  
Then the environment structure $\Dbl{\Mat_S}$ on $\Quant{S}$ is easily seen to have the superposition properties.
\end{example}

\section{The Operational Principles} \label{sec:principles}

To obtain a full reconstruction it remains to find further conditions making the embedding $\Quant{S} \hookrightarrow \catC$ an equivalence. As well as this it would be desirable to use principles of a more operational nature than the superposition properties. In fact we already explored several suitable such principles in Chapter~\ref{chap:principles}.

Firstly, let us call a pair of effects $d, e$ of an object $A$ \indef{total} when they satisfy
\[
\left(
\scalebox{0.8}{\input{./figures/total.tikz}} 
\ \ \text{ and } \ \ 
\scalebox{0.8}{\input{./figures/total2.tikz}}
\right)
\implies
 \scalebox{0.8}{\input{./figures/total3.tikz}}
\]
for all $f, g \colon B \to A$. For example in any operational theory in the sense of Chapter~\ref{chap:OpCategories} this will be the case whenever $d$ and $e$ form the outcomes of some binary test that we may perform on the system $A$. 

Also, recall that morphisms $f, g$ are said to be orthogonal when $f^\dagger \circ g = 0$. Let us call a pair of states $\ket{0}, \ket{1}$ \indef{orthonormal} \index{orthonormal states} when they are orthogonal isometries. 

\subsection{The principles}

We will consider dagger theories with the following properties, many of which we have met already, which we spell out in more detail after the definition.

\begin{definition}[\textbf{Operational Principles}] \label{def:op-principles-concise}
A dagger theory $(\catC, \discard{})$ satisfies the \deff{operational principles} \index{operational principles} when it is non-trivial and satisfies the following.
\begin{enumerate}[label=\arabic*., ref=\arabic*, itemsep=1pt]
\item \label{princ:strongpurif}
\deff{Strong Purification}: \index{purification!strong} The collection $\catC_\pure$ of $\otimes$-pure morphisms form an environment structure on $(\catC, \discard{})$. Moreover every non-zero object has a causal $\otimes$-pure state, and purifications are essentially unique.

\item \label{princ:pureexcl}
\deff{Pure exclusion} is satisfied. 

\item \label{princ:ker}
\deff{Kernels}: The category $\catC$ has dagger kernels, and these are \deff{causally complemented} \index{causally complemented} meaning that for all dagger kernels $k \colon K \to A$ the following pair of effects is total:
\begin{equation} \label{eq:causal-compl}
\scalebox{0.8}{\input{./figures/causal-comp-pair1.tikz}}
\quad
\text{ and }
\ 
\quad
\scalebox{0.8}{\input{./figures/causal-comp-pair2.tikz}}
\end{equation}
\item \label{princ:cond}
\deff{Conditioning}: \index{conditioning}
For every pair of orthonormal states $\ket{0}, \ket{1}$ of any object $A$ and states $\rho, \sigma$ of any $B$ there is a morphism $f \colon A \to B$ with
\[
\scalebox{0.8}{\input{./figures/conditioning1.tikz}}
\quad
\text{ and }
\qquad
\scalebox{0.8}{\input{./figures/conditioning2.tikz}}
\]
\end{enumerate}
\end{definition}

Let us go through these principles in detail.
From now on we will call any $\otimes$-pure morphism simply \emps{pure}. We already met the various aspects of strong purification in Chapter~\ref{chap:principles}. Recall that it means that every non-zero morphism $f$ in $\catC$ has a purification:
\[
\scalebox{0.8}{\input{./figures/purif-simple1a.tikz}}
\quad
\text{ where  } \ \ 
\quad
\scalebox{0.8}{\input{./figures/purif-simple2ia.tikz}}
\text{ for some causal $\rho$}
\]
and also that pure morphisms are closed under $\circ, \otimes$, $\dagger$, contain all identity morphisms, and satisfy the CP axiom~\eqref{eq:CP axiom} and essential uniqueness, which here are together equivalent to the rule 
\[
\scalebox{0.8}{\input{./figures/EU-CP-combine.tikz}}
\]
for some unitary $U$ on $B$, for all pure $f, g$. Note here that all unitaries are pure. 
Recall that pure exclusion states that any object $A$ with a pure state $\psi$ for which for all effects $e$ we have
\[
 \scalebox{0.8}{\input{./figures/cokerzeroi.tikz}} \ 0 
 \implies
  \scalebox{0.8}{\input{./figures/cokerzeroii.tikz}} \ 0 
\]
is in fact trivial, meaning that $\discard{A}$ is an isomorphism (or equivalently a unitary) or $A$ is a zero object. Non-triviality of $\catC$ means there is some object $A$ for which neither is the case. As before the existence of dagger kernels means that every morphism $f$ comes with an isometry $\ker(f)$ satisfying 
\[
\scalebox{0.8}{\input{./figures/kernel-def1.tikz}} \ 0 
\iff
(\exists ! h) \ 
\scalebox{0.8}{\input{./figures/kernel-def2.tikz}}
\]

Let us recap some immediate consequences of these principles from Chapter~\ref{chap:principles}.

\begin{proposition} \label{prop:OpTheoryProperties}
Let $\catC$ be a non-trivial compact dagger theory with dagger kernels satisfying principles~\ref{princ:strongpurif} and~\ref{princ:pureexcl}. Then the following hold.
\begin{enumerate}[label=\arabic*., ref=\arabic*, itemsep=1pt]
\item \label{enumi:zero-monn}
$\discard{} \circ f = 0 \implies f = 0$ for all morphisms $f$.
\item \label{enumi:zero-cancell}
Zero-cancellativity: $f \otimes g = 0 \implies f = 0$ or $g = 0$, for all morphisms $f, g$.
\item \label{enumi:dag-ker-pure}
Every dagger kernel in $\catC$ is pure and causal and is a kernel in $\catC_{\pure}$.
\item \label{enumi:normalisation}
Normalisation: every non-zero state is a scalar multiple of a causal one.
\item \label{enumi:scalars-are-pure}
All scalars are pure and satisfy $r=r^\dagger$.
\item \label{enumi:caus-pure-state-kernel}
Every causal pure state is a kernel. 
\item \label{enumi:pairOfPure}
Every non-trivial non-zero object has an orthonormal pair of pure states.
\item \label{enumi:specialObjectC}
There is an object with a pair of causal pure states $\ket{0}, \ket{1}$ with $\ket{0} = \ket{1}^\bot$. 
\end{enumerate}
\end{proposition}
\begin{proof}

\ref{enumi:zero-monn}. Let $g$ be a purification of $f$. Then $\discard{} \circ g = 0 = \discard{} \circ 0$. Then since $0$ is pure by definition we have $g = U \circ 0$ for some unitary $U$. Then $g = 0$ and so $f = 0$ also. 

\ref{enumi:zero-cancell}. By strong purification every object has an isometric pure state. Then use Lemma~\ref{lem:dagzero}~\ref{enum:zero-mult}.


\ref{enumi:dag-ker-pure} holds by both statements of Lemma~\ref{lem:whatswhole}~\ref{enum:ker-whole}, with causality of kernels following from the CP axiom.
\ref{enumi:normalisation} and the first part of~\ref{enumi:scalars-are-pure} are equivalent statements and hold by Proposition~\ref{prop:kerpure}~\ref{enum:kerpure2}, and every (pure) scalar has $r = r^\dagger$ by the CP axiom. \ref{enumi:caus-pure-state-kernel} holds by Lemma~\ref{lem:PureExclusion}.

\ref{enumi:pairOfPure}.
Let $A$ be any non-trivial non-zero object, and $\psi$ any causal pure state of $A$. By pure exclusion, $\Coker(\psi)$ is non-zero and so has a causal pure state $\eta$. Then $\phi = \coker(\psi)^{\dagger} \circ \eta$ is also a causal pure state of $A$ and $\psi$ and $\phi$ are orthonormal.

\ref{enumi:specialObjectC}. Let $A$ have a pair of orthonormal pure states $\psi_0, \psi_1$, as in the previous part. Since the dagger kernels on $A$ form an orthomodular lattice~\cite{heunen2010quantum}, we may define 
\[
B = \Img(\psi_0) \vee \Img(\psi_1)
\] 
and $i := \img(\psi_0) \vee \img(\psi_1) \colon B \to A$. Then $\psi_0 = i \circ \ket{0}$ and $\psi_1 = i \circ \ket{1}$ for some unique pure isometric states $\ket{0}, \ket{1}$, which are kernels by pure exclusion. Furthermore these are orthogonal and by orthomodularity we have $\ket{0} = \ket{1}^{\bot}$. 
\end{proof}

Next, let us consider the two new conditions in the operational principles. 
\begin{itemize}
\item 
Firstly, causal completeness of dagger kernels is new here. It is natural if we imagine that for each kernel $k$ one may perform a test on the system $A$ with two outcomes given by the effects~\eqref{eq:causal-compl}, which intuitively aims to determine whether a state belongs to the image of $k$ or of its complement $k^\bot$. 
 \item 
 Conditioning is also a new property, but is an extremely mild one. We may think of it as asserting the ability to form a conditioned process $f$ which prepares either state $\rho$ or $\sigma$ depending upon receiving input $\ket{0}$ or $\ket{1}$, much like the controlled tests from Chapter~\ref{chap:OpCategories}. 
\end{itemize}

 In fact, in this setting, conditioning is equivalent simply to the ability to coarse-grain processes in our earlier sense. Recall that we say that $\catC$ has addition when it has an operation $+$ making it dagger monoidally enriched in commutative monoids.

\begin{proposition} \label{prop:coarse-graining}
In the presence of the other operational principles, $\catC$ satisfies conditioning iff it has a unique addition operation. Moreover, in a theory with addition, causal complementation holds iff all dagger kernels $k \colon K \to A$ are causal and satisfy 
\begin{equation} \label{eq:causally-complemented}
\scalebox{0.8}{\input{./figures/causal-comp.tikz}}
\end{equation}
\end{proposition}

\begin{proof}
Suppose that $\catC$ has addition. Then conditioning follows automatically by setting 
\[
\scalebox{0.8}{\input{./figures/condition-dagger.tikz}}
\]
Conversely, suppose that $\catC$ satisfies the operational principles. By Proposition~\ref{prop:OpTheoryProperties}~\ref{enumi:specialObjectC} it contains an object $C$ with a pair of causal dagger kernel states $\ket{0}, \ket{1} \colon I \to C$ with $\ket{0} = \ket{1}^\bot$. Now given any $f, g \colon A \to B$, using conditioning and compactness let $h \colon A \to B \otimes C$ be any morphism with 
\begin{equation} \label{eq:cg-def1}
\scalebox{0.8}{\input{./figures/cg-1.tikz}}
\end{equation}
We then define 
\begin{equation} \label{eq:cg-def2}
\scalebox{0.8}{\input{./figures/cg-2i.tikz}}
\end{equation}
By causal completeness this is independent of our choice of $h$. Moreover it is straightforward to verify that it respects $\circ, \otimes$ and $\dagger$ and has unit $0$, and so indeed gives $\catC$ addition. For example to check that $j \circ (f + g) = j \circ f + j \circ g$ for all $j \colon B \to C$, note that 
\[
\scalebox{0.8}{\input{./figures/cg-3.tikz}}
\quad
\text{ and }
\quad
\scalebox{0.8}{\input{./figures/cg-3i.tikz}}
\]
Let us now note the second statement. Firstly~\eqref{eq:causally-complemented} is easily seen to ensure causal complementation. Conversely, for any kernel $k$ as above let $f = k \circ k^\dagger + k^{\bot} \circ k^{\bot \dagger} \colon A \to A$. Then we have 
\[
\scalebox{0.8}{\input{./figures/cg-8.tikz}}
\]
and so by causal complementation, $f$ is causal. But since all dagger kernels are causal we have $\discard{} \circ f 
= 
\discard{} \circ k^\dagger + \discard{} \circ k^{\bot \dagger}$
and so~\eqref{eq:causally-complemented} holds. 

Finally let us show that $+$ as defined above is unique. Indeed if $\catC$ comes with any other addition then by~\eqref{eq:causally-complemented} for any object $C$ as above we have 
\[
\scalebox{0.8}{\input{./figures/cg-7.tikz}}
\]
It follows that any morphism $h$ satisfying~\eqref{eq:cg-def1} will then automatically have marginal $f + g$, and so $+$ coincides with our definition above. In particular $+$ is independent of our choice of $C$.
\end{proof}

We may thus see conditioning as a convenient diagrammatic way of encoding coarse-graining, and in place of our pair of new conditions have equivalently required the presence of such an operation $+$ satisfying~\eqref{eq:causally-complemented}. 

\begin{examples} \label{examples:OfPrinciples}
$\Quant{\mathbb{C}}$ and $\Quant{\mathbb{R}}$ each satisfy the operational principles, as we will prove in Section~\ref{sec:PureProcessProperties}. 
\end{examples}

\section{Deriving Superpositions}

Let us now begin our reconstruction by using the operational principles to derive superposition-like features in our theory.

Our first result strengthens the observation that, by essential uniqueness, any pair of causal pure states of the same object are related by a unitary.

\begin{lemma} \label{lem:strong_homog}
In any dagger theory satisfying the operational principles, for any pairs $\{\ket{0}, \ket{1}\}$ and $\{\ket{0'}, \ket{1'}\}$ of orthonormal pure states of an object $A$, there is a unitary $U$ on $A$ with $U \circ \ket{0} = \ket{0'}$ and $U \circ \ket{1} = \ket{1'}$. 
\end{lemma}

\begin{proof}
By essential uniqueness there is a unitary $U$ on $A$ with $U \circ \ket{0} = \ket{0'}$. Since every causal pure state is a dagger kernel we may define a new dagger kernel $k = \ket{0'}^\bot \colon K \to A$. 

Since unitaries preserve orthogonality, $U \circ \ket{1}$ is orthogonal to $\ket{0'}$, so that $U \circ \ket{1}  = k \circ \psi$ for the causal pure state $\psi = k^\dagger \circ U \circ \ket{1}$. Similarly we always have $\ket{1'} = k \circ \phi$ for some causal pure state $\phi$. By essential uniqueness there is then a unitary $V$ on $K$ with $V \circ \psi = \phi$, and in turn a unitary $W$ on $A$ with $W \circ k = k \circ V$. 

One may then verify that $W^{\dagger} \circ \ket{0'}$ is orthogonal to $k$ and so factors over $k^{\bot} = \ket{0'}^{\bot \bot} = \ket{0'}$. Hence we have $W^{\dagger} \circ \ket{0'} = \ket{0'} \circ z$ for some scalar $z$. Then since $\ket{0'}$ is an isometry so is the scalar $z$, and so, since all scalars are pure, by the CP axiom we have $z = \id{I}$. Finally since $W$ preserves $\ket{0'}$ we have that $W \circ U$ is the desired unitary on $A$.
\end{proof}

In just the same way one may show that any orthonormal collections of pure states of the same size $\{\ket{i}\}^n_{i=1}$ and $\{\ket{i'}\}^n_{i=1}$ are related by a causal isomorphism; this is called \emps{strong symmetry} in~\cite{barnum2014higher}. The result also allows us to extend conditioning to pure morphisms as follows.

\begin{lemma} \label{lem:existence}
In any dagger theory satisfying the operational principles, for any orthonormal pure states $\ket{0}, \ket{1}$ of an object $A$ and pair of pure states $\psi, \phi$ of an object $B$ there is a pure $f \colon A \to B$ with $f \circ \ket{0} = \psi$ and $f \circ \ket{1} = \phi$. 
\end{lemma}
\begin{proof}
If $\psi = 0$ then we may take $f = \phi \circ \ket{1}^\dagger$, and similarly if $\phi = 0$ the result is trivial. Otherwise assume that $\psi$ and $\phi$ are non-zero. Using conditioning, let $h$ be any morphism satisfying 
\[
\scalebox{0.8}{\input{./figures/purif-in-proof-program-depure.tikz}}
\]
and let $g$ be any purification of $h$ via some object $C$. Then since all morphisms involved are pure it follows that 
\[
\scalebox{0.8}{\input{./figures/purif-in-proof-program-2.tikz}}
\]
for some causal states $a, b$, which must be pure by Lemma \ref{lem:caustensclosedexamples}. Then by Lemma~\ref{lem:strong_homog} there is a unitary $U$ with 
\[
\scalebox{0.8}{\input{./figures/U-function.tikz}}
\]
Finally the pure morphism defined by
\[
\scalebox{0.8}{\input{./figures/proof-fancy-construction.tikz}}
\]
then has $f \circ \ket{0} = \psi$ and $f \circ \ket{1} = \phi$. 
\end{proof}

We are now able to show that $\catC_{\pure}$ has a qubit-like object.

\begin{proposition} \label{prop:PhBiprodII}
Let $\catC$ satisfy the operational principles. Then
$\catC_{\pure}$ has a phased dagger biproduct $B = I \pbiprod I$ for which all phases are unitary. 
\end{proposition}
\begin{proof}
Let $B$ be any object with a pair of pure causal states $\ket{0}, \ket{1}$ with $\ket{0} = \ket{1}^\bot$ as dagger kernels, as in Proposition~\ref{prop:OpTheoryProperties}~\ref{enumi:specialObjectC}. Then $\ket{0}, \ket{1} \colon I \to B$ satisfy the existence property of a phased coproduct by Lemma~\ref{lem:existence}. 

We now establish the uniqueness property. Let $\tinymultflip[whitedot] \colon B \to B \otimes B$ be a pure morphism with $\tinymultflip[whitedot] \circ \ket{i} = \ket{i} \otimes \ket{i}$ for $i = 0,1$. 
Then since $\ket{0} = \ket{1}^{\bot}$ and 
\[
\scalebox{0.8}{\input{./figures/copier-zero-0.tikz}} \  0 
\quad
\text{ we have }
\quad
\scalebox{0.8}{\input{./figures/copier-top-0.tikz}}
\]
along with the similar equation for $\ket{1}$. Now let $f, g \colon B \to A$ be pure with $f \circ \ket{i} = g \circ \ket{i}$ for $i= 0, 1$. If $f \circ \ket{0} = 0$ then  since $\ket{1} = \ket{0}^\bot$ we get $f = f \circ \ket{1} \circ \ket{1}^\dagger = g$, and similarly $f = g$ if $f \circ \ket{1} = 0$. So now suppose that $f \circ \ket{i} \neq 0$ for $i = 0,1$. By the above we have 
\[
\scalebox{0.8}{\input{./figures/copier-EUPurif-1a.tikz}}
\]
and so bending wires and using causal complementation we get
\[
\scalebox{0.8}{\input{./figures/copier-EUPurif-1.tikz}}
\quad
\text{ and so }
\quad
\scalebox{0.8}{\input{./figures/copier-EUPurif2.tikz}}
\]
for some unitary $U$ by essential uniqueness. But then 
\[
0 \ \ \scalebox{0.8}{\input{./figures/copier-fapply-merge.tikz}}
\]
Hence by zero-cancellativity we have $\ket{1}^\dagger \circ U \circ \ket{0} = 0$ and so $U \circ \ket{0} = \ket{0} \circ z$ for some scalar $z$. But then $z$ is an isometry and so by the CP axiom $z = \id{I}$. Hence $U$ preserves the states $\ket{0}$, and $\ket{1}$, i.e.~is a phase. Now letting $\tinycounit \colon B \to I$ be any pure effect with $\tinycounit \circ \ket{0} = \id{I} = \tinycounit \circ \ket{1}$ we have 
\begin{equation} \label{eq:witjcounit}
\scalebox{0.8}{\input{./figures/copier-last.tikz}} 
\end{equation}
where each of the endomorphisms of $B$ below $f, g$ above are also phases.

Finally note that any phase $W$ is unitary, since we have that $\ket{i}^\dagger \circ W^\dagger = \ket{i}^\dagger$ for $i =0,1$ and so $W^\dagger$ is causal by causal complementation and hence unitary by essential uniqueness. In particular this makes phases invertible, so that by~\eqref{eq:witjcounit} $f$ and $g$ are equal up to phase, making $B$ a phased coproduct, and closed under the dagger, so that $B$ is a phased dagger biproduct by Lemma~\ref{lem:ph-d-biprod-help}.     
\end{proof}

\begin{corollary} \label{cor:daggerphasedbiproducts}
Let $\catC$ satisfy the operational principles. Then $\catC_{\pure}$ has the superposition properties.
\end{corollary}

\begin{proof}
By the previous result, $\catC_\pure$ has a phased dagger biproduct $I \pbiprod I$. Then just as in Lemma~\ref{lem:comp-closed}~\ref{enum:distmonic} by compactness $A \otimes (I \pbiprod I)$ is a phased dagger biproduct $A \pbiprod A$, for all objects $A$. Hence since all kernels in $\catC$ are also kernels in $\catC_\pure$, by Proposition~\ref{prop:kernelsGiveBiproducts} to show that $\catC_\pure$ has phased dagger biproducts it suffices to show for all objects $A$, $B$ that there are orthogonal kernels $k \colon A \to C$ and $l \colon B \to C$.

Now if ether $A$ or $B$ is a zero object the result is trivial. Otherwise let $\psi$ and $\phi$ be causal pure states of $A, B$ respectively, and let $C$ be an object with two orthogonal causal pure states $\ket{0}$, $\ket{1}$, such as $I \pbiprod I$. Then these states are all kernels and so by Proposition~\ref{prop:kernels-in-compact} so are the morphisms
\[
\scalebox{0.8}{\input{./figures/ker1.tikz}}
\]
which are indeed orthogonal. Hence $\catC_\pure$ has finite phased dagger biproducts. 

We now verify the positive cancellation property. First consider a pure positive endomorphism $f^{\dagger} \circ f$ of $A \pbiprod B$ which is diagonal so that $f_A := f \circ \pcoproj_A$ and $f_B := f \circ \pcoproj_B$ are orthogonal. Letting $c_A = \img(f_A)^{\dagger}$ and $c_B = {c_A}^{\bot}$ we have
\begin{align*}
f_A^\dagger \circ f_B = 0 
\implies 
c_A \circ f_B = 0
&\implies 
c_A \circ f = c_A \circ f \circ \coproj_A \circ \coproj_A^\dagger 
\\ 
c_B \circ f \circ \coproj_A
=
\coker(f_A) \circ f_A
=
0
&\implies
c_B \circ f = c_B \circ f \circ \coproj_B \circ \coproj_B^\dagger 
\end{align*}
using that $\pcoproj_A$ and $\pcoproj_B$ are dagger kernels by Lemma~\ref{lem:dag-ker-coproj}. Hence we have 
\begin{align} \label{eq:poscancel2a}
\discard{} \circ c_A \circ f 
&=
\discard{} \circ f_A \circ \coproj_A^{\dagger}
\\   \label{eq:poscancel2b}
\discard{} \circ c_B \circ f 
&=
\discard{} \circ f_B \circ \coproj_B^{\dagger}
\end{align}
Now if any other pure diagonal endomorphism $g$ has $f^{\dagger} \circ f = g^{\dagger} \circ g \circ U$ for some phase $U$, defining $g_A = g \circ \coproj_A$ and $g_B = g \circ \coproj_B$ we have that $f_A^{\dagger} \circ f_A = g_A^{\dagger} \circ g_A$, and the similar equation holds for $B$. Then by the CP axiom 
\begin{align*}
\discard{} \circ f_A = \discard{} \circ g_A 
& & 
\discard{} \circ f_B = \discard{} \circ g_B
\end{align*}
Since $c_B = {c_A}^{\bot}$, by causal complementation and~\eqref{eq:poscancel2a}, \eqref{eq:poscancel2b} we have $\discard{} \circ f = \discard{} \circ g$. Finally $f^{\dagger} \circ f = g^{\dagger} \circ g$ by the CP axiom again.
\end{proof}

Hence we can conclude that whenever $\catC$ satisfies the operational principles it comes with an embedding $\Quant{S} \hookrightarrow \catC$. However by studying the properties of $\catC_{\pure}$ in detail we will be able to say much more. 

\section{Properties of Pure Morphisms} \label{sec:PureProcessProperties}

Whenever $\catC$ is a dagger compact category satisfying the operational principles, we can capture the properties of $\catC_\pure$ and $\plusIdag{\catC_\pure}$ as follows. 

\begin{definition}
Consider dagger compact categories with dagger kernels satisfying the following: 
\begin{itemize}
\item 
\deff{state habitation}: \index{state habitation} every non-zero object has a non-zero state;
\item 
\deff{dagger normalisation}: \index{dagger normalisation} every non-zero state $\psi \colon I \to A$ has $\psi = \sigma \circ r$ for some isometry $\sigma \colon I \to A$ and scalar $r$;
\item 
\deff{homogeneity}: \index{homogeneity} for all $f, g \colon A \to B$ we have
\[
\scalebox{0.8}{\input{./figures/homog.tikz}}
\]
for some unitary $U$ on $B$.
\end{itemize}
A  \deff{pre-quantum category} \index{pre-quantum category} $\catB$ is one which furthermore has finite phased dagger biproducts with positive-free phases, and that $\id{I}$ is its only unitary scalar. 

Alternatively, a \deff{quantum category} \index{quantum category} $\catA$ is one which satisfies the above and has dagger biproducts. 
\end{definition}

In fact we will see that any pre-quantum category has the stronger property of positive-cancellation for phases. Now from the results of the previous section, essential uniqueness and the CP axiom, we immediately have the following. 

\begin{proposition} \label{prop:OpToquantWOGlobal}
Let $\catC$ satisfy the operational principles. Then $\catC_{\pure}$ is a pre-quantum category.
\end{proposition}

Just as $\Hilb$ is typically studied in place of $\HilbP$, we will be able to learn more by passing from $\catC_\pure$ to a category with proper biproducts.

\begin{proposition} \label{prop:QuantWOGToQuantBiprod}
Let $\catB$ be a pre-quantum category. Then $\plusIdag{\catB}$ is a quantum category, with its canonical choice of global phases $\mathbb{P}$ consisting of its unitary scalars.
\end{proposition}

\begin{proof}
Biproducts in a compact category are distributive by Lemma~\ref{lem:comp-closed}. Hence by Corollary~\ref{cor:daggerbiproducts} and Proposition~\ref{prop:constr-compact} $\plusIdag{\catB}$ is dagger compact with dagger biproducts, and we may identify $\catB$ with its category $\plusIdag{\catB}_\mathbb{P}$ of equivalence classes $[-]$ under $f \sim g$ whenever $f = u \cdot g$ for $u \in \mathbb{P}$, with all such $u$ being unitary. In fact every unitary scalar $u$  in $\plusIdag{\catB}$ has that $[u]$ is unitary in $\catB$ and so $[u] = \id{I}$, giving $u \in \mathbb{P}$. 

Now since all phases are positive-free we have $[f^\dagger \circ f] = \id{} \implies f^\dagger \circ f = \id{}$ for all morphisms $f \in \plusIdag{\catB}$. In particular a morphism $f$ in $\plusIdag{\catB}$ is an isometry or unitary iff $[f]$ is in $\catB$. This lets one straightforwardly deduce dagger normalisation and homogeneity in $\plusIdag{\catB}$ using that they hold in $\catB$. Noting that $[f] = 0 \implies f = 0$ it follows that if $[f] = \ker([g])$ in $\catB$ then $f = \ker(g)$ in $\plusIdag{\catB}$, and so $\plusIdag{\catB}$ has dagger kernels.
\end{proof}

\begin{examples}
$\FHilbP$ is a pre-quantum category, with homogeneity easily seen to follow from the polar decomposition of a complex matrix. Hence by the previous result $\FHilb \simeq \plusIdag{\FHilbP}$ is a quantum category. 

Note that in contrast homogeneity fails in $\Hilb$; for example on $l^2(\mathbb{N})$ the shift operator $(a_0, a_1, \dots) \mapsto (0, a_0, \dots)$ is an isometry but not unitary.
\end{examples}

Quantum categories have a rich structure, generalising that of $\FHilb$, which we now explore. Recall that since they have biproducts they come with an addition $+$ on morphisms, generalising the superpositions in $\Hilb$. In fact they surprisingly also come with a notion of subtraction.

\begin{proposition} \label{prop:quant-cat-properties}
In any quantum category $\catA$, the following hold.
\begin{enumerate}[label=\arabic*., ref=\arabic*]
\item \label{enum:purepropnegatives}
Every morphism $f$ has an additive inverse $-f$;
\item \label{enum:daggerequalisers}
Every pair of morphisms $f, g$ have a \emps{dagger equaliser} in the sense of~\cite{vicary2011categorical};
\item \label{enum:dagmonodagkernel}
Every isometry is a kernel; 
\item \label{enum:decomp-property}
For every kernel $k \colon K \to A$ the morphism $[k, k^{\bot}] \colon K \biprod K^{\bot} \to A$ is unitary;
\item \label{enum:wellpted}
\emps{Well-pointedness}: ($f \circ \psi = g \circ \psi$ $\forall$ states $\psi$) $\implies$ $f =g$; 
\item \label{enum:bounded}
Every morphism $f \colon A \to B$ has a \emps{bound} in the sense of~\cite{heunen2009embedding}: a scalar $s$ such that for every state $\psi$ of $A$ we have $\psi^{\dagger} \circ f^{\dagger} \circ f \circ \psi = (s^{\dagger} \circ s) \circ (\psi^{\dagger} \circ \psi) + r$ for some positive scalar $r$.
\end{enumerate}
\end{proposition}

\begin{proof}
\ref{enum:purepropnegatives}. 
It suffices to find a scalar $t$ with $t + \id{I} = 0$, since then for all $f$ we have $f + (t \cdot f) = (\id{I} + t) \cdot f = 0$. 
As is standard we write $\langle a_1, a_2 \rangle \colon I \to I \biprod I$ for the unique state with $\pproj_i \circ \langle a_1, a_2 \rangle = a_i $ for $i=1,2$. Now let
\[
\begin{tikzcd}[column sep = huge]
I \rar{\Delta = \langle \id{I}, \id{I} \rangle} & I \biprod I
\end{tikzcd}
\]
  have $\Delta = \psi \circ s$ for some scalar $s$ and isometric state $\psi$. By homogeneity there is a unitary $U$ with $U \circ \pcoproj_1 = \psi$. Then let $\langle a, b \rangle = U \circ \pcoproj_2$. Since $U \circ \pcoproj_2$ is an isometry we have $a^\dagger \circ a + b^\dagger \circ b = \id{I}$ and also 
\begin{align*}
a + b 
&=
\Delta^{\dagger} \circ U \circ \pcoproj_2
\\ 
&= 
(s^{\dagger} \circ \pproj_1 \circ U^{\dagger}) \circ (U \circ \pcoproj_2)
\\ 
&= 
s^{\dagger} \circ (
\pproj_1^\dagger \circ \pcoproj_2 
)
= 
0
\end{align*}
Then $t = a^\dagger \circ b + b^\dagger \circ a$ is the required scalar since
\begin{align*}
\id{I} + t
&=
a^\dagger \circ a + b^\dagger \circ b + a^\dagger \circ b + b^\dagger \circ a 
\\
&=
(a + b)^{\dagger} \circ (a + b)
=
0
\end{align*}

\ref{enum:daggerequalisers}. 
This follows immediately with $f, g$ having dagger equaliser $\ker(f-g)$.

\ref{enum:dagmonodagkernel}. 
Thanks to~\ref{enum:purepropnegatives} a morphism $m$ is monic iff $\ker(m) = 0$. But then any isometry $i$ has $i = \img(i) \circ e$ with $\coker(e) = 0$, and so dually $e$ is an epimorphism. But since $i$ and $\img(i)$ are isometries, so is $e$, making it unitary and $i$ a kernel. 

\ref{enum:decomp-property}. 
$i = [k,k^{\bot}]$ is an isometry since $k$ and $k^{\bot}$ are orthogonal isometries. But if $f \circ i = 0$ then $f \circ k = 0$ and $f \circ k^{\bot} = 0$, so that $\img(f) = 0$ giving $f = 0$. Hence as in the previous part $i$ is epic, and so unitary. 

\ref{enum:wellpted}. 
Suppose that $f \circ \psi = g \circ \psi$ for all states $\psi$. Then $h = (f-g)$ has $h \circ \psi = 0$  and so $\coim(h) \circ \psi = 0$ for all states $\psi$. But if $\CoIm(h)$ is non-zero then it possesses a non-zero state $\phi$, and then $\coim(h) \circ \coim(h)^{\dagger} \circ \phi = \phi \neq 0$, a contradiction. Hence $\coim(h) = 0$ so that $h = 0$ and $f = g$.

\ref{enum:bounded}. 
Thanks to dagger normalisation it suffices to consider when $\psi$ is an isometry and hence a kernel. Then letting $c = \coker(\psi)$, by Proposition~\ref{prop:quant-cat-properties}~\ref{enum:decomp-property} we have $\id{A} = \psi \circ \psi^{\dagger} + c^\dagger \circ c$ so that:
\[
\scalebox{0.8}{\input{./figures/Trace-pic.tikz}}
\]
since the right-hand scalar is positive, we may take as $s$ the left-hand side scalar.
\end{proof}

\begin{remark}[Hilbert Categories]
Properties~\ref{enum:daggerequalisers},~\ref{enum:dagmonodagkernel} and~\ref{enum:bounded} and the presence of dagger biproducts make any quantum category a \emps{Hilbert category} in the sense of Heunen~\cite{heunen2009embedding}. By well-pointedness and~\cite[Theorem~4]{heunen2009embedding} this means that when $\catA$ is locally small and has that its ring of scalars $S$ is a field of at most continuum cardinality, there is a lax dagger monoidal embedding 
\[
\catA \hookrightarrow \Hilb
\]
up to some isomorphism of $S$. We will not rely on this result explicitly, but it would be interesting to further explore connections between our results and Heunen's.
\end{remark}

We can now in fact precisely characterise theories $\catC$ satisfying the operational principles in terms of quantum categories. Call a pre-quantum or quantum category \indef{non-trivial} \index{non-trivial!quantum category} \index{non-trivial!pre-quantum category} when it has $\id{I} \neq 0$.

\begin{proposition} \label{prop:quant-to-theory}
Let $\catA$ be a non-trivial quantum category. Then $\CPM(\catA)$ forms a dagger theory satisfying the operational principles. 
\end{proposition}

\begin{proof}
By Examples~\ref{examples:kernels}~\ref{ex:CPMD-kernels} and \ref{examples:Purif}~\ref{examples:purif-in-CPM} $\CPM(\catA)$ has dagger kernels and essentially unique dilations with respect to its environment structure $\Dbl{\catA}$, within which every object has a causal state by state habitation and dagger normalisation in $\catA$. To show that $\CPM(\catA)$ has strong purification, we need to show that a morphism belongs to $\Dbl{\catA}$ iff it is pure. 

By Proposition~\ref{prop:pure-lem} and compactness it suffices to show in $\CPM(\catA)$ that, for any non-zero state $\rho$ and causal state $\sigma$, that if $\rho \otimes \sigma \in \Dbl{\catA}$ then so does $\rho$. So suppose that this holds. It follows from well-pointedness in $\catA$ and the rule $f^\dagger \circ f = 0 \implies f = 0$ that there is some effect $\psi \in \catA$ for which 
\[
\scalebox{0.8}{\input{./figures/statesinA-alt.tikz}} \in \Dbl{\catA}
\]
is non-zero. Now by dagger normalisation and Proposition~\ref{prop:quant-cat-properties}~\ref{enum:dagmonodagkernel} in $\catA$ every state is kernel-pure, i.e.~of the form $k \circ s$ for some dagger kernel state $k$ and scalar $s$. So then in $\CPM(\catA)$ 
\[
\scalebox{0.8}{\input{./figures/statesinA3alt.tikz}}
\]
Since $\Dbl{k}$ is again a kernel in $\CPM(\catA)$ it follows from zero-cancellativity that $\img(\rho) = \Dbl{k}$, and so $\rho = \Dbl{k} \circ t$ for some scalar $t$. 
Then dagger normalisation in $\catA$ states that every scalar in $\CPM(\catA)$ belongs to $\Dbl{\catA}$. In particular so does $t$ and hence so does $\rho$, as required. 
Hence $\CPM(\catA)$ has strong purification.

In particular we've just seen that all scalars are pure, and so $\CPM(\catA)$ has normalisation, and by the CP axiom and Proposition~\ref{prop:quant-cat-properties}~\ref{enum:dagmonodagkernel} every causal pure state is a kernel. Hence by Lemma~\ref{lem:Pure-Excl-Is-Easy} pure exclusion holds also.

Next we show that non-triviality of $\catA$ ensures non-triviality of the dagger theory $\CPM(\catA)$. Let $A=I \biprod I$ in $\catA$. Then if $\discard{A}$ is an isomorphism in $\CPM(\catA)$ it is pure and hence unitary, giving a unitary $\psi = [a,b] \colon I \biprod I \to I$ in $\catA$. But $\psi$ being an isometry is equivalent to $a$ and $b$ being unitary scalars in $\catA$ with $a^{\dagger} \cdot b = 0$. But then $a= b =0$ and so $\id{I} = 0$, a contradiction.

Finally, by Proposition~\ref{prop:addition-in-CPM} the addition in $\catA$ provides $\CPM(\catA)$ with addition also. Moreover by Proposition~\ref{prop:quant-cat-properties}~\ref{enum:decomp-property} in $\catA$ all dagger kernels $k \colon K \to A$ satisfy 
\[
\scalebox{0.8}{\input{./figures/kersumrule.tikz}}
\]
which translates precisely to~\eqref{eq:causally-complemented} in $\CPM(\catA)$. Hence by Proposition~\ref{prop:coarse-graining} $\CPM(\catA)$ satisfies the remaining operational principles.
\end{proof}

\begin{theorem} \label{thm:mainPropToQuantumCat}
There are one-to-one correspondences between non-trivial:
\begin{itemize}
\item quantum categories $\catA$;
\item pre-quantum categories $\catB$;
\item dagger theories $\catC$ satisfying the operational principles;
\end{itemize}
up to equivalence, via $\catA = \plusIdag{\catB}$, $\catB = \catC_{\pure}$, $\catC = \CPM(\catA)$. Moreover, the equivalence $\catC \simeq \CPM(\catA)$ preserves addition.
\end{theorem}

\begin{proof}
The assignments are well-defined by Propositions~\ref{prop:OpToquantWOGlobal},~\ref{prop:QuantWOGToQuantBiprod} and~\ref{prop:quant-to-theory} along with the observation that if such a category $\catB$ is non-trivial then so is $\plusIdag{\catB}$.

 First let $\catA$ be a quantum category, and choose as global phases $\mathbb{P}$ all of its unitary scalars, writing $f \sim g$ when $f = u \cdot g$ for some $u \in \mathbb{P}$. Then by Corollary~\ref{cor:daggerbiproducts} we have $\catA \simeq \plusIdag{\catA_{\quotP}}$. On the other hand by homogeneity we in fact have $\Dbl{f} = \Dbl{g} \iff f \sim g$ since:
\[
\scalebox{0.8}{\input{./figures/Dbl-argument.tikz}}
\]
for some unitary scalar $u$. But as in the proof of Proposition~\ref{prop:quant-to-theory} we have $\Dbl{\catA} = \CPM(\catA)_\pure$, giving an equivalence $\catA_\quotP \simeq \CPM(\catA)_{\pure}$. Hence we obtain a dagger monoidal equivalence $\catA \simeq \plusIdag{\CPM(\catA)_{\pure}}$.  

Next, let $\catB$ be a pre-quantum category and consider the quantum category $\catA = \plusIdag{\catB}$. Then as above $\catB \simeq \catA_\quotP \simeq \CPM(\catA)_{\pure}$ as required. In particular by Corollary~\ref{cor:daggerphasedbiproducts} any pre-quantum category $\catB$ has the strong superposition properties.

Now if $\catC$ satisfies the operational principles, by Lemma~\ref{lem:CPMsCoincide} the functor $[-]$ extends to an equivalence of dagger theories $\CPM(\plusIdag{\catC_\pure}) \simeq \catC$. 

Finally we check that this equivalence preserves addition. Since we've seen that all kernels in $\CPM(\catA)$ are of the form $\Dbl{k}$ for a kernel $k$ in $\catA$, one may check that the addition in $\CPM(\catA)$ from Proposition~\ref{prop:addition-in-CPM} satisfies~\eqref{eq:causally-complemented} by Proposition~\ref{prop:quant-cat-properties}~\ref{enum:decomp-property}. Hence since this makes the operation unique by Proposition~\ref{prop:coarse-graining} it is preserved by any equivalence of dagger theories. 
\end{proof}

This is a strong result, since for general $\catC$ with an environment structure there may be many $\catA$ with $\catC \simeq \CPM(\catA)$. 

\subsection{The extended scalars}

Our characterisation of theories satisfying the operational principles motivates further study of the scalars in a quantum category $\catA$, which we describe as follows. 


\begin{definition} 
A \deff{phased ring} \index{phased ring} is a commutative involutive ring $(S, \dagger)$ 
which is an integral domain (with $a \cdot b = 0 \implies a = 0 \text{ or } b = 0$) 
such that $\forall a, b$
\[
  a^{\dagger} \cdot a + b^{\dagger} \cdot b = c^{\dagger} \cdot c
\] 
for some $c \in S$, with any such $c$ having $a = c \cdot d$ and $b = c \cdot e$ for some $d, e \in S$.
\end{definition}

\begin{examples}
$\mathbb{C}$ forms a phased ring, as does $\mathbb{R}$ under the trivial involution.
\end{examples}

\begin{proposition} \label{prop:scalars_are_phased_ring}
Let $\catA$ be a quantum category. Then $\catA(I,I)$ is a phased ring.
\end{proposition}

\begin{proof}
$S = \catA(I,I)$ forms a commutative semi-ring since $\catA$ has dagger biproducts, and $S$ is a ring by Proposition~\ref{prop:quant-cat-properties}~\ref{enum:purepropnegatives} and an integral domain by Lemma~\ref{lem:dagzero}~\ref{enum:zero-mult}. 

Now given $a, b \in S$ let $\psi = \langle a, b \rangle \colon I \to I \biprod I$. Using normalisation let $\psi = \phi \circ c$ where $\phi$ is an isometry. Then 
 \[
 c^\dagger \cdot c = \psi^{\dagger} \circ \psi = a^\dagger \cdot a + b^\dagger \cdot b
 \]
 Furthermore $d = \pproj_1 \circ \phi \in S$ has $a = \pproj_1 \circ \psi = c \cdot d$, and similarly $c$ divides $b$. Moreover any other $e \in S$ with $e^{\dagger} \cdot e = c^{\dagger} \cdot c$ has $e = c \cdot u$ for a unitary $u$ by homogeneity, and so also divides $a$ and $b$.
\end{proof}

Phased rings provide us at last with our main examples of quantum categories. By a \indef{phased field} \index{phased field} we mean a phased ring which is also a field.

\begin{example} \label{example:MatofPhasedField}
Let $S$ be a phased field. Then $\Mat_S$ is a quantum category. 
In particular so are  $\Mat_{\mathbb{C}}$ and $\Mat_{\mathbb{R}}$. 
\end{example}

Hence, for any such $S$, the dagger theory $\Quant{S}$ satisfies the operational principles, as do $\Quant{\mathbb{C}}$ and $\Quant{\mathbb{R}}$.

\begin{proof}
We've seen that $\Mat_S$ is always dagger compact with dagger biproducts. We now establish dagger normalisation. 
For any state $\psi = (a_i)^n_{i=1} \colon 1 \to n$, since $S$ is a phased ring we have 
\[
\psi^\dagger \circ \psi = \sum^n_{i=1} a_i^\dagger \cdot a_i = a^\dagger \cdot a
\] 
for some $a \in S$. Then if $\psi \neq 0$ also $a \neq 0$ and so $\phi = (\frac{a_i}{a})^n_{i=1}$ is an isometry with $\psi = \phi \circ a$. We now show that $\Mat_S$ has dagger kernels. Note that the states on any object $n \in \mathbb{N}$ form the vector space $S^n$ and also come with the `inner product' 
\[
\langle \psi, \phi \rangle := \psi^{\dagger} \circ \phi
\] 
for $\psi, \phi \colon 1 \to n$. Since $S$ is a phased ring this satisfies $\langle \psi, \psi \rangle = 0 \implies \psi = 0$. 

Now for any morphism $M \colon n \to m$, the set $\{ \psi \mid M \circ \psi = 0 \}$ is a subspace of $S^n$ and so has a finite basis $\{\psi_i\}^r_{i=1}$ for some $r \leq n$. Using the well-known Gram-Schmidt algorithm (see e.g.~\cite[p.544]{cheney2009linear}) we may replace this by another basis $\{\phi_i\}^r_{i=1}$ which is orthonormal in that $\langle \phi_j^{\dagger}, \phi_i \rangle = \delta_{i,j}$. Then $k := (\phi_i)^r_{i=1} \colon r \to n$ in $\MatS$ is an isometry with $k = \ker(M)$. 

Next we verify homogeneity. Let $M, N \colon n \to m$ satisfy $M^{\dagger} \circ M = N^{\dagger} \circ N$. It follows that $\coim(M) = \coim(N)$ and so after restricting along these we may assume that $\ker(M) = \ker(N) = 0$. Now define a modified `inner product' by 
\[
\langle \psi, \phi \rangle' := \langle M \circ \psi, M \circ \phi \rangle = \langle N \circ \psi, N \circ \phi \rangle
\] 
Again this satisfies $\langle \psi, \psi \rangle' = 0 \implies \psi = 0$. Hence we may again apply the Gram-Schmidt algorithm to find an orthonormal basis $\{e_i\}^n_{i=1}$ with respect to $\langle - , - \rangle'$. Then $\{M \circ e_i\}^n_{i=1}$ and $\{N \circ e_i\}^n_{i=1}$ are each orthonormal collections of states of $m$ and so may be extended to orthonormal bases $\{\psi_i\}^m_{i=1}$ and $\{\phi_i\}^m_{i=1}$ respectively. Finally, any matrix $U \colon m \to m$ with $U \circ \psi_i = \phi_i$ for all $i$ is unitary and satisfies $U \circ M = N$.
\end{proof}

From the definition, we see that the positive elements $R=S^\pos$ of a phased ring are always closed under addition, forming a sub-semi-ring of $S$, and have nice properties: they have characteristic $0$, and that $a$ is divisible by $a + b$ for all $a, b$, hence coming with an embedding $\Qpos \hookrightarrow R$ of the positive rationals.

Under one extra assumption we obtain a converse to the above result, telling us when a quantum category $\catA$ arises as such a matrix category. Call a semi-ring $R$ \indef{bounded} \index{bounded semi-ring} when no element $r$ has that for all $n \in \mathbb{N}$ there is some $r_n \in R$ with $r = n + r_n$. For example $\Rpos$ and $\Qpos$ are certainly bounded. Boundedness is similar to the \emps{Archimedean} property for totally ordered groups~\cite{EncyMath}. 

\begin{lemma} \label{lem:Boundedequivalence}
Let $\catA$ be a quantum category and $S$ its ring of scalars. If $S^{\pos}$ is bounded then $\catA \simeq \Mat_{S}$.
\end{lemma}
\begin{proof}
Consider the full embedding $\MatS \hookrightarrow \catA$ given by $n \mapsto n \cdot I$. We now show that any object $A$ has a unitary $A \simeq n \cdot I$ for some $n \in \mathbb{N}$. If $A$ is a zero object we are done, otherwise there is an isometry $\psi \colon I \to A$, which is a kernel by Proposition~\ref{prop:quant-cat-properties}. Then by the same result, letting $B = \coker(\psi_1)$ the morphism $[\psi^{\bot}, \psi] \colon B \biprod I \simeq A$ is unitary. Setting $B_1 = B$ and proceeding similarly we get a sequence $B_1, B_2, \dots$ with $A \simeq B_n \biprod n \cdot I$ for each $n$. Then if $B_n \simeq 0$ for some $n$ we are done. Otherwise for all $n \in \mathbb{N}$
\[
\scalebox{0.8}{\input{./figures/dimAbig.tikz}}
\!
=
\ 
\scalebox{0.8}{\input{./figures/dimBnbig.tikz}}
= \  
\scalebox{0.8}{\input{./figures/dimBbig.tikz}}
+\  
\sum^n_{i=1}
\scalebox{0.8}{\input{./figures/dimSum.tikz}}
= \  
\scalebox{0.8}{\input{./figures/dimBbig.tikz}}
+
\
n
\]
contradicting boundedness.
\end{proof}

\paragraph{Towards real or complex structure}

Under another condition we can show that a phased ring $S$ resembles one of our motivating examples of $\mathbb{R}$ or $\mathbb{C}$. 

First, note that the semi-ring $R = S^\pos$ may be freely extended to a ring $\dring{R}$, the \indef{difference ring} \index{difference ring}\label{not:differencering} of $R$. Formally $\dring{R}$ consists of pairs $(a,b)$ of elements of $R$ after identifying
\[
(a,b) = (c,d) \iff a + d = b + c
\]
Addition and multiplication are defined in the obvious way when interpreting $(a,b)$ as `$a-b$'. For example $\dring{\Rpos} = \mathbb{R}$. Next, for any ring $S$ we write $S[i]$\label{not:adjoin-i} for the involutive ring with elements of the form
\[
a + b \cdot i
\]
for $a, b \in S$, where $1 = -i^2 = i \cdot i^{\dagger}$, and we define $a^{\dagger} = a$ for all $a \in S$. Now say that a semi-ring $R$ has \indef{square roots} \index{square roots} when every $a \in R$ has $a = b^2$ for some $b \in R$.

\begin{lemma} \label{lem:square roots}
Let $S$ be a phased ring for which $R = S^{\pos}$ has square roots. 
\begin{enumerate}
\item \label{enum:posAndUn}
Every non-zero $s \in S$ has $s = r \cdot u$ for a unique $r \in R$ and unitary $u \in S$.
\item \label{enum:TotOrdered}
$R$ is totally ordered under $a \leq b$ whenever $a + c = b$ for some $c \in R$.
\item \label{enum:equalsSa}
$\dring{R} \simeq S^{\sa} := \{s \in S \mid s^{\dagger} = s \}$\label{not:sadjoint}.
\item \label{enum:TwoOptions}
Either $S = S^{\sa}$ with trivial involution, or $S$ has square roots and $S = S^{\sa}[i]$.
\end{enumerate}
\end{lemma}

\begin{proof}
\ref{enum:posAndUn}
For uniqueness, suppose that $p \cdot u = l \cdot v$ for $p, l \in R$ and $u, v$ unitary. Then $p = l \cdot w$ where $w = v \cdot u^{-1}$ is unitary. So 
\[l \cdot w^{\dagger} = p^{\dagger} = p = l \cdot w\] Since $S$ is an integral domain, multiplication is cancellative so $w = w^{\dagger}$ and $w^2 = 1$. If $w = 1$ we are done, otherwise $w = -1$ and so $p + l = 0$. But then $p = l = 0$ by the definition of a phased ring. 
For existence, given $s \in S$ let $r = s^\dagger \cdot s \in S^\pos$. Since $S^\pos$ has square roots, $r = t^2$ for some $t \in S^\pos$. Then $s^\dagger \cdot s = t^\dagger \cdot t$, so $s = t \cdot u$ for some $u$, which is easily seen to be unitary.

\ref{enum:TotOrdered} Let $s \in S^{\sa}$ be non-zero with $s = t \cdot u$ as above. Then $t \cdot u = t \cdot u^{\dagger}$ and so $u =u^{\dagger}$ giving $u = \pm 1$. Hence either $s \in R$ or $-s \in R$. Then $\forall a, b \in R$ either 
\[a - b \in R\quad \text{ or } \quad b - a \in R\] making $R$ totally ordered in the above manner. 

\ref{enum:equalsSa} We may identify $\dring{R}$ with the set of elements $a-b \in S$ for $a, b \in R$. Then $\dring{R} \subseteq S^{\sa}$ always, but by the previous part $S^{\sa} \subseteq \dring{R}$. 

\ref{enum:TwoOptions} Suppose that $S \neq S^{\sa}$. We will show that $S$ has square roots  using techniques adapted from Vicary~\cite[Thm.~4.2]{vicary2011categorical}. Thanks to the first part it suffices to find a unitary square root of any unitary $u \in S$. Fix some unitary $u$. Suppose that for all $s \in S$ we have $s + u \cdot s^{\dagger} = 0$. Then putting $s = 1$ shows that $u = -1$, and so $S = S^{\sa}$, a contradiction. Hence there is $s \in S$ such that 
\[x := s + u \cdot s^{\dagger} \neq 0
\]
Then $x^\dagger = u^\dagger \cdot x$. Letting $x = r \cdot v$ for a unitary $v$ and $r \in R$ we have $r \cdot v^\dagger = r \cdot v \cdot u^\dagger$, and so $v^\dagger = v \cdot u^\dagger$ and hence $v^2 = u$ as desired. In particular, $-1$ has a unitary square root $i$. Finally note that $2$ is divisible in $R$ thanks to the embedding $\mathbb{Q}_{\geq 0} \hookrightarrow R$. Then for any $s \in S$ defining elements of $S^\sa$ by
\[
\mathsf{R}(s) := \frac{1}{2} \cdot (s + s^{\dagger}) \qquad \mathsf{I}(s) := \frac{i}{2}(s^{\dagger} - s)\]
we have $s = \mathsf{R}(s) + i \cdot \mathsf{I}(s)$ so that $S = S^{\sa}[i]$. 
\end{proof}

\paragraph{Phased fields}
To close in on our examples of $\mathbb{R}$ and $\mathbb{C}$ further, we may wonder when a phased ring $S$ is in fact a field. Indeed in any phased ring $S$ every element of the form $1 + s^{\dagger} \cdot s$ is invertible, and it is only for fields $S$ that we showed that $\Quant{S}$ satisfies our principles. 

We leave open the question of determining a phased ring $S$ which is not a field, or proving that none such exists, but note the following sufficient conditions for $S$ to be one. Recall that in any category a \emps{sub-object} of an object $B$ is an (isomorphism class of a) monic $m \colon A \rightarrowtail B$.


\begin{lemma} \label{lem:no_nontriv_subob}
Let $\catA$ be a quantum category, $S$ its ring of scalars and $R = S^{\pos}$. 
The following are equivalent:
\begin{enumerate}[label=\arabic*., ref=\arabic*]
\item \label{enum:R-divis}
$R$ is a semi-field;
\item \label{enum:S-divis}
$S$ is a field;
\item \label{enum:Subob}
In $\catA$ the only sub-objects of $I$ are $\{0, I\}$. 
\end{enumerate}
\end{lemma}
\begin{proof}
\ref{enum:R-divis} $\iff$ \ref{enum:S-divis}: $R \subseteq S$ and an element $s \in S$ is invertible iff $s^{\dagger} \cdot s$ is. 

\ref{enum:S-divis} $\implies$ \ref{enum:Subob}:
Let $m \colon A \rightarrowtail I$ be monic. Then if $r: = m \circ m^\dagger$ is zero then $m = 0$ by Lemma~\ref{lem:dagzero}~\ref{enum:dag-law}. Otherwise $r$ is invertible and hence so is $m$. 

\ref{enum:Subob} $\implies$ \ref{enum:S-divis}: 
Thanks to zero-cancellativity, every non-zero scalar $r$ has $\ker(r) = 0$. Since $\catA$ has negatives by Proposition~\ref{prop:quant-cat-properties}, this makes $r$ monic and hence an isomorphism by assumption. 
\end{proof}


Let us say that an ordered semi-ring $R$ has \indef{no infinitesimals} if whenever $a \leq \frac{1}{n}$ for all $n \in \mathbb{N}$ we have $a = 0$. 

\begin{lemma} \label{lem:its_a_field}
Let $S$ be a phased ring and suppose that $S^\pos$ is totally ordered and has no infinitesimals. Then $S$ is a field.
\end{lemma}
\begin{proof}
From the definition of a phased ring, we have for $a, b \in S^{\pos}$ that whenever $a \leq b$ then $a$ is divisible by $b$. If $a \neq 0$ then, by assumption on $S^\pos$, $1 \leq a \cdot n$ for some $n \in \mathbb{N}$. This makes $a \cdot n$ invertible and hence $a$ also. Hence every $s \in S$ is invertible since $s^{\dagger} \cdot s$ is. 
\end{proof}

\section{Reconstruction}

Let us now spell out our main result.

\begin{theorem} \label{thm:main-reconstruction}
Let $\catC$ be a dagger theory satisfying the operational principles and $R = \catC(I,I)$. Then there is an embedding of dagger theories
\begin{equation} \label{eq:main-embed-in-recons} 
\Quant{S} \hookrightarrow \catC
\end{equation}
which preserves addition, for some phased ring $S$ with $R \simeq S^{\pos}$ as semi-rings. Moreover if $R$ is bounded this is an equivalence of theories $\catC \simeq \Quant{S}$.
\end{theorem}

\begin{proof}
By Theorem~\ref{thm:mainPropToQuantumCat} there is an addition-preserving equivalence $\catC \simeq \CPM(\catA)$ where $\catA$ is a quantum category, so it suffices to assume that $\catC$ is of this form. But now $S = \catA(I,I)$ is a phased ring, and by dagger normalisation in $\catA$ we always have $\CPM(\catA)(I,I) \simeq S^\pos$.

The embedding $\Mat_{S} \hookrightarrow \catA$ is an equivalence when $R$ is bounded thanks to Lemma~\ref{lem:Boundedequivalence}, and it induces the respective embedding or equivalence~\eqref{eq:main-embed-in-recons}. Since the former preserves biproducts, the latter preserves addition. 
\end{proof}

We can often furthermore give the theory $\Quant{S}$ structure resembling real or complex quantum theory.

\begin{corollary} \label{cor:precise-reconstruction}
Let $\catC$ be a dagger theory satisfying the operational principles whose scalars $R$ have square roots and are bounded. Then $\catC$ is equivalent to 
\[
\Quant{\dring{R}} \quad \text{ or } \quad \Quant{\dring{R}[i]}
\]
\end{corollary}
\begin{proof}
Theorem~\ref{thm:main-reconstruction} and Lemma~\ref{lem:square roots}.
\end{proof}

Having reached this general categorical result, let us now consider the typical physical setting in which scalars correspond to (unnormalised) probabilities. 
\begin{definition} 
A dagger theory with addition $\catC$ is \deff{probabilistic} \index{dagger theory!probabilistic} when it comes with an isomorphism of semi-rings $\catC(I,I) \simeq \Rpos$.
\end{definition}
Note that this is a weaker definition than typical in the literature (such as~\cite{Barrett2007InfoGPTs,PhysRevA.84.012311InfoDerivQT}) since we have not made any assumptions relating to tomography, finite-dimensionality or topological closure. One may in fact identify such theories intrinsically, as in the following observation for which we thank John van de Wetering. 

\begin{lemma} \label{lem:prob}
Let $\catC$ be a dagger theory satisfying the operational principles. Then $\catC$ is probabilistic iff $R=\catC(I,I)$ has square roots, no infinitesimals, and that every bounded increasing sequence has a supremum.
\end{lemma}
\begin{proof}
Clearly $\Rpos$ satisfies these properties. Conversely if they hold then, by Lemmas~\ref{lem:square roots} and~\ref{lem:its_a_field}, $\dring{R}$ is a totally ordered Archimedean field~\cite{hall2011completeness} with $R$ as its positive elements. Let $(x_n)^{\infty}_{n=1}$ be any bounded monotonic sequence in $\dring{R}$. Then for some $r \in R$ and $t=\pm 1$, the bounded sequence $(t x_n + r)^\infty_{n=1}$ is increasing and belongs to $R$, and so converges there. Hence $(x_n)^{\infty}_{n=1}$ also converges in $D(R)$, making the latter monotone complete. But then by~\cite[Theorem 3.11]{hall2011completeness} there is an isomorphism $D(R) \simeq \mathbb{R}$ and hence $R \simeq \Rpos$. 
\end{proof}

Now immediately our earlier reconstruction yields one for probabilistic theories. 

\begin{corollary} \label{cor:prob-reconstruction-1}
Any dagger theory which satisfies the operational principles and is probabilistic is equivalent to $\Quant{\mathbb{R}}$ or $\Quant{\mathbb{C}}$.
\end{corollary}
\begin{proof}
By Corollary~\ref{cor:precise-reconstruction}, since $\dring{\Rpos} \simeq \mathbb{R}$ and $\mathbb{R}[i] \simeq \mathbb{C}$.
\end{proof}

To distinguish between real and complex quantum theory one should add an extra principle. 
An example which is known to be satisfied by $\Quant{\mathbb{C}}$ but not $\Quant{\mathbb{R}}$ is \emps{local tomography}, which asserts that any pair of bipartite states may be separated by product effects~\cite{Hardy2012Holism}:
\[
\left(
\scalebox{0.8}{\input{./figures/local-tomographylabels.tikz}}
\quad
{\forall a, b}
\right)
\quad
\implies
\quad
\scalebox{0.8}{\input{./figures/local-tomography2.tikz}}
\]
Note that for a compact theory this is simply equivalent to well-pointedness. Alternatively we may identify complex quantum theory without any tomography assumptions by postulating that in $\catC_\pure$ every phase of a phased biproduct has a square root. It would also be desirable to find a more generic categorical property separating these theories.

Recovering $\Quant{\mathbb{R}}$ is a pleasing consequence of our tomography-free approach, with most reconstructions ruling it out from the outset by assuming local tomography; an example of a probabilistic reconstruction which does recover both theories is~\cite{hohn2017quantum}.

\section{Further Reconstructions}

The operational principles were chosen to be as broad as possible while allowing for our main result to hold. The results of this thesis allow us to now also deduce some alternative sets of axioms for reconstructions.

\subsection{Using coarse-graining}

We saw that the operational principles provide a `coarse-graining' addition operation $f + g$ on morphisms. In fact this is surprisingly well-behaved.

\begin{proposition} \label{cor:cancellative}
In any dagger theory $\catC$ satisfying the operational principles, the following hold for all morphisms $f, g, h$: 
\begin{itemize}
\item
$f + g = f + h \implies g = h$;
\item
$f \otimes g = f \otimes h \implies f = 0 \text{ or } g = h$.
\end{itemize}
\end{proposition}
\begin{proof}
We have $\catC \simeq \CPM(\catA)$ for a quantum category $\catA$. But the definition of addition in $\CPM(\catA)$ is simply addition in $\catA$. Since $\catA$ has negatives $-f$ for all morphisms $f$ it satisfies both properties. 
\end{proof}

Under even milder assumptions we obtain another property of coarse-graining.

\begin{lemma} \label{lem:OpLemma-Redux}
In any non-trivial compact dagger theory with addition and dagger kernels satisfying strong purification and pure exclusion, all morphisms $f, g$ satisfy
\[
f + g = 0 \implies f = g = 0
\] 
\end{lemma}
\begin{proof}
By Proposition~\ref{prop:OpTheoryProperties} all kernels are pure isometries and causal, and any non-trivial object $C$ has an orthonormal pair of pure states $\ket{0}, \ket{1}$. Now suppose that $f, g \colon A \to B$ have $f + g = 0$. Then
\[
\scalebox{0.8}{\input{./figures/pos-arg.tikz}}
\quad 
\text{ has }
\quad 
\scalebox{0.8}{\input{./figures/pos-arg2.tikz}}
\ 0 
\]
and so $h = 0$ by the same proposition. But then applying $\ket{0}$ we obtain $f = 0$, and similarly $g = 0$ also.
\end{proof}

If we instead take the operation $+$ as primitive, as is typical in the study of operational theories (see Chapters~\ref{chap:OpCategories} and~\ref{chap:totalisation}), several of our principles follow almost automatically. Recall that here the physically meaningful morphisms $f$ are those which are sub-causal, with $\discard{} \circ f + e = \discard{}$ for some effect $e$.  

\begin{lemma} \label{lem:caus-complemented}
Let $\catC$ be a dagger theory with addition satisfying
\begin{align}
\discard{} + e = \discard{} &\implies e = 0 \label{eq:effect-zero-law}
\\ 
d + e = 0 &\implies d = e = 0
\end{align}
for all effects $d, e$. 
Then all kernels satisfy~\eqref{eq:causally-complemented} iff all kernels and cokernels are sub-causal. Hence in this case they are causally complemented.
\end{lemma}

\begin{proof}
The equation~\eqref{eq:causally-complemented} makes all cokernels sub-causal, and composing with any kernel $k$ shows that it is causal. Conversely let $k \colon K \to A$ be a kernel with 
\[
\scalebox{0.8}{\input{./figures/kersc.tikz}}  \qquad \qquad \scalebox{0.8}{\input{./figures/kersc2.tikz}}
\]
for some effects $a, b$. Then since $k$ is an isometry we obtain $\discard{K} = \discard{K} + b \circ k + a$ and so $b \circ k = 0 = 0$. Hence all kernels are causal. It follows that $c = k^{\bot \dagger}$ has
\begin{align*}
\discard{K^\bot} \circ c
&= \discard{A} \circ c^\dagger \circ c  & (\text{c causal}) \\ 
&= b \circ c^\dagger \circ c  & (c \circ k = 0) \\  
&= b  &  (b \circ k = 0)
\end{align*}
as required. The final statement is from Proposition~\ref{prop:coarse-graining}.
\end{proof}

Next, pure exclusion can also be deduced easily.

\begin{lemma} \label{lem:orthogonality}
Let $\catC$ be a compact dagger theory with strong purification, dagger kernels, normalisation, and addition satisfying 
\begin{equation} \label{eq:zero-scalar-law}
\scalebox{0.8}{\input{./figures/scalar-weird-law.tikz}} \ \ 0
\end{equation}
for all scalars $r$. Suppose also that $\psi^{\dagger}$ is sub-causal for every causal pure state $\psi$.
Then $\catC$ satisfies pure exclusion. 
\end{lemma}
\begin{proof}
By Lemma~\ref{lem:Pure-Excl-Is-Easy} it remains to show that every causal pure state $\psi \colon I \to A$ is a kernel. It suffices to assume $\coker(\psi) = 0$ and show that $\psi$ is unitary. Then there is an effect $e$ for which 
\[
\scalebox{0.8}{\input{./figures/pureexcl1.tikz}} \quad \text{ and so } \quad \scalebox{0.8}{\input{./figures/pureexcl12.tikz}}
\] 
since $\psi$ is a causal isometry. Hence $e \circ \psi = 0$ and so $e = 0$ giving $\psi^{\dagger} = \discard{A}$. But then by essential uniqueness $\psi \circ \psi^{\dagger}$ is unitary, making $\psi$ unitary also.
\end{proof}

We can now present our principles in a new equivalent manner in terms of coarse-graining.

\begin{theorem} \label{thm:alternate-axioms}
A non-trivial dagger theory $\catC$ satisfies the operational principles iff it has the properties of Lemma~\ref{lem:orthogonality} and that every dagger cokernel is sub-causal.
\end{theorem}

\begin{proof}
If $\catC$ satisfies the principles then~\eqref{eq:zero-scalar-law} holds by Corollary~\ref{cor:cancellative}, normalisation by Lemma~\ref{lem:Pure-Excl-Is-Easy}, and if $k$ is a kernel then $k^{\dagger}$ is sub-causal by~\eqref{eq:causally-complemented}. In particular if $\psi$ is a causal pure state then $\psi^{\dagger}$ is sub-causal.

Conversely, if these hold then by Lemma~\ref{lem:orthogonality} pure exclusion holds and by Proposition~\ref{prop:OpTheoryProperties} all kernels are causal. Hence by Lemmas~\ref{lem:OpLemma-Redux} and~\ref{lem:caus-complemented} it remains to check~\eqref{eq:effect-zero-law} for all effects $e$. But if $\discard{} + e = \discard{}$ then $e \circ \rho = 0$ for any causal state $\rho$. In particular, since kernels are causal, for any causal pure state $\psi$ of $\CoIm(e)$ we have $e \circ \coim(e)^{\dagger} \circ \psi = 0$ and so $\coim(e)^\dagger \circ \psi = 0$ giving $\psi = 0$. Hence $\CoIm(e) = 0$ and so $e = 0$. 
\end{proof}

This result lets us deduce a simpler reconstruction than Corollary~\ref{cor:prob-reconstruction-1} for probabilistic theories with coarse-graining.

\begin{corollary} \label{cor:simple-prob-reconstr}
Let $\catC$ be a compact dagger theory with addition and which is probabilistic. Suppose that $\catC$ has strong purification, dagger kernels, and that $f^{\dagger}$ is sub-causal for every dagger kernel or causal pure state $f$. Then $\catC$ is equivalent to $\Quant{\mathbb{R}}$ or $\Quant{\mathbb{C}}$.
\end{corollary}

\subsection{Alternative notions of purity}

Rather than using $\otimes$-purity, one may wish to instead consider the more typical notion of purification in terms of morphisms which are $+$-pure, as for example used when the principle was introduced in~\cite{chiribella2010purification}.  

In fact by Proposition~\ref{prop:pure-lem} and Lemma~\ref{lem:caustensclosedexamples} for any dagger theory (with addition) $\catC$ it is equivalent to consider purifications satisfying the properties of Principle~\ref{princ:strongpurif} with respect to morphisms which are $\otimes$-pure, $+$-pure, or meet any of the other notions of purity we met in Chapter~\ref{chap:principles}. Moreover, in any theory satisfying the operational principles, all of these in fact coincide. 

\begin{lemma} \label{lem:OpPrinciplesPurityCollapse}
In any dagger theory $\catC$ satisfying the operational principles the classes of 
$\otimes$-pure, $+$-pure, copure and kernel-pure morphisms all coincide.
\end{lemma}
\begin{proof}
$\otimes$-purity coincides with copurity by Lemma~\ref{lem:copure}, with kernel-purity by Lemma~\ref{lem:PureExclusion} and Proposition~\ref{prop:kerpure}~\ref{enum:kerpure}, and by Proposition~\ref{prop:wholetoatomic} any $\otimes$-pure morphism is $+$-pure. Finally we show that any $+$-pure morphism is kernel-pure. Bending wires it suffices to consider the case of a $+$-pure effect $e \colon A \to I$. 

Let $p \colon A \to B$ be a purification of $e$. If $\Img(p)$ is a zero object then $p = 0$ and so $e = 0$. Otherwise it has some causal pure state $\phi$. Then $\psi := \img(p) \circ \phi$ is a causal pure state with $\psi^{\dagger} \circ p \neq 0$. 
Now since $\psi$ is a kernel, by~\eqref{eq:causally-complemented} we in particular have $\psi^{\dagger} + d = \discard{B}$ for some effect $d$. But then 
\[
\scalebox{0.8}{\input{./figures/atomic-arg-effect.tikz}}
\quad
\text{ and so  }
\quad
\scalebox{0.8}{\input{./figures/atomic-arg2-effect.tikz}}
\]
for some scalar $r$, since $e$ is $+$-pure. Then since $p$ and $\psi$ is pure, $r \cdot e$ is also pure and hence kernel-pure. But by zero-cancellativity we have $\CoIm(r \cdot e) = \CoIm(e)$ and so if $r \cdot e$ is kernel-pure then so is $e$. 
\end{proof}

\subsection{Principles on kernels}

One of the less clearly operationally motivated of our principles is the CP axiom~\eqref{eq:CP axiom}, which it would be desirable to replace with more physical assumptions. Indeed we explored this earlier in Section~\ref{sec:CP}, where we saw how to derive~\eqref{eq:CP axiom} instead from the  internal isomorphism property, along with homogeneity of kernels. 

We can use these to give alternative reconstruction principles, making no reference to the CP axiom or even purification. Several of these are in a `quantum logic style', referring to kernels and their associated orthomodular lattices.

\begin{definition}[\textbf{Kernel Principles}]
We say that a compact dagger theory $\catC$ satisfies the \deff{kernel principles} \index{kernel principles} when it satisfies the following.
\begin{enumerate}[label=\arabic*)]
\item \label{enum:caus-comp}
$\catC$ has dagger kernels which are causal, homogeneous and causally complemented.
\item Every non-zero object has a state which is a dagger kernel.
\item 
The internal isomorphism property holds.
\item 
Conditioning holds.
\item \label{enum:ker-dil}
Every causal morphism has a dilation which is a dagger kernel.
\item \label{enum:everypure}
Every identity morphism is $\otimes$-pure. 
\item \label{enum:cov-law}
Each lattice $\DKer(A)$ satisfies the covering law.
\end{enumerate}

Alternatively, we'll see that one may replace the final condition by the following:
\begin{enumerate}[label=\arabic*')]
\setcounter{enumi}{6}
\item \label{enum:pure-under-comp}
Whenever $f, g$ are $\otimes$-pure morphisms then so is $g \circ f$.
\end{enumerate}
\end{definition}

\begin{example}
$\Quant{S}$ satisfies the kernel principles, for any phased field $S$.
\end{example}
\begin{proof}
By Example~\ref{example:MatofPhasedField} and the next result it suffices to verify the internal isomorphism property. Since all kernels are pure, it is then sufficient to show that every morphism $f$ in $\Mat_S$ with $\coker(f) = 0$ and $\ker(f) = 0$ is an isomorphism, which follows from standard linear algebra.
\end{proof}

\begin{theorem} \label{prop:Compare-Isom-Principles}
For any compact dagger theory $\catC$, the following are equivalent:
\begin{enumerate}[label=\arabic*., ref=\arabic*]
\item \label{enum:ker-princ}
The kernel principles;
\item \label{enum:alt-princ}
The kernel principles, replacing condition \ref{enum:cov-law} with \ref{enum:pure-under-comp};
\item \label{enum:opandisomprop}
The operational principles along with the internal isomorphism property. 
\end{enumerate}
Moreover, when these hold and $\catC(I,I)$ is bounded we have $\catC \simeq \Quant{S}$ for some phased field $S$. 
\end{theorem}

\begin{proof}
\ref{enum:ker-princ} $\iff$ \ref{enum:alt-princ}: Suppose that $\catC$ satisfies principles \ref{enum:caus-comp} -- \ref{enum:ker-dil}. Then every non-zero scalar $r$ is invertible, since by zero-cancellativity (Lemma~\ref{lem:dagzero}~\ref{enum:zero-mult}) we have $\coker(r) = \ker(r) = 0$, and so by the internal isomorphism property $r$ is an isomorphism. Hence in particular $\catC$ has normalisation.

Then by Proposition~\ref{prop:kerpure} any kernel-pure morphism is $\otimes$-pure. Conversely, we claim that any $\otimes$-pure morphism is kernel-pure. Thanks to normalisation, it suffices to show that any causal $\otimes$-pure state $\rho$ is kernel-pure. But since $\rho$ has a kernel dilation this follows from Proposition~\ref{prop:kerpure}~\ref{enum:ker-pure-2-out-3}. Hence by Theorem~\ref{thm:closedUnderComp-Equiv} conditions \ref{enum:pure-under-comp} and \ref{enum:cov-law} are equivalent.

\ref{enum:alt-princ} $\implies$ \ref{enum:opandisomprop}:
We've just seen that a morphism is $\otimes$-pure iff it is kernel-pure, in particular making pure exclusion hold and all pure morphisms closed under $\otimes$. Since every identity is $\otimes$-pure, so is every dagger (co)kernel.

 Finally by Theorem~\ref{thm:DeduceCPM} to deduce strong purification it suffices to show that every $\otimes$-pure causal morphism $f \colon A \to B$ is an isometry. But since $f$ has a kernel dilation, for some causal state $\psi$ we have that $f \otimes \psi$ is a dagger kernel. Since every dagger kernel is kernel-pure, then so is $\psi$ by Proposition~\ref{prop:kerpure}~\ref{enum:ker-pure-2-out-3}. But then $\psi$ is a dagger kernel and so an isometry, and then it follows that $f$ is also.

\ref{enum:opandisomprop} $\implies$ \ref{enum:alt-princ}:
If $\catC$ satisfies the operational principles then since all identity morphisms and kernels are $\otimes$-pure, homogeneity of kernels is a special case of essential uniqueness of purification. Moreover as in Corollary~\ref{cor:rev-dilation} every causal morphism has a reversible dilation which is in particular a kernel since every causal $\otimes$-pure state is.

The final statement follows from Theorem~\ref{thm:main-reconstruction} and Lemma~\ref{lem:no_nontriv_subob} along with the above fact that $\catC(I,I)$ is a semi-field.
\end{proof}

Using the relations between the various principles we gathered in Chapter~\ref{chap:principles}, it will be possible to put together numerous other reconstructions in a similar vein.

%% file: figures/cokerzeroi.tikz
\begin{tikzpicture}
	\begin{pgfonlayer}{nodelayer}
		\node [style=dagpointadj] (0) at (0, 1) {$e$};
		\node [style=dagpoint] (1) at (0, -1) {$\psi$};
		\node [style=none] (2) at (1.5, -0) {$=$};
	\end{pgfonlayer}
	\begin{pgfonlayer}{edgelayer}
		\draw [style=none] (0) to (1);
	\end{pgfonlayer}
\end{tikzpicture}

%% file: figures/cokerzeroii.tikz
\begin{tikzpicture}
	\begin{pgfonlayer}{nodelayer}
		\node [style=dagpointadj] (0) at (5.25, 1) {$e$};
		\node [style=none] (1) at (6.75, -0) {$=$};
		\node [style=none] (2) at (5.25, -0.5) {};
	\end{pgfonlayer}
	\begin{pgfonlayer}{edgelayer}
		\draw [style=none] (2.center) to (0);
	\end{pgfonlayer}
\end{tikzpicture}

%% file: outlook.tex
\chapter*{Outlook}  \addcontentsline{toc}{chapter}{Outlook} \markboth{Outlook}{Outlook}

The aim of this project was to develop a categorical approach to the study of operational theories of physics. In particular we wished to show that many of the features typically associated with general probabilistic theories may in fact be treated and understood in this elementary categorical framework, without requiring any of the usual technical assumptions relating to ordered vector spaces. 

 We saw how the framework of operational theories can be captured by basic categorical properties (Chapter~\ref{chap:OpCategories}) and related these with the usual formalism of categorical quantum mechanics (Chapter~\ref{chap:totalisation}). Numerous principles considered in the study of probabilistic theories were found to treatable categorically and typically even in the basic language of diagrams (Chapter~\ref{chap:principles}), along with a more novel account of superpositions (Chapter~\ref{chap:superpositions}).

 Most convincingly, we were able to use these to provide a reconstruction of (finite-dimensional) quantum theory itself (Chapter~\ref{chap:recons}), with principles and proof both given in the basic setting of dagger compact categories with discarding. To our knowledge this is the first quantum reconstruction which does not rely on any vector space assumptions from the outset. Other comparable results are due to Soler~\cite{soler1995characterization}, who reconstructs infinite-dimensional Hilbert space from its lattice of subspaces but does not include any compositional or measurement-based features, and Heunen who axiomatizes $\Hilb$~\cite{heunen2009embedding} but in terms of its own features rather than those of more `operational' categories such as $\HilbP$ or $\Quant{}$.

 Our results suggest many new potential avenues of research in the categorical study of operational theories; let us close by discussing a few.

\subsection*{Extending the notion of operational theory}
As mentioned there, it would be interesting to extend our approach in Chapter~\ref{chap:OpCategories} beyond tests simply having finitely many outcomes, to allow for tests of various types, such as real-valued ones with infinitely many outcomes.
At risk of a high level of abstraction, we suggested that this could be possible by viewing tests as arrows in some form of generalised multicategory. This should at least allow us to unify our two approaches to tests based on varying $(f_\x \colon A \to B_\x)_{\x \in X}$ or non-varying $(f_\x \colon A \to B)_{\x \in X}$ output systems.

\subsection*{Categorifying probabilistic results}
Combining the results of Chapters~\ref{chap:OpCategories} and~\ref{chap:principles} it should be possible to adapt many existing results and proofs about probabilistic theories into simple categorical ones.
Though we did not go into this in detail, it is routine to translate most arguments from e.g.~\cite{chiribella2010purification,PhysRevA.84.012311InfoDerivQT} into the setting of Chapter~\ref{chap:OpCategories} or of more general categories with discarding. There has been a history of such simplified categorical proofs in the literature, such as the categorical form~\cite{CoeckeNoBroadcasting} of the `No-Broadcasting' theorem~\cite{Barnum2007NoBroadcast,barnum1996noncommuting}.

\subsection*{Superpositions in operational theories}
 We introduced phased coproducts mainly to allow us to define the category $\plusI{\catC}$ for use in our reconstructions in Chapter~\ref{chap:recons}. Their applications to the study of superpositions in physical theories are promising and remain to be explored fully. 

From a mathematical perspective, we did not yet find many well-motivated examples of non-monoidal categories with phased coproducts; if these can be found then the one-to-one correspondence between phased coproducts and coproducts from Corollary~\ref{cor:phcoprodcorrespon} should be extended to this setting. It would also be interesting to compare them with other weak notions of limit, such as those in 2-categories~\cite{lack20102}. 

\subsection*{Reconstruction principles}
The principles used in our reconstruction were chosen to be as weak as possible while allowing for the result. It should be possible to find a smaller and more natural, though potentially stronger, set of assumptions as we began exploring at the end of Chapter~\ref{chap:recons}. 

\subsection*{Including classical systems}

The notion of purification we have considered applies only to categories like $\Quant{}$ which are $\otimes$-pure in our sense, with all identity morphisms being $\otimes$-pure, ruling out the inclusion of classical systems and biproducts. Eventually we should extend our approach to include such systems, and so potentially reach a reconstruction of (some generalisation of) $\FCStar$, recovering our current reconstruction by restricting to such `pure' objects.

Notions of purification which hold classically can be found in our concept of minimal dilations, the definition of purity due to Cunningham and Heunen~\cite{cunningham2017purity}, and in Selby, Scandolo and Coecke's reconstruction~\cite{selby2018reconstructing}. 

\subsection*{Purifying objects}

Related to the previous goal, it would be interesting to extend purification to objects. Given a (finite-dimensional) C*-algebra $\bigoplus^n_{i=1} B(\hilbH_i)$ this would return its smallest extension to a purely quantum algebra $B(\bigoplus^n_{i=1} \hilbH_i)$. Rather than using phased coproducts we could then simply describe $B(\hilbH \oplus \hilbK)$ by purifying the algebra $B(\hilbH) \oplus B(\hilbK)$.

\subsection*{Removing daggers}

The most significant open area left from our reconstruction lies in its extensive use of the dagger in $\Quant{}$. Though the dagger has a direct operational meaning on pure states, it lacks one for more general mixed states and processes, and this is reflected in its failure to exist in infinite dimensions. For our reconstruction to be as truly operational as those of e.g.~\cite{PhysRevA.84.012311InfoDerivQT,Hardy2011a} it should thus not require the dagger from the outset, instead being given simply in the language of monoidal categories with discarding. This could be achieved in at least two ways.
\begin{enumerate}
\item 
Avoiding any use of the dagger in our theory $\catC$ itself, and merely establishing its existence in the extended ring of scalars $S$. This would still allow us to define $\Quant{S}$ and its embedding into $\catC$.
\item 
Deriving the presence of the dagger from other more operational principles. Equipped with a suitable characterisation of the dagger on pure states, it should be possible to use purification and compactness to extend the dagger to all morphisms. Alternatively, we could aim to find conditions on a sub-causal category $\catC$ which ensure that its totalisation $\To(\catC)$ has a dagger, as we found for compactness in Theorem~\ref{thm:CompactInterp}.
\end{enumerate}

\subsection*{Towards infinite dimensions}

As well as the dagger, it should in fact be possible to derive compactness itself from principles such as purification. An ideal reconstruction would apply simply to monoidal categories which come with discarding and also a distinguished `maximally mixed' state on each object, from which the cup states arise by (minimal) purification:
\[
\scalebox{0.8}{\input{./figures/derive-cup.tikz}}
\]
 Avoiding compactness from the outset in this way should also allow for a reconstruction involving only the physically meaningful sub-causal processes, applicable for example to effectuses.
 
Finally, no longer assuming the presence of such maximally mixed states should yield axioms which hold even in infinite-dimensional settings such as $\CStarop$ and $\vNop$. Drawing on results such as our own, developments from effectus theory~\cite{EffectusIntro}, and Heunen's axiomatization of $\Hilb$~\cite{heunen2009embedding}, one day we can hope to arrive at such a reconstruction in infinite dimensions. This would be a major success for the categorical framework, being the first such result of this kind even under the usual assumptions of general probabilistic theories. 

Ultimately, such totally new results will be necessary to demonstrate that categorical methods have a role to play in the physics of tomorrow.

%% file: categories.tex
\chapter*{Index of Categories}\addcontentsline{toc}{chapter}{Index of Categories}
  \markboth{Index of Categories}{Index of Categories}

\begin{tabular}{l l r}
\emph{Notation} & \emph{Description} & \emph{Page} \\ 

$\KlDT$ & Sets and $\Rpos$-distributions & \pageref{cat:KlDT}\\

$\KlSD$ & Sets and sub-distributions & \pageref{cat:KlsD}\\


$\CStarop$ & C*-algebras and completely positive maps (opposite direction) & \pageref{cat:cstarop} \\ 

$\CStarSUop$ & Subcategory of sub-unital morphisms in $\CStarop$ & \pageref{cat:cstarsu}\\ 

$\CStarUop$ & Subcategory of unital morphisms $\CStarop$ & \pageref{cat:cstaru}\\ 

$\DCM$ & Commutative monoids with specified downset & \pageref{cat:DCM} \\ 

$\Class$ & $\Rpos$-valued matrices (finite classical physics) & \pageref{cat:class} \\ 

$\FCStar$ & Finite-dimensional C*-algebras and completely positive maps & \pageref{cat:FCStar} \\ 

$\FHilb$ & Finite-dimensional Hilbert spaces and linear maps & \pageref{cat:fhilb} \\ 

$\FHilbP$ & $\FHilb$ modulo global phases & \pageref{cat:fhilbP} \\ 

$\FVeck$ & Finite-dimensional vector spaces over $\fieldk$ & \pageref{cat:FVeck}\\


$\Hilb$ & Hilbert spaces and continuous linear maps & \pageref{cat:hilb}\\ 

$\HilbP$ & $\Hilb$ modulo global phases & \pageref{cat:hilbP}\\

$\KlD$ & Kleisli category of distribution monad & \pageref{cat:KlD}\\

$\KlMn$ & Kleisli category of multiset monad & \pageref{cat:KlMn}\\

$\MatS$ & Matrices over semi-ring $S$ & \pageref{cat:Mats}\\

$\Mat_{S^{\leq 1}}$ & Matrices with values in $S^{\leq 1}$ & \pageref{cat:MatSless}\\ 

$\SpekOrMSpek$ & Spekkens toy model (resp.~including mixtures)& \pageref{cat:mspek}\\

$\OpCats$ & Operational categories & \pageref{cat:OpCats}\\


$\OTCatOrRep$ & (Representable) operational theories & \pageref{cat:OTCat}\\


$\PCM$ & Partial commutative monoids & \pageref{cat:PCM} \\ 

$\PCMD$ & Sub-causal categories & \pageref{cat:PCMD} \\  

$\PFun$ & Sets and partial functions & \pageref{cat:pfun} \\ 

$\VecProj$ & $\fieldk$-Vector spaces modulo global non-zero scalars & \pageref{cat:vecproj}\\

$\Quant{}$ & Finite-dimensional Hilbert spaces and completely positive maps & \pageref{cat:quant}\\

$\QuantSU$ & F.d.~Hilbert spaces and trace non-increasing c.p.~maps & \pageref{cat:quantsu}\\

$\Quant{S}$ & Quantum theory over involutive semi-ring $S$ & 
\pageref{cat:QuantS}\\


$\Rel$ & Sets and relations & \pageref{cat:rel}\\

$\Rel(\catC)$ & Relations in regular category $\catC$ & \pageref{cat:relC}\\

$\Rep(G)$ & Unitary representations of $G$ and intertwiners & \pageref{cat:repg}\\

$\Set$ & Sets and functions & \pageref{cat:set}\\


$\TestCats$ & Test categories & \pageref{cat:testcats}\\

$\CMonD$ & Categories with addition and discarding & \pageref{cat:CMonD} \\ 

$\Veck$ & Vector spaces over $\fieldk$ and linear maps & \pageref{cat:veck}\\

$\VecP$ & $\VecC$ modulo global phases & \pageref{cat:vecP}\\

$\vNop$ & von Neumann algebras, normal c.p.~maps (opposite direction) & \pageref{cat:vnop} \\ 

$\vNSUop$ & Subcategory of $\vNop$ of sub-unital morphisms & \pageref{cat:vnsu}\\

$\vNUop$ & Subcategory of $\vNop$ of unital morphisms & \pageref{cat:vnU} \\ 

\end{tabular}

%% file: notation.tex
\twocolumn

\chapter*{Index of Notation}  \addcontentsline{toc}{chapter}{Index of Notation}
\markboth{Index of Notation}{Index of Notation}

\noindent
$\catC, \catD \dots$, 
categories, \pageref{not:category} 

\noindent
$A, B, C \dots$, 
objects, \pageref{not:ob} 

\noindent
$f \colon A \to B$,
morphism, \pageref{not:morphism}  

\noindent
$\catC \simeq \catD$
equivalence of categories, \pageref{not:equivalence}

\noindent
$\otimes$, 
monoidal tensor, \pageref{not:moncat} 

\noindent
$I$, monoidal unit, \pageref{not:unitobject}  

\noindent
$\alpha, \rho, \lambda, \sigma$, coherence morphisms, \pageref{not:coherenceiso}, \pageref{not:swap} 

\noindent
$\discard{}$,  discarding effect, \pageref{not:discarding} 

\noindent
$\catC_\causal$,  causal subcategory ,\pageref{not:causalsubcat}

\noindent
$\ts{f_\x}_{\x \in \X}$, partial test, \pageref{not:test}, \pageref{not:test2}

\noindent
$0 \colon A \to B$,  zero morphism, \pageref{not:zeromorphism} 

\noindent
$\Events_\Theta$, category of events, \pageref{not:catevents}

\noindent
$0$, initial or zero object, \pageref{not:initialobject} 

\noindent
$1$,  terminal object, \pageref{not:terminal}  

\noindent
$! \colon 0 \to A$, initial object morphism, \pageref{not:initialobject}

\noindent
$! \colon A \to 1$, terminal object morphism, \pageref{not:terminal} 

\noindent
$f \ovee g$,  coarse-graining, \pageref{not:coarse-grain}, \pageref{not:coarse-grain2} 

\noindent
$f \ovee g$, PCM addition, \pageref{not:PCM}, \pageref{not:PCMenriched}   

\noindent
$X \cdot A$, $n \cdot A$,  copower, \pageref{not:copower}, \pageref{not:copowerOfI} 

\noindent
$A + B$,  coproduct, \pageref{not:coproduct} 

\noindent
$[f,g]$,  cotuple, \pageref{not:cotuple} 

\noindent
$f + g$,  diagonal morphism, \pageref{not:diagonal}  

\noindent
$f + g$,  addition of morphisms, \pageref{not:addition}  

\noindent
$\coproj_i$,  coprojection, \pageref{not:coprojection}, \pageref{not:phcoproj} 

\noindent
$\triangleright_i$,  projection from coproduct, \pageref{not:coprodproj} 

\noindent
$\PTest(\Theta)$, category of partial tests, \pageref{not:PtestTheta} 

\noindent
$\Test(\Theta)$,  category of tests, \pageref{not:testcat} 

\noindent
$\Theta^\otcplus$,  representable completion, \pageref{not:repcompletion} 

\noindent
$e^\bot$,  complement of effect $e$, \pageref{not:complements} 

\noindent
$k^\bot$,  complement of dagger kernel, \pageref{not:orthocomp} 

\noindent
$f \colon A \partialto B$,  morphism in $\Partial(\catB)$, \pageref{not:partialarrow}

\noindent
$A \biprod B$,  biproduct, \pageref{not:biproduct} 

\noindent
$\pproj_i$, projection, \pageref{not:projection}, \pageref{not:phbiprod} 

\noindent
$\catC^\oplus$,  biproduct completion, \pageref{not:biproductcompletion}  

\noindent
$\catC_\subcausal$,  sub-causal subcategory, \pageref{not:subcausalsubcat} 

\noindent
$\bot$,  summable elements of PCM, \pageref{not:compatPCM}  

\noindent
$\To(\catC)$,  totalisation of category, \pageref{not:totalisationcat} 

\noindent
$A^*$,  dual object, \pageref{not:dualobject} 

\noindent
$\varepsilon, \eta$,  dual pair, \pageref{not:cupcap} 

\noindent
$\dagger$,  dagger, \pageref{not:daggercat}

\noindent
$\dagger$, involution of monoid, semi-ring, \pageref{not:involsring} 

\noindent
$\discardflip{}$,  dagger of discarding, \pageref{not:upsidediscard} 

\noindent
$\CPM(\catC)$,  CPM construction, \pageref{not:CPM}

\noindent
$\Dbl{\catC}$,  `doubled' CPM subcategory, \pageref{not:doublingfunctor}, \pageref{not:doublesubcat} 

\noindent
$\leq$,  order in ordered theory, \pageref{not:ordertheory} 

\noindent
$\mdil{f}$,  minimal dilation, \pageref{not:mindil} 

\noindent
$\ker(f)$,  kernel, \pageref{not:kernel} 

\noindent
$\coker(f)$,  cokernel, \pageref{not:cokernel} 

\noindent
$\img(f)$,  image, \pageref{not:image} 

\noindent
$\coim(f)$,  coimage, \pageref{not:coimage} 

\noindent
$\DKer$,  dagger kernels, \pageref{not:dker}

\noindent
$\preceq_F$,  face pre-order, \pageref{not:facepreorder} 

\noindent
$\preceq_K$,  pre-order of kernel inclusion, \pageref{not:ker-preorder} 

\noindent
$\idob_\rho$, $\enc_\rho$, $\dec_\rho$,  ideal compression, \pageref{not:idealcomp} 

\noindent
$\catC_\prepure$, class of morphisms, \pageref{not:prepuresubcat}, 

\noindent
$\catC_\prepure$, environment structure,  \pageref{not:envstruc}

\noindent
$\catC_\pure$,  pure morphisms, \pageref{not:puresubcat} 

\noindent
$A \pcoprod B$,  phased coproduct, \pageref{not:phcoprod} 


\noindent
$A \pprod B$,  phased product, \pageref{not:phproduct} 

\noindent
$\TrI_{A}$,  group of trivial isomorphisms, \pageref{not:trivisom} 

\noindent
$[f]_\tc$,  equivalence class, \pageref{not:equivclass} 

\noindent
$\quot{\cta}{\sim}$,  quotient category, \pageref{not:quotcat} 

\noindent
$\plusI{\ctb}$,  GP construction, \pageref{not:GP} 

\noindent
$\plusIdag{\ctb}$,  dagger GP construction, \pageref{not:GPdag} 

\noindent
$\obb{A}, \obb{B}, \dots$,  objects in $\plusI{\catB}$, \pageref{not:GP}  

\noindent
$\mathbb{P}$,  global phase group, \pageref{not:globalphases} 

\noindent
$\cta_\quotP$,  quotient by global phases, \pageref{not:quotglobphase} 

\noindent
$A \pbiprod B$,  phased biproduct, \pageref{not:phbiprod} 

\noindent
$S^\pos$,  positive elements, \pageref{not:poselements} 

\noindent
$S^\sa$,  self-adjoint elements, \pageref{not:sadjoint} 

\noindent
$\dring{R}$,  difference ring, \pageref{not:differencering} 

\noindent
$R[i]$,  adjoin element $i$ to ring $R$, \pageref{not:adjoin-i}

\onecolumn